\documentclass[letterpaper,11pt]{article}
\pdfoutput=1
\usepackage[margin=1in]{geometry}
\usepackage[utf8]{inputenc}
\usepackage[T1]{fontenc}
\usepackage{microtype}
\usepackage{amsfonts}
\usepackage{amssymb}
\usepackage{amsmath}
\usepackage{amsthm}
\usepackage{comment}
\usepackage{fullpage}
\usepackage{booktabs}
\usepackage{enumitem}
\usepackage[draft]{fixme}
\usepackage{mathtools}
\usepackage{tikz}
\usepackage[font=small,labelfont=bf]{caption}
\usepackage{titlesec}
\usepackage{thmtools,thm-restate}
\usepackage{wrapfig}
\usepackage[ruled,noend,linesnumbered]{algorithm2e}
\usepackage[colorlinks,bookmarksopen=true]{hyperref}
\usepackage[capitalize]{cleveref}
\usepackage{soul}

\usetikzlibrary{patterns}

\newcommand{\Symb}{\mathcal{S}}
\newcommand{\sm}{\setminus}
\newcommand{\Act}{\mathcal{A}}
\newcommand{\sig}{\mathsf{sig}}

\newcommand{\zero}{\mathsf{0}}
\newcommand{\one}{\mathsf{1}}
\newcommand{\Zz}{\mathbb{Z}_{\ge 0}}
\newcommand{\Zp}{\mathbb{Z}_{+}}
\newcommand{\str}{\mathsf{str}}

\newcommand{\ceil}[1]{\lceil #1 \rceil}
\newcommand{\floor}[1]{\lfloor #1 \rfloor}
\newcommand{\tow}[1]{\Sigma_{#1}}
\newcommand{\ptow}[1]{\overline{\Sigma}_{#1}}
\newcommand{\RUNS}{\mathsf{RUNS}}

\newcommand{\Tr}{\mathcal{T}}
\newcommand{\Gr}{\mathcal{G}}
\renewcommand{\root}{\mathsf{root}}
\newcommand{\degree}{\mathsf{deg}}
\newcommand{\parent}{\mathsf{parent}}
\newcommand{\child}{\mathsf{child}}
\newcommand{\symb}{\mathsf{symb}}
\newcommand{\level}{\mathsf{level}}
\newcommand{\indexop}{\mathsf{index}}

\newcommand{\rightop}{\mathsf{right}}

\newcommand{\letter}{\mathsf{letter}}
\newcommand{\length}{\mathsf{len}}
\newcommand{\id}{\mathsf{id}}
\newcommand{\block}{\mathsf{block}}
\newcommand{\preflength}{\mathsf{pLen}}

\newcommand{\X}{\mathcal{X}}
\newcommand{\Y}{\mathcal{Y}}
\newcommand{\Pts}{\mathcal{P}}
\newcommand{\Points}{\mathrm{Points}}
\newcommand{\rot}{\mathsf{rot}}
\newcommand{\Ints}{\mathcal{I}}

\newcommand{\LW}{\mathcal{L}}
\newcommand{\val}{\mathsf{val}}
\newcommand{\labelop}{\mathsf{label}}
\newcommand{\deleteop}{\mathsf{delete}}
\newcommand{\unlabel}{\mathsf{unlabel}}
\newcommand{\swapop}{\mathsf{swap}}
\newcommand{\initop}{\mathsf{initialize}}

\newcommand{\W}{\mathcal{W}}
\newcommand{\makeop}{\mathsf{insert}}
\newcommand{\lcpop}{\mathsf{lcp}}
\newcommand{\perop}{\mathsf{period}}
\newcommand{\canshop}{\mathsf{canShift}}
\newcommand{\ipmop}{\mathsf{ipm}}
\newcommand{\lcsop}{\mathsf{lcs}}
\newcommand{\splitop}{\mathsf{split}}
\newcommand{\concatop}{\mathsf{concat}}
\newcommand{\accessop}{\mathsf{access}}

\numberwithin{equation}{section}
\definecolor{mypink}{RGB}{199, 21 133}
\hypersetup{
  colorlinks,
  urlcolor = blue!70!black,
  linkcolor = mypink,
  citecolor = blue!70!black
}

\fxsetup{mode=multiuser,theme=color, layout=inline}
\FXRegisterAuthor{dk}{adk}{\color{red}\textbf{D.}}
\FXRegisterAuthor{tk}{atk}{\color{blue}\textbf{T.}}

\newtheorem{theorem}{Theorem}[section]
\newtheorem{construction}[theorem]{Construction}
\newtheorem{corollary}[theorem]{Corollary}
\newtheorem{conjecture}[theorem]{Conjecture}
\newtheorem{lemma}[theorem]{Lemma}
\newtheorem{fact}[theorem]{Fact}
\newtheorem{observation}[theorem]{Observation}
\newtheorem{proposition}[theorem]{Proposition}

\newtheorem{assumption}[theorem]{Assumption}
\newtheorem{problem}[theorem]{Problem}
\theoremstyle{definition}
\newtheorem{definition}[theorem]{Definition}
\theoremstyle{remark}
\newtheorem{remark}[theorem]{Remark}

\crefname{fact}{Fact}{Facts}
\crefname{assumption}{Assumption}{Assumptions}

\newcommand{\LCE}{\mathrm{LCE}}
\newcommand{\lcp}{\mathrm{lcp}}
\newcommand{\lcs}{\mathrm{lcs}}
\newcommand{\per}{\mathrm{per}}
\newcommand{\bigO}{\mathcal{O}}

\newcommand{\Oh}{\bigO}
\newcommand{\Ohtilde}{\widetilde{\Oh}}
\DeclareMathOperator{\polylog}{polylog}

\newcommand{\revstr}[1]{\overline{#1}}
\newcommand{\Lroot}{\mathrm{root}}
\newcommand{\Lrootstrgen}[2]{\mathrm{root}_{#1}(#2)}
\newcommand{\Lrootgen}[4]{\mathrm{root}_{#1}(#2,#3,#4)}

\newcommand{\Lhead}{\mathrm{head}}
\newcommand{\Lheadgen}[4]{\mathrm{head}_{#1}(#2,#3,#4)}

\newcommand{\Ltail}{\mathrm{tail}}
\newcommand{\Ltailgen}[4]{\mathrm{tail}_{#1}(#2,#3,#4)}
\newcommand{\Lexp}{\mathrm{exp}}
\newcommand{\Lexpgen}[4]{\mathrm{exp}_{#1}(#2,#3,#4)}
\newcommand{\Lexpcut}[2]{\mathrm{exp}^{\rm cut}(#1,#2)}
\newcommand{\Lexpcutgen}[5]{\mathrm{exp}^{\rm cut}_{#1}(#2,#3,#4,#5)}
\newcommand{\type}{\mathrm{type}}
\newcommand{\typegen}[3]{\mathrm{type}(#1,#2,#3)}

\newcommand{\Posalll}{\mathrm{Pos}_{\ell}}
\newcommand{\Posallg}{\mathrm{Pos}'_{\ell}}
\newcommand{\Poslow}{\mathrm{Pos}_{\ell}^{\rm low}}
\newcommand{\Posmid}{\mathrm{Pos}_{\ell}^{\rm mid}}
\newcommand{\Poshigh}{\mathrm{Pos}_{\ell}^{\rm high}}
\newcommand{\deltal}{\delta}
\newcommand{\deltag}{\delta'}
\newcommand{\deltalow}{\delta_{\ell}^{\rm low}}
\newcommand{\deltamid}{\delta_{\ell}^{\rm mid}}
\newcommand{\deltahigh}{\delta_{\ell}^{\rm high}}
\newcommand{\Occ}{\mathrm{Occ}}
\newcommand{\Occm}{\mathrm{Occ}^{-}}
\newcommand{\Occp}{\mathrm{Occ}^{+}}
\newcommand{\Occeq}{\mathrm{Occ}^{\rm eq}}
\newcommand{\Occgt}{\mathrm{Occ}^{\rm gt}}
\newcommand{\Occeqm}{\mathrm{Occ}^{{\rm eq}-}}
\newcommand{\Occgtm}{\mathrm{Occ}^{{\rm gt}-}}
\newcommand{\Occeqp}{\mathrm{Occ}^{{\rm eq}+}}
\newcommand{\Occgtp}{\mathrm{Occ}^{{\rm gt}+}}
\newcommand{\rendgen}[3]{e(#1,#2,#3)}
\newcommand{\rend}[1]{e(#1)}
\newcommand{\rendfull}[1]{e^{\rm full}(#1)}
\newcommand{\rendfullgen}[4]{e^{\rm full}_{#1}(#2,#3,#4)}
\newcommand{\rendcut}[2]{e^{\rm cut}(#1,#2)}
\newcommand{\rendcutgen}[5]{e^{\rm cut}_{#1}(#2,#3,#4,#5)}
\newcommand{\rendlow}[1]{e^{\rm low}(#1)}
\newcommand{\rendlowgen}[4]{e^{\rm low}_{#1}(#2,#3,#4)}
\newcommand{\rendhigh}[1]{e^{\rm high}(#1)}
\newcommand{\rendhighgen}[4]{e^{\rm high}_{#1}(#2,#3,#4)}
\newcommand{\LB}{\mathrm{RangeBeg}}
\newcommand{\UB}{\mathrm{RangeEnd}}

\newcommand{\Successor}{\mathrm{succ}}
\newcommand{\Predecessor}{\mathrm{pred}}
\newcommand{\rcountb}[3]{\mathsf{r\mbox{-}count}_{#1}(#2, #3)}
\newcommand{\rcounti}[4]{\mathsf{r\mbox{-}count}^{\rm inc}_{#1}(#2, #3, #4)}
\newcommand{\rcount}[4]{\mathsf{r\mbox{-}count}_{#1}(#2, #3, #4)}
\newcommand{\rselect}[4]{\mathsf{r\mbox{-}select}_{#1}(#2, #3, #4)}
\newcommand{\mcountb}[3]{\mathsf{mod\mbox{-}count}_{#1}(#2,#3)}
\newcommand{\mcount}[4]{\mathsf{mod\mbox{-}count}_{#1}(#2, #3, #4)}
\newcommand{\mselect}[4]{\mathsf{mod\mbox{-}select}_{#1}(#2, #3, #4)}

\newcommand{\T}{T}
\newcommand{\Pat}{P}
\renewcommand{\S}{\mathsf{S}}
\newcommand{\R}{\mathsf{R}}

\newcommand{\SA}{\mathrm{SA}}
\newcommand{\ISA}{\mathrm{ISA}}
\newcommand{\BWT}{\mathrm{BWT}}

\newcommand{\Z}{\mathbb{Z}}

\newcommand{\Lroots}{\mathrm{Roots}}
\newcommand{\Lrootsgen}[3]{\mathrm{Roots}_{#1}(#2,#3)}

\newcommand{\sub}{\subseteq}
\newcommand{\dd}{\mathinner{.\,.}}
\newcommand{\eps}{\varepsilon}

\setlist[enumerate]{nosep, topsep=1ex}
\setlist[itemize]{nosep, topsep=1ex}
\setlist[description]{nosep,topsep=1ex}

\allowdisplaybreaks

\let\OLDthebibliography\thebibliography
  \renewcommand\thebibliography[1]{
  \OLDthebibliography{#1}
  \setlength{\parskip}{0pt}
  \setlength{\itemsep}{0pt plus 0.3ex}
}

\begin{document}

\title{Dynamic Suffix Array with Polylogarithmic Queries and Updates}

\author{
  \normalsize Dominik Kempa\\[-0.2ex]
  \normalsize Stony Brook University\\[-0.2ex]
  \normalsize \texttt{kempa@cs.stonybrook.edu}
  \and
  \normalsize Tomasz Kociumaka\thanks{Supported by
    NSF 1652303, 1909046, and HDR TRIPODS 1934846 grants,
    and an Alfred P. Sloan Fellowship.}\\[-0.2ex]
  \normalsize University of California, Berkeley\\[-0.2ex]
  \normalsize \texttt{kociumaka@berkeley.edu}
}

\date{\vspace{-0.5cm}}
\maketitle

\begin{abstract}
  The suffix array $\SA[1 \dd n]$ of a text $\T$ of length $n$ is a
  permutation of $\{1, \ldots, n\}$ describing the lexicographical
  ordering of suffixes of $\T$, and it is considered to be among of
  the most important data structures in string algorithms, with dozens
  of applications in data compression, bioinformatics, and information
  retrieval.  One of the biggest drawbacks of the suffix array is that
  it is very difficult to maintain under text updates: even a single
  character substitution can completely change the contents of the
  suffix array. Thus, the suffix array of a dynamic text is modelled
  using \emph{suffix array queries}, which return the value $\SA[i]$
  given any $i \in [1 \dd n]$.

  Prior to this work, the fastest dynamic suffix array implementations
  were by Amir and Boneh.  At ISAAC 2020, they showed how to answer
  suffix array queries in $\Ohtilde(k)$ time, where $k\in [1\dd n]$ is
  a trade-off parameter, with $\Ohtilde(\frac{n}{k})$-time text
  updates.  In a very recent preprint [arXiv, 2021], they also
  provided a solution with $\Oh(\log^5 n)$-time queries and
  $\Ohtilde(n^{2/3})$-time updates.

  We propose the first data structure that supports both suffix array
  queries and text updates in $\bigO(\polylog n)$ time (achieving
  $\Oh(\log^4 n)$ and $\Oh(\log^{3+o(1)} n)$ time, respectively). Our
  data structure is deterministic and the running times for all
  operations are worst-case.  In addition to the standard
  single-character edits (character insertions, deletions, and
  substitutions), we support (also in $\Oh(\log^{3+o(1)} n)$ time) the
  ``cut-paste'' operation that moves any (arbitrarily long) substring
  of $\T$ to any place in $\T$. To achieve our result, we develop a
  number of new techniques which are of independent interest.  This
  includes a new flavor of dynamic locally consistent parsing, as well
  as a dynamic construction of string synchronizing sets with an extra
  \emph{local sparsity} property; this significantly generalizes the
  sampling technique introduced at STOC 2019.  We complement our
  structure by a hardness result: unless the Online Matrix-Vector
  Multiplication (OMv) Conjecture fails, no data structure with
  $\bigO(\polylog n)$-time suffix array queries can support the
  ``copy-paste'' operation in $\bigO(n^{1-\eps})$ time for any $\eps >
  0$.
\end{abstract}

\pagenumbering{arabic}

\section{Introduction}

For a text $T$ of length $n$, the suffix array $\SA[1 \dd n]$ stores
the permutation of $\{1, \ldots, n\}$ such that $T[\SA[i]\dd n]$ is
the $i$th lexicographically smallest suffix of~$T$.  The original
application of SA~\cite{mm1993} was to solve the \emph{text indexing}
problem: construct a data structure such that, given a pattern $\Pat[1
\dd m]$ (typically with $m \ll n$), can quickly count (and optionally
list) all occurrences of $\Pat$ in $\T$.  Since the sought set of
positions occupies a contiguous block $\SA[b \dd e)$ (for some $b,e
\in [1 \dd n+1]$) and since, given $j \in [1 \dd n]$, we can in
$\bigO(m)$ time check if the value $\SA[j]$ is before, inside, or
after this block, the indices $b$ and $e$ can be computed in $\bigO(m
\log n)$ time with binary search. If $b < e$, we can then report all
occurrences of $\Pat$ in~$\T$ at the extra cost of $\bigO(e-b)$.  Soon
after the discovery of $\SA$ it was realized that a very large set of
problems on strings is essentially solved (or at least becomes much
easier) once we have a suffix array (often augmented with the
\emph{LCP array}~\cite{mm1993,klaap2001}). The textbook of
Gusfield~\cite{gusfield} lists many such problems including:
\begin{itemize}
\item finding repeats (e.g., \textsc{MaximalRepeats},
  \textsc{LongestRepeatedFactor}, \textsc{TandemRepeats});
\item computing special subwords (e.g., \textsc{MinimalAbsentWord},
  \textsc{ShortestUniqueSubstring});
\item sequence comparisons (e.g., \textsc{LongestCommonSubstring},
  \textsc{MaximalUniqueMatches}); and
\item data compression (e.g., \textsc{LZ77Factorization},
  \textsc{CdawgCompression}).
\end{itemize}
This trend continues in more recent textbooks~\cite{ohl2013,
MBCT2015,navarrobook}, with the latest suffix array representations
(such as \emph{FM-index}~\cite{FerraginaM05}, \emph{compressed suffix
array}~\cite{GrossiV05}, and \emph{$r$-index}~\cite{Gagie2020}) as
central data structures.  There are even textbooks such
as~\cite{bwtbook} solely dedicated to the applications of suffix
arrays and the closely related Burrows--Wheeler transform
(BWT)~\cite{bwt}.

The power of suffix array comes with one caveat: It is very difficult
to maintain it for a text undergoing updates.  For example, for $\T =
\texttt{b}^n$ (symbol $\texttt{b}$ repeated $n$ times) we have
$\SA_{\T} = [n, \dots, 2, 1]$, whereas for $\T' =
\texttt{b}^{n-1}\texttt{c}$ (obtained from $\T$ with a single
substitution), it holds $\SA_{\T'} = [1, 2, \dots, n]$, i.e., the
complete reversal. Even worse, if $n=2m+1$ and $\T'' =
\texttt{b}^{m}\texttt{a}\texttt{b}^{m}$ (again, a single
substitution), then $\SA_{\T''} = [m{+}1, n, m, n-1, m-1, \ldots, m+2,
1]$, i.e., the near-perfect interleaving of two halves of $\SA_{\T}$.
In general, even a single character substitution may permute $\SA$ in
a very complex manner.  Thus, if one wishes to maintain the suffix
array of a dynamic text, its entries cannot be stored in plain form
but must be obtained by querying a data structure. The quest for such
\emph{dynamic suffix array} is open since the birth of suffix array
nearly three decades ago. We thus pose our problem:

\begin{problem}\label{pr:main}
  Can we support efficient $\SA$ queries for a dynamically changing
  text?
\end{problem}

\paragraph*{Previous Work}

One of the first attempts to tackle the above problem is due to
Salson, Lecroq, L{\'e}onard, and Mouchard~\cite{salson}.  Their
dynamic suffix array achieves good practical performance on real-world
texts (including English text and DNA sequences), but its update takes
$\Theta(n)$ in the worst-case. A decade later, Amir and
Boneh~\cite{AmirB20} proposed a structure that, for any parameter $k
\in [1 \dd n]$, supports $\SA$ and $\SA^{-1}$ queries in $\Ohtilde(k)$
time and character substitutions in $\Ohtilde(\frac{n}{k})$ time. This
implies the first nontrivial trade-off for the dynamic suffix array,
e.g., $\Ohtilde(\sqrt{n})$-time operations. Very recently, Amir and
Boneh~\cite{AmirB21} described a dynamic suffix array that supports
updates (insertions and deletions) in the text in $\Ohtilde(n^{2/3})$
time and $\SA$ queries in $\bigO(\log^5 n)$ time. They also gave a
structure that supports substitutions in $\Ohtilde(n^{1/2})$ time and
$\SA^{-1}$ queries in $\bigO(\log^4 n)$ time.

A separate line of research focused on the related problem of
\emph{dynamic text indexing} introduced by Gu, Farach, and,
Beigel~\cite{GuFB94}. This problem aims to design a data structure
that permits updates to the text $\T$ and pattern searches (asking for
all occurrences of a given pattern $\Pat$ in $\T$).  As noted
in~\cite{AmirB20}, the solution in~\cite{GuFB94} achieves
$\Ohtilde(\sqrt{n})$-time updates to text and $\Ohtilde(m\sqrt{n} +
{\rm occ} \log n)$ pattern search query (where ${\rm occ}$ is the
number of occurrences of $\Pat$ in~$\T$).  Sahinalp and
Vishkin~\cite{SahinalpV96} then proposed a solution based on the idea
of \emph{locally consistent parsing} that achieves $\bigO(\log^3
n)$-time update and $\bigO(m + {\rm occ} + \log n)$ pattern searching
time. The update time was later improved by Alstrup, Brodal, and
Rauhe~\cite{Alstrup2000} to $\bigO(\log^2 n \log \log n \log^{*} n)$
at the expense of the slightly slower query $\bigO(m + {\rm occ} +
\log n \log \log n)$.  This last result was achieved by building on
techniques for the \emph{dynamic string equality} problem proposed by
Mehlhorn, Sundar, and Uhrig~\cite{Mehlhorn}.  This was improved
in~\cite{dynstr2} to $\bigO(\log^2 n)$-time update and $\bigO(m + {\rm
occ})$-time search.  A slightly different approach to dynamic text
indexing, based on \emph{dynamic position heaps} was proposed
in~\cite{EhrenfeuchtMOW11}. It achieves $\bigO(m \log n + {\rm occ})$
amortized search, but its update take $\Theta(n)$ time in the worst
case. There is also work on \emph{dynamic compressed text
indexing}. Chan, Hon, Lam, and Sadakane~\cite{ChanHLS07} proposed an
index that uses $\bigO(\tfrac{1}{\eps}(nH_0(\T) + n))$ bits of space
(where $H_0(\T)$ is the zeroth order empirical entropy of $\T$),
searches in $\bigO(m \log^2 n (\log^{\eps} n + \log \sigma) + {\rm
occ} \log^{1 + \eps} n)$ time (where $\sigma$ is the alphabet size),
and performs updates in $\bigO(\sqrt{n} \log^{2+\eps} n)$ amortized
time, where $0 < \eps \leq 1$.  Recently, Nishimoto, I, Inenaga,
Bannai, and Takeda~\cite{NishimotoDAM} proposed a faster index for a
text $\T$ with LZ77 factorization of size $z$. Assuming for simplicity
$z \log n \log^{*} n = \bigO(n)$, their index occupies $\bigO(z \log n
\log^{*} n)$ space, performs updates in amortized $\bigO((\log n
\log^{*} n)^2 \log z)$ time (in addition to edits, they also support
the ``cut-paste'' operation that moves a substring of $\T$ from one
place to another), and searches in $\bigO\left(m\min\left\{\tfrac{\log
\log n \log \log z}{\log \log \log n}, \sqrt{\tfrac{\log z}{\log \log
z}}\right\} + \log z \log m \log^{*}n (\log n + \log m \log^{*} n) +
{\rm occ} \log n\right)$ time.

We point out that although the (compressed) dynamic text indexing
problem~\cite{GuFB94,NishimotoDAM} discussed above is related to
dynamic suffix array, it is not the same.  Assuming one accepts
additional log factors, the dynamic suffix array problem is
\textbf{\emph{strictly harder}}: it solves dynamic indexing (by simply
adding binary search on top), but no converse reduction is known; such
a reduction would compute values of $\SA$ in $\bigO(\polylog n)$ time
using searches for short patterns.  Thus, the many applications of
$\SA$ listed above cannot be solved with these indexes. Unfortunately,
due to lack of techniques, the dynamic suffix array problem has seen
very little progress; as noted by Amir and Boneh~\cite{AmirB20},
``(\ldots) although a dynamic suffix array algorithm would be
extremely useful to automatically adapt many static pattern matching
algorithms to a dynamic setting, other techniques had to be
sought''. They remark, however, that ``Throughout all this time, an
algorithm for maintaining the suffix tree or suffix array of a
dynamically changing text had been sought''. To sum up, until now, the
best dynamic suffix arrays have been those of~\cite{AmirB20}, taking
$\Ohtilde(k)$ time to answer $\SA$ and $\SA^{-1}$ queries and
$\Ohtilde(\frac{n}{k})$ time for substitutions, and~\cite{AmirB21},
taking $\Ohtilde(n^{2/3})$ time for insertions/deletions and
$\bigO(\log^5 n)$ time for $\SA$ queries, or $\Ohtilde(n^{1/2})$ time
for substitutions and $\bigO(\log^4 n)$ time for $\SA^{-1}$
queries. No solution with $\polylog n$-time queries and updates (even
amortized or expected) was known, not even for character substitutions
only.

\paragraph{Our Results}

We propose the first dynamic suffix array with all operations (queries
\textbf{\textit{and}} updates) taking only $\bigO(\polylog n)$
time. Our data structure is \ul{deterministic} and the complexities of
both queries and updates are \ul{worst-case}. Thus, we leap directly
to a solution satisfying the commonly desired properties on the query
and update complexity for this almost thirty-years old open
problem. In addition to single-character edits, our structure supports
the powerful ``cut-paste'' operation, matching the functionality of
state-of-the-art indexes~\cite{Alstrup2000,NishimotoDAM}.  More
precisely, our structure supports the following operations (the bounds
below are simplified overestimates; see \cref{thm:saisa}):
\begin{itemize}
\item We support $\SA$ queries (given $i \in [1 \dd n]$, return
  $\SA[i]$) in $\bigO(\log^4 N)$ time.
\item We support $\SA^{-1}$ queries (given $j \in [1 \dd n]$, return
  $\SA^{-1}[j]$) in $\bigO(\log^5 N)$ time.
\item We support updates (insertion and deletion of a single symbol in
  $\T$ as well as the ``cut-paste'' operation, moving any block of
  $\T$ to any other place in $\T$) in $\bigO(\log^3 N (\log \log
  N)^2)$ time.
\end{itemize}
Here, $N = n\sigma$ is the product of the current text length and the
size of the alphabet $\Sigma = [0 \dd \sigma)$.

The above result may leave a sense of incompleteness regarding further
updates such as ``copy-paste''. We show that, despite its similarity
with ``cut-paste'', supporting this operation in the dynamic setting
is most likely very costly. More precisely, we prove that, unless the
Online Matrix-Vector Multiplication (OMv)
Conjecture~\cite{DBLP:conf/stoc/HenzingerKNS15} fails, no data
structure that supports $\SA$ or $\SA^{-1}$ queries in $\bigO(\polylog
n)$ time can support the ``copy-paste'' operation in
$\bigO(n^{1-\eps})$ time for any $\eps > 0$.\footnote{In fact, we
prove a more general result that the \emph{product} of time needed to
support $\SA$/$\SA^{-1}$ queries and copy-paste operation cannot be
within $\bigO(n^{1-\eps})$ for any $\eps > 0$. Thus, the trade-off
similar to the one achieved by Amir and Boneh~\cite{AmirB20}
($\Ohtilde(\tfrac{n}{k})$-time query and $\Ohtilde(k)$-time update)
may still be possible for a dynamic suffix array with copy-pastes; we
leave proving such upper bound as a possible direction for future
work.}

\paragraph{Technical Contributions}

To achieve our result, we develop new techniques and significantly
generalize several existing ones.  Our first novel technique is the
notion of \emph{locally sparse} synchronizing sets. String
synchronizing sets~\cite{sss} have recently been introduced in the
context of efficient construction of $\BWT$ and $\LCE$ queries for
``packed strings'' (where a single machine word stores multiple
characters). Since then, they have found many other
applications~\cite{circfactor,Charalampopoulos21,resolution,sss-index,quantum-sss}.
Loosely speaking (a formal definition follows in \cref{sec:prelim}),
for any fixed $\tau \geq 1$, a set $\S \sub [1 \dd n]$ is a
$\tau$-synchronizing set of string $\T \in \Sigma^{n}$ if $\S$ samples
positions of $\T$ \emph{consistently} (i.e., whether $i \in [1 \dd n]$
is sampled depends only on the length-$\Theta(\tau)$ context of $i$ in
$\T$) and does not sample positions inside highly periodic fragments
(the so-called \emph{density} condition). In all prior applications
utilizing synchronizing sets, the goal is to ensure that $\S$ is
\emph{sparse on average}, i.e., that the size $|\S|$ is minimized.  In
this paper, we prove that at the cost of increasing the size of $|\S|$
by a mere factor $\bigO(\log^{*}(\tau\sigma))$, we can additionally
ensure that $\S$ is also \emph{locally sparse} (see
\cref{lem:sss}). We then show that such $\S$ can be maintained
dynamically using a new construction of $\S$ from the signature
parsing (a technique introduced in~\cite{Mehlhorn} and used, for
example, in~\cite{Alstrup2000,NishimotoDAM}).  The crucial benefit of
local sparsity is that any auxiliary structure that depends on the
length-$\Theta(\tau)$ contexts of positions in $\S$, including $\S$
itself, can be updated efficiently.  Another result of independent
interest is the first dynamic construction of string synchronizing
sets.  For this, we internally represent some substrings of $\T$ using
the abstraction of \emph{dynamic
strings}~\cite{Alstrup2000,SahinalpV94,Mehlhorn,dynstr}.  The problem
with all existing variants of dynamic strings, however, is that they
rely on representing the strings using a hierarchical representation
whose only goal is to ensure the string shrinks by a constant factor
at each level. This, however, is not sufficient for our purpose: in
order to satisfy the density condition for the resulting synchronizing
set, we also need the expansions of all symbols at any given level to
have some common upper bound on the expansion length. The notion of
such ``balanced'' parsing is known~\cite{SahinalpV94,BirenzwigeGP20},
but has only been implemented in static settings. We show the first
variant of dynamic strings that maintains a ``balanced'' parsing at
every level, and consequently lets us dynamically maintain a locally
sparse string synchronizing set.

The mainstay among data structures for pattern matching or $\SA$
queries is the use of (typically 2D) orthogonal range
searching~\cite{Chazelle1988}.  Example indexes include nearly all
indexes based on LZ77~\cite{ArroyueloNS12, DCC2015, BilleEGV18,
  ChristiansenEKN21, GagieGKNP14, phdjuha, kreft2013compressing,
  NishimotoDAM}, context-free grammars~\cite{BLRSRW15, ClaudeN11,
  ClaudeNP21, GagieGKNP12, MaruyamaNKS13, TakabatakeTS14, LyndonSLP},
and some recent BWT-based
indexes~\cite{sss-index,MunroNN20b,ChienHSTV15}.  Nearly all these
structures maintain a set of points corresponding to some set of
suffixes of $\T$ (identified with their starting positions $\mathsf{P}
\sub [1 \dd n]$) ordered lexicographically. The problem with adapting
this to the dynamic setting is that once we modify $\T$ near the end,
the order of (potentially all) elements in $\mathsf{P}$ changes.  We
overcome this obstacle as follows. Rather than on sampling of
suffixes, we rely on $\log n$ levels of sampling of \emph{substrings}
(identified by the set $\S_k$ of the starting positions of their
occurrences) of length roughly $2^k$, where $k \in [0 \dd \lfloor \log
n \rfloor)$ implemented using dynamic locally sparse synchronizing
sets.  With such structure, we can efficiently update the sets $\S_k$
and the associated geometric structures, but we lose the ability to
compare suffixes among each other. In \cref{sec:sa}, we show that even
with such partial sampling, we can still implement $\SA$ queries.  In
$\log n$ steps, our query algorithm successively narrows the range of
$\SA$ to contain only suffixes prefixed with $\T[\SA[i] \dd \SA[i] +
\ell)$, while also maintaining the starting position of some arbitrary
suffix in the range, where $\ell = 2^k$ is the current prefix
length. In the technical overview, we sketch the main ideas of this
reduction, but we remark that the details of this approach are
nontrivial and require multiple technical results (see, e.g.,
Section~\ref{sec:sa-periodic-pos-decomposition}--\ref{sec:sa-periodic-occ-pos}).

Finally, we remark that, even though (as noted earlier), dynamic
suffix array is a strictly harder problem that text indexing (since
one can be reduced to the other, but not the other way around), our
result has implications even for the text indexing problem.  The
existing dynamic text indexes with fast queries and updates (such
as~\cite{dynstr2,NishimotoDAM}) can only \emph{list} all $k$
occurrences of any pattern.  One however, cannot obtain the number of
occurrences (which is the standard operation supported by most of the
static indexes~\cite{NavarroM07}) in time that is always bounded by
$\bigO(m \polylog n)$ (since $k$ can be arbitrarily
large).\footnote{Efficient counting for indexes relying on orthogonal
range searching is a technically challenging problem that has been
solved only recently for the static case~\cite{ChristiansenEKN21}. It
is possible that these ideas can be combined with~\cite{NishimotoDAM},
but the result will nevertheless be significantly more complicated
than counting using the suffix array.}  Our dynamic suffix array, on
the other hand, implements counting seamlessly: it suffices to perform
the standard binary search~\cite{mm1993} over $\SA$ for the pattern
$\Pat$, resulting in the endpoints of range $\SA[b \dd e)$ of suffixes
of $\T$ prefixed with $\Pat$, and return $e - b$.

\paragraph{Related Work}

Chan, Hon, Lam, and Sadakane~\cite{ChanHLS07} introduced a problem of
\emph{indexing text collections}, which asks to store a dynamically
changing collection $\mathcal{C}$ of strings (under insertions and
deletions of \emph{entire strings}) so that pattern matching queries
on $\mathcal{C}$ can be answered efficiently.  Although the resulting
data structures are also called \emph{dynamic full-text
indexes}~\cite{MakinenN08} or \emph{dynamic suffix
trees}~\cite{RussoNO08}, we remark that they are solving a different
problem that what, by analogy to dynamic suffix array, we would mean
by ``dynamic suffix tree''. Observe that we cannot simulate the
insertion of a symbol in the middle of $\T$ using a collection of
strings. Maintaining symbols of $\T$ as a collection of length-1
strings will not work because the algorithms
in~\cite{ChanHLS07,MakinenN08,RussoNO08} only report occurrences
entirely inside one of the strings. Since the insertion of a string
$S$ of length $m$ into $\mathcal{C}$ takes $\Omega(m)$ time
in~\cite{ChanHLS07,MakinenN08,RussoNO08} (similarly for deletion), one
also cannot efficiently use $\mathcal{C}$ to represent the entire text
$\T$ as a single element of $\mathcal{C}$.

Tangentially related to the problem of text indexing is the problem of
\emph{sequence representation}~\cite{NavarroN14}, where we store a
string $S[1 \dd n]$ under character insertions and deletions and
support the access to $S$, ${\rm rank}(i,c)$ (returning $|\{j \in [1
\dd i] : S[j] = c\}|$) and ${\rm select}(i,c)$ (returning the $i$th
occurrence of $c$ in $S$). A long line of research,
including~\cite{MakinenN08, GonzalezN09, HeM10, NavarroS14},
culminated with the work of Navarro and Nekrich, who achieved
$\bigO(\log n/\log \log n)$ time for all operations (amortized for
updates), while using space close to $nH_0(S)$, where $H_0$ is the
zeroth order~empirical~entropy.

Finally, there is a line of research focusing on storing a text $\T$
under updates so that we can support efficient \emph{longest common
extension} queries $\LCE_{\T}(i,j)$ that return the length of the
longest common prefix of $\T[i \dd |\T|]$ and $\T[j \dd |\T|]$.  The
research was initiated with the seminal work
in~\cite{Mehlhorn,Alstrup2000} (recently improved
in~\cite{dynstr}). More recently, a solution working in the compressed
space (achieving similar runtimes as the index in~\cite{NishimotoDAM})
was proposed in~\cite{NishimotoMFCS}.

\section{Preliminaries}\label{sec:prelim}

A \emph{string} is a finite sequence of characters from some set
$\Sigma$ called the \emph{alphabet}. We denote the length of a string
$S$ as $|S|$. For any index $i\in [1\dd |S|]$,\footnote{ For $i,j\in
\mathbb{Z}$, denote $[i\dd j]=\{k \mkern 1.5mu {\in} \mkern 1.5mu
\mathbb{Z} : i \le k \le j\}$, $[i\dd j)=\{k \mkern 1.5mu {\in} \mkern
1.5mu \mathbb{Z} : i \le k < j\}$, and $(i\dd j]={\{k \mkern 1.5mu
{\in} \mkern 1.5mu \mathbb{Z}: i < k \le j\}}$.  } we denote the $i$th
character of $S$ as $S[i]$.  A string of the form $S[i \dd
j)=S[i]S[i+1]\dots S[j-1]$, where $i \leq i \leq j \leq |S|+1$ is
called a \emph{fragment} of $S$.  If $S$ is known, we will encode
$S[i\dd j)$ in $\bigO(1)$ space as a pair $(i,j)$. \emph{Prefixes} and
\emph{suffixes} are of the form $S[1\dd j)$ and $S[i\dd |S|]$,
respectively. By $\revstr{S}$ we denote the \emph{reverse} of $S$,
i.e., $\revstr{S} = S[|S|]\dots S[2][1]$.  The \emph{concatenation} of
two strings $U$ and $V$ is denoted $UV$ or $U\cdot V$.  Moreover,
$S^k=\bigodot_{i=1}^k S$ is the concatenation of $k$ copies of $S$;
note that $S^0=\eps$ is the \emph{empty string}. An integer $p\in
[1\dd |S|]$ is a called a \emph{period} of $S$ if $S[1\dd
|S|-p]=S[1+p\dd |S|]$; we denote the shortest period of $S$
as~$\per(S)$.  We use $\preceq$ to denote the order on $\Sigma$,
extended to the \emph{lexicographic} order on $\Sigma^*$ (the set of
strings over $\Sigma$) so that $U,V\in \Sigma^*$ satisfy $U\preceq V$
if and only if either $U$ is a prefix of $V$, or $U[1\dd i)=V[1\dd i)$
and $U[i]\prec V[i]$ holds for some $i\in [1\dd \min(|U|,|V|)]$.

\begin{wrapfigure}{r}{0.34\textwidth}
  \vspace{-.75cm}
  \begin{tikzpicture}[yscale=0.39]
    \foreach \x [count=\i] in {a, aababa, aababababaababa, aba,
      abaababa, abaababababaababa, ababa, ababaababa, abababaababa,
      ababababaababa, ba, baababa, baababababaababa, baba, babaababa,
      babaababababaababa, bababaababa, babababaababa,bbabaababababaababa}
        \draw (1.9, -\i) node[right]
          {$\texttt{\x}$};
    \draw(1.9,0) node[right] {\scriptsize $\T[\SA[i]\dd n]$};
    \foreach \x [count=\i] in {b, b, b, b, b, b, a, b, b,
                               a, a, a, a, a, a, b, a, a,\$}
      \draw (0.7, -\i) node {\footnotesize $\i$};
    \draw(0.7,0) node{\scriptsize $i\vphantom{\BWT[]}$};
    \foreach \x [count=\i] in {19,14,5,17,12,3,15,10,
                               8,6,18,13,4,16,11,2,9,7,1}
      \draw (1.4, -\i) node {$\x\vphantom{\textbf{\underline{7}}}$};
    \draw(1.4,0) node{\scriptsize $\SA[i]$};
  \end{tikzpicture}
  \caption{A list of all sorted suffixes of $\T=
    \texttt{bbabaababababaababa}$ along with
    the suffix array.}\label{fig:example}
  \vspace{-.5cm}
\end{wrapfigure}

Throughout, we consider a string (called the \emph{text}) $\T$ of
length $n\geq 1$ over an integer alphabet $\Sigma = [0 \dd
\sigma)$. We assume that $\T[n] = \texttt{\$}$ (where $\$ =
\min\Sigma$) and $\texttt{\$}$ does not occur in $\T[1 \dd n)$.

The \emph{suffix array} $\SA[1\dd n]$ of $\T$ is a permutation of
$[1\dd n]$ such that $\T[\SA[1]\dd n] \prec \T[\SA[2]\dd n] \prec
\cdots \prec \T[\SA[n]\dd n]$, i.e., $\SA[i]$ is the starting position
of the lexicographically $i$th suffix of $\T$; see \cref{fig:example}
for an example.  The \emph{inverse suffix array} $\ISA[1 \dd n]$ is
the inverse permutation of $\SA$, i.e., $\ISA[j] = i$ holds if and
only if $\SA[i] = j$. By $\lcp(U,V)$ (resp.\ $\lcs(U,V)$) we denote
the length of the longest common prefix (resp.\ suffix) of $U$ and
$V$. For $j_1, j_2 \in [1 \dd n]$, we let $\LCE_{\T}(j_1,j_2) =
\lcp(\T[j_1 \dd ], \T[j_2 \dd ])$.

The \emph{rotation} operation $\rot(\cdot)$, given a string $S\in
\Sigma^+$, moves the last character of $S$ to the front so that
$\rot(S)=S[|S|]\cdot S[1\dd |S|-1]$.  The inverse operation
$\rot^{-1}(\cdot)$ moves the first character of $S$ to the back so
that $\rot^{-1}(S)=S[2\dd |S|]\cdot S[1]$.  For an integer $s\in \Z$,
the operation $\rot^s(\cdot)$ denotes the $|s|$-time composition of
$\rot(\cdot)$ (if $s\ge 0$) or $\rot^{-1}(\cdot)$ (if $s\le 0$).
Strings $S,S'$ are \emph{cyclically equivalent} if $S'=\rot^{s}(S)$
for some $s\in \Z$.

We use the word RAM model of computation~\cite{Hagerup98} with $w$-bit
\emph{machine words}, where $w \ge \log n$.

\begin{definition}[{\fontfamily{lmss}\selectfont $\tau$-synchronizing
      set~\cite{sss}}{\fontfamily{lmr}\selectfont}]\label{def:sss} Let
  $\T\in \Sigma^n$ be a string and let $\tau \in [1\dd
  \lfloor\frac{n}{2}\rfloor]$ be a parameter. A set $\S \subseteq [1
  \dd n - 2\tau + 1]$ is called a \emph{$\tau$-synchronizing set} of
  $\T$ if it satisfies the following \emph{consistency} and
  \emph{density} conditions:
  \begin{enumerate}
  \item If $\T[i \dd i + 2\tau) = \T[j\dd j + 2\tau)$, then $i \in \S$
    holds if and only if $j \in \S$ (for $i, j \in [1 \dd n - 2\tau +
    1]$),
  \item $\S\cap[i \dd i + \tau)=\emptyset$ if and only if $i \in
    \R(\tau,\T)$ (for $i \in [1 \dd n - 3\tau + 2]$), where
  \end{enumerate}
  \[
    \R(\tau,\T) := \{i \in [1 \dd |\T| - 3\tau + 2] : \per(\T[i \dd i
      + 3\tau - 2]) \leq \tfrac{1}{3}\tau\}.
  \]
\end{definition}

In most applications, we want to minimize $|\S|$. The density
condition imposes a lower bound $|\S| = \Omega(\frac{n}{\tau})$ for
strings of length $n \ge 3\tau - 1$ that do not contain highly
periodic substrings of length $3\tau-1$. Hence, in the worst case, we
cannot hope to improve upon the following bound.

\begin{theorem}[{\cite[Proposition~8.10]
      {sss}}]\label{th:sss-existence-and-construction}
  For every $\T\in \Sigma^n$ and $\tau \in [1\dd
  \lfloor\frac{n}{2}\rfloor]$, there exists a $\tau$-synchronizing set
  $\S$ of size $|\S| = \bigO \left( \frac{n}{\tau} \right)$. Such $\S$
  can be deterministically constructed in $\bigO(n)$ time.
\end{theorem}

\section{Technical Overview}\label{sec:tech}

We derive our dynamic suffix array gradually. We start
(\cref{sec:range-queries,sec:mod-queries}), by introducing the
auxiliary tools. In the first part of the paper (\cref{sec:sa}) we
prove that if we have $\Theta(\log n)$ synchronizing sets for values
of $\tau$ spanning across the whole range $[1 \dd n]$, and we can
support some set of queries (stated as ``assumptions'') concerning
either positions in those synchronizing sets, gaps across successive
elements, or their length-$\Theta(\tau)$ contexts in $\T$, then we can
support $\SA$ queries on $\T$.  In the second part of the paper
(\cref{sec:parsing,sec:ds,sec:bsp2sss,sec:app}) we then show how to
satisfy these assumptions for text supporting update operations. This
approach lets us separate the clean (combinatorial) details concerning
$\SA$ queries from (more algorithmic and technical) details concerning
the maintenance of the underlying synchronizing sets and the
associated structures.  This also lets us for now treat $\T$ as well
as each synchronizing sets as static, since the reduction works for
\emph{any} collection of such sets, and thus if after the update these
sets (and the associated structures) change, the query algorithm is
unaffected.

We start with an overview of \cref{sec:sa}. Let $\ell \geq 1$. For any
$j \in [1 \dd n]$, we define
\begin{align*}
  \Occ_{\ell}(j)
    &= \{j' \in [1 \dd n] : \T^{\infty}[j' \dd j' + \ell) = 
       T^{\infty}[j \dd j + \ell)\},\\
  \LB_{\ell}(j)
    &= |\{j' \in [1 \dd n] : \T[j' \dd n] \prec \T[j \dd n]
       \text{ and }\LCE_{\T}(j, j) < \ell\}|, \text{ and }\\
  \UB_{\ell}(j)
    &= \LB_{\ell}(j) + |\Occ_{\ell}(j)|.
\end{align*}
For the motivation of the above names, note that viewing $P :=
\T^{\infty}[j \dd j + \ell)$ as a pattern, we have $\{\SA[i] : i \in
(\LB_{\ell}(j) \dd \UB_{\ell}(j)]\} = \{i \in [1 \dd n] :
\T^{\infty}[i \dd i + |P|) = P\}$.

Moreover, for any $j \in [1 \dd n]$, we define
\begin{align*}
  \Posalll(j) &=
    \{j' \in [1 \dd n] : \T[j' \dd n] \prec \T[j \dd n]
      \text{ and }\LCE_{\T}(j',j) \in [\ell \dd 2\ell)\}\text{ and }\\
  \Posallg(j) &=
    \{j' \in [1 \dd n] : \T[j' \dd n] \succ \T[j \dd n]
      \text{ and }\LCE_{\T}(j',j) \in [\ell \dd 2\ell)\}.
\end{align*}
We denote the size of $\Posalll(j)$ and $\Posallg(j)$ as
$\deltal_{\ell}(j) = |\Posalll(j)|$ and $\deltag_{\ell}(j) =
|\Posallg(j)|$. Note that by definition, $(\LB_{2\ell}(j),
\UB_{2\ell}(j)) = (\LB_{\ell}(j) + \deltal_{\ell}(j), \LB_{\ell}(j) +
\deltal_{\ell}(j) + |\Occ_{2\ell}(j)|)$ and $(\LB_{2\ell}(j),
\UB_{2\ell}(j)) = (\UB_{\ell}(j) - \deltag_{\ell}(j) -
|\Occ_{2\ell}(j)|, \UB_{\ell}(j) - \deltag_{\ell}(j))$.

The main idea of our algorithm is as follows.  Suppose that we have
obtained some $j \in \Occ_{16}(\SA[i])$ and $(\LB_{16}(\SA[i]),
\UB_{16}(\SA[i]))$ (\cref{as:small}). Then, for $q = 4, \ldots, \lceil
\log n \rceil - 1$, denoting $\ell = 2^q$, we compute
$(\LB_{2\ell}(\SA[i]), \UB_{2\ell}(\SA[i]))$ and some $j' \in
\Occ_{2\ell}(\SA[i])$, by using as input some position $j \in
\Occ_{\ell}(\SA[i])$ and the pair $(\LB_{\ell}(\SA[i]),
\UB_{\ell}(\SA[i]))$, i.e., the output of the earlier step. This lets
us compute $\SA[i]$, since eventually we obtain some $j' \in
\Occ_{2^{\lceil \log n \rceil}}(\SA[i])$, and for any $k \geq n$,
$\Occ_{k}(\SA[i]) = \{\SA[i]\}$, i.e., we must have $j' = \SA[i]$.

Our goal is to show how to implement a single step of the above
process.  Let $\ell \in [16 \dd n)$ and $i \in [1 \dd n]$. Suppose
that we are given $(\LB_{\ell}(\SA[i]), \UB_{\ell}(\SA[i]))$ and some
$j \in \Occ_{\ell}(\SA[i])$ as input.  We aim to show that under
specific assumptions about the ability to perform some queries, we can
compute some $j' \in \Occ_{2\ell}(\SA[i])$ and the pair
$(\LB_{2\ell}(\SA[i]), \UB_{2\ell}(\SA[i]))$.

Let $\tau := \lfloor \tfrac{\ell}{3} \rfloor$. Our algorithm works
differently, depending on whether it holds $\SA[i] \in \R(\tau,\T)$ or
$\SA[i] \in [1 \dd n] \setminus \R(\tau,\T$). Thus, we first need to
efficiently implement this check.  Observe that whether or not it
holds $j \in \R(\tau,\T)$ depends only on $\T[j \dd j + 3\tau -
1)$. Therefore, by $3\tau - 1 \leq \ell$, if $j \in
\Occ_{\ell}(\SA[i])$ then $\SA[i] \in \R(\tau,\T)$ holds if and only
if $j \in \R(\tau,\T)$. Consequently, given any $j \in
\Occ_{\ell}(\SA[i])$, using \cref{as:core} we can efficiently check if
$\SA[i] \in \R(\tau,\T)$ (such $\SA[i]$ is called \emph{periodic}) or
$\SA[i] \in [1 \dd n] \sm \R(\tau,\T)$ (i.e., the position $\SA[i]$ is
\emph{nonperiodic}).

\paragraph{The Nonperiodic Positions}\label{sec:to-sa-nonperiodic}

Assume $\SA[i] \in [1 \dd n] \setminus \R(\tau,\T)$. We proceed in two
steps. First, we show how to compute $|\Posalll(\SA[i])|$ and
$|\Occ_{2\ell}(\SA[i])|$ assuming we have some $j' \in
\Occ_{2\ell}(\SA[i])$.  By the earlier observation about
$\deltal(\SA[i]) = |\Posalll(\SA[i])|$ and the definition of
$\Occ_{2\ell}(\SA[i])$, this gives us the pair $(\LB_{2\ell}(\SA[i]),
\UB_{2\ell}(\SA[i]))$. We then explain how to find $j'$.

Let $j' \in \Occ_{2\ell}(\SA[i])$.  Observe that by $\T^{\infty}[j'
\dd j' + 2\ell) = \T[\SA[i] \dd \SA[i] + 2\ell)$, we have
$|\Posalll(\SA[i])| = |\Posalll(j')|$, $|\Occ_{2\ell}(\SA[i])| =
|\Occ_{2\ell}(j')|$, and $j' \not \in \R(\tau,\T)$. We can thus focus
on computing $|\Posalll(j')|$ and $|\Occ_{2\ell}(j')|$.  Let $\S$ be
any $\tau$-synchronizing set of $\T$.  First, observe that by $j'
\not\in \R(\tau,\T)$, the position $s' = \Successor_{\S}(j')$
satisfies $s' - j' < \tau$. Thus, by the consistency of $\S$ and
$3\tau \leq \ell$, all $j'' \in \Posalll(j')$ share a common offset
$\delta_{\S} = s' - j'$ such that $j'' + \delta_{\S} = \min(\S \cap
[j'' \dd j''+\tau))$ and hence the relative lexicographical order
between $\T[j'' \dd n]$ and $\T[j' \dd n]$ is the same as between
$\T[j''+\delta_{\S} \dd n]$ and $\T[j' + \delta_{\S} \dd n]$.  Hence,
to compute $|\Posalll(j')|$ it suffices to count $s'' \in \S$ that
are:
\begin{enumerate}
\item\label{to-sa-nonperiodic-condition-1} Preceded in $\T$ by the
  string $\T[j' \dd s')$, and
\item\label{to-sa-nonperiodic-condition-2} For which it holds $\T[s''
  \dd n] \prec \T[s' \dd n]$ and $\LCE_{\T}(s'', s') \in [\ell -
  \delta_{\S} \dd 2\ell - \delta_{\S})$.
\end{enumerate}
Observe, that for any $q \geq 2\ell$,
Condition~\ref{to-sa-nonperiodic-condition-1} is equivalent to
position $s''$ having a reversed left length-$q$ context in the
lexicographical range $[X \dd X')$, where $\revstr{X} = \T[j' \dd s')$
and $X' = Xc^{\infty}$ (where $c = \max\Sigma$), and
Condition~\ref{to-sa-nonperiodic-condition-2} is equivalent to
position $s''$ having a right length-$q$ context in $[Y \dd Y')$,
where $Y = \T^{\infty}[s' \dd j' + \ell)$ and $Y' = \T^{\infty}[s' \dd
j' + 2\ell)$ (\cref{lm:utils}). Consequently, the only queries needed
to compute $|\Posalll(j')|$ are $\Successor_{\S}(j')$ and generalized
range queries on a set of points $\Pts = \{(\revstr{\T^{\infty}[s'-q
\dd s')}, \T^{\infty}[s' \dd s'+q)) : s' \in \S\}$.  Since $\tau
\,{=}\, \lfloor \tfrac{\ell}{3} \rfloor$ and $\ell \geq 16$, it
suffices to choose $q = 7\tau$ to satisfy $q \geq 2\ell$. Thus, under
\cref{as:nonperiodic}, we can efficiently compute $|\Posalll(j')| =
|\Posalll(\SA[i])|$.

The intuition for $|\Occ_{2\ell}(j')|$ is similar, except we observe
that Condition~\ref{to-sa-nonperiodic-condition-2} is that $s''$
satisfies $\T^{\infty}[s'' \dd j'' + 2\ell) = \T^{\infty}[s' \dd j' +
2\ell)$, which is equivalent to $s''$ having a right length-$q$
context in $[Y' \dd Y'c^{\infty})$. Thus, we can also count such $s''$
(and consequently, compute $|\Occ_{2\ell}(\SA[i])|$) using $\Pts$.

The above reductions are proved in
\cref{lm:sa-nonperiodic-pos,lm:sa-nonperiodic-occ} and lead to the
following result.

\begin{proposition}\label{pr:to-nonperiodic-first}
  Under \cref{as:nonperiodic}, assuming $\SA[i] \in [1 \dd n] \sm
  \R(\tau,\T)$, given a position $j' \in \Occ_{2\ell}(\SA[i])$, we can
  efficiently compute $|\Posalll(\SA[i])|$ and
  $|\Occ_{2\ell}(\SA[i])|$.
\end{proposition}

It remains to show how to find some $j' \in \Occ_{2\ell}(\SA[i])$.
For this, observe that if we sort all $j'' \in \Occ_{\ell}(\SA[i])$ by
their right length-$2\ell$ context in $\T^{\infty}$ then for the $k$th
position $j''$ in this order we have $\T^{\infty}[j'' \dd j'' + 2\ell)
= \T^{\infty}[\SA[b + k] \dd \SA[b + k] + 2\ell)$, since $\SA(b \dd
e]$ also contains all $j'' \in \Occ_{2\ell}(\SA[i])$ sorted by their
length-$2\ell$ right context, although potentially in a different
order.  Note, however, that $j' \in \Occ_{2\ell}(\SA[i])$ only
requires $\T^{\infty}[j' \dd j' + 2\ell) = \T^{\infty}[\SA[i] \dd
\SA[i] + 2\ell)$.  Thus, the ability to find the $k$th element in the
sequence of all $j'' \in \Occ_{\ell}(\SA[i])$ sorted by
$\T^{\infty}[j'' \dd j'' + 2\ell)$ (with ties resolved arbitrarily) is
all we need to compute $j'$.  Note now that to show a common offset
$\delta_{\S}$ above, we only used the fact that suffixes shared a
prefix of length at least $3\tau$. Thus, by $3\tau \leq \ell$, here we
also have that all $j'' \in \Occ_{\ell}(\SA[i])$ share a common offset
$\delta_{\S} = \Successor_{\S}(\SA[i]) - \SA[i]$ such that $j'' +
\delta_{\S} = \min(\S \cap [j'' \dd j'' + \tau))$. Consequently, to
find $j'$ we take some $q \geq 2\ell$ and:
\begin{enumerate}
\item First, letting $\delta_{\S} = \Successor_{\S}(\SA[i]) - \SA[i]$,
  we compute the number $m$ of positions $s' \in \S$ that have their
  reversed left length-$q$ context in the range $[X \dd X')$ (where
  $\revstr{X} = \T[\SA[i] \dd \SA[i] + \delta_{\S})$ and $X' =
  Xc^{\infty}$) and their right length-$q$ context in $[\eps \dd Y)$
  (where $Y = \T^{\infty}[\SA[i] + \delta_{\S} \dd \SA[i] + \ell)$).
\item Then, for any $k \in [1 \dd |\Occ_{\ell}(\SA[i])|]$, the
  $(m+k)$th element in the sequence of all $s' \in \S$ sorted by the
  length-$2\ell$ right context which simultaneously has its left
  length-$q$ context in $[X \dd X')$, satisfies $\T^{\infty}[s' -
  \delta_{\S} \dd s' - \delta_{\S} + 2\ell) = \T^{\infty}[\SA[b + k]
  \dd \SA[b + k] + 2\ell)$.  In particular, the position $s'$ for $k =
  i - b$ satisfies $\T^{\infty}[s' - \delta_{\S} \dd s' - \delta_{\S}
  + 2\ell) = \T^{\infty}[\SA[i] \dd \SA[i] + 2\ell)$, i.e., $s' -
  \delta_{\S} \in \Occ_{2\ell}(\SA[i])$. Thus, finding $j'$ reduces to
  a range selection query on $\Pts$.
\end{enumerate}

The above reduction is proved in \cref{lm:sa-nonperiodic-occ-pos}.
One last detail is that we need $X$, $Y$, and $\delta_{\S}$. We note,
however, that they all depend only on $\T^{\infty}[\SA[i] \dd \SA[i] +
\ell)$, and thus can be computed using $j \in \Occ_{\ell}(\SA[i])$
(which we have as input). We have thus proved the following result,
which combined with \cref{pr:to-nonperiodic-first} concludes the
description of the $\SA$ query for nonperiodic $\SA[i]$.

\begin{proposition}
  Under \cref{as:nonperiodic}, assuming $\SA[i] \in [1 \dd n] \sm
  \R(\tau,\T)$, given a position $j \in \Occ_{\ell}(\SA[i])$ and the
  pair $(\LB_{\ell}(\SA[i]), \UB_{\ell}(\SA[i]))$, we can efficiently
  find some $j' \in \Occ_{2\ell}(\SA[i])$.
\end{proposition}

\paragraph{The Periodic Positions}

Assume $\SA[i] \in \R(\tau,\T)$.  The standard way to introduce
structure among periodic positions (see, e.g.,~\cite{sss}) is as
follows.  Observe that if $j, j + 1 \in \R(\tau,\T)$, then $\per(\T[j
\dd j {+} 3\tau {-} 1)) = \per(\T[j {+} 1 \dd j {+} 1 {+} 3\tau {-}
1))$. This implies that any maximal block of positions in
$\R(\tau,\T)$ defines a highly periodic fragment of $\T$ (called a
``run'') with an associated period $p$. To compare runs, we ``anchor''
each run $\T[j \dd j')$ by selecting some $H \in \Sigma^{+}$, called
its \emph{root}, so that $\T[j \dd j')$ is a substring of $H^{\infty}$
and no nontrivial rotation of $H$ is selected for other runs.  Every
run $\T[j \dd j')$ can then be uniquely written as $\T[j \dd j') =
H'H^{k}H''$, where $k \geq 1$ and $H'$ (resp.\ $H''$) is a proper
suffix (resp.\ prefix) of $H$.  The value $k$ is called the
\emph{exponent} of $j$. We finally classify $\type(j) {=} -1$ if
$\T[j'] \,{\prec}\, \T[j'-p]$ and $\type(j) = +1$ otherwise, where
$p=|H|$.

\begin{figure}[t!]
  \vspace{-0.5cm}
  \begin{center}
    \begin{tikzpicture}[scale=0.25]

      \fill[pattern=north west lines, pattern color=blue!50]
        (0,1) rectangle (40,17);
      \fill[pattern=north east lines, pattern color=red!50]
        (0,4) rectangle (40,11);

      \scriptsize
      \draw[|-|] (44,9+0.05) -- (44,11-0.05) node[right,midway]
           {$\Poshigh(\SA[i])$};
      \draw[|-|] (44,11+0.05) -- (44,17-0.05) node[right,midway]
           {$\Posalll(\SA[i])$};
      \draw[|-|] (-1,9+0.05) -- (-1,14-0.05) node[left,midway]
           {$\Posmid(\SA[i])$};
      \draw[|-|] (-1,14+0.05) -- (-1,17-0.05) node[left,midway]
           {$\Poslow(\SA[i])$};

      \foreach \v [count=\x from 0] in
          {13,15,16,20,30,35,38,35,33,31,30,29,23,18,16,16,15,14,13,12}{
        \fill[black!25] (0,\x) rectangle (\v, \x+1);
        \pgfmathtruncatemacro{\d}{0.25*(\v-5)}
        \foreach \r in {0,1,...,\d}{
          \draw[thin,fill=black!50] (1+4*\r,\x) rectangle
            node[white]{\scriptsize$H$} (5+4*\r, \x+1);
        }
        \draw (40, \x) -- (0,\x) -- (0,\x+1) -- (40, \x+1);
        \draw[dotted] (40, \x) -- (43, \x) (40, \x+1) -- (43, \x+1);
      }
      \draw[dotted] (0,20) -- (0,21.5);
      \draw[densely dotted] (15, 0) -- (15,21.5) (30, 0) -- (30,21.5);

      \draw[latex-latex] (0,21) -- node[above]{$\ell$} (15, 21);
      \draw[latex-latex] (15,21) -- node[above]{$\ell$} (30, 21);

      \draw[ultra thick] (0,9) rectangle (40,10);

    \end{tikzpicture}
  \end{center}
  \vspace{-.5cm}
  \caption{Example showing the sets $\Posalll(\SA[i])$,
    $\Poslow(\SA[i])$, $\Posmid(\SA[i])$, and
    $\Poshigh(\SA[i])$. Note, that $|\Posalll(\SA[i])| =
    |\Poslow(\SA[i])| + |\Posmid(\SA[i])| - |\Poshigh(\SA[i])|$ (see
    \cref{lm:delta}). The suffix $\T[\SA[i] \dd n]$ is highlighted in
    bold. The sets $\Occ_{\ell}(\SA[i])$ and $\Occ_{2\ell}(\SA[i])$
    and marked with blue and red (respectively). }
  \label{fig:ranges}
\end{figure}

The first major challenge in the periodic case is selecting roots
efficiently.  An easy solution in the static case
(e.g.~\cite{sss,sss-index}) is to choose the lexicographically
smallest rotation of $H$ (known as the \emph{Lyndon root}). This seems
very difficult in the dynamic case, however.  We instead show a
construction that exploits the presence of symmetry-breaking component
in the signature parsing (i.e., the deterministic coin
tossing~\cite{Cole1986}) to develop a custom computation of roots
(\cref{sec:ccs}).

Denote $b = \LB_{\ell}(\SA[i])$ and $e = \UB_{\ell}(\SA[i])$. Observe
that all $i' \in (b \dd e]$ with $\type(\SA[i'])<0$ precede those with
$\type(\SA[i'])>0$. Moreover, $\Lexp(\SA[i'])$ among those with
$\type(\SA[i'])<0$ (resp.\ $\type(\SA[i'])>0$) is increasing
(resp.\ decreasing) as we increase $i'$; see also \cref{lm:utils}.
This structure of the periodic block has led to relatively
straightforward processing of periodic positions in previous
applications (e.g., the BWT construction in~\cite{sss}), where the
positions were first grouped by the type, and then by the exponent. In
our case, however, we need to exclude the positions with very small
and very large exponents. We thus employ the following modification to
the above scheme: Rather than computing $|\Posalll(\SA[i])|$ in one
step, we prove (\cref{lm:delta}) that it can be expressed as a
combination of three sets, $\Poslow(\SA[i])$ (containing exponents
``truncated'' at length $\ell$), $\Poshigh(\SA[i])$ (where truncation
occurs for length $2\ell$), and $\Posmid(\SA[i])$ (all exponents in
between); see \cref{fig:ranges} for an example. We then compute the
following values: the type $\type(\SA[i])$, the size
$|\Poslow(\SA[i])|$, the exponent $\Lexp(\SA[i])$, the size
$|\Posmid(\SA[i])|$, a position $j'\in \Occ_{2\ell}(\SA[i])$, the size
$|\Poshigh(\SA[i])|$, and the size $|\Occ_{2\ell}(\SA[i])|$.  See the
proof of \cref{pr:sa-periodic} for the high-level algorithm and
\cref{sec:sa-periodic-occ-size,sec:sa-periodic-type,sec:sa-periodic-poslow-poshigh,sec:sa-periodic-exp,sec:sa-periodic-posmid,sec:sa-periodic-occ-pos}
for the implementation of subsequent steps.  Finally, we derive
$|\Posalll(\SA[i])|$, which equals
$|\Poslow(\SA[i])|+|\Posmid(\SA[i])|-|\Poshigh(\SA[i])|$, if
$\type(\SA[i])=-1$, and the values
$\LB_{2\ell}(\SA[i])=\LB_\ell(\SA[i])+|\Posalll(\SA[i])|$ as well as
$\UB_{2\ell}(\SA[i])=\LB_{2\ell}(\SA[i])+|\Occ_{2\ell}(\SA[i])|$.  The
case of $\type(\SA[i])=+1$ is symmetric: we then compute
$|\Posallg(\SA[i])|$ instead of $|\Posalll(\SA[i])|$.

\paragraph{Dynamic Text Implementation}

In the second part of the paper
(\cref{sec:parsing,sec:ds,sec:bsp2sss,sec:app}), we develop a data
structure that maintains a dynamic text (subject to updates listed
below) and provides efficient implementation of the auxiliary queries
specified in the four assumptions of \cref{sec:sa}.

\begin{definition}\label{def:dt}
  We say that a \emph{dynamic text} over an integer alphabet $\Sigma$
  (with $\$ = \min \Sigma$) is a data structure that maintains a text
  $T\in \Sigma^+$ subject to the following updates:
\begin{description}
  \item[$\initop(\sigma)$:] Given the alphabet size $\sigma$,
    initialize the data structure, setting $T:=\$$;
  \item[$\makeop(i, a)$:] Given $i\in [1\dd |T|]$ and $a\in \Sigma \sm
    \{\$\}$, set $T := T[1\dd i) \cdot a \cdot T[i\dd |T|]$.
  \item[$\deleteop(i)$:] Given $i\in [1\dd |T|)$, set $T:= T[1\dd
    i)\cdot T(i\dd |T|]$.
  \item[$\swapop(i,j,k)$:] Given $i,j,k\in [1\dd |T|]$ with $i\le j
    \le k$, set $T:= T[1\dd i) \cdot T[j\dd k) \cdot T[i\dd j) \cdot
    T[k\dd |T|]$.
\end{description}
\end{definition}
Observe that our interface does not directly support character
substitutions; this is because $\mathsf{substitute}(i,a)$ can be
implemented as $\deleteop(i)$ followed by $\makeop(i, a)$.  Moreover,
note that the interface enforces the requirement of \cref{sec:sa} that
$\{i\in [1\dd |T|] : T[i]=\$\} = \{|T|\}$.

Various components of our data structure, described in \cref{sec:app},
maintain auxiliary information associated to individual positions of
the text $T$ (such as whether the position belongs to a synchronizing
set). Individual updates typically alter the positions of many
characters in $T$ (for example, insertion moves the character at
position $j$ to position $j+1$ for each $j\in [i\dd |T|]$), so
addressing the characters by their position is volatile. To address
this issue, we use a formalism of \emph{labelled strings}, where each
character is associated to a unique label that is fixed throughout
character's lifetime. This provides a clean realization of the concept
of \emph{pointers to characters}, introduced in~\cite{Alstrup2000}
along with $\Oh(\log n)$-time conversion between labels (pointers) and
positions.

The main challenge that arises in our implementation of a dynamic text
is to maintain synchronizing sets.  As for the $\Successor_{\S}$
queries alone, we could generate synchronizing positions at query
time.  Nevertheless, \cref{as:nonperiodic} also entails supporting
range queries over a set of points in one-to-one correspondence with
the synchronizing positions, we cannot evade robust maintenance.

\paragraph{Dynamic Synchronizing Sets}

By the consistency condition in \cref{def:sss}, whether a position $j$
belongs to a synchronizing set $\S$ is decided based on the right
context $T[j\dd j+2\tau)$.  This means that the entire synchronizing
set can be described using a family $\mathcal{F} \sub \Sigma^{2\tau}$
such that $j\in \S$ if and only if $T[j\dd j+2\tau)\in \mathcal{F}$.
The construction algorithms in~\cite{sss} select such $\mathcal{F}$
adaptively (based on the contents of $T$) to guarantee that $|\S|$ is
small.  In a dynamic scenario, we could either select $\mathcal{F}$
non-adaptively and keep it fixed; or adaptively modify $\mathcal{F}$
as the text $T$ changes.

The second approach poses a very difficult task: the procedure
maintaining $\mathcal{F}$ does not know the future updates, yet it
needs to be robust against any malicious sequence of updates that an
adversary could devise. This is especially hard in the deterministic
setting, where we cannot hide $\mathcal{F}$ from the adversary.  Thus,
we aim for non-adaptivity, which comes at the price of increasing the
synchronizing set size by a factor of $\Oh(\log^*(\sigma \tau))$.  On
the positive side, a non-adaptive choice of $\mathcal{F}$ means that
$\S$ only undergoes local changes; for example, a substitution of
$T[i]$ may only affect $\S\cap (i-2\tau\dd i]$.  Moreover, since
$\mathcal{F}$ yields small synchronizing sets for all strings, this is
in particular true for all substrings $T[i\dd j)$, whose synchronizing
sets are in one-to-one correspondence with $\S\cap [i\dd j-2\tau]$.
This means that $\S$ is \emph{locally sparse} and that each update
incurs $\Oh(\log^*(\sigma\tau))$ changes to $\S$.

The remaining challenge is thus to devise a non-adaptive synchronizing
set construction.  Although all existing constructions are adaptive,
Birenzwige, Golan, and Porat~\cite{BirenzwigeGP20} provided a
non-adaptive construction of a related notion of \emph{partitioning
sets}. While partitioning sets and synchronizing sets satisfy similar
consistency conditions, the density condition of synchronizing sets is
significantly stronger, which is crucial for a clean separation
between periodic and nonperiodic positions. Thus, in
\cref{sec:bsp2sss}, we strengthen the construction
of~\cite{BirenzwigeGP20} so that it produces synchronizing sets. This
boils down to `fixing' the set in the vicinity of positions in
$\R(\tau,T)$.

\paragraph{Dynamic Strings over Balanced Signature Parsing}

Unfortunately, the approach of~\cite{BirenzwigeGP20} only comes with a
static implementation.  Thus, we need to dive into their techniques
and provide an efficient dynamic implementation.  Their central tool
is a locally consistent parsing algorithm that iteratively parses the
text using deterministic coin tossing~\cite{Cole1986} to determine
phrase boundaries. Similar techniques have been used many times (see
e.g.~\cite{SahinalpV94,SahinalpV96,Mehlhorn,Alstrup2000,NishimotoDAM}),
but the particular flavor employed in~\cite{BirenzwigeGP20} involves a
mechanism that, up to date, has not been adapted to the dynamic
setting. Namely, as subsequent levels of the parsing provide coarser
and coarser partitions into phrases, the procedure grouping phrases
into blocks (to be merged in the next level) takes into account the
lengths of the phrases, enforcing very long phrases to form
single-element blocks.  This trick makes the phrase lengths much more
balanced, which is crucial in controlling the context size that
governs the local consistency of the synchronizing sets derived from
the parsing.

In \cref{sec:parsing}, we develop \emph{balanced signature parsing}: a
version of signature parsing (originating the early works on dynamic
strings~\cite{Mehlhorn,Alstrup2000}) that involves the phrase
balancing mechanism. Then, in \cref{sec:ds}, we provide a dynamic
strings implementation based on the balanced signature parsing.  The
main difference compared to the previous
work~\cite{Mehlhorn,Alstrup2000,dynstr} stems from the fact that the
size of the \emph{context-sensitive} part of the parsing (that may
change depending on the context surrounding a given substring) is
bounded in terms of the number of individual letters rather than the
number of phrases at the respective level of the parsing.  Due to
this, we need to provide new (slightly modified) implementations of
the basic operations (such as updates and longest common prefix
queries). However, we also benefit from this feature, obtaining faster
running times for more advanced queries, such as the period queries
(which we utilize in \cref{sec:bsp2sss} to retrieve $\R(\tau,T)$)
compared to solutions using existing dynamic strings
implementations~\cite{CKW20}.

\section{Generalized Range Counting and Selection Queries}\label{sec:range-queries}

In this section we introduce generalized range counting and selection
queries. The generalization lies in the fact that the ``coordinates''
of points can come from any ordered set. In particular, coordinates of
points in some of our structures will be elements of $\Zp$ or
$\Sigma^{*}$ (i.e., the strings over alphabet $\Sigma$). Furthermore,
the points in our data structures are labelled with distinct integer
identifiers; we allow multiple points with the same coordinates,
though.

The section is organized as follows. We start with the definition of
range counting/selection queries. In the following two sections
(\cref{sec:str-str,sec:int-str}), we present two data structures
supporting efficient range counting/selection queries over instances
that will be of interest to us.

Let $\X$ and $\Y$ be some linearly ordered sets (we denote the order
on both sets using $\prec$ or $\preceq$).  Let $\Pts \sub \X\times \Y
\times \Z$ be a finite set of points with distinct integer labels.  We
define the notation for \emph{range counting} and \emph{range
selection} queries as follows.

\begin{description}[style=sameline,itemsep=1ex]
\item[Range counting query $\rcount{\Pts}{X_l}{X_u}{Y_u}$:] Given
  $X_l, X_r {\in} \X$ and $Y_u {\in} \Y$, return $|\{(X,Y,\ell)\in
  \Pts : X_l \preceq X \prec X_r\text{ and }Y \prec Y_u\}|$. We also
  let $\rcounti{\Pts}{X_l}{X_r}{Y_u} = |\{(X,Y,\ell)\in \Pts : X_l
  \preceq X \prec X_r\text{ and } Y\preceq Y_r\}|$ and
  $\rcountb{\Pts}{X_l}{X_r} = |\{(X,Y,\ell) \in \Pts : X_l \preceq X
  \prec X_r\}|$.
\item[Range selection query $\rselect{\Pts}{X_{l}}{X_{u}}{r}$:] Given
  $X_{l}, X_{u} \in \mathcal{X}$ and $r\in [1 \dd
  \rcountb{\Pts}{X_{\ell}}{X_{r}}]$, return any $\ell \in
  \rselect{\Pts}{X_l}{X_u}{r} := \{\ell : (X,Y,\ell) \in
  \Pts\text,\ X_l \preceq X \prec X_r,\text{ and }Y = Y_u\}$, where
  $Y_u \in \Y$ is such that $r \in (\rcount{\Pts}{X_l}{X_r}{Y_u} \dd
  \rcounti{\Pts}{X_l}{X_r}{Y_u}]$.
\end{description}

\begin{theorem}\label{thm:range}
  Suppose that the elements of $\X$ and $\Y$ can be compared in
  $\Oh(t)$ time.  Then, there is a deterministic data structure that
  maintains a set $\Pts\sub \X\times \Y \times [0\dd 2^w)$ of size
  $n$, with insertions in $\Oh((t + \log n)\log n)$ time and deletions
  in $\Oh(\log^2 n)$ time so that range queries are answered in
  $\Oh((t + \log^2 n)\log n)$ time.
\end{theorem}
\begin{proof}
  The set $\Pts$ is stored in a data structure of~\cite{Lueker1982}
  (see also~\cite{Willard1985,Chazelle1988}) for dynamic range
  counting queries; this component supports updates and queries in
  $\Oh(\log^2 n)$ assuming $\Oh(1)$-time comparisons.  As for range
  selection queries, we resort to binary search with range counting
  queries as an oracle (the universe searched consists of the second
  coordinates of all points in $\Pts$).  Thus, range selection queries
  cost $\Oh(\log^3 n)$ time.

  In order to substantiate the assumption on constant comparison time,
  we additionally maintain $\Pts$ in two instances of the
  order-maintenance data structure~\cite{Dietz1987,Bender2017}, with
  points $(X,Y,\ell)$ ordered according to $X$ in the first instance
  and according to $Y$ in the second instance (ties are resolved
  arbitrarily). This allows for $\Oh(1)$-time comparisons between
  points in $\Pts$.  The overhead for deletions is $\Oh(1)$, but
  insertions to the order-maintenance structure require specifying the
  predecessor of the newly inserted element; we find the predecessor
  in $\Oh(t\log n)$ time using binary search.  Similarly, at query
  time, we temporarily add the query coordinates ($X_{l}$, $X_r$, and,
  if specified, $Y_u$) to the appropriate order-maintenance structure,
  also at the cost of an extra $\Oh(t\log n)$ term in the query time.
\end{proof}

\subsection{A String-String Instance}\label{sec:str-str}

\begin{definition}\label{def:p-context}
  Let $\T \in \Sigma^n$. For any $q \in \Zp$ and any
  set $\mathsf{P} \sub [1 \dd n]$, we define
  \[
    \Points_{q}(\T, \mathsf{P}) = \{(\revstr{\T^{\infty}[p \, {-} \, q
    \,{\dd}\, p)}, \T^{\infty}[p \, {\dd} \, p \,{+}\, q),p) : p \in
    \mathsf{P}\}.
  \]
\end{definition}

In other words, $\Points_{q}(\T, \mathsf{P})$ represents the
collection of string-pairs $(X, Y)$ composed of reversed length-$q$
left context and a length-$q$ right context (in $\T^{\infty}$) of
every $p \in \mathsf{P}$, and for any $(X, Y, \ell) \in
\Points_{q}(\T, \mathsf{P})$, the label $\ell \in \mathsf{P}$ is the
underlying position.  Equivalently, $\Points_{q}(\T, \mathsf{P})$ can
be interpreted as a set of labelled points $(X,Y,\ell)$ with
coordinates $X,Y$ and an integer label $\ell$.  The order among
strings on the $\mathcal{X} = \Sigma^{*}$ and $\mathcal{Y} =
\Sigma^{*}$ axis is the lexicographical order.

Next, we define the problem of supporting range counting/selection
queries on $\Points_{q}(\T, \mathsf{P})$ for some special family of
queries that can be succinctly represented as substrings of
$\T^{\infty}$ or $\revstr{\T^{\infty}}$.

\begin{problem}[String-String Range Queries]\label{prob:str-str}
  Let $\T \in \Sigma^n$, $q \in [1 \dd 3n]$, and $\mathsf{P} \sub [1
  \dd n]$ be a set satisfying $|\mathsf{P}| = m$. Denote $\Pts =
  \Points_{q}(\T, \mathsf{P})$ and $c = \max\Sigma$. Provide efficient
  support for the following queries:
  \begin{enumerate}
  \item\label{prob:str-str-it-1} Given $i \in [1 \dd n]$ and $q_{l},
    q_r \in [0 \dd 2n]$, return $\rcount{\Pts}{X_l}{X_u}{Y_l}$ and
    $\rcount{\Pts}{X_l}{X_u}{Y_u}$, where $\revstr{X_l} =
    \T^{\infty}[i - q_{l} \dd i)$, $X_u = X_lc^{\infty}$, $Y_l =
    \T^{\infty}[i \dd i + q_r)$, and $Y_u = Y_lc^{\infty}$, and
  \item\label{prob:str-str-it-2} Given $i \,{\in}\, [1 \dd n]$, $q_{l}
    \,{\in}\, [0 \dd 2n]$, and $r \,{\in}\, [1 \dd
    \rcountb{\Pts}{X_l}{X_u}]$ (where \mbox{$\revstr{X_l} \,{=}\,
    \T^{\infty}[i \,{-}\, q_{l} \dd i)$} and $X_u = X_lc^{\infty}$),
    return some position $p \in \rselect{\Pts}{X_l}{X_u}{r}$.
  \end{enumerate}
\end{problem}

\subsection{An Int-String Instance}\label{sec:int-str}

\begin{definition}\label{def:p-right-context}
  Let $\T \in \Sigma^n$, $q \in \Zp$, and suppose that $\mathsf{P}
  \sub [1 \dd n] \times \Zz$ contains pairs with distinct first
  coordinates. We define
  \[
    \Points_{q}(\T, \mathsf{P})=\{(d, \T^{\infty}[p \,{\dd}\, p
    \,{+}\, q),p) : (p, d) \in \mathsf{P}\}.
  \]
\end{definition}

In other words, $\Points_{q}(\T, \mathsf{P})$ represents the
collection of pairs $(X, Y)$ composed of an integer $X = d$ and a
length-$q$ right context (in $\T^{\infty}$) of position $p$ for every
$(p, d) \in \mathsf{P}$.  The label of each point is set to $p$; the
assumption on $\mathsf{P}$ guarantees that points have distinct
labels.  Equivalently, $\Points_{q}(\T, \mathsf{P})$ can be
interpreted as a set of labelled points from $\mathcal{X} \times
\mathcal{Y}\times \Z$, where the coordinates of each point
$(X,Y,\ell)\in \Pts$ are an integer $X$, and a string $Y$, whereas the
label is $\ell$.  We use the standard order on $\mathcal{X} = \Zp$ and
the lexicographic order on $\mathcal{Y} = \Sigma^{*}$.

Next, we define the problem of efficiently supporting range counting
and selection queries on $\Points_{q}(\T, \mathsf{P})$ for some
restricted family of queries that can be succinctly represented as
substrings of $\T^{\infty}$.

\begin{problem}[Int-String Range Queries]\label{prob:int-str}
  Let $\T \in \Sigma^n$, $q \in [1 \dd 3n]$, and $\mathsf{P} \sub [1
  \dd n] \times [0 \dd n]$ be a set of size $|\mathsf{P}| = m$.
  Denote $\Pts = \Points_{q}(\T, \mathsf{P})$ and $c =
  \max\Sigma$. Provide efficient support for the following queries:
  \begin{enumerate}
  \item\label{prob:int-str-it-1} Given $i \in [1 \dd n]$, $x \in [0
    \dd n]$, and $q_r \in [0 \dd 2n]$, return
    $\rcount{\Pts}{x}{n}{Y_l}$, $\rcount{\Pts}{x}{n}{Y_u}$, and
    $\rcountb{\Pts}{x}{n}$, where $Y_l = \T^{\infty}[i \dd i + q_r)$,
    $Y_u = Y_lc^{\infty}$, and
  \item\label{prob:int-str-it-2} Given $x \in [0 \dd n]$ and $r \in [1
    \dd \rcountb{\Pts}{x}{n}]$, return some position $p \in
    \rselect{\Pts}{x}{n}{r}$.
  \end{enumerate}
\end{problem}

\section{Modular Constraint Queries}\label{sec:mod-queries}

Let $\Ints \sub \Zz^2 \times \Z$ be a finite set of integer intervals
(represented by endpoints) with distinct integer labels (we allow
multiple intervals with the same endpoints).  We define the
\emph{modular constraint counting} and \emph{modular constraint
selection} queries as follows:

\begin{description}[style=sameline,itemsep=1ex]
\item[Modular constraint count query $\mcount{\mathcal{I}}{h}{r}{q}$:]
  Given $h \in \Zp$, $r \in [0 \dd h)$ and $q \in \Zz$, return
  $\sum_{(b,e,\ell) \in \Ints}|\{j \in [b \dd e) : j \bmod h = r\text{
  and }\lfloor \tfrac{j}{h} \rfloor \leq q\}|$.  As a notational
  convenience, we define $\mcountb{\mathcal{I}}{h}{r} =
  \sum_{(b,e,\ell) \in \Ints}|\{j \in [b \dd e) : j \bmod h = r\}|$.
\item[Modular constraint selection query
  $\mselect{\mathcal{I}}{h}{r}{c}$:] Given $h \in \Zp$, $r \in [0 \dd
  h)$ and $c \in [1 \dd \mcountb{\mathcal{I}}{h}{r}]$, return (the
  unique) positive integer $q$ satisfying the condition $c \in
  (\mcount{\mathcal{I}}{h}{r}{q-1} \dd
  \mcount{\mathcal{I}}{h}{r}{q}]$.
\end{description}

We extend the above notation to sets of labelled one-sided intervals
(identified for simplicity with integers) as follows.  For any finite
set $Q \sub \Zz \times \Z$ of coordinates labelled with distinct
integer labels (we allow equal coordinates), we define
$\mcount{Q}{h}{r}{q} := \mcount{\Ints}{h}{r}{q}$, $\mcountb{Q}{h}{r}
:= \mcountb{\Ints}{h}{r}$, and $\mselect{Q}{h}{r}{c} :=
\mselect{\Ints}{h}{r}{c}$, where $\Ints = \{(0, e, \ell) : (e, \ell)
\in Q\}$.

\begin{proposition}\label{pr:mod-queries}
  Let $h \in [0 \dd 2^w)$. There is a deterministic data structure
  that maintains a set of labelled coordinates $Q \sub [0 \dd 2^w)
  \times \Z$ of size $|Q| = n$ so that updates $($insertions and
  deletions in $Q)$ and modular constraint counting queries
  $($returning $\mcount{Q}{h}{r}{q}$, given any $r \in [0 \dd h)$ and
  $q \in \Zz)$ take $\bigO(\log^2 n)$ time.
\end{proposition}
\begin{proof}

  Observe that for any $e \in \Zz$, letting $x := e \bmod n$ and $y :=
  \lfloor \tfrac{e}{h} \rfloor$, it holds
  \begin{equation}\label{eq:mod-queries}
    |\{j \in [0 \dd e) : j \bmod h = r\text{ and }\lfloor
    \tfrac{j}{h} \rfloor \leq q\}| = \min(y, q) + \delta,
  \end{equation}
  where $\delta = 1$ if and only if $r \leq x$ and $y \leq q$.  The
  basic idea is to maintain two data structures, each responsible for
  computing the aggregate value (summed over all coordinates in $Q$)
  of one of the terms on the right-hand side of \cref{eq:mod-queries}.
  The key property of this separation is that computing the sum of the
  first terms does not depend on the query argument $r$, and hence can
  be reduced to a prefix sum query. On the other hand, the second term
  for each integer in $Q$ contributes either zero or one to the total
  sum, and hence can be reduced to an orthogonal range counting query.

  Denote $Q = \{(e_1, \ell_1), \ldots, (e_k, \ell_k)\}$, where $e_1
  \leq \dots \leq e_k$. For any $i \in [1 \dd k]$, let $x_i = e_i
  \bmod h$ and $y_i = \lfloor \tfrac{e_i}{h} \rfloor$.  Denote
  $\mathcal{S} = ((y_1, \ell_1), \ldots, (y_k, \ell_k))$ and $\Pts =
  \{(x_i, y_i, \ell_i) : i \in [1 \dd k]\}$. Then, by
  \cref{eq:mod-queries}, if we let $j \in [0 \dd k]$ be the largest
  integer satisfying $\max\{y_1, \ldots, y_j\} < q$, it holds
  \begin{align*}
    \mcount{Q}{h}{r}{q}
      &= \sum_{i=1}^{k}\min(y_i, q) +
         {\big|}\{i \in [1 \dd k] : r \leq x_i < h\text{ and }y_i
         \leq q\}{\big|}\\
      &= \sum_{i=1}^{j}y_i + q(k-j) +
         \rcount{\Pts}{r}{h}{q {+} 1}.
  \end{align*}

  To compute the first two terms, we store elements $(y_i,\ell_i)$ of
  $\mathcal{S}$ using a balanced binary search tree (e.g., an AVL
  tree~\cite{AVL}) ordered by $y_i$. A node $v_i$ corresponding to
  $(y_i,\ell_i)$ is augmented with the value $\sum_{j \in J(v_i)}y_j$,
  where $J(v_i)$ contains all $j \in [1 \dd k]$ such $v_j$ is in the
  subtree rooted in $v_i$. Given such representation of $\mathcal{S}$,
  we can compute $j$ together with $\sum_{i=1}^{j}y_i$ in $\bigO(\log
  n)$ time. To compute the third term, we apply \cref{thm:range} for
  $\Pts$. The query takes $\bigO(\log^2 n)$ time.

  To insert or delete a labelled coordinate $(v,\ell) \in [0 \dd 2^w)
  \times \Z$ from $Q$, we first compute $x = v \bmod h$ and $y =
  \lfloor \tfrac{v}{h} \rfloor$. Adding/deleting $(y,\ell)$ from the
  first component of our structure is straightforward and takes
  $\bigO(\log n)$ time. The extra values in each node can be computed
  using the information stored in the node and its children, and thus
  can be maintained during rotations. By \cref{thm:range}, insertions
  and deletions on the second structure take $\bigO(\log^2 n)$ time.
\end{proof}

\begin{corollary}\label{cor:mod-queries}
  Let $h, n \in [0 \dd 2^w)$. There is a deterministic data structure
  that maintains a set of labelled intervals $\Ints \sub [0 \dd n)^2
  \times \Z$ of size $|\Ints| \leq n$ so that updates $($insertions
  and deletions in $\Ints)$ take $\bigO(\log^2 n)$ time and modular
  constraint counting (resp.\ selection) queries are answered in
  $\bigO(\log^2 n)$ (resp.\ $\bigO(\log^3 n)$) time.
\end{corollary}
\begin{proof}

  We keep two instances of the structure from \cref{pr:mod-queries},
  one for the set of labelled coordinates $Q^{-}$ containing left
  endpoints of intervals in $\mathcal{I}$ and one for the set $Q^{+}$
  of right endpoints.

  Given any $r \in [0 \dd h)$ and $q \in \Zz$, we return
  $\mcount{\mathcal{I}}{h}{r}{q} = \mcount{Q^{+}}{h}{r}{q} -
  \mcount{Q^{-}}{h}{r}{q}$ in $\bigO(\log^2 n)$ time. The select query
  on $\mathcal{I}$ is implemented in $\bigO(\log^3 n)$ time by using
  binary search (in the range $[0 \dd n)$) and modular constraint
  counting queries.

  Updates simply insert/delete the left (resp.\ right) endpoint of the
  input interval into the structure for $Q^{-}$ (resp.\ $Q^{+}$).  By
  \cref{pr:mod-queries}, each update to $\mathcal{I}$ takes
  $\bigO(\log^2 n)$ time.
\end{proof}

\section{\texorpdfstring{$\SA$}{SA} Query Algorithm}\label{sec:sa}

Let $\T \in \Sigma^n$. In this section we show that under some small
set of assumptions about the ability to perform queries on string
synchronizing sets~\cite{sss}, we can perform $\SA$ queries (i.e.,
returning the value $\SA[i]$, given any position $i \in [1 \dd n]$)
for $\T$.

Let $\ell \geq 1$. For any $j \in [1 \dd n]$, we define
\begin{align*}
  \Occ_{\ell}(j)
    &= \{j' \in [1 \dd n] : \T^{\infty}[j' \dd j' + \ell) = 
       T^{\infty}[j \dd j + \ell)\},\\
  \LB_{\ell}(j)
    &= |\{j' \in [1 \dd n] : \T[j' \dd n] \prec \T[j \dd n]
       \text{ and }\LCE_{\T}(j, j) < \ell\}|, \text{ and }\\
  \UB_{\ell}(j)
    &= \LB_{\ell}(j) + |\Occ_{\ell}(j)|.
\end{align*}
For the motivation of the above names, note that viewing $P :=
\T^{\infty}[j \dd j + \ell)$ as a pattern, we have $\{\SA[i] : i \in
(\LB_{\ell}(j) \dd \UB_{\ell}(j)]\} = \{i \in [1 \dd n] :
\T^{\infty}[i \dd i + |P|) = P]\}$.

Moreover, for any $j \in [1 \dd n]$, we define
\begin{align*}
  \Posalll(j) &=
    \{j' \in [1 \dd n] : \T[j' \dd n] \prec \T[j \dd n]
      \text{ and }\LCE_{\T}(j',j) \in [\ell \dd 2\ell)\}\text{ and }\\
  \Posallg(j) &=
    \{j' \in [1 \dd n] : \T[j' \dd n] \succ \T[j \dd n]
      \text{ and }\LCE_{\T}(j',j) \in [\ell \dd 2\ell)\}.
\end{align*}
We denote the size of $\Posalll(j)$ and $\Posallg(j)$ as
$\deltal_{\ell}(j) = |\Posalll(j)|$ and $\deltag_{\ell}(j) =
|\Posallg(j)|$.  Then, it holds $(\LB_{2\ell}(j), \UB_{2\ell}(j)) =
(\LB_{\ell}(j) + \deltal_{\ell}(j), \LB_{\ell}(j) + \deltal_{\ell}(j)
+ |\Occ_{2\ell}(j)|)$, $(\LB_{2\ell}(j), \UB_{2\ell}(j)) =
(\UB_{\ell}(j) - \deltag_{\ell}(j) - |\Occ_{2\ell}(j)|, \UB_{\ell}(j)
- \deltag_{\ell}(j))$.

\begin{assumption}\label{as:small}
  For any $i \in [1 \dd n]$, we can compute some $j \in
  \Occ_{16}(\SA[i])$ and the pair $(\LB_{16}(\SA[i]),
  \UB_{16}(\SA[i]))$ in $\bigO(t)$ time.
\end{assumption}

The main idea of our algorithm to compute $\SA[i]$ is as follows.
Suppose that we have obtained some $j \in \Occ_{16}(\SA[i])$ and
$(\LB_{16}(\SA[i]), \UB_{16}(\SA[i]))$ using \cref{as:small}. Then,
for $q = 4, \ldots, \lceil \log n \rceil - 1$, we compute
$(\LB_{2^{q+1}}(\SA[i]), (\UB_{2^{q+1}}(\SA[i]))$ and some $j' \in
\Occ_{2^{q+1}}(\SA[i])$, by using as input some position $j \in
\Occ_{2^q}(\SA[i])$ and the pair of integers $(\LB_{2^q}(\SA[i]),
\UB_{2^q}(\SA[i]))$, i.e., the output of the earlier step. This
algorithm lets us compute $\SA[i]$, since eventually we obtain some
$j' \in \Occ_{2^{\lceil \log n \rceil}}(\SA[i])$, and for any $k \geq
n$, $\Occ_{k}(\SA[i]) = \{\SA[i]\}$, i.e., the final computed position
must satisfy $j' = \SA[i]$.  To implement the above strategy, we now
need to present a query algorithm parameterized by any $\ell \geq 16$
that, given $i \in [1 \dd n]$ along with some $j \in
\Occ_{\ell}(\SA[i])$ and the pair $(\LB_{\ell}(\SA[i]),
\UB_{\ell}(\SA[i]))$ as input, returns some $j' \in
\Occ_{2\ell}(\SA[i])$ and the pair $(\LB_{2\ell}(\SA[i]),
\UB_{2\ell}(\SA[i]))$.

Let us now fix some $\ell \in [16 \dd n)$. For any $\tau \in \Zp$ we
define
\begin{equation}\label{eq:R}
  \R(\tau,\T) := \{i \in [1 \dd |\T| - 3\tau + 2] : \per(\T[i \dd i +
  3\tau - 2]) \leq \tfrac{1}{3}\tau\}.
\end{equation}

The rest of this section is organized into fours subsections. In the
first three subsections, we show that under three assumptions
(\cref{as:core,as:nonperiodic,as:periodic}) about the ability to
perform some queries in a black-box manner, given any index $i \in [1
\dd n]$ along with some position $j \in \Occ_{\ell}(\SA[i])$ and the
pair of integers $(\LB_{\ell}(\SA[i]), \UB_{\ell}(\SA[i]))$ as input,
we can efficiently compute some $j' \in \Occ_{2\ell}(\SA[i])$ and the
pair $(\LB_{2\ell}(\SA[i]), \UB_{2\ell}(\SA[i]))$ as output.  In the
first of these three subsections (\cref{sec:sa-core}) we show the
algorithm to efficiently determine if the query position $i \in [1 \dd
  n]$ satisfies $\SA[i] \in \R(\lfloor \tfrac{\ell}{3} \rfloor, \T)$.
In \cref{sec:sa-nonperiodic} (resp.\ \cref{sec:sa-periodic}) we then
present the query algorithm to efficiently compute some position $j'
\in \Occ_{2\ell}(\SA[i])$ and the pair $(\LB_{2\ell}(\SA[i]),
\UB_{2\ell}(\SA[i]))$ from $j \in \Occ_{2\ell}(\SA[i])$ and
$(\LB_{\ell}(\SA[i]), \UB_{\ell}(\SA[i]))$ for the case $\SA[i] \in [1
\dd n] \sm \R(\lfloor \tfrac{\ell}{3} \rfloor, \T)$ (resp.\ $\SA[i]
\in \R(\lfloor \tfrac{\ell}{3} \rfloor, \T)$).  All steps of the query
algorithms are put together in \cref{sec:sa-final}.

\vspace{1ex}
\begin{remark}
  Note that this section only presents a \emph{reduction} of $\SA$
  query to some other queries. To obtain the data structure for $\SA$
  queries, we need to still prove that such queries can be
  implemented, i.e., that
  \cref{as:small,as:core,as:nonperiodic,as:periodic} can be satisfied
  with some small query time $t$. We will prove this in later
  sections, and for now focus only on showing the reduction.
\end{remark}

\subsection{The Index Core}\label{sec:sa-core}

Assume $\ell \in [16 \dd n)$. In this section, we show that if for any
$i \in [1 \dd n]$ we can efficiently check if $i \in \R(\lfloor
\tfrac{\ell}{3} \rfloor, \T)$ (\cref{as:core}), then given any $j \in
\Occ_{\ell}(\SA[i])$ we can equally efficiently determine if $\SA[i]
\in \R(\lfloor \tfrac{\ell}{3} \rfloor,\T)$ (i.e., if $\SA[i]$ is a
\emph{periodic position}) or $\SA[i] \in [1 \dd n] \sm \R(\lfloor
\tfrac{\ell}{3} \rfloor,\T)$ (such $\SA[i]$ is called a
\emph{nonperiodic position}).

\begin{assumption}\label{as:core}
  For any $i \in [1 \dd n]$, we can in $\bigO(t)$ time check if $i \in
  \R(\tau,\T)$, where $\tau = \lfloor \tfrac{\ell}{3} \rfloor$.
\end{assumption}

\begin{proposition}\label{pr:sa-core}
  Let $i \in [1 \dd n]$. Under \cref{as:core}, given any position $j
  \in \Occ_{\ell}(\SA[i])$, we can in $\bigO(t)$ time compute a bit
  indicating whether $\SA[i] \in \R(\lfloor \tfrac{\ell}{3} \rfloor,
  \T)$ holds.
\end{proposition}
\begin{proof}
  Let $\tau = \lfloor \tfrac{\ell}{3} \rfloor$. Observe that for any
  $j \in [1 \dd n - 3\tau + 2]$, $j \in \R$ holds if and only if
  $\per(\T[i \dd i + 3\tau - 1)) \leq \tfrac{1}{3}\tau$.  In other
  words, $j \in \R(\tau,\T)$ depends only on the substring $\T[j \dd j
  + 3\tau - 1)$. Therefore, by $3\tau - 1 \leq \ell$ and $j \in
  \Occ_{\ell}(\SA[i])$, $\SA[i] \in \R(\tau,\T)$ holds if and only if
  $j \in \R(\tau,\T)$.
\end{proof}

\subsection{The Nonperiodic Positions}\label{sec:sa-nonperiodic}

Let $\tau = \lfloor \tfrac{\ell}{3} \rfloor$.  In this section, we
show that assuming that for some $\tau$-synchronizing set $\S$ we can
efficiently perform $\Successor_{\S}$ queries (with
$\Successor_{\S}(i) := \min\{j \in \S\cup\{|T|-2\tau+2 : j \geq i\}$
for any $i \in [1 \dd n - 2\tau + 1]$), and support some string-string
range queries (\cref{as:nonperiodic}), given a position $i \in [1 \dd
n]$ satisfying $\SA[i] \in [1 \dd n] \sm \R(\tau,\T)$, along with
some $j \in \Occ_{\ell}(\SA[i])$ and the pair $(\LB_{\ell}(\SA[i]),
\UB_{\ell}(\SA[i]))$ as input, we can efficiently compute some $j' \in
\Occ_{2\ell}(\SA[i])$ and the pair $(\LB_{2\ell}(\SA[i]),
\UB_{2\ell}(\SA[i]))$.

\begin{assumption}\label{as:nonperiodic}
  For some $\tau$-synchronizing set $\S$ of $\T$ $($where $\tau =
  \lfloor \tfrac{\ell}{3} \rfloor)$ the queries $\Successor_{\S}(i)$
  $($where $i \in [1 \dd n - 3\tau + 1] \sm \R(\tau,\T))$ and the
  string-string range queries (\cref{prob:str-str}) for text $\T$,
  integer $q = 7\tau$, and the set of positions $\mathsf{P} = \S$ can
  be supported in $\bigO(t)$ time.
\end{assumption}

\begin{remark}
  Note that by $\ell < n$ and $3\tau \leq \ell$, the value $q = 7\tau
  \leq 2\ell + \tau < 3n$ in the above assumption satisfies the
  requirement of \cref{prob:str-str}.
\end{remark}

The section is organized into four parts. First, we present
combinatorial results showing how to reduce LCE queries and prefix
conditions to substring inequalities
(\cref{sec:sa-nonperiodic-prelim}).  Next, we show how under
\cref{as:nonperiodic} to combine the properties of synchronizing sets
with these reductions to compute the cardinalities of sets
$\Posalll(j)$ and $\Occ_{2\ell}(j)$
(\cref{sec:sa-nonperiodic-pos-and-occ-size}). Then, in
\cref{sec:sa-nonperiodic-occ-pos}, we show how under the same
assumptions to efficiently compute some $j' \in \Occ_{2\ell}(\SA[i])$.
Finally, in \cref{sec:sa-nonperiodic-structure}, we put everything
together.

\subsubsection{Preliminaries}\label{sec:sa-nonperiodic-prelim}

We start with a combinatorial result that shows the three fundamental
reductions. In the first and second, we show an equivalence between
LCE queries for suffixes of $\T$ and comparisons of substrings of
$\T$. These results will be used to derive the characterization of the
set $\Posalll(j)$ and its components. In the third reduction we show
an equivalence, where LCE queries are replaced with substring
equalities. This will be used to characterize the set
$\Occ_{2\ell}(j)$.

\begin{lemma}\label{lm:utils}
  Let $j \in [1 \dd n]$ and $c = \max\Sigma$. Then:
  \begin{enumerate}
  \item\label{lm:utils-1} If $0 \leq \ell_1 < \ell_2 \leq \ell_3$ then
    for any $j' \in [1 \dd n]$, $\T[j' \dd n] \prec \T[j \dd n]$ and
    $\LCE_{\T}(j, j') \in [\ell_1 \dd \ell_2)$ holds if and only if
    $\T^{\infty}[j \dd j + \ell_1) \preceq \T^{\infty}[j' \dd j' +
    \ell_3) \prec \T^{\infty}[j \dd j + \ell_2)$.
  \item\label{lm:utils-2} If $0 \leq \ell_1 \leq \ell_2$ then for any
    $j' \in [1 \dd n]$, $\T[j' \dd n] \succeq \T[j \dd n]$ or
    $\LCE_{\T}(j, j') \geq \ell_1$ holds if and only of
    $\T^{\infty}[j' \dd j' + \ell_2) \succeq \T^{\infty}[j \dd j +
    \ell_1)$.
  \item\label{lm:utils-3} If $0 \leq \ell_1 \leq \ell_2$ then for any
    $j' \in [1 \dd n]$, $\T^{\infty}[j' \dd j' + \ell_1) =
    \T^{\infty}[j \dd j + \ell_1)$ holds if and only if $\T^{\infty}[j
    \dd j + \ell_1) \preceq \T^{\infty}[j' \dd j' + \ell_2) \prec
    \T^{\infty}[j \dd j + \ell_1)c^{\infty}$.
  \end{enumerate}
\end{lemma}
\begin{proof}

  1. Assume $\T[j' \dd n] \prec \T[j \dd n]$ and $\LCE_{\T}(j, j') \in
  [\ell_1 \dd \ell_2)$. By $\ell_1 \leq \ell_3$, this implies
  $\T^{\infty}[j \dd j + \ell_1) = \T^{\infty}[j' \dd j' + \ell_1)
  \preceq \T^{\infty}[j' \dd j' + \ell_3)$.  To show the second
  condition, denote $\ell = \LCE_{\T}(j, j')$. From $\T[j' \dd n]
  \prec \T[j \dd n]$ we obtain $j \neq j'$.  Thus, by the uniqueness
  of $\T[n] = \texttt{\$}$ in $\T$, we must have $\T[j' \dd j' + \ell)
  = \T[j \dd j + \ell)$ and $\T[j' + \ell] \prec \T[j + \ell]$, or
  equivalently, $\T[j' \dd j' + \ell + 1) \prec \T[j \dd j + \ell +
  1)$.  By $\ell + 1 \leq \ell_3$ and $\ell + 1 \leq \ell_2$, and
  since appending symbols at the end of the distinct equal-length
  substrings does not change their lexicographical order, we obtain
  $\T^{\infty}[j' \dd j' + \ell_3) \prec \T^{\infty}[j \dd j +
  \ell_2)$.

  To show the opposite implication, assume $\T^{\infty}[j \dd j +
  \ell_1) \preceq \T^{\infty}[j' \dd j' + \ell_3) \prec \T^{\infty}[j
  \dd j + \ell_2)$. First, note that $j \neq j'$, since otherwise, by
  $\ell_2 \leq \ell_3$, $\T^{\infty}[j \dd j + \ell_2)$ would be a
  prefix of $\T^{\infty}[j' \dd j' + \ell_3)$ and hence
  $\T^{\infty}[j' \dd j' + \ell_3) \succeq \T^{\infty}[j \dd j +
  \ell_2)$.  Denote $\ell = \LCE_{\T}(j, j')$. By $j \neq j'$ and the
  uniqueness of $\T[n] = \texttt{\$}$, we have $\max(j + \ell, j' +
  \ell) \leq n$. Suppose $\ell < \ell_1$ and consider two cases:
  \begin{itemize}
  \item If $\T[j + \ell] \prec \T[j + \ell]$ then $\T[j' \dd j' + \ell
    + 1) \succ \T[j \dd j + \ell + 1)$, which by $\ell + 1 \leq
    \ell_3$ and $\ell + 1 \leq \ell_2$ implies $\T^{\infty}[j' \dd j'
    + \ell_3) \succ \T^{\infty}[j \dd j + \ell_2)$, contradicting the
    assumption.
  \item On the other hand, if $\T[j + \ell] \succ \T[j' + \ell]$, then
    by $\ell + 1 \leq \ell_1$ and $\ell + 1 \leq \ell_3$ we obtain
    $\T^{\infty}[j \dd j + \ell_1) \succ \T^{\infty}[j' \dd j' +
      \ell_3)$.
  \end{itemize}
  Therefore, we must have $\ell \geq \ell_1$. Next, we note that
  $\T^{\infty}[j' \dd j' + \ell_3) \prec \T^{\infty}[j \dd j +
  \ell_2)$ implies $\T^{\infty}[j' \dd j' + \ell_2) \prec
  \T^{\infty}[j \dd j + \ell_2)$, which in turn gives $\ell <
  \ell_2$. Thus, it remains to show $\T[j' \dd n] \prec \T[j \dd
  n]$. For that it suffices to notice that knowing $\ell \in [\ell_1
  \dd \ell_2)$, the assumption $\T[j' + \ell] \succ[j + \ell]$ implies
  $\T^{\infty}[j' \dd j' + \ell + 1) \succ \T^{\infty}[j \dd j + \ell
  + 1)$, which in turn implies $\T^{\infty}[j' \dd j' + \ell_3) \succ
  \T^{\infty}[j \dd j + \ell_2)$, contradicting the assumption.

  2. Assume that it holds $\T[j' \dd n] \succeq \T[j \dd n]$ or
  $\LCE_{\T}(j, j') \geq \ell_1$. Let us first assume $\T[j' \dd n]
  \succeq \T[j \dd n]$. If $j' = j$, then by $\ell_1 \leq \ell_2$ we
  obtain $\T^{\infty}[j' \dd j' + \ell_2) \succeq \T^{\infty}[j' \dd
  j' + \ell_1) = \T^{\infty}[j \dd j + \ell_1)$.  Let us thus assume
  $j \neq j'$ and let $\ell = \LCE_{\T}(j, j')$.  By $j \neq j'$ and
  the uniqueness of $\T[n] = \texttt{\$}$, we have $\max(j + \ell, j'
  + \ell) \leq n$. Consider two cases:
  \begin{itemize}
  \item If $\ell \geq \ell_1$, then it holds $\T^{\infty}[j' \dd j' +
    \ell_2) \succeq \T^{\infty}[j' \dd j' + \ell_1) = \T[j' \dd j' +
    \ell_1) = \T[j \dd j + \ell_1) = \T^{\infty}[j \dd j + \ell_1)$,
    i.e., we obtain the claim.
  \item Otherwise, (i.e., $\ell < \ell_1$), by $\T[j' \dd n] \succeq
    \T[j \dd n]$, we must have $\T[j' + \ell] \succ \T[j +
    \ell]$. Thus, $\T[j' \dd j' + \ell + 1) \succ \T[j \dd j + \ell +
    1)$.  Appending $\ell_1 - (\ell + 1)$ symbols after the mismatch
    does not change the other between suffixes.  Thus, $T^{\infty}[j'
    \dd j' + \ell_2) \succeq \T^{\infty}[j' \dd j' + \ell_1) \succ
    \T^{\infty}[j \dd j + \ell_1)$.
  \end{itemize}
  Let us now assume $\LCE_{\T}(j, j') \geq \ell_1$. Then we obtain
  $\T^{\infty}[j' \dd j' + \ell_2) \succeq \T^{\infty}[j' \dd j' +
  \ell_1) = \T[j' \dd j' + \ell_1) = \T[j \dd j + \ell_1) =
  \T^{\infty}[j \dd j + \ell_1)$.

  To show the opposite implication, assume $\T^{\infty}[j' \dd j' +
  \ell_2) \succeq \T^{\infty}[j \dd j + \ell_1)$.  If $j = j'$, then
  we immediately obtain the claim, since $\T[j' \dd n] \succeq \T[j
  \dd n]$. Let us thus assume $j \neq j'$ and let $\ell = \LCE_{\T}(j,
  j')$. By the uniqueness of $\T[n] = \texttt{\$}$, it holds $\max(j +
  \ell, j' + \ell) \leq n$. If $\ell \geq \ell_1$, then we obtain the
  claim. Let us thus assume $\ell < \ell_1$. Then, $\T[j' + \ell] \neq
  \T[j' + \ell]$. If $\T[j' + \ell] \prec \T[j + \ell]$, then $\T[j'
  \dd j' + \ell + 1) \prec \T[j \dd j + \ell + 1)$. Appending symbols
  after the mismatch does not change the other between suffixes. Thus,
  this implies $\T^{\infty}[j' \dd j' + \ell_2) \prec \T^{\infty}[j
  \dd j + \ell_1)$, contradicting the assumption. Thus, we must have
  $\T[j' + \ell] \succ \T[j' + \ell]$. This implies $\T[j' \dd n]
  \succ \T[j \dd n]$.

  3. Assume $\T^{\infty}[j' \dd j' + \ell_1) = \T^{\infty}[j \dd j +
  \ell_1)$. Then, by $\ell_1 \leq \ell_2$, we first obtain
  $\T^{\infty}[j \dd j + \ell_1) = \T^{\infty}[j' \dd j' + \ell_1)
  \preceq \T^{\infty}[j' \dd j' + \ell_2)$. On the other hand, by $c =
  \max\Sigma$, $\T^{\infty}[j' \dd j' + \ell_2) = \T^{\infty}[j' \dd
  j' + \ell_1) \cdot \T^{\infty}[j' + \ell_1 \dd j' + \ell_2) \preceq
  \T^{\infty}[j \dd j + \ell_1) \cdot c^{\ell_2 - \ell_1} \prec
  \T^{\infty}[j \dd j + \ell_1) \cdot c^{\infty}$.

  To show the opposite implication, assume $\T^{\infty}[j \dd j +
  \ell_1) \preceq \T^{\infty}[j' \dd j' + \ell_2) \prec \T^{\infty}[j
  \dd j + \ell_1)c^{\infty}$. If $j' = j$, then we immediately obtain
  the claim. Let thus assume $j \neq j'$ and denote $\ell =
  \LCE_{\T}(j, j')$. By the uniqueness of $\T[n] = \texttt{\$}$, we
  then have $\max(j + \ell, j' + \ell) \leq n$.  Suppose $\ell <
  \ell_1$ and let us consider two cases:
  \begin{itemize}
  \item If $\T[j + \ell] \prec \T[j' + \ell]$, then we have $\T[j \dd
    j + \ell + 1) \prec \T[j' \dd j' + \ell + 1)$, which by $\ell + 1
    \leq \ell_1$ implies $\T^{\infty}[j \dd j + \ell_1) \prec
    \T^{\infty}[j' \dd j' + \ell_1)$.  Thus, in turn, by $\ell_1 \leq
    \ell_2$, implies $\T^{\infty}[j \dd j + \ell_1)c^{\infty} \prec
    \T^{\infty}[j' \dd j' + \ell_2)$, contradicting our assumption.
  \item On the other hand, if $\T[j + \ell] \succ \T[j + \ell]$, then
    $\T[j \dd j + \ell + 1) \succ \T[j' \dd j' + \ell + 1)$, which by
    $\ell + 1 \leq \ell_1$ implies $\T^{\infty}[j \dd j + \ell_1)
    \succ \T^{\infty}[j' \dd j' + \ell_1)$.  This in turn, by $\ell_1
    \leq \ell_2$, implies $\T^{\infty}[j \dd j + \ell_1) \succ
    \T^{\infty}[j' \dd j' + \ell_2)$, contradicting the assumption.
  \end{itemize}
  We have thus proved $\ell \geq \ell_1$, i.e., $\T^{\infty}[j \dd j +
  \ell_1) = \T^{\infty}[j' \dd j' + \ell_1)$.
\end{proof}

\begin{corollary}\label{cor:lbub}
  For every $j\in [1\dd n]$ and $\ell \in \Zp$, we have $\LB_{\ell}(j)
  = |\{j' \in [1 \dd n] : \T^{\infty}[j' \dd j'+\ell) \prec
  \T^\infty[j \dd j+\ell)\}|$ and $\UB_{\ell}(j) = |\{j' \in [1 \dd n]
  : \T^{\infty}[j' \dd j'+\ell) \preceq \T^\infty[j \dd j+\ell)\}| =
  \LB_{\ell}(j) + |\Occ_{\ell}(j)|$.
\end{corollary}

\subsubsection{Computing the Size of
  \texorpdfstring{$\Posalll(j)$}{Pos} and
  \texorpdfstring{$\Occ_{2\ell}(j)$}{Occ}}\label{sec:sa-nonperiodic-pos-and-occ-size}

Let $j \in [1 \dd n] \sm \R(\tau,\T)$. In this section, we show how
under \cref{as:nonperiodic} to efficiently compute the values
$|\Posalll(j)|$ and $|\Occ_{2\ell}(j)|$.

The section is organized as follows. First, we present two
combinatorial results
(\cref{lm:sa-nonperiodic-pos,lm:sa-nonperiodic-occ}) characterizing
the sets $\Posalll(j)$ and $\Occ_{2\ell}(j)$ using the string
synchronizing set $\S$. We then use these characterizations to prove a
formula for the cardinality of these sets
(\cref{lm:sa-nonperiodic-2}). We conclude with
\cref{pr:sa-nonperiodic-pos-and-occ-size} showing how under
\cref{as:nonperiodic}, to utilize this formula to quickly compute
values $|\Posalll(j)|$ and $|\Occ_{2\ell}(j)|$ given position $j$.

\begin{remark}

  Before formally characterizing sets $\Posalll(j)$ and
  $\Occ_{2\ell}(j)$, we first provide some intuition. By $j \not\in
  \R(\tau,\T)$, the position $s = \Successor_{\S}(j)$ satisfies $s - j
  < \tau$. Thus, by the consistency of $\S$ and $3\tau \leq \ell$, all
  $j' \in \Posalll(j)$ share a common offset $\delta_{\S} = s - j$
  such that $j' + \delta_{\S} = \min(\S \cap [j' \dd j'+\tau))$ and
  hence the relative lexicographical order between $\T[j' \dd n]$ and
  $\T[j \dd n]$ is the same as between $\T[j'+\delta_{\S} \dd n]$ and
  $\T[j + \delta_{\S} \dd n]$.  Consequently, it suffices to only
  consider positions in $\S$. Then, computing $|\Posalll(j)|$ reduces
  to finding those $s' \in \S$ that are:
  \begin{enumerate}
  \item\label{sa-nonperiodic-condition-1} Preceded in $\T$ by the
    string $\T[j \dd s)$, and
  \item\label{sa-nonperiodic-condition-2} For which it holds $\T[s'
    \dd n] \prec \T[s \dd n]$ and $\LCE_{\T}(s', s) \in [\ell -
    \delta_{\S} \dd 2\ell - \delta_{\S})$.
  \end{enumerate}
  For any $q \geq 2\ell$, the above
  Condition~\ref{sa-nonperiodic-condition-1} is equivalent to position
  $s'$ having a reversed left length-$q$ context in the
  lexicographical range $[X \dd X')$, where $\revstr{X} = \T[j \dd s)$
  and $X' = Xc^{\infty}$ (where $c = \max\Sigma$), and
  Condition~\ref{sa-nonperiodic-condition-2} is equivalent to position
  $s'$ having a right length-$q$ context in $[Y \dd Y')$, where $Y =
  \T^{\infty}[s \dd j + \ell)$ and $Y' = \T^{\infty}[s \dd j +
  2\ell)$.  Thus, counting such $s'$ can be reduced to a
  two-dimensional weighted orthogonal range counting query on a set of
  points having length-$q$ substrings of $\T^{\infty}$ or
  $\revstr{\T^{\infty}}$ as coordinates (see
  \cref{sec:range-queries}). Crucially, since $\tau = \lfloor
  \tfrac{\ell}{3} \rfloor$ and $\ell \geq 16$, it suffices to choose
  $q = 7\tau$ to guarantee $q \geq 2\ell$, and we will later see that
  $q = \bigO(\tau)$ is the crucial property of this reduction that
  lets us generalize the index to the dynamic case.

  The intuition for the computation of $|\Occ_{2\ell}(j)|$ is similar,
  except we observe that Condition~\ref{sa-nonperiodic-condition-2} is
  that $s'$ satisfies $\T^{\infty}[s' \dd j' + 2\ell) = \T^{\infty}[s
  \dd j + 2\ell)$, which is equivalent to $s'$ having a right
  length-$q$ context in $[Y' \dd Y'c^{\infty})$. Thus, we can count
  such $s'$ (and consequently, compute $|\Occ_{2\ell}(j)|$), using the
  same set of two-dimensional points as in the computation of
  $|\Posalll(j)|$.
\end{remark}

\begin{lemma}\label{lm:sa-nonperiodic-pos}
  Let $j \in [1 \dd n - 3\tau + 1] \sm \R(\tau,\T)$ and $s =
  \Successor_{\S}(j)$. Let $X \in \Sigma^{*}$ and $Y, Y' \in
  \Sigma^{+}$ be such that $\revstr{X} = \T[j \dd s)$, $\T^{\infty}[j
  \dd j {+} \ell) = \revstr{X}Y$, and $\T^{\infty}[j \dd j {+} 2\ell)
  = \revstr{X}Y'$. Then, for any $j' \in [1 \dd n]$, letting $s' = j'
  + |X|$, it holds
  \begin{equation*}
    j' \in \Posalll(j)\text{ if and only if }s' \,{\in}\,
    \S,\ Y\,{\preceq}\,\T^{\infty}[s' \dd s' +
    7\tau)\,{\prec}\,Y',\text{ and } \T^{\infty}[s' {-} |X| \dd s') =
    \revstr{X}.
  \end{equation*}
\end{lemma}
\begin{proof}

  Note that by the uniqueness of $\T[n] = \texttt{\$}$, it holds
  $\per(\T[n - 3\tau + 2 \dd n]) = 3\tau - 1$ and hence $\S \cap [n -
  3\tau + 2 \dd n - 2\tau + 2) \neq \emptyset$. Thus, by $j < n -
  3\tau + 2$, $s = \Successor_{\S}(j)$ is well-defined. Moreover, by
  $j \not\in \R(\tau,\T)$, it holds $\S \cap [j \dd j + \tau) \neq
  \emptyset$.  Therefore, we have $|X| = s - j < \tau \leq \ell$, and
  hence the strings $Y$ and $Y'$ (of length $\ell - |X|$ and $2\ell -
  |X|$, respectively) are well-defined and nonempty.

  Let $j' \in \Posalll(j)$, i.e., $\T[j' \dd n] \prec \T[j \dd n]$ and
  $\LCE_{\T}(j, j') \in [\ell \dd 2\ell)$. This implies $j \neq j'$
  and $\T[j' \dd j' + \ell) = \T[j \dd j + \ell)$. Therefore, by $\ell
  - (s - j) \geq \ell - \tau \geq 2\tau$ and the consistency of the
  string synchronizing set $\S$ (\cref{def:sss}) applied for positions
  $j_1 = j + t$ and $j_2 = j' + t$, where $t \in [0 \dd |X|)$, we
  obtain $\Successor_{\S}(j') - j' = \Successor_{\S}(j) - j$, or
  equivalently, $\Successor_{\S}(j') = j' + (s - j) = s'$. Thus, it
  holds $s' \in \S$. Next, by $\T[j' \dd j' + \ell) = \T[j \dd j +
  \ell)$ and $|X| < \tau \leq \ell$, it holds $\LCE_{\T}(j, j') = |X|
  + \LCE_{\T}(s, s')$. Thus, $\LCE_{\T}(s, s') \in [\ell - |X| \dd
  2\ell - |X|)$. By \cref{lm:utils-1} in \cref{lm:utils} (with
  parameters $\ell_1 = |Y| = \ell - |X|$, $\ell_2 = |Y'| = 2\ell -
  |X|$, and $\ell_3 = 7\tau$) and $2\ell \leq 7\tau$ (holding for
  $\ell \geq 16$), this implies $Y \preceq \T^{\infty}[s' \dd s' +
  7\tau) \prec Y'$.  Finally, the equality $\T^{\infty}[j' \dd j' +
  \ell) = \T^{\infty}[j \dd j + \ell) = \revstr{X}Y$ implies
  $\T^{\infty}[s' - |X| \dd s') = \T^{\infty}[j' \dd s') =
  \revstr{X}$.

  For the opposite implication, assume $s' \in \S$, $Y \preceq
  \T^{\infty}[s' \dd s' + 7\tau) \prec Y'$, and $\T^{\infty}[s' - |X|
  \dd s') = \revstr{X}$. By \cref{lm:utils-1} in \cref{lm:utils} (with
  the same parameters as above) and $2\ell \leq 7\tau$, this implies
  $\T[s' \dd n] \prec \T[s \dd n]$ and $\LCE_{\T}(s, s') \in [\ell -
  |X| \dd 2\ell - |X|)$.  Since $\T[j \dd s) = \revstr{X}$ and by $s
  \in \S$ we have $s < n$, $X$ does not contain the symbol
  $\texttt{\$}$, and hence $j' \geq 1$. Thus, $\T^{\infty}[j' \dd s')
  = X$ implies $\T[j \dd s) = \T[j' \dd s')$, and consequently, $\T[j'
  \dd n] \prec \T[j \dd n]$ and $\LCE_{\T}(j, j') = |X| + \LCE_{\T}(s,
  s') \in [\ell \dd 2\ell)$. Thus, $j' \in \Posalll(j)$.
\end{proof}

\begin{lemma}\label{lm:sa-nonperiodic-occ}
  Let $j \in [1 \dd n - 3\tau + 1] \sm \R(\tau,\T)$, $s =
  \Successor_{\S}(j)$, and $d \in [\ell \dd 2\ell]$.  Let $X \in
  \Sigma^{*}$ and $Y, Y' \in \Sigma^{+}$ be such that $\revstr{X} =
  \T[j \dd s)$, $\T^{\infty}[j \dd j {+} \ell) = \revstr{X}Y$, and
  $\T^{\infty}[j \dd j {+} d) = \revstr{X}Y'$. Then, for any $j' \in
  [1 \dd n]$, letting $s' = j' + |X|$ and $c = \max\Sigma$, it holds
  \begin{equation*}
    j' \in \Occ_{d}(j)\text{ if and only if }s' \,{\in}\,
    \S,\ Y'\,{\preceq}\,\T^{\infty}[s' \dd s' +
    7\tau)\,{\prec}\,Y'c^{\infty},\,\text{ and }\T^{\infty}[s' {-} |X|
    \dd s') = \revstr{X}.
  \end{equation*}
\end{lemma}
\begin{proof}

  Similarly as in \cref{lm:sa-nonperiodic-pos}, we first observe that
  by $j < n - 3\tau + 2$, $\Successor_{\S}(j)$ is well-defined.
  Moreover, by $j \not\in \R(\tau,\T)$, it holds $\S \cap [j \dd j +
  \tau) \neq \emptyset$.  Therefore, $|X| = s - j < \tau \leq \ell$,
  and hence the strings $Y$ and $Y'$ (of length $\ell - |X|$ and $d -
  |X|$, respectively) are well-defined and nonempty.

  Let $j' \in \Occ_{d}(j)$, i.e., $\T^{\infty}[j' \dd j' + d) =
  \T^{\infty}[j \dd j + d)$.  To show $s' \in \S$, we consider two
  cases. If $j = j'$ then $s' = s \in \S$ holds by definition.
  Otherwise, by $\T^{\infty}[j \dd j + d) = \T^{\infty}[j' \dd j' +
  d)$ and the uniqueness of $\T[n] = \texttt{\$}$, we must have $\T[j
  \dd j + \ell) = \T[j' \dd j' + \ell)$.  Thus, by $\ell - (s - j)
  \geq \ell - \tau \geq 2\tau$ and the consistency of $\S$ (applied as
  in the proof of \cref{lm:sa-nonperiodic-pos}) for $j_1 = j + t$ and
  $j_2 = j' + t$, where $t \in [0 \dd |X|)$, we obtain
  $\Successor_{\S}(j') = s'$. Thus, $s' \in \S$.  Next, by
  $\T^{\infty}[s' \dd j' + d) = \T^{\infty}[s \dd j + d)$ and $d \leq
  2\ell \leq 7\tau$, we obtain from \cref{lm:utils-3} in
  \cref{lm:utils} (with parameters $\ell_1 = |Y'| = d - |X|$ and
  $\ell_2 = 7\tau$) that $Y' \preceq \T^{\infty}[s' \dd s' + 7\tau)
  \prec Y'c^{\infty}$. Finally, $\T^{\infty}[j' \dd j' + \ell) =
  \T^{\infty}[j \dd j + \ell) = \revstr{X}Y$ implies $\T^{\infty}[s' -
  |X| \dd s') = \T^{\infty}[j' \dd s') = \revstr{X}$, i.e., the third
  condition.

  For the opposite implication, assume $s' \in \S$, $Y' \preceq
  \T^{\infty}[s' \dd s' + 7\tau) \prec Y'c^{\infty}$, and
  $\T^{\infty}[s' - |X| \dd s') = \revstr{X}$.  By \cref{lm:utils-3}
  in \cref{lm:utils} (with the same parameters as above) and $d \leq
  2\ell \leq 7\tau$, this implies $\T^{\infty}[s' \dd j' + d) =
  \T^{\infty}[s \dd j + d)$.  Combining this with the assumption
  $\T^{\infty}[j' \dd s') = \revstr{X} = \T^{\infty}[j \dd s)$, we
  obtain $\T^{\infty}[j' \dd j' + d) = \T^{\infty}[j \dd j + d)$,
  i.e., $j' \in \Occ_{d}(j)$.
\end{proof}

\begin{lemma}\label{lm:sa-nonperiodic-2}
  Let $j \in [1 \dd n - 3\tau + 1] \sm \R(\tau,\T)$ and $s =
  \Successor_{\S}(j)$. Let $X \in \Sigma^{*}$ and $X', Y, Y' \in
  \Sigma^{+}$ be such that $\revstr{X} = \T[j \dd s)$, $X' =
  Xc^{\infty}$ (with $c = \max\Sigma$), $\T^{\infty}[j \dd j {+} \ell)
  = \revstr{X}Y$, and $\T^{\infty}[j \dd j {+} 2\ell) =
  \revstr{X}Y'$. Then, letting $q = 7\tau$ and $\Pts = \Points_{q}(\T,
  \S)$, it holds:
  \begin{enumerate}
  \item $|\Posalll(j)| = \rcount{\Pts}{X}{X'}{Y'} -
    \rcount{\Pts}{X}{X'}{Y}$ and
  \item $|\Occ_{2\ell}(j)| = \rcount{\Pts}{X}{X'}{Y'c^{\infty}} -
    \rcount{\Pts}{X}{X'}{Y'}$.
  \end{enumerate}
\end{lemma}
\begin{proof}

  1. By \cref{lm:sa-nonperiodic-pos}, we can write $\Posalll(j) = \{s'
  - |X| : s' \in \S,\ Y \preceq \T^{\infty}[s' \dd s' + 7\tau) \prec
  Y', \text{ and } \T^{\infty}[s' {-} |X| \dd s') = \revstr{X}\}$.  On
  the other hand, by \cref{def:p-context} and the definition of the
  $\mathsf{rcount}$ query, we have:
  \begin{align*}
    \rcount{\Pts}{X}{X'}{Y}
      &= |\{s' \in \S : \T^{\infty}[s' \dd s' {+} 7\tau) \prec Y
         \text{ and }X \preceq \revstr{\T^{\infty}[s' {-} 7\tau
         \dd s')} \prec X'\}| \\
      &= |\{s' \in \S : \T^{\infty}[s' \dd s' {+} 7\tau) \prec Y
         \text{ and }X\text{ is a prefix of }
         \revstr{\T^{\infty}[s' {-} 7\tau \dd s')}\}| \\
      &= |\{s' \in \S : \T^{\infty}[s' \dd s' {+} 7\tau) \prec Y
         \text{ and }\T^{\infty}[s' {-} |X| \dd s') = \revstr{X}\}|.
  \end{align*}
  Analogously, $\rcount{\Pts}{X}{X'}{Y'} = |\{s' \in \S :
  \T^{\infty}[s' \dd s' {+} 7\tau) \prec Y' \text{ and }
  \T^{\infty}[s' {-} |X| \dd s') = \revstr{X}\}|$.  Since any position
  $s' \in \S$ that satisfies $\T^{\infty}[s' \dd s' + 7\tau) \prec Y$
  also satisfies $\T^{\infty}[s' \dd s' + 7\tau) \prec Y'$, we obtain
  $\rcount{\Pts}{X}{X'}{Y'} - \rcount{\Pts}{X}{X'}{Y} = |\{s' \in \S :
  Y \preceq \T^{\infty}[s' \dd s' {+} 7\tau) \prec Y'\text{ and
  }\T^{\infty}[s' {-} |X| \dd s') = \revstr{X}\}|$. The cardinality of
  this set is clearly the same as the earlier set characterizing
  $\Posalll(j)$. Thus, $\rcount{\Pts}{X}{X'}{Y'} -
  \rcount{\Pts}{X}{X'}{Y} = |\Posalll(j)|$.

  2. By \cref{lm:sa-nonperiodic-occ}, we have $\Occ_{2\ell}(j) = \{s'
  - |X| : s' \in \S,\ Y' \preceq \T^{\infty}[s' \dd s' + 7\tau) \prec
  Y'c^{\infty}, \text{ and } \T^{\infty}[s' {-} |X| \dd s') =
  \revstr{X}\}$. On the other hand, by \cref{def:p-context} and the
  definition of $\mathsf{rcount}$, we have $\rcount{\Pts}{X}{X'}{Y'} =
  |\{s' \in \S : \T^{\infty}[s' \dd s' {+} 7\tau) \prec Y' \text{ and
  }\T^{\infty}[s' {-} |X| \dd s') = \revstr{X}\}|$ and
  $\rcount{\Pts}{X}{X'}{Y'c^{\infty}} = |\{s' \in \S : \T^{\infty}[s'
  \dd s' {+} 7\tau) \prec Y'c^{\infty} \text{ and } \T^{\infty}[s' {-}
  |X| \dd s') = \revstr{X}\}|$. Since any position $s' \in \S$ that
  satisfies $\T^{\infty}[s' \dd s' + 7\tau) \prec Y'$ also satisfies
  $\T^{\infty}[s' \dd s' + 7\tau) \prec Y'c^{\infty}$, we thus obtain
  $\rcount{\Pts}{X}{X'}{Y'c^{\infty}} - \rcount{\Pts}{X}{X'}{Y'} =
  |\{s' \in \S : Y' \preceq \T^{\infty}[s' \dd s' {+} 7\tau) \prec
  Y'c^{\infty}\text{ and }\T^{\infty}[s' {-} |X| \dd s') =
  \revstr{X}\}|$. The cardinality of this set is clearly the same as
  the earlier set characterizing $\Occ_{2\ell}(j)$. Thus,
  $\rcount{\Pts}{X}{X'}{Y'c^{\infty}} - \rcount{\Pts}{X}{X'}{Y'} =
  |\Occ_{2\ell}(j)|$.
\end{proof}

\begin{proposition}\label{pr:sa-nonperiodic-pos-and-occ-size}
  Under \cref{as:nonperiodic}, given a position $j \in [1 \dd n] \sm
  \R(\tau,\T)$, we can compute $|\Posalll(j)|$ and $|\Occ_{2\ell}(j)|$
  in $\bigO(t)$ time.
\end{proposition}
\begin{proof}

  We first check if $j > n - 3\tau + 1$. If yes, then by the
  uniqueness of $\T[n] = \texttt{\$}$, it holds $|\Occ_{\ell}(j)| =
  1$. By $\Occ_{2\ell}(j) \neq \emptyset$, $\Occ_{2\ell}(j) \sub
  \Occ_{\ell}(j)$, we can therefore return $|\Posalll(j)| = 0$ and
  $|\Occ_{2\ell}(j)| = 1$.  Let us thus assume $j \leq n - 3\tau + 1$
  and recall from the proof of \cref{lm:sa-nonperiodic-pos} that then
  $\Successor_{\S}(j)$ is well-defined.  Using \cref{as:nonperiodic}
  we compute $s = \Successor_{\S}(j)$. Let $X \in \Sigma^{*}$ and $X',
  Y, Y' \in \Sigma^{+}$ be such that $\revstr{X} = \T[j \dd s)$, $X' =
  Xc^{\infty}$, $\T^{\infty}[j \dd j {+} \ell) = \revstr{X}Y$, and
  $\T^{\infty}[j \dd j {+} 2\ell) = \revstr{X}Y'$. Then:
  \begin{enumerate}
  \item By \cref{lm:sa-nonperiodic-2}, we have $|\Posalll(j)| =
    \rcount{\Pts}{X}{X'}{Y'} - \rcount{\Pts}{X}{X'}{Y}$ (where $\Pts =
    \Points_{7\tau}(\T, \S)$) which under \cref{as:nonperiodic} we can
    efficiently compute using the query arguments $(i, q_{l}, q_{r}) =
    (s, s - j, 2\ell - (s - j))$ and then with arguments $(i, q_{l},
    q_{r}) = (s, s - j, \ell - (s - j))$ (see \cref{prob:str-str}). By
    $j \not\in \R(\tau,\T)$ and the density property of $\S$, we have
    $s - j < \tau \leq \ell < n$. On the other hand, $q_{r} \leq 2\ell
    - (s - j) \leq 2\ell < 2n$. Thus, the arguments $q_{l}$ and $q_r$
    of both queries satisfy the requirements in \cref{prob:str-str}.
  \item By \cref{lm:sa-nonperiodic-2}, we also have $|\Occ_{2\ell}(j)|
    = \rcount{\Pts}{X}{X'}{Y'c^{\infty}} - \rcount{\Pts}{X}{X'}{Y'}$
    (where $\Pts$ is defined as above), which under
    \cref{as:nonperiodic} we can compute using the query arguments
    $(i, q_{l}, q_{r}) = (s, s - j, 2\ell - (s - j))$ (see
    \cref{prob:str-str}). As noted above, these arguments satisfy the
    requirements in \cref{prob:str-str}.
  \end{enumerate}
  By \cref{as:nonperiodic}, the query takes $\bigO(t)$ time in total.
\end{proof}

\subsubsection{Computing a Position in
  \texorpdfstring{$\Occ_{2\ell}(\SA[i])$}{Occ}}\label{sec:sa-nonperiodic-occ-pos}

Assume that $i \in [1 \dd n]$ satisfies $\SA[i] \in [1 \dd n] \sm
\R(\tau,\T)$. In this section, we show how under
\cref{as:nonperiodic}, given the index $i$ along with values
$\LB_{\ell}(\SA[i])$, $\UB_{\ell}(\SA[i])$, and some position $j \in
\Occ_{\ell}(\SA[i])$, to efficiently compute some position $j' \in
\Occ_{2\ell}(\SA[i])$.

The section is organized as follows. First, we present a combinatorial
result (\cref{lm:sa-nonperiodic-occ-pos}) that reduces the computation
of $j' \in \Occ_{2\ell}(\SA[i])$ to a generalized range selection
query (see \cref{sec:range-queries}).  We then use this reduction to
present the query algorithm for the computation of some $j' \in
\Occ_{2\ell}(\SA[i])$ in \cref{pr:sa-nonperiodic-occ-pos}.

\begin{lemma}\label{lm:sa-nonperiodic-occ-pos}
  Assume $i \in [1 \dd n]$ is such that $\SA[i] \in [1 \dd n - 3\tau +
  1] \sm \R(\tau,\T)$.  Denote $b = \LB_{\ell}(\SA[i])$, $d =
  |\Occ_{\ell}(\SA[i])|$, and $s = \Successor_{\S}(\SA[i])$.  Let $X
  \in \Sigma^{*}$ and $X', Y \in \Sigma^{+}$ be such that $\revstr{X}
  = \T[\SA[i] \dd s)$, $X' = Xc^{\infty}$ (with $c = \max\Sigma$), and
  $\T^{\infty}[\SA[i] \dd \SA[i] {+} \ell) = \revstr{X}Y$. Let also
  $\Pts = \Points_{7\tau}(\T, \S)$, $m = \rcount{\Pts}{X}{X'}{Y}$, and
  $m' = \rcountb{\Pts}{X}{X'}$. Then, $m + d \leq m'$.  Moreover:
  \begin{enumerate}
  \item\label{lm:sa-nonperiodic-occ-pos-it-1} For $\delta \in [1 \dd
    d]$, any position $p \in \rselect{\Pts}{X}{X'}{m + \delta}$
    satisfies $\T^{\infty}[p - |X| \dd p - |X| + 2\ell) =
    \T^{\infty}[\SA[b + \delta] \dd \SA[b + \delta] + 2\ell)$.
  \item\label{lm:sa-nonperiodic-occ-pos-it-2} For $\delta = i - b$,
    any position $p \in \rselect{\Pts}{X}{X'}{m + \delta}$ satisfies
    $p - |X| \in \Occ_{2\ell}(\SA[i])$.
  \end{enumerate}
\end{lemma}
\begin{proof}

  Note that by the uniqueness of $\T[n] = \texttt{\$}$, it holds
  $\per(\T[n - 3\tau + 2 \dd n]) = 3\tau - 1$ and hence $\S \cap [n -
  3\tau + 2 \dd n - 2\tau + 2) \neq \emptyset$. Thus, by $\SA[i] < n -
  3\tau + 2$, $s = \Successor_{\S}(\SA[i])$ is well-defined.  Denote
  $q = |\S|$. Let $(a_j)_{j \in [1 \dd q]}$ be a sequence containing
  all positions $p \in \S$ ordered according to the string
  $\T^{\infty}[p \dd p + 7\tau)$.  In other words, for any $j, j' \in
  [1 \dd q]$, $j < j'$ implies $\T^{\infty}[a_j\dd a_j + 7\tau)
  \preceq \T^{\infty}[a_{j'} \dd a_{j'} + 7\tau)$.  Note, that the
  sequence $(a_j)_{j \in [1 \dd q]}$ is not unique.  Since $\{a_j : j
  \in [1 \dd q]\} = \S$, it holds $|\{a_j - |X| : j \in [1 \dd
  q]\text{ and } \T^{\infty}[a_j - |X| \dd a_j) = \revstr{X}\}| = |\{j
  \in [1 \dd q] : \T^{\infty}[a_j - |X| \dd a_j) = \revstr{X}\}| =
  |\{p \in \S : \T^{\infty}[p - |X| \dd p) = \revstr{X}\}| = m'$,
  where the last equality follows by \cref{lm:utils-3} of
  \cref{lm:utils} and the definition of $\rcountb{\Pts}{X}{X'}$ (see
  \cref{sec:str-str}).  Moreover, by the same argument (utilizing the
  definition of $\rcount{\Pts}{X}{X'}{Y}$ instead of
  $\rcountb{\Pts}{X}{X'}$), we have $|\{a_j - |X| : j \in [1 \dd q],\,
  \T^{\infty}[a_j - |X| \dd a_j) = \revstr{X}\text{ and }
  \T^{\infty}[a_j \dd a_j + 7\tau) \prec Y\}| = |\{j \in [1 \dd q] :
  \T^{\infty}[a_j - |X| \dd a_j) = \revstr{X}\text{ and }
  \T^{\infty}[a_j \dd a_j + 7\tau) \prec Y\}| = m$.  On the other
  hand, observe that by \cref{lm:sa-nonperiodic-occ}, for any $j \in
  [1 \dd n]$, it holds $j \in \Occ_{\ell}(\SA[i])$ if and only if $j +
  |X| \in \S$, $\T^{\infty}[j \dd j + |X|) = \revstr{X}$, and $Y
  \preceq \T^{\infty}[j + |X| \dd j + |X| + 7\tau) \prec Yc^{\infty}$.
  In other words, $\Occ_{\ell}(\SA[i]) = \{a_j - |X| : j \in [1 \dd
  q],\, \T^{\infty}[a_j - |X| \dd a_j) = \revstr{X},\text{ and } Y
  \preceq \T^{\infty}[a_j \dd a_j + 7\tau) \prec Yc^{\infty}\}$.  The
  latter set (whose cardinality is equal to $d$) is clearly a subset
  of $\{a_j - |X| : j \in [1 \dd q]\text{ and }\T^{\infty}[a_j - |X|
  \dd a_j) = \revstr{X}\}$ (whose cardinality, as shown above, is
  equal to $m'$). Thus, $d \leq m'$.  On the other hand, the set
  $\{a_j - |X| : j \in [1 \dd q],\, \T^{\infty}[a_j - |X| \dd a_j) =
  \revstr{X},\text{ and }\T^{\infty}[a_j \dd a_j + 7\tau) \prec Y\}$
  (whose cardinality, as shown above, is $m$) is also clearly a subset
  of $\{a_j - |X| : j \in [1 \dd q]\text{ and }\T^{\infty}[a_j - |X|
  \dd a_j) = \revstr{X}\}$ (whose cardinality is $m'$). Thus, $m \leq
  m'$. Since $j \in [1 \dd q]$ cannot simultaneously satisfy
  $\T^{\infty}[a_j \dd a_j + 7\tau) \prec Y$ and $Y \preceq
  \T^{\infty}[a_j \dd a_j + 7\tau)$, these subsets are disjoint.
  Hence, $m + d \leq m'$.

  1. As shown above, $|\{j \in [1 \dd q] : \T^{\infty}[a_j - |X| \dd
  a_j) = \revstr{X}\}| = m'$.  Let $(b_j)_{j \in [1 \dd m']}$ be a
  subsequence of $(a_j)_{j \in [1 \dd q]}$ containing all elements of
  $\{a_j : j \in [1 \dd q]\text{ and } \T^{\infty}[a_j - |X| \dd a_j)
  = \revstr{X}\}$ (in the same order as they appear in the sequence
  $(a_j)_{j \in [1 \dd q]}$). Our proof consists of three steps:
  \begin{enumerate}[label=(\roman*)]
  \item\label{lm:sa-nonperiodic-occ-pos-it-1-1} Let $j \in [1 \dd
    m']$. We start by showing that $b_j \in \rselect{\Pts}{X}{X'}{j}$.
    Let $Q, Q'$ be such that $\revstr{Q} = \T^{\infty}[b_j - 7\tau \dd
    b_j)$ and $Q' = \T^{\infty}[b_j \dd b_j + 7\tau)$. Let
    \begin{align*}
      r_{\rm beg} &=
        |\{a_t : t \in [1 \dd q],
          \T^{\infty}[a_t - |X| \dd a_t) = \revstr{X},\text{ and }
          \T^{\infty}[a_t \dd a_t + 7\tau)\prec Q'\}|\text{ and}\\
      r_{\rm end} &=
        |\{a_t : t \in [1 \dd q],
          \T^{\infty}[a_t - |X| \dd a_t) = \revstr{X},\text{ and }
          \T^{\infty}[a_t \dd a_t + 7\tau)\preceq Q'\}|.
    \end{align*}
    If $t \in [1 \dd q]$ satisfies $\T^{\infty}[a_t - |X| \dd a_t) =
    \revstr{X}$, then $a_t \in \{b_1, \ldots, b_m\}$.  Moreover, since
    for any $t, t' \in [1 \dd m]$, $t < t'$ implies $\T^{\infty}[b_t
    \dd b_t + 7\tau) \preceq \T^{\infty}[b_{t'} \dd b_{t'} + 7\tau)$,
    any $t \in [1 \dd q]$ that additionally satisfies $\T^{\infty}[a_t
    \dd a_t + 7\tau) \prec \T^{\infty}[b_j \dd b_j + 7\tau)$, also
    satisfies $a_t \in \{b_1, \ldots, b_{j-1}\}$. Thus, $r_{\rm beg} <
    j$.  On the other hand, every $t \in [1 \dd j]$ satisfies
    $\T^{\infty}[b_t - |X| \dd b_t) = \revstr{X}$ and $\T^{\infty}[b_t
    \dd b_t + 7\tau) \preceq \T^{\infty}[b_j \dd b_j + 7\tau)$. Thus,
    $j \leq r_{\rm end}$. Altogether, $j \in (r_{\rm beg} \dd r_{\rm
    end}]$. Recall now the definition $\Pts = \Points_{7\tau}(\T, \S)$
    (\cref{def:p-right-context}) and note that by \cref{lm:utils-3} of
    \cref{lm:utils}, it holds $r_{\rm beg} = \rcount{\Pts}{X}{X'}{Q'}$
    and $r_{\rm end} = \rcounti{\Pts}{X}{X'}{Q'}$.  We thus obtain
    that $j \in (\rcount{\Pts}{X}{X'}{Q'} \dd
    \rcounti{\Pts}{X}{X'}{Q'}]$. On the other hand, $(Q,Q',b_j)\in
    \Pts$ and $T^\infty[b_j-|X|\dd b_j)=\revstr{X}$, so $b_j \in
    \rselect{\Pts}{X}{X'}{j}$ holds as claimed.
  \item\label{lm:sa-nonperiodic-occ-pos-it-1-2} Let $j \in [1 \dd
    m']$. We will now show that, for any $p \in
    \rselect{\Pts}{X}{X'}{j}$, it holds $\T^{\infty}[p - |X| \dd p +
    7\tau) = \T^{\infty}[b_j - |X| \dd b_j + 7\tau)$. By
    \cref{lm:sa-nonperiodic-occ-pos-it-1-1} and the definition of
    $\rselect{\Pts}{X}{X'}{j}$, the assumption $p \in
    \rselect{\Pts}{X}{X'}{j}$ implies $\T^{\infty}[p \dd p + 7\tau) =
    \T^{\infty}[b_j \dd b_j + 7\tau)$.  Moreover, letting $Q$ be such
    that $\revstr{Q} = \T^{\infty}[p - |X| \dd p)$, it also implies $X
    \preceq Q \prec X'$.  Thus, by \cref{lm:utils-3} of
    \cref{lm:utils}, $p$ is preceded by $X$ in $\T$.  Since by
    definition of $(b_j)_{j\in [1 \dd m']}$, the position $b_j$ is
    also preceded by $X$ in $\T$, we obtain $\T^{\infty}[p - |X| \dd p
    + 7\tau) = \T^{\infty}[b_j - |X| \dd b_j + 7\tau)$.
  \item We are now ready to prove the main claim.  As observed above,
    $\Occ_{\ell}(\SA[i]) = \{a_j - |X| : j \in [1 \dd q],\,
    \T^{\infty}[a_j - |X| \dd a_j) = \revstr{X},\text{ and } Y \preceq
    \T^{\infty}[a_j \dd a_j + 7\tau) \prec Yc^{\infty}\}$. Note, that
    since the positions $k$ in the sequence $(a_j)_{j \in [1 \dd q]}$
    are sorted by $\T^{\infty}[k \dd k + 7\tau)$, we can simplify the
    second condition. Denoting $j_{\rm skip} = |\{j \in [1 \dd q] :
    \T^{\infty}[a_j \dd a_j + 7\tau) \prec Y\}|$, we have
    \begin{equation*}
      \Occ_{\ell}(\SA[i]) = \left\{a_j - |X| : \substack{j \in (j_{\rm
      skip} \dd q],\ \T^{\infty}[a_j - |X| \dd a_j) =
      \revstr{X},\\ \text{ and }\T^{\infty}[a_j \dd a_j + 7\tau) \prec
      Yc^{\infty}}\right\}.
    \end{equation*}
    Let us now estimate $|\{j \in [1 \dd j_{\rm skip}] :
    \T^{\infty}[a_j - |X| \dd a_j) = \revstr{X}\}|$.  Any $j$ in this
    set satisfies $j \in [1 \dd q]$, $\T^{\infty}[a_j - |X| \dd a_j) =
    \revstr{X}$, and $\T^{\infty}[a_j \dd a_j + 7\tau) \prec
    Y$. Earlier we observed that the number of such $j$ is precisely
    $m$. Combining this fact with the above formula for
    $\Occ_{\ell}(\SA[i])$ and the definition of $(b_j)_{j \in [1 \dd
    m']}$, we have $\Occ_{\ell}(\SA[i]) = \{b_j - |X| : j \in (m \dd m
    + d]\}$.  On the other hand, $b + d = \LB_{\ell}(\SA[i]) +
    |\Occ_{\ell}(\SA[i])| = \UB_{\ell}(\SA[i])$.  Thus,
    $\Occ_{\ell}(\SA[i]) = \{\SA[j] : j \in (b \dd b + d]\}$.  We now
    observe:
    \begin{itemize}
    \item Let $j_1, j_2 \in (m \dd m+d]$ and assume $j_1 < j_2$. Since
      the elements of $(b_j)$ occur in the same order as in $(a_j)$,
      and positions $p$ in $(a_j)$ are sorted by $\T[p \dd p +
      7\tau)$, it follows that $\T^{\infty}[b_{j_1} \dd b_{j_1} +
      7\tau) \preceq \T^{\infty}[b_{j_2} \dd b_{j_2} + 7\tau)$. On the
      other hand, by definition of $(b_j)$, both positions $b_{j_1}$
      and $b_{j_2}$ are preceded in $\T$ by the string $\revstr{X}$.
      Thus, $\T^{\infty}[b_{j_1} - |X| \dd b_{j_2} + 7\tau) \preceq
      \T^{\infty}[b_{j_2} - |X| \dd b_{j_2} + 7\tau)$.
    \item On the other hand, by definition of lexicographical order,
      for any $j_1, j_2 \in [1 \dd d]$, the assumption $j_1 < j_2$
      implies $\T^{\infty}[\SA[b + j_1] \dd \SA[b + j_1] + |X| +
      7\tau) \preceq \T^{\infty}[\SA[b + j_2] \dd \SA[b + j_2] + |X| +
      7\tau)$.
    \end{itemize}
    We have thus shown that both sequences $\SA[b + 1], \ldots,
    \SA[b+d]$ and $b_{m+1} - |X|, \ldots, b_{m+d} - |X|$ contain the
    same set of positions $\Occ_{\ell}(\SA[i])$ ordered according to
    the length-$(|X|+7\tau)$ right context in $\T^{\infty}$.
    Therefore, regardless of how ties are resolved in each sequence,
    for any $\delta \in [1 \dd d]$, we have
    \begin{equation*}
      \T^{\infty}[\SA[b + \delta] \dd \SA[b + \delta] + |X| + 7\tau) =
      \T^{\infty}[b_{m+\delta} - |X| \dd b_{m+\delta} + 7\tau).
    \end{equation*}
    To finalize the proof of the claim, take any $p \in
    \rselect{\Pts}{X}{X'}{m+\delta}$.  By
    \cref{lm:sa-nonperiodic-occ-pos-it-1-2}, for $j = m + \delta$, we
    have $\T^{\infty}[p - |X| \dd p + 7\tau) = \T^{\infty}[b_{m +
    \delta} - |X| \dd b_{m + \delta} + 7\tau) = \T^{\infty}[\SA[b +
    \delta] \dd \SA[b + \delta] + |X| + 7\tau)$. In particular, by $0
    \leq |X|$ and $2\ell \leq 7\tau$, we obtain $2\ell \leq |X| +
    7\tau$ and $\T^{\infty}[p - |X| \dd p - |X| + 2\ell) =
    \T^{\infty}[\SA[b+ \delta] \dd \SA[b + \delta] + 2\ell)$, i.e.,
    the claim.
  \end{enumerate}

  \noindent
  2. Applying \cref{lm:sa-nonperiodic-occ-pos-it-1} for $\delta = i -
  b$, we conclude that any $p \in \rselect{\Pts}{X}{X'}{m+\delta}$,
  satisfies $\T^{\infty}[p - |X| \dd p - |X| + 2\ell) =
  \T^{\infty}[\SA[b + \delta] \dd \SA[b +\delta] + 2\ell) =
  \T^{\infty}[\SA[i] \dd \SA[i] + 2\ell)$, i.e., $p - |X| \in
  \Occ_{2\ell}(\SA[i])$.
\end{proof}

\begin{proposition}\label{pr:sa-nonperiodic-occ-pos}
  Let $i \in [1 \dd n]$ be such that $\SA[i] \in [1 \dd n] \sm
  \R(\tau,\T)$. Under \cref{as:nonperiodic}, given the values $i$,
  $\LB_{\ell}(\SA[i])$, $\UB_{\ell}(\SA[i])$, and some $j \in
  \Occ_{\ell}(\SA[i])$ as input, we can compute some $j' \in
  \Occ_{2\ell}(\SA[i])$ in $\bigO(t)$ time.
\end{proposition}
\begin{proof}

  We first calculate $|\Occ_{\ell}(\SA[i])| = \UB_{\ell}(\SA[i]) -
  \LB_{\ell}(\SA[i])$ using the input arguments. If
  $|\Occ_{\ell}(\SA[i])| = 1$, then by $\Occ_{2\ell}(\SA[i]) \neq
  \emptyset$ and $\Occ_{2\ell}(\SA[i]) \sub \Occ_{\ell}(\SA[i])$ we
  must have $\LB_{2\ell}(\SA[i]) = \LB_{\ell}(\SA[i])$,
  $\UB_{2\ell}(\SA[i]) = \UB_{\ell}(\SA[i])$ and $\Occ_{2\ell} =
  \{j\}$. Thus, we return $j' := j$. Let us thus assume
  $|\Occ_{\ell}(\SA[i])| > 1$, and observe that by the uniqueness of
  $\T[n] = \texttt{\$}$, this implies $\SA[i], j \in [1 \dd n - 3\tau
  + 1]$. Moreover, by $3\tau - 1 \leq \ell$, $j \in
  \Occ_{\ell}(\SA[i])$, and $\SA[i] \not\in \R(\tau,\T)$, we also have
  $j \not\in \R(\tau,\T)$. Therefore, as noted in the proof of
  \cref{lm:sa-nonperiodic-pos}, $\Successor_{\S}(j)$ is well-defined
  and using \cref{as:nonperiodic} we can compute $s =
  \Successor_{\S}(j)$ in $\bigO(t)$ time. Let $X \in \Sigma^{*}$ and
  $X', Y \in \Sigma^{+}$ be such that $\revstr{X} = \T[j \dd s)$, $X'
  = Xc^{\infty}$, and $\T^{\infty}[j \dd j {+} \ell) = \revstr{X}Y$
  (where $c = \max\Sigma$). Observe now that by $j \in
  \Occ_{\ell}(\SA[i])$, we have $\T^{\infty}[j \dd j + \ell) =
  \T^{\infty}[\SA[i] \dd \SA[i] + \ell)$.  We also have $3\tau - 1
  \leq \ell$. Thus, by the consistency condition of $\S$
  (\cref{def:sss}), for any $t \in [0 \dd \tau)$, $\SA[i] + t \in \S$
  holds if and only if $j + t \in \S$. Recall, however, that we
  assumed $\SA[i] \not\in \R(\tau,\T)$, i.e., $\S \cap [\SA[i] \dd
  \SA[i] + \tau) \neq \emptyset$.  Therefore, denoting $s' =
  \Successor_{\S}(\SA[i])$, it holds $s' - \SA[i] = s -
  j$. Consequently, it holds $\T^{\infty}[\SA[i] \dd s') =
  \T^{\infty}[j \dd s) = \revstr{X}$ and $\T^{\infty}[s' \dd \SA[i] +
  \ell) = \T^{\infty}[s \dd j + \ell) = Y$.  We can thus apply
  \cref{lm:sa-nonperiodic-occ-pos} without knowing the values of
  $\SA[i]$ or $s'$ (we only need to know $|X|$ and the starting
  position of some occurrence of $\revstr{X}Y$ in $\T$). First, using
  \cref{prob:str-str-it-1} of \cref{prob:str-str} with the query
  arguments $(i, q_{l}, q_{r}) = (s, s - j, \ell - (s - j))$ we
  compute in $\bigO(t)$ time (which is possible under
  \cref{as:nonperiodic}) the value $m := \rcount{\Pts}{X}{X'}{Y}$ (the
  arguments satisfy the requirements of \cref{prob:str-str} since
  $q_{l} = s - j < \tau \leq \ell < n$ and $q_{r} = \ell - (s - j)
  \leq \ell < n$). We then calculate $\delta = i - \LB_{\ell}(\SA[i])$
  and using \cref{prob:str-str-it-2} of \cref{prob:str-str} with the
  query arguments $(i, q_{l}, r) = (s, j - s, m + \delta)$ we compute
  in $\bigO(t)$ time some position $p \in \rselect{\Pts}{X}{X'}{m +
  \delta}$ (the arguments satisfy the requirements of
  \cref{prob:str-str} by the argument as above). By
  \cref{lm:sa-nonperiodic-occ-pos}, we then have $p - q_{l} \in
  \Occ_{2\ell}(\SA[i])$. In total, we spend $\bigO(t)$ time.
\end{proof}

\subsubsection{The Data Structure}\label{sec:sa-nonperiodic-structure}

By combining the above results, we obtain that under
\cref{as:nonperiodic}, given an index $i \in [1 \dd n]$ satisfying
$\SA[i] \in [1 \dd n] \sm \R(\tau,\T)$, along with
$\LB_{\ell}(\SA[i])$, $\UB_{\ell}(\SA[i])$ and some $j \in
\Occ_{\ell}(\SA[i])$ as input, we can efficiently compute
$(\LB_{2\ell}(\SA[i]), \UB_{2\ell}(\SA[i]))$ and some $j' \in
\Occ_{2\ell}(\SA[i])$.

\begin{proposition}\label{pr:sa-nonperiodic}
  Let $i \in [1 \dd n]$ be such that $\SA[i] \in [1 \dd n] \sm
  \R(\tau,\T)$.  Under \cref{as:nonperiodic}, given the index $i$
  along with values $\LB_{\ell}(\SA[i])$, $\UB_{\ell}(\SA[i])$, and
  some position $j \in \Occ_{\ell}(\SA[i])$, we can compute
  $(\LB_{2\ell}(\SA[i]), \UB_{2\ell}(\SA[i]))$ and some $j' \in
  \Occ_{2\ell}(\SA[i])$ in $\bigO(t)$ time.
\end{proposition}
\begin{proof}

  First, using \cref{pr:sa-nonperiodic-occ-pos}, we compute some $j'
  \in \Occ_{2\ell}(\SA[i])$. This takes $\bigO(t)$ time and all the
  required values ($i$, $\LB_{2\ell}(\SA[i])$, $\UB_{\ell}(\SA[i])$,
  and some $j \in \Occ_{\ell}(\SA[i])$) are given as input. We now
  observe that since for $j'$ we have $\T^{\infty}[j' \dd j' + 2\ell)
  = \T^{\infty}[\SA[i] \dd \SA[i] + 2\ell)$, we have $j'' \in
  \Occ_{2\ell}(\SA[i])$ if and only if $j'' \in \Occ_{2\ell}(j')$.
  Thus, $\Occ_{2\ell}(\SA[i]) = \Occ_{2\ell}(j')$. On the other hand,
  by \cref{lm:utils-1} of \cref{lm:utils}, we have $j'' \in
  \Posalll(\SA[i])$ if and only if $\T^{\infty}[\SA[i] \dd \SA[i] +
  \ell) \preceq \T^{\infty}[j'' \dd j'' + 2\ell) \prec
  \T^{\infty}[\SA[i] \dd \SA[i] + 2\ell)$, i.e., whether $j'' \in
  \Posalll(\SA[i])$ depends only on $\T^{\infty}[\SA[i] \dd \SA[i] +
  2\ell)$.  Thus, $j' \in \Occ_{2\ell}(\SA[i])$ implies
  $\Posalll(\SA[i]) = \Posalll(j')$.  Therefore, in the second step of
  the query, using \cref{pr:sa-nonperiodic-pos-and-occ-size} we
  compute in $\bigO(t)$ time the values
  \begin{align*}
    \deltal_{\ell}(\SA[i]) &:= |\Posalll(j')| =
    |\Posalll(\SA[i])|\text{ and }\\ m &:= |\Occ_{2\ell}(j')| =
    |\Occ_{2\ell}(\SA[i])|.
  \end{align*}
  Letting $b = \LB_{\ell}(\SA[i])$, it then holds
  $(\LB_{2\ell}(\SA[i]), \UB_{2\ell}(\SA[i])) = (b +
  \deltal_{\ell}(\SA[i]), b + \deltal_{\ell}(\SA[i]) + m)$. In total,
  the query takes $\bigO(t)$ time.
\end{proof}

\subsection{The Periodic Positions}\label{sec:sa-periodic}

In this section, we show that assuming we can compute basic
characteristic of periodic positions and perform some int-string range
and modular constraint queries (\cref{as:periodic}), given any
position $i \in [1 \dd n]$ satisfying $\SA[i] \in \R(\lfloor
\tfrac{\ell}{3} \rfloor, \T)$, along with some position $j \in
\Occ_{\ell}(\SA[i])$ and the pair $(\LB_{\ell}(\SA[i]),
\UB_{\ell}(\SA[i]))$ as input, we can efficiently compute some $j'
\,{\in}\, \Occ_{2\ell}(\SA[i])$ and the pair $(\LB_{2\ell}(\SA[i]),
\UB_{2\ell}(\SA[i]))$.

The section is organized into nine parts. First
(\cref{sec:sa-periodic-prelim}) we recall the standard definitions and
notation for string synchronizing sets~\cite{sss}, that we will use to
formalize the combinatorial properties used in our algorithms and then
present a main assumption that we will use throughout this
section. Next (\cref{sec:sa-periodic-pos-decomposition}), we formulate
the main idea of the algorithm by describing three sets
$\Poslow(\SA[i])$, $\Posmid(\SA[i])$, and $\Poshigh(\SA[i])$, which
play a central role in the SA query algorithm, and prove the
combinatorial result showing how their sizes relate to the size of
$\Posalll(\SA[i])$. The following six sections
(Section~\ref{sec:sa-periodic-occ-size}--Section~\ref{sec:sa-periodic-occ-pos})
describe the individual steps of the query algorithm in the order in
which are used (with some minor exceptions). In
\cref{sec:sa-periodic-structure} we put everything together to obtain
the algorithm that computes the pair $(\LB_{2\ell}(\SA[i]),
\UB_{2\ell}(\SA[i]))$ and some $j' \in \Occ_{2\ell}(\SA[i])$ for $i
\in [1 \dd n]$ satisfying $\SA[i] \in \R$.

\subsubsection{Preliminaries}\label{sec:sa-periodic-prelim}

In this section, we define preliminary concepts used for answering
$\SA$ queries for periodic positions.

The section is organized as follows. First, we recall basic
definitions and notation for periodic positions in string
synchronizing sets. We then recall two results characterizing blocks
of periodic positions in lexicographical order (\cref{lm:exp}) and in
text order (\cref{lm:R-block}). We conclude with the central
assumption that we will use all through the section.

\begin{definition}\label{def:nc}
  A function $f: \Sigma^{+} \rightarrow \Sigma^{+}$ is said to be
  \emph{necklace-consistent} if it satisfies the following conditions
  for every $S,S'\in \Sigma^+$:
  \begin{enumerate}
  \item The strings $f(S)$ and $S$ are cyclically equivalent.
  \item If $S$ and $S'$ are cyclically equivalent, then $f(S)= f(S')$.
  \end{enumerate}
\end{definition}

Let $S \in \Sigma^{+}$ and $\tau \in \Zp$, and let $f$ be some
necklace-consistent function.  For any $X \in \Sigma^{+}$, we define
$\Lrootstrgen{f}{X} := f(X[1 \dd p])$, where $p =
\per(X)$.\footnote{Note that this definition of $\Lroot$ generalizes
the original definition~\cite{sss}, in which $f(S)$ always returned
the lexicographically minimal string that is cyclically equivalent
with $S$. Such function is clearly necklace-consistent.}  For any
position $j \in \R(\tau,S)$, we then let $\Lrootgen{f}{\tau}{S}{j} =
\Lrootstrgen{f}{S[j \dd j + 3\tau - 1)}$.  We denote
$\Lrootsgen{f}{\tau}{S} = \{\Lrootgen{f}{\tau}{S}{j} : j \in
\R(\tau,S)\}$.  Next, we define the \emph{run-decomposition}.  For any
$j \in \R(\tau,S)$, let $\rendgen{\tau}{S}{j} := \min\{j' \in [j \dd
n] : j' \not\in \R(\tau,S)\} + 3\tau - 2$.  By~\cite[Fact~3.2]{sss},
for $j \in \R(\tau,S)$, the substring $S[j \dd \rendgen{\tau}{S}{j})$
is the longest prefix of $S[j \dd n]$ that has a period
$|\Lrootgen{f}{\tau}{S}{j}|$. Moreover, by definition of $\Lroot$,
letting $p = |\Lrootgen{f}{\tau}{S}{j}|$, there exists $s \in [0 \dd
p)$ such that $S[j + s \dd j + s + p) = \Lrootgen{f}{\tau}{S}{j})$.
Thus, for every $j \in \R(\tau,S)$, we can write $S[j \dd
\rendgen{\tau}{S}{j}) = H'H^{k}H''$, where $H =
\Lrootgen{f}{\tau}{S}{j}$, and $H'$ (resp.\ $H''$) is a proper suffix
(resp.\ prefix) of $H$. Note that there is always only one way to
write $S[j \dd \rendgen{\tau}{S}{j})$ in this way, since the opposite
would contradict the synchronization property of primitive
strings~\cite[Lemma~1.11]{AlgorithmsOnStrings}. In other words, for
any fixed $f$, the run-decomposition in unique. We denote
$\Lheadgen{f}{\tau}{S}{j} = |H'|$, $\Lexpgen{f}{\tau}{S}{j} = k$, and
$\Ltailgen{f}{\tau}{S}{j} = |H''|$.  For $j \in \R(\tau,S,j)$, we let
$\typegen{\tau}{S}{j} = +1$ if $\rendgen{\tau}{S}{j} \leq |S|$ and
$S[\rendgen{\tau}{S}{j}] \succ S[\rendgen{\tau}{S}{j} - p]$ (where $p
= |\Lrootgen{f}{\tau}{S}{j}|$), and $\typegen{\tau}{S}{j} = - 1$
otherwise.\footnote{Although the string $\Lrootgen{f}{\tau}{S}{j}$
depends on the function $f$, the value $|\Lrootgen{f}{\tau}{S}{j}| =
\per(S[j \dd j + 3\tau - 1))$ does not. Thus, it is correct to drop
$f$ in the notation for $\typegen{\tau}{S}{j}$.} For any $j \in
\R(\tau,S)$ and $t \geq 3\tau - 1$, we define
$\Lexpcutgen{f}{\tau}{S}{j}{t} := \min(\Lexpgen{f}{\tau}{S}{j},
\lfloor \tfrac{t - s}{|H|} \rfloor)$ and
$\rendcutgen{f}{\tau}{S}{j}{t} := j + s +
\Lexpcutgen{f}{\tau}{S}{j}{t}|H|$, where $s =
\Lheadgen{f}{\tau}{S}{j}$ and $H = \Lrootgen{f}{\tau}{S}{j}$.  We
denote $\rendfullgen{f}{\tau}{S}{j} := \rendcutgen{f}{\tau}{S}{j}{n}$,
$\rendlowgen{f}{\tau}{S}{j} := \rendcutgen{f}{\tau}{S}{j}{\ell}$, and
$\rendhighgen{f}{\tau}{S}{j} :=
\rendcutgen{f}{\tau}{S}{j}{2\ell}$. Observe that it holds
$\rendfullgen{f}{\tau}{S}{j} = j + s + \Lexpgen{f}{\tau}{S}{j}|H|$.

We will repeatedly refer to the following subsets of $\R(\tau,S)$.
First, we denote $\R^{-}(\tau,S) = \{j \in \R(\tau,S) :
\typegen{\tau}{S}{j} = -1\}$ and $\R^{+}(\tau,S) = \R(\tau,S) \sm
\R^{-}(\tau,S)$. For any $H \in \Sigma^{+}$ and $s \in \Zz$, we then
let $\R_{f,H}(\tau,S) = \{j \in \R(\tau,S) : \Lrootgen{f}{\tau}{S}{j}
= H\}$, $\R^{-}_{f,H}(\tau,S) = \R^{-}(\tau,S) \cap \R_{f,H}(\tau,S)$,
$\R^{+}_{f,H}(\tau,S) = \R^{+}(\tau,S) \cap \R_{f,H}(\tau,S)$,
$\R_{f,s,H}(\tau,S) = \{j \in \R_{f,H}(\tau,S) :
\Lheadgen{f}{\tau}{S}{j}=s\}$, $\R^{-}_{f,s,H}(\tau,S) =
\R^{-}(\tau,S) \cap \R_{f,s,H}(\tau,S)$, and $\R^{+}_{f,s,H}(\tau,S) =
\R^{+}(\tau,S) \cap \R_{f,s,H}(\tau,S)$.

The following two lemmas establish the basic properties of periodic
positions. First, we show that for any necklace-consistent $f$, the
set of positions $\R_{f,s,H}(\tau,\T)$ occupies a contiguous block in
the $\SA$ of $\T$ and describe the structure of such
\emph{lexicographical} block. The following lemma was proved
in~\cite[Lemma~3.5]{sss-index}, but for completeness, we provide its
proof in \cref{app:extra-proofs}.

\begin{restatable}{lemma}{lmexp}\label{lm:exp}
  Let $S \in \Sigma^{k}$, $\tau \in \Zp$, and let $f$ be any
  necklace-consistent function. If $j \in \R_{f,s,H}(\tau,S)$ then for
  any $j' \in [1 \dd k]$, $\LCE_{S}(j, j') \geq 3\tau - 1$ holds if
  and only if $j' \in \R_{f,s,H}(\tau,S)$. Moreover, if $j'\in
  \R_{f,s,H}(\tau, S)$, then:
  \begin{enumerate}
  \item\label{lm:SA-block-1} If $\typegen{\tau}{S}{j} = -1$ and
    $\typegen{\tau}{S}{j'} = +1$, then $S[j \dd] \prec S[j' \dd]$,
  \item\label{lm:SA-block-2} If $\typegen{\tau}{S}{j} =
    \typegen{\tau}{S}{j'} = -1$ and $\rendgen{\tau}{S}{j} - j <
    \rendgen{\tau}{S}{j'} - j'$, then $S[j \dd] \prec S[j' \dd]$,
  \item If $\typegen{\tau}{S}{j} = \typegen{\tau}{S}{j'} = +1$ and
    $\rendgen{\tau}{S}{j} - j > \rendgen{\tau}{S}{j'} - j'$, then $S[j
    \dd] \prec S[j' \dd]$.
  \end{enumerate}
\end{restatable}

The key to the efficient computation of $\deltal_{\ell}(j)$ is
processing of the elements of $\R(\tau,S)$ in blocks (note that here
by ``blocks'' we mean blocks of positions in \emph{text order}, as
opposed to blocks in \emph{lexicographical order} (i.e., in the suffix
array) considered in the previous lemma). The starting positions of
these blocks are defined as
\begin{equation}\label{eq:Rprim}
  \R'(\tau,S) := \{j\in \R(\tau,S) : j-1 \notin \R(\tau,S)\}.
\end{equation}
We also let $\R'^{-}(\tau,S) = \R^{-}(\tau,S) \cap \R'(\tau,S)$,
$\R'^{+}(\tau,S) = \R^{+}(\tau,S) \cap \R'(\tau,S)$,
$\R'^{-}_{f,H}(\tau,S) = \R'(\tau,S) \cap \R^{-}_{f,H}(\tau,S)$, and
$\R'^{+}_{f,H}(\tau,S) = \R'(\tau,S) \cap \R^{+}_{f,H}(\tau,S)$.  For
any $H \in \Sigma^{+}$ we also denote:
\[
  E^{-}_{f,H}(\tau,S) := \{(\rendfullgen{f}{\tau}{S}{j} - j,
  \rendfullgen{f}{\tau}{S}{j}) : j \in \R'^{-}_{f,H}(\tau,S)\}.
\]
The set $E^{+}_{f,H}(\tau,S)$ is defined analogously, but with
$\R'^{-}_{f,H}(\tau,S)$ replaced by $\R'^{+}_{f,H}(\tau,S)$.  The next
lemma justifies our strategy. As with the previous lemma, this is a
standard result characterizing periodic positions of string
synchronizing sets and was proved in~\cite[Lemma~3.6]{sss-index} for
the original definition of $\Lroot$. For completeness, in
\cref{app:extra-proofs} we provide its minimally modified proof
adapting it to the more general version of $\Lroot$ used in this
paper.

\begin{restatable}{lemma}{lmRblock}\label{lm:R-block}
  Let $S \in \Sigma^{+}$, $\tau \in \Zp$, and assume that $f$ is a
  necklace-consistent function. For any position $j \in \R(\tau,S) \sm
  \R'(\tau,S)$ it holds
  \begin{itemize}
  \item $\Lrootgen{f}{\tau}{S}{j {-} 1} = \Lrootgen{f}{\tau}{S}{j}$,
  \item $\rendgen{\tau}{S}{j {-} 1} = \rendgen{\tau}{S}{j}$, and
  \item $\typegen{\tau}{S}{j {-} 1} = \typegen{\tau}{S}{j}$.
  \end{itemize}
\end{restatable}

\begin{definition}\label{def:intervals}
  For any $j \in \R(\tau,S)$, denote $I_j(\tau,S) = (b + 1, e + 1,j)$,
  where $e = \rendfullgen{f}{\tau}{S}{j} - j$, $t =
  \rendgen{\tau}{s}{j} - j - 3\tau + 2$, and $b = e - t$. Let $H \in
  \Sigma^{+}$.  We define
  \begin{align*}
    \mathcal{I}^{-}_{f,H}(\tau,S) &
      :=\{I_{j}(\tau,S): j \in \R'^{-}_{f,H}(\tau,S)\},\\ 
    \mathcal{I}^{+}_{f,H}(\tau,S) &
      :=\{I_{j}(\tau,S): j \in \R'^{+}_{f,H}(\tau,S)\}.
  \end{align*}
\end{definition}

We have now defined all the necessary notation to express the
assumption that we will use in the rest of this section.

\begin{assumption}\label{as:periodic}
  For $\tau = \lfloor \tfrac{\ell}{3} \rfloor$, some
  necklace-consistent function $f$, and any $H \in \Sigma^{+}$, the
  following queries on text $\T$ can be supported in $\bigO(t)$ time:
  \begin{enumerate}
  \item\label{as:periodic-it-1}Given any $j \in \R(\tau,\T)$, return
    $|\Lrootgen{f}{\tau}{\T}{j}|$, $\Lheadgen{f}{\tau}{\T}{j}$, and
    $\rendgen{\tau}{\T}{j}$.
  \item\label{as:periodic-it-2}Range queries (\cref{prob:int-str}) on
    $\Points_{7\tau}(\T, E^{-}_{f,H}(\tau,\T))$ and
    $\Points_{7\tau}(\T, E^{+}_{f,H}(\tau,\T))$.
  \item\label{as:periodic-it-3}Modular constraint queries
    (\cref{sec:mod-queries}) on $\mathcal{I}^{-}_{f,H}(\tau,\T)$ and
    $\mathcal{I}^{+}_{f,H}(\tau,\T)$.
  \end{enumerate}
  The string $H$ for queries 2.\ and 3.\ is specified by some $i, p
  \in [1 \dd n]$ such that $H = \T[i \dd i + p)$.
\end{assumption}

The above notation was introduced in the general form, because it is
needed in the section implementing \cref{as:periodic} which supports
these queries for \emph{collections of strings}.  For the rest of the
section, we define $\Lroot(j)$, $\Lhead(j)$, $\Lexp(j)$,
$\Lexpcut{j}{t}$, $\Ltail(j)$, $\rend{j}$, $\rendfull{j}$,
$\rendlow{j}$, $\rendhigh{j}$, $\rendcut{j}{t}$, and $\type(j)$ as a
shorthand for, respectively, $\Lrootgen{f}{\tau}{S}{j}$,
$\Lheadgen{f}{\tau}{S}{j}$, $\Lexpgen{f}{\tau}{S}{j}$,
$\Lexpcutgen{f}{\tau}{S}{j}{t}$, $\Ltailgen{f}{\tau}{S}{j}$,
$\rendgen{\tau}{S}{j}$, $\rendfullgen{f}{\tau}{S}{j}$,
$\rendlowgen{f}{\tau}{S}{j}$, $\rendhighgen{f}{\tau}{S}{j}$,
$\rendcutgen{f}{\tau}{S}{j}{t}$, and $\typegen{\tau}{S}{j}$ with $S =
\T$, $\tau = \lfloor \tfrac{\ell}{3} \rfloor$, and $f$ being some
necklace-consistent function (recall that $\T$ is the main text
defined globally all through the paper).

We also define $\R$, $\R^{-}$, $\R^{+}$, $\R_H$, $\R^{-}_H$,
$\R^{+}_H$, $\R_{s,H}$, $\R^{-}_{s,H}$, $\R^{+}_{s,H}$, $\R'$,
$\R'^{-}$, $\R'^{+}$, $\R'^{-}_{s,H}$, $\R'^{+}_{s,H}$, $E^{-}_{H}$,
$E^{+}_{H}$, $\mathcal{I}^{-}_{H}$, $\mathcal{I}^{+}_{H}$, and
$\Lroots$ as a shorthand notation for, respectively, the corresponding
sets $\R(\tau,S)$, $\R^{-}(\tau,S)$, $\R^{+}(\tau,S)$,
$\R_{f,H}(\tau,S)$, $\R^{-}_{f,H}(\tau,S)$, $\R^{+}_{f,H}(\tau,S)$,
$\R_{f,s,H}(\tau,S)$, $\R^{-}_{f,s,H}(\tau,S)$,
$\R^{+}_{f,s,H}(\tau,S)$, $\R'(\tau,S)$, $\R'^{-}(\tau,S)$,
$\R'^{+}(\tau,S)$, $\R'^{-}_{f,s,H}(\tau,S)$,
$\R'^{+}_{f,s,H}(\tau,S)$, $E^{-}_{f,H}(\tau,S)$,
$E^{+}_{f,H}(\tau,S)$, $\mathcal{I}^{-}_{f,H}(\tau,S)$,
$\mathcal{I}^{+}_{f,H}(\tau,S)$, and finally $\Lrootsgen{f}{\tau}{S}$
with $S = \T$, $\tau = \lfloor \tfrac{\ell}{3} \rfloor$, and $f$ being
some necklace-consistent function.

\subsubsection{Decomposition of Pos}\label{sec:sa-periodic-pos-decomposition}

Let $j \in \R^{-}$. In this section we introduce the three sets
$\Poslow(j)$, $\Posmid(j)$, and $\Poshigh(j)$ playing a central role
in the SA query algorithm (positions $j \in \R^{+}$ are processed
symmetrically; the details are provided in the proof of
\cref{pr:sa-periodic}).

The section is organized as follows. First, we define $\Poslow(j)$,
$\Posmid(j)$, and $\Poshigh(j)$. We then prove the central
combinatorial result (\cref{lm:delta}) of the $\SA$ query, showing how
the size $\deltal_{\ell}(j) = |\Posalll(j)|$ relates to the sizes of
the other three sets.

\begin{definition}
  Assume $H \in \Lroots$, $s \in [0 \dd |H|)$, and $j \in
  \R_{s,H}^{-}$. Denote $k_1 := \Lexpcut{j}{\ell} = \min(\Lexp(j),
  \lfloor \frac{\ell-s}{|H|} \rfloor)$ and $k_2 := \Lexpcut{j}{2\ell}
  = \min(\Lexp(j), \lfloor \frac{2\ell-s}{|H|} \rfloor)$.  We define
  \begin{align*}
    \Poslow(j) &=
      \{j' \in \R_{s,H}^{-} : \Lexp(j') = k_1\text{ and }(\T[j' \dd n]
      \succeq \T[j \dd n]\text{ or }\LCE_{\T}(j,j') \geq \ell)\},\\
    \Posmid(j) &=
      \{j' \in \R_{s,H}^{-} : \Lexp(j') \in (k_1 \dd k_2]\},\text{
      and}\\
    \Poshigh(j) &=
      \{j' \in \R_{s,H}^{-} : \Lexp(j') = k_2\text{ and }(\T[j' \dd n]
      \succeq \T[j \dd n]\text{ or }\LCE_{\T}(j, j') \geq 2\ell)\}.
  \end{align*}
  We denote $\deltalow(j) = |\Poslow(j)|$, $\deltamid(j) =
  |\Posmid(j)|$, and $\deltahigh(j) = |\Poshigh(j)|$.
\end{definition}

\begin{lemma}\label{lm:delta}
  For any $j \in \R^{-}$, it holds $\deltal_{\ell}(j) = \deltalow(j) +
  \deltamid(j) - \deltahigh(j)$.
\end{lemma}
\begin{proof}

  By definition, it holds $\Poslow(j) \cap \Posmid(j) = \emptyset$.
  On the other hand, by definition, we also have $\Posalll(j) \cap
  \Poshigh(j) = \emptyset$.  The main strategy of the proof is to show
  the equality $\Poslow(j) \cup \Posmid(j) = \Posalll(j) \cup
  \Poshigh(j)$. Since both elements of the equation are disjoint
  unions, this implies
  \begin{align*}
    \deltalow(j) + \deltamid(j)
      &= |\Poslow(j)| + |\Posmid(j)|
       = |\Poslow(j) \cup \Posmid(j)|\\
      &= |\Posalll(j) \cup \Poshigh(j)|
       = |\Posalll(j)| + |\Poshigh(j)|\\
      &= \deltal_{\ell}(j) + \deltahigh(j),
  \end{align*}
  which yields the claim. It thus remains to show $\Poslow(j) \cup
  \Posmid(j) = \Posalll(j) \cup \Poshigh(j)$. Let $H \in \Lroots$ and
  $s \in [0 \dd |H|)$ be such that $j \in \R_{s,H}^{-}$. Before we
  move on to proving the main equality, we observe that since for any
  $j' \in \Posalll(j)$, it holds $\LCE_{\T}(j, j') \geq \ell \geq
  3\tau - 1$, by \cref{lm:exp}, it holds $j' \in \R_{s,H}$ and
  $\type(j') = -1$. Thus, we can equivalently write $\Posalll(j) =
  \{j' \in \R_{s,H}^{-} : \T[j' \dd n] \prec \T[j \dd n]\text{ and }
  \LCE_{\T}(j, j') \in [\ell \dd 2\ell)\}$.

  We first show the inclusion $\Poslow(j) \cup \Posmid(j) \sub
  \Posalll(j) \cup \Poshigh(j)$.  Let us first take a position $j' \in
  \Poslow(j)$. Observe that it holds $k_1 \leq k_2$. We consider two
  cases:
  \begin{itemize}
  \item If $k_1 < k_2$, or equivalently, $\min(\Lexp(j), \lfloor
    \frac{\ell-s}{|H|} \rfloor) < \min(\Lexp(j), \lfloor
    \frac{2\ell-s}{|H|} \rfloor)$ then $\lfloor \frac{\ell-s}{|H|}
    \rfloor < \Lexp(j)$ and hence $k_1 = \lfloor \frac{\ell-s}{|H|}
    \rfloor < \Lexp(j)$. Therefore, since for $j'$ we have $\type(j')
    = \type(j) = -1$ and $\rend{j'} - j' = s + \Lexp(j')|H| +
    \Ltail(j') = s + k_1|H| + \Ltail(j') < s + (k_1 + 1)|H| \leq s +
    \Lexp(j)|H| \leq s + \Lexp(j)|H| + \Ltail(j) = \rend{j} - j$, we
    obtain from \cref{lm:SA-block-2} of \cref{lm:exp} that $\T[j' \dd
    n] \prec \T[j \dd n]$.  Thus, by definition of $\Poslow(j)$ we
    must also have $\LCE_{\T}(j,j') \geq \ell$. On the other hand, by
    $j,j' \in \R_{s,H}$ and $\rend{j'} - j' < \rend{j} - j$, we have
    $\T[j \dd j + t) = \T[j' \dd j' + t)$ (where $t = \rend{j'} - j'$)
    and $\T[j' + t] \neq \T[j' + t - |H|] = \T[j + t - |H|] = \T[j +
    t]$.  Thus, $\LCE_{\T}(j, j') = \rend{j'} - j' = s + k_1|H| +
    \Ltail(j') = s + \lfloor \frac{\ell - s}{|H|} \rfloor|H| +
    \Ltail(j') \leq \ell + \Ltail(j') < \ell + \tau < 2\ell$. We thus
    proved $\T[j' \dd n] \prec \T[j \dd n]$ and $\LCE_{\T}(j, j') \in
    [\ell \dd 2\ell)$, i.e., $j' \in \Posalll(j)$.
  \item If $k_1 = k_2$, then it suffices to consider two subcases. If
    $\T[j' \dd n] \succeq \T[j \dd n]$ or $\LCE_{\T}(j, j') \geq
    2\ell$, then we immediately obtain $j' \in \Poshigh(j)$.  The
    other possibility is that $\T[j' \dd n] \prec \T[j \dd n]$ and
    $\LCE_{\T}(j, j') < 2\ell$. Combining this with $\LCE_{\T}(j, j')
    \geq \ell$ (from the definition of $\Poslow(j)$), we obtain $j'
    \in \Posalll(j)$.
  \end{itemize}
  We have thus proved that $\Poslow(j) \sub \Posalll(j) \cup
  \Poshigh(j)$. Next, we show $\Posmid(j) \sub \Posalll(j) \cup
  \Poshigh(j)$. Consider $j' \in \Posmid(j)$. If $k_1 = k_2$, the
  claim follows trivially, since $\Posmid(j) = \emptyset$. Let us thus
  assume $k_1 < k_2$. As noted above, this implies $k_1 = \lfloor
  \frac{\ell - s}{|H|} \rfloor < \Lexp(j)$. Then, by $\Lexp(j') >
  k_1$, we obtain $\rend{j'} - j' = s + \Lexp(j')|H| + \Ltail(j') \geq
  s + (k_1 + 1)|H| = (s + |H|) + \lfloor \frac{\ell - s}{|H|} \rfloor
  |H| \geq (s + |H|) + (\ell - s - |H|) = \ell$. Similarly, by
  $\lfloor \frac{\ell-s}{|H|} \rfloor < \Lexp(j)$, we have $\rend{j} -
  j = s + \Lexp(j)|H| + \Ltail(j) \geq s + (\lfloor \frac{\ell-s}{|H|}
  \rfloor + 1)|H| \geq \ell$.  Thus, by definition of
  run-decomposition, we obtain $\LCE_{\T}(j, j') \geq \min(\rend{j} -
  j, \rend{j'} - j') \geq \ell$. Recall now that we have $\Lexp(j')
  \in (k_1 \dd k_2]$. Consider now two cases:
  \begin{itemize}
  \item Let $\Lexp(j') = k_2$. If $\T[j' \dd n] \succeq \T[j \dd n]$
    or $\LCE_{\T}(j, j') \geq 2\ell$, then we immediately obtain $j'
    \in \Poshigh(j)$.  The remaining case is $\T[j' \dd n] \prec \T[j
    \dd n]$ and $\LCE_{\T}(j, j') < 2\ell$.  Combining with the above
    observation $\LCE_{\T}(j, j') \geq \ell$, we obtain $j' \in
    \Posalll(j)$.
  \item Let $\Lexp(j') \in (k_1 \dd k_2)$. We will show that it holds
    $j' \in \Posalll(j)$.  By $\Lexp(j') < k_2 = \min(\Lexp(j),
    \lfloor \frac{2\ell-s}{|H|} \rfloor)$, it follows in particular
    that $\Lexp(j') < \lfloor \frac{2\ell-s}{|H|} \rfloor$.  Thus,
    $\rend{j'} - j' = s + \Lexp(j')|H| + \Ltail(j') < s + (\Lexp(j') +
    1)|H| \leq s + \lfloor \frac{2\ell - s}{|H|} \rfloor|H| \leq
    2\ell$.  On the other hand, by $\Lexp(j') < k_2$, it follows
    $\Lexp(j') < \Lexp(j)$.  Thus, $\rend{j'} - j' = s + \Lexp(j')|H|
    + \Ltail(j') < s + (\Lexp(j') + 1)|H| \leq s + \Lexp(j)|H| \leq s
    + \Lexp(j)|H| + \Ltail(j) = \rend{j} - j$. Therefore, for $t =
    \rend{j'} - j'$, we have $\T[j' \dd j' + t) = \T[j \dd j + t)$ and
    $\T[j' + t] \neq \T[j' + t - |H|] = \T[j + t - |H|] = \T[j +
    t]$. Thus, $\LCE_{\T}(j, j') = \rend{j'} - j' < 2\ell$. Combining
    with the above observation (holding for all $j' \in (k_1 \dd
    k_2]$) $\LCE_{\T}(j, j') \geq \ell$, we obtain $\LCE_{\T}(j, j')
    \in [\ell \dd 2\ell)$.  Finally, by $\rend{j'} - j' < \rend{j} -
    j$ and $\type(j') = \type(j) = -1$, it follows from
    \cref{lm:SA-block-2} of \cref{lm:exp}, that $\T[j' \dd n] \prec
    \T[j \dd n]$. We have thus proved $j' \in \Posalll(j)$.
  \end{itemize}
  Thus, $\Posmid(j) \sub \Posalll(j) \cup \Poshigh(j)$. Combining with
  the above proof of $\Poslow(j) \sub \Posalll(j) \cup \Poshigh(j)$,
  this establishes the inclusion $\Poslow(j) \cup \Posmid(j) \sub
  \Posalll(j) \cup \Poshigh(j)$.

  We now show the opposite inclusion, i.e., $\Posalll(j) \cup
  \Poshigh(j) \sub \Poslow(j) \cup \Posmid(j)$. Let us first assume
  $j' \in \Posalll(j)$. By $j \in \R_{s,H}^{-}$ and $\T[j' \dd n]
  \prec \T[j \dd n]$, we obtain from \cref{lm:SA-block-2} of
  \cref{lm:exp}, that $\rend{j'} - j' \leq \rend{j} - j$. Therefore,
  $\Lexp(j') = \lfloor \tfrac{\rend{j'} - j' - s}{|H|} \rfloor \leq
  \lfloor \frac{\rend{j} - j - s}{|H|} \rfloor = \Lexp(j)$. On the
  other hand, by definition of the run-decomposition, we have
  $\LCE_{\T}(j, j') \geq \min(\rend{j'} - j', \rend{j} - j)$. Thus, by
  $\rend{j'} - j' < \rend{j} - j$ and $\LCE_{\T}(j, j') < 2\ell$, we
  obtain $\rend{j'} - j' = \min(\rend{j'} - j', \rend{j} - j) \leq
  \LCE_{\T}(j, j') < 2\ell$, and consequently, $\Lexp(j') = \lfloor
  \frac{\rend{j'} - j' - s}{|H|} \rfloor \leq \lfloor \frac{2\ell -
  s}{|H|} \rfloor$.  Combining the two upper bounds on $\Lexp(j')$, we
  thus obtain $\Lexp(j') \leq \min(\Lexp(j), \lfloor \frac{2\ell -
  s}{|H|} \rfloor) = k_2$. Next, we prove $\Lexp(j') \geq k_1$. For
  this, we consider two cases:
  \begin{itemize}
  \item Let us first assume $\rend{j} - j < \ell$. By definition of
    $k_1$, we then have $k_1 = \min(\Lexp(j), \lfloor \tfrac{\ell -
    s}{|H|} \rfloor) = \min(\lfloor \frac{\rend{j} - j - s}{|H|}
    \rfloor, \lfloor \frac{\ell - s}{|H|} \rfloor) = \lfloor
    \frac{\rend{j} - j - s}{|H|} \rfloor = \Lexp(j)$. Observe now that
    by $\LCE_{\T}(j, j') \geq \ell$, we have $\T[j \dd j + \ell) =
    \T[j' \dd j' + \ell)$. Recall that by~\cite[Fact~3.2]{sss}, $\T[j
    \dd \rend{j})$ is the longest prefix of $\T[j \dd n]$ having
    period $|\Lroot(j)|$.  Therefore, since $\Lroot(j') = \Lroot(j)$
    and $\T[j \dd \rend{j})$ is a proper prefix of $\T[j \dd j +
    \ell)$, applying this equivalent definition to $\T[j' \dd n]$, we
    obtain $\rend{j'} - j' = \rend{j} - j$. Thus, $\Lexp(j') = \lfloor
    \frac{\rend{j'} - j' - s}{|H|} \rfloor = \lfloor \frac{\rend{j} -
    j - s}{|H|} \rfloor = \Lexp(j) = k_1$.
  \item Let us now assume $\rend{j} - j \geq \ell$. Then, by the above
    alternative definition of $\T[j \dd \rend{j})$, $\T[j \dd j +
    \ell)$ has a period $|\Lroot(j)|$. Thus, by $\Lroot(j') =
    \Lroot(j)$ and $\T[j \dd j + \ell) = \T[j' \dd j' + \ell)$, the
    string $\T[j' \dd j' + \ell)$ has period $|\Lroot(j')|$, and
    consequently, $\rend{j'} - j' \geq \ell$.  Therefore, $\Lexp(j') =
    \lfloor \frac{\rend{j'} - j' - s}{|H|} \rfloor \geq \lfloor
    \frac{\ell - s}{|H|} \rfloor = \min(\lfloor \frac{\ell - s}{|H|}
    \rfloor, \lfloor \frac{\rend{j} - j - s}{|H|} \rfloor) = k_1$.
  \end{itemize}
  We have thus shown than in both cases, it holds $\Lexp(j') \geq
  k_1$.  Combining with the earlier upper bound on $\Lexp(j')$, we
  therefore obtain $\Lexp(j') \in [k_1 \dd k_2]$. To see that this
  immediately implies the inclusion $\Posalll(j) \sub \Poslow(j) \cup
  \Posmid(j)$, it suffices to consider two cases. If $\Lexp(j') >
  k_1$, then $j' \in \Posmid(j)$ holds by definition of
  $\Posmid(j)$. On the other hand, if $\Lexp(j') = k_1$, then $j' \in
  \Poslow(j)$ follows from $\LCE_{\T}(j, j') \in [\ell \dd 2\ell)$. It
  remains to show $\Poshigh(j) \sub \Poslow(j) \cup \Posmid(j)$.  Let
  $j' \in \Poshigh(j)$. We again consider two cases. If $k_1 < k_2$,
  then we immediately have $j' \in \Posmid(j)$ since $\Lexp(j') \in
  (k_1 \dd k_2]$. On the other hand, if $k_1 = k_2$, then either
  $\T[j' \dd n] \succeq \T[j \dd n]$ or $\LCE_{\T}(j, j') \geq 2\ell
  \geq \ell$. In either case, $j' \in \Poslow(j)$. This concludes the
  proof of the inclusion $\Posalll(j) \cup \Poshigh(j) \sub \Poslow(j)
  \cup \Posmid(j)$.
\end{proof}

\subsubsection{Computing the Size of Occ}\label{sec:sa-periodic-occ-size}

Let $j \in \R$ and $d \geq \ell$. We define $\Occm_{d}(j) :=
\Occ_{d}(j) \cap \R^{-}$ and $\Occp_{d}(j) := \Occ_{d}(j) \cap
\R^{+}$.  In this section, we show how under \cref{as:periodic}, given
any position $j \in \R$, to efficiently compute the cardinalities of
the four sets $\Occm_{\ell}$, $\Occp_{\ell}(j)$, $\Occm_{2\ell}(j)$,
and $\Occp_{2\ell}(j)$. In the rest of this section, we focus on the
computation of $|\Occm_{\ell}(j)|$ and $|\Occm_{2\ell}(j)|$. The
cardinalities of $\Occp_{\ell}(j)$ and $\Occp_{2\ell}(j)$ are computed
analogously (see the proof of \cref{pr:sa-periodic}).

The section is organized into four parts. First, we prove the
combinatorial results (\cref{lm:occ-exp,cor:occ}) characterizing the
set $\Occm_{d}(j)$, $d \geq \ell$, as a disjoint union of two sets
$\Occeqm_{d}(j)$ and $\Occgtm_{d}(j)$ (in particular, each of the sets
$\Occm_{\ell}(j)$ and $\Occm_{2\ell}(j)$ admits such
decomposition). In the following two parts, we present a query
algorithms (\cref{pr:sa-occ-eq} and \cref{pr:sa-occ-gt}) to compute
the cardinality of $\Occeqm_{d}(j)$ for any $d \in [\ell \dd 2\ell]$
(in particular, we can compute $|\Occeqm_{\ell}(j)|$ and
$|\Occeqm_{2\ell}(j)|$) and the cardinality of $\Occgtm_{d}(j)$ for
any $d \geq \ell$ (in particular, we can compute $|\Occgtm_{\ell}(j)|$
and $|\Occgtm_{2\ell}(j)|$). In \cref{pr:sa-occ-size}, we put them
together to obtain the final query algorithm.

\paragraph{The Main Idea}

Let $d \geq \ell$.  Assume $H \in \Lroots$, $s \in [0 \dd |H|)$, and
$j \in \R_{s,H}$. Denote $k = \Lexpcut{j}{d} = \min(\Lexp(j), \lfloor
\frac{d - s}{|H|} \rfloor)$. We define
\begin{align*}
  \Occeqm_{d}(j)
    &= \{j' \in \R_{s,H} \cap \Occm_{d}(j) : \Lexp(j') = k\},\\
  \Occgtm_{d}(j)
    &= \{j' \in \R_{s,H} \cap \Occm_{d}(j) : \Lexp(j') > k\}.
\end{align*}

\begin{lemma}\label{lm:occ-exp}
  Assume $d \geq \ell$. If $j \in R_{s,H}$ then for any $j' \in
  \Occ_{d}(j)$, it holds:
  \begin{enumerate}
  \item\label{lm:occ-exp-item-1} It holds $j' \in \R^{-}_{s,H}$,
  \item\label{lm:occ-exp-item-2} $\rendcut{j'}{d} - j' =
    \rendcut{j}{d} - j$, and
  \item\label{lm:occ-exp-item-3} $T^{\infty}[\rendcut{j'}{d} \dd j' +
    d) = T^{\infty}[\rendcut{j}{d} \dd j + d)$.
  \end{enumerate}
\end{lemma}
\begin{proof}

  1. Consider two cases. If $j' = j$, then immediately $j' \in \R$,
  $\Lhead(j') = s$ and $\Lroot(j') = H$. Otherwise, by definition of
  $\Occ_{d}(j)$, it holds $\T^{\infty}[j \dd j + d) = \T^{\infty}[j'
  \dd j' + d)$. By $j \neq j'$ and the uniqueness of $\T[n] =
  \texttt{\$}$, this is equivalent to $\LCE_{\T}(j, j') \geq d$. Since
  by definition of $\tau$, it holds $3\tau \leq \ell \leq d$, we
  obtain $\LCE_{T}(j, j') \geq 3\tau - 1$. Thus, by \cref{lm:exp}, $j'
  \in \R_{s,H}$.

  2. We first prove that $\Lexpcut{j'}{d} = \Lexpcut{j}{d}$.  Recall
  that by~\cite[Fact~3.2]{sss}, for $i \in \R$, the substring $\T[i
  \dd \rend{i})$ is the longest prefix of $\T[i \dd n]$ that has a
  period $|\Lroot(i)|$. By $\T^{\infty}[j \dd j + d) = \T^{\infty}[j'
  \dd j' + d)$, this immediately implies $\min(\rend{j'} - j', d) =
  \min(\rend{j} - j, d)$.  Therefore, we obtain
  \begin{align*}
    \Lexpcut{j'}{d}
      &= \min(\Lexp(j'), \lfloor \tfrac{d - s}{|H|} \rfloor)
       = \min(\lfloor \tfrac{\rend{j'} - j' - s}{|H|} \rfloor,
              \lfloor \tfrac{d - s}{|H|} \rfloor)\\
      &= \lfloor \tfrac{\min(\rend{j'} - j', d) - s}{|H|} \rfloor
       = \lfloor \tfrac{\min(\rend{j} - j, d) - s}{|H|} \rfloor\\
      &= \min(\lfloor \tfrac{\rend{j} - j - s}{|H|} \rfloor,
              \lfloor \tfrac{d - s}{|H|} \rfloor)
       = \min(\Lexp(j), \lfloor \tfrac{d - s}{|H|} \rfloor)\\
      &= \Lexpcut{j}{d}
  \end{align*}
  Thus, $\rendcut{j'}{d} - j' = \Lhead(j') + \Lexpcut{j'}{d}|H| =
  \Lhead(j) + \Lexpcut{j}{d}|H| = \rendcut{j}{d} - j$.

  3. By $j' \in \Occ_{d}(j)$, $\T^{\infty}[j' \dd j' + d) =
  \T^{\infty}[j \dd j + d)$. By \cref{lm:occ-exp-item-2} this yields
  the claim.
\end{proof}

\begin{corollary}\label{cor:occ}
  For any $d \geq \ell$ and any $j \in \R$, the set $\Occm_{d}(j)$ is
  a disjoint union of $\Occeqm_{d}(j)$ and $\Occgtm_{d}(j)$.
\end{corollary}
\begin{proof}
  By definition, it holds $\Occeqm_{d}(j) \cap \Occgtm_{d}(j) =
  \emptyset$ and $\Occeqm_{d}(j) \cup \Occgtm_{d}(j) \sub
  \Occm_{d}(j)$. It remains to show $\Occm_{d}(j) \sub \Occeqm_{d}(j)
  \cup \Occgtm_{d}(j)$.  Let $j' \in \Occm_{d}(j)$.  By
  \cref{lm:occ-exp}, this implies $j' \in \R_{s,H}$ and
  $\Lexpcut{j'}{d} = \Lexpcut{j}{d}$. Thus, by $\Lexp(j') \geq
  \min(\Lexp(j'), \lfloor \tfrac{d-s}{|H|} \rfloor) =
  \Lexpcut{j'}{d}$, we obtain that $j' \in \Occeqm_{d}(j)$ or $j' \in
  \Occgtm_{d}(j)$.
\end{proof}

\begin{remark}

  Note that the differentiation between $j \in \R$ satisfying
  $\type(j) = -1$ and $\type(j) = +1$ used (without the loss of
  generality) in \cref{sec:sa-periodic-pos-decomposition} is different
  from the symmetry used in this section.  Here we partition the
  output set $\Occeq_{2\ell}(j)$ (resp.\ $\Occgt_{2\ell}(j)$) into two
  subsets $\Occeqm_{2\ell}(j)$ and $\Occeqp_{2\ell}(j)$
  (resp.\ $\Occgtm_{2\ell}(j)$ and $\Occgtp_{2\ell}(j)$), but the
  computation is always performed regardless of $\type(j)$, leading to
  two queries for each $j \in \R$. In
  \cref{sec:sa-periodic-pos-decomposition}, on the other hand, the
  computation is performed separately for $j \in \R^{-}$ and $j \in
  \R^{+}$, without the need to partition $\Posalll(j)$ within each
  case, leading to a single query but only on the appropriate
  structure depending on $\type(j)$.

  This subtle difference follows from the fact, then when computing
  $\deltal_{\ell}(j) = |\Posalll(j)|$ by \cref{lm:exp} we have
  $\Posalll(j) \sub \R^{-}$ for any $j \in \R^{-}$ (and $\Posalll(j)
  \sub \R^{+}$ for $j \in \R^{+}$). However, when computing $|{\rm
  Occ}_{2\ell}(j)|$ for $j \in \R^{-}$, it is possible that ${\rm
  Occ}_{2\ell}(j) \cap \R^{-} \neq \emptyset$ and ${\rm
  Occ}_{2\ell}(j) \cap \R^{+} \neq \emptyset$ hold. This is the reason
  for why the seemingly related computation of $|{\rm
  Occ}_{2\ell}(j)|$ and $|\Posalll(j)|$ is (unlike for the nonperiodic
  positions; see \cref{sec:sa-nonperiodic-pos-and-occ-size}) described
  separately.
\end{remark}

\paragraph{Computing $|\Occeqm_{\ell}(j)|$ and $|\Occeqm_{2\ell}(j)|$}

We now describe the query algorithm that, given any position $j \in
\R$ and any $d \in [\ell \dd 2\ell]$, returns the cardinality of the
sets $\Occeqm_{d}(j)$ (in particular, it can return the cardinality of
$\Occeqm_{\ell}(j)$ and $\Occeqm_{2\ell}(j)$).  We start with a
combinatorial result (\cref{lm:occ-eq}) that shows how to count the
elements of $\Occeqm_{d}(j)$ (where $d \in [\ell \dd 2\ell]$) that
belong to any right-maximal contiguous block of positions in
$\R^{-}$. The query algorithm to compute the size of $\Occeqm_{d}(j)$
is presented next (\cref{pr:sa-occ-eq}).

\begin{lemma}\label{lm:occ-eq}
  Let $d \in [\ell \dd 2\ell]$ and $j \in \R_H$.  Assume $i \in
  \R^{-}_{H}$ and denote $t = \rend{i} - i - 3\tau + 2$.  Then,
  $|\Occeqm_{d}(j) \cap [i \dd i + t)| \leq 1$.  Moreover,
  $|\Occeqm_{d}(j) \cap [i \dd i + t)| = 1$ holds if and only if
  $\rendfull{i} - i \geq \rendcut{j}{d} - j$ and
  $\T^{\infty}[\rendcut{j}{d} \dd j + d)$ is a prefix of
  $\T^{\infty}[\rendfull{i} \dd \rendfull{i} + 7\tau)$.
\end{lemma}
\begin{proof}

  Denote $k_d = \Lexpcut{j}{d}$.  Recall that $\rend{i} = \min\{i' \in
  [i \dd n] : i' \not\in\R\} + 3\tau - 2$.  Thus, $t > 0$ and $[i \dd
  \rend{i} - 3\tau + 2) = [i \dd i + t) \sub \R$. For any $\delta \in
  [0 \dd t)$, by \cref{lm:R-block}, we have $\Lroot(i + \delta) =
  \Lroot(i)$ and $\type(i + \delta) = \type(i)$. Thus, $[i \dd i + t)
  \sub \R^{-}_H$.  Moreover, by \cref{lm:R-block}, $\rend{i + \delta}
  = \rend{i}$. Therefore, by the uniqueness of run-decomposition, it
  holds $\Ltail(i + \delta) = \Ltail(i)$. Together, these facts imply
  $\rendfull{i + \delta} = \rendfull{i}$, and consequently, that
  $\rendfull{i + \delta} - (i + \delta) = \rendfull{i} - i -
  \delta$. Let $s = \Lhead(j)$. Recall that for any $j' \in {\rm
  Occ}_{d}^{{\rm eq}-}(j)$, we have $\rendfull{j'} - j' = s +
  \Lexp(j')|H| = s + k_d(j)|H| = \rendcut{j}{d} - j$.  Thus, $i +
  \delta \in \Occeqm_{d}(j)$ implies $\rendfull{i + \delta} - (i +
  \delta) = \rendfull{i} - i - \delta = \rendcut{j}{d} - j$, or
  equivalently, $\delta = (\rendfull{i} - i) - (\rendcut{j}{d} - j)$,
  and thus $|\Occeqm_{d}(j) \cap [i \dd i + t)| \leq 1$.

  We now prove the equivalence. Let us first assume $|{\rm
  Occ}_{d}^{{\rm eq}-}(j) \cap [i \dd i + t)| = 1$, i.e., that there
  exists $\delta \in [0 \dd t)$ such that $i + \delta \in {\rm
  Occ}_{d}^{{\rm eq}-}(j)$. Then, as observed above, $\rendcut{j}{d} -
  j = \rendfull{i} - i - \delta \leq \rendfull{i} - i$.  This proves
  the first condition. To show the second one, let $s = \Lhead(j)$. By
  definition of $\Occeqm_{d}(j)$, it holds $i + \delta \in
  \R^{-}_{s,H}$, $\Lexp(i + \delta) = k$, and $\T^{\infty}[i + \delta
  \dd i + \delta + d) = \T^{\infty}[j \dd j + d)$.  Therefore, by $s +
  k_d|H| \leq d$, we obtain $\T^{\infty}[i + \delta \dd \rendfull{i +
  \delta}) = \T^{\infty}[i + \delta \dd \rendfull{i}) = \T[j \dd
  \rendcut{j}{d}) = H'H^{k}$, where $H'$ is a length-$s$ suffix of
  $H$. Thus, $\T^{\infty}[\rendfull{i} \dd i + \delta + d) =
  \T^{\infty}[\rendcut{j}{d} \dd j + d)$. Therefore, by $i + \delta +
  d \leq \rendfull{i} + 7\tau$, $\T^{\infty}[\rendcut{j}{d} \dd j +
  d)$ is a prefix of $\T^{\infty}[\rendfull{i} \dd \rendfull{i} +
  7\tau)$.

  To prove the opposite implication, assume $\rendfull{i} - i \geq
  \rendcut{j}{d} - j$ and that $\T^{\infty}[\rendcut{j}{d} \dd j + d)$
  is a prefix of $\T^{\infty}[\rendfull{i} \dd \rendfull{i} + 7\tau)$.
  Let $\delta = (\rendfull{i} - i) - (\rendcut{j}{d} - j)$. We will
  show that $\delta \in [0 \dd t)$ and $i + \delta \in {\rm
  Occ}_{d}^{{\rm eq}-}(j)$. The inequality $\delta \geq 0$ follows
  from the definition of $\delta$ and our assumptions. To show $\delta
  < t$, we consider two cases:
  \begin{itemize}
  \item Let $\rend{j} - j < d$. We start by noting $k_d =
    \min(\Lexp(j), \lfloor \frac{d-s}{|H|} \rfloor) = \min(\lfloor
    \frac{\rend{j} - j - s}{|H|} \rfloor, \lfloor \frac{d - s}{|H|}
    \rfloor) = \lfloor \frac{\rend{j} - j - s}{|H|} \rfloor =
    \Lexp(j)$. Thus, $\rendcut{j}{d} = j + s + k_d|H| = j + s +
    \Lexp(j)|H| = \rendfull{j}$.  We next show that it holds $\rend{i}
    - \rendfull{i} \geq \rend{j} - \rendfull{j}$. Let $q = \rend{i} -
    \rendfull{i}$ and suppose $q < \rend{j} - \rendfull{j}$.  Then, by
    definition of run-decomposition, it holds $\LCE_{\T}(\rendfull{i},
    \rendfull{j}) \geq \min(\rend{i} - \rendfull{i}, \rend{j} -
    \rendfull{j}) = q$.  We claim that it holds $\T[\rendfull{i} + q]
    \prec \T[\rendfull{j} + q]$. To show this, we first note that by
    $\T[n] = \texttt{\$}$ being unique in $\T$, we have $\rendfull{i}
    + q = \rend{i} \leq n$. We then consider two subcases:
    \begin{enumerate}
    \item If $\rendfull{i} + q = n$, then by $\rendfull{j} + q <
      \rend{j} \leq n$, we have $\T[\rendfull{i} + q] = \texttt{\$}
      \prec \T[\rendfull{j} + q]$.
    \item On the other hand, if $\rendfull{i} + q < n$, then by $i \in
      \R^{-}$, we again have $\T[\rend{i}] = \T[\rendfull{i} + q]
      \prec \T[\rendfull{i} + q - |H|] = \T[\rendfull{j} + q - |H|] =
      \T[\rendfull{j} + q]$.
    \end{enumerate}
    We have thus obtained $\T[\rendfull{i} \dd \rendfull{i} + q] \prec
    \T[\rendfull{j} \dd \rendfull{j} + q]$. By $q < \rend{j} -
    \rendfull{j} < j + d - \rendfull{j} = j + d - \rendcut{j}{d} \leq
    d \leq 2\ell \leq 7\tau$, this contradicts
    $\T^{\infty}[\rendcut{j}{d} \dd j + d)$ being a prefix of
    $\T^{\infty}[\rendfull{i} \dd \rendfull{i} + 7\tau)$.  Thus, it
    holds $\rend{i} - \rendfull{i} \geq \rend{j} - \rendfull{j}$.
    Applying the definition of $\delta$ and using the above
    inequality, we obtain
    \begin{align*}
      \rend{i} - (i + \delta)
        &=    (\rendfull{i} - (i + \delta)) + (\rend{i} - \rendfull{i})
         =    (\rendcut{j}{d} - j) + (\rend{i} - \rendfull{i})\\
        &=    (\rendfull{j} - j) + (\rend{i} - \rendfull{i})
         \geq (\rendfull{j} - j) + (\rend{j} - \rendfull{j})\\
        &=    \rend{j} - j
         \geq 3\tau - 1.
    \end{align*}
  \item Let $\rend{j} - j \geq d$. We first show that in this case it
    holds $\rend{i} - \rendfull{i} \geq j + d - \rendcut{j}{d}$.  Let
    $q = \rend{i} - \rendfull{i}$ and suppose $q < j + d -
    \rendcut{j}{d}$. Then, by definition of run-decomposition, it
    holds $\LCE_{\T}(\rendfull{i}, \rendcut{j}{d}) \geq \min(\rend{i}
    - \rendfull{i}, \rend{j} - \rendcut{j}{d}) = q$.  We claim that
    $\T[\rendfull{i} + q] \prec \T[\rendcut{j}{d} + q]$.  To show
    this, we first note that by $\T[n] = \texttt{\$}$, we have
    $\rendfull{i} + q = \rend{i} \leq n$. Thus:
    \begin{enumerate}
    \item If $\rendfull{i} + q = n$, then by $\rendcut{j}{d} + q < d +
      j \leq \rend{j} \leq n$, we have $\T[\rendfull{i} + q] =
      \texttt{\$} \prec \T[\rendcut{j}{d} + q]$.
    \item On the other hand, if $\rendfull{i} + q < n$, then by $i \in
      \R^{-}$, we again have $\T[\rend{i}] = \T[\rendfull{i} + q]
      \prec \T[\rendfull{i} + q - |H|] = \T[\rendcut{j}{d} + q - |H|]
      = \T[\rendcut{j}{d} + q]$.
    \end{enumerate}
    We thus obtained $\T[\rendfull{i} \dd \rendfull{i} + q] \prec
    \T[\rendcut{j}{d} \dd \rendcut{j}{d} + q]$.  By $q < j + d -
    \rendcut{j}{d} \leq d \leq 7\tau$, this contradicts the string
    $\T^{\infty}[\rendcut{j}{d} \dd j + d)$ being a prefix of
    $\T^{\infty}[\rendfull{i} \dd \rendfull{i} + 7\tau)$.  Thus,
    $\rend{i} - \rendfull{i} \geq j + d - \rendcut{j}{d}$.  Applying
    the definition of $\delta$ and using the above inequality, we
    obtain
    \begin{align*}
      \rend{i} - (i + \delta)
        &= (\rendfull{i} - (i + \delta)) + (\rend{i} - \rendfull{i})
         = (\rendcut{j}{d} - j) + (\rend{i} - \rendfull{i})\\
        &\geq (\rendcut{j}{d} - j) + (j + d - \rendcut{j}{d})
         = d \geq \ell \geq 3\tau - 1.
    \end{align*}
  \end{itemize}
  Thus, in both cases, we have shown $\rend{i} - (i + \delta) \geq
  3\tau - 1$, or equivalently, $\delta \leq \rend{i} - i - 3\tau + 1 <
  t$.

  It remains to show $i + \delta \in \Occeqm_{d}(j)$. For this, we
  first note $\rendfull{i + \delta} - (i + \delta) = \rendfull{i} - (i
  + \delta) = \rendcut{j}{d} - j$.  Combining this with $i + \delta, j
  \in \R_H$ gives $\T[i + \delta \dd \rendfull{i + \delta}) = \T[j \dd
  \rendcut{j}{d})$. This implies $\Lhead(i + \delta) = \Lhead(j)$ and
  $\Lexp(i + \delta) = \lfloor (\rendcut{j}{d} - j)/|H| \rfloor =
  k_d$. Moreover, by \cref{lm:R-block}, $\type(i + \delta) = \type(i)
  = -1$. Thus, $i + \delta \in \R^{-}_{s,H}$.  It remains to show
  $\T^{\infty}[i + \delta \dd i + \delta + d) = \T^{\infty}[j \dd j +
  d)$. For this, it suffices to combine $\T^{\infty}[i + \delta \dd
  \rendfull{i}) = \T^{\infty}[j \dd \rendcut{j}{d})$ (shown above) and
  $\T^{\infty}[\rendfull{i} \dd i + \delta + d) =
  \T^{\infty}[\rendcut{j}{d} \dd j + d)$ (following from
  $\T^{\infty}[\rendcut{j}{d} \dd j + d)$ being a prefix of
  $\T^{\infty}[\rendfull{i} \dd \rendfull{i} + 7\tau)$ and $i + \delta
  + d = \rendfull{i} - (\rendcut{j}{d} - j) + d \leq \rendfull{i} + d
  \leq \rendfull{i} + 7\tau$).
\end{proof}

\begin{proposition}\label{pr:sa-occ-eq}
  Under \cref{as:periodic}, given any position $j \in \R$ and an
  integer $d \in [\ell \dd 2\ell]$, we can compute $|\Occeqm_{d}(j)|$
  in $\bigO(t)$ time.
\end{proposition}
\begin{proof}

  The main idea of the query algorithm is to group all positions $i
  \in \R'^{-}$ by $\Lroot(i)$ and then sort all positions within each
  group by $\T^{\infty}[\rendfull{i} \dd \rendfull{i} + 7\tau)$.
  Then, by \cref{lm:occ-eq}, in order to compute $|\Occeqm_{d}(j)|$
  given $j \in \R_{H}$, it suffices to count the all positions $i \in
  \R'^{-}_{H}$ that satisfy $\rendfull{i} - i \geq \rendcut{j}{d} - j$
  and for which $\T^{\infty}[\rendcut{j}{d} \dd j + d)$ is a prefix of
  $\T^{\infty}[\rendfull{i} \dd \rendfull{i} + 7\tau)$.  The latter
  task is equivalent to the general range counting queries
  (\cref{sec:range-queries}) on a set of points $\Pts_H \sub
  \mathcal{X} \times \mathcal{Y}$ (with $\mathcal{X} = \Zz$ and
  $\mathcal{Y} = \Sigma^{*}$) containing the value $\rendfull{i} - i$
  on the $\mathcal{X}$-coordinate and the string
  $\T^{\infty}[\rendfull{i} \dd \rendfull{i} + 7\tau)$ on the
  $\mathcal{Y}$-coordinate for every $i \in \R'^{-}_{H}$.  Note that
  we need to provide efficient range counting queries on each
  collection $\Pts_H$ separately for every $H \in \Lroots$.  Since
  $|H|$ could be large, $H$ cannot be given explicitly during the
  query. Thus, we specify $H$ using its length and the position of
  some occurrence in $\T$.

  We use the following definitions. Let $q = 7\tau$ and $c =
  \max\Sigma$. Then, for any $H \in \Lroots$, we let $\Pts_H =
  \Points_{q}(\T, E^{-}_H)$ (see \cref{sec:sa-periodic-prelim} for the
  definition of $E^{-}_{H}$ and \cref{def:p-right-context} for the
  definition of $\Pts_H$). Note that by $\ell < n$ and $\tau = \lfloor
  \tfrac{\ell}{3} \rfloor$, the value $q = 7\tau \leq 2\ell + \tau <
  3n$ satisfies the requirement in \cref{prob:int-str}.

  Given any position $j \in \R$, we compute $|\Occeqm_{d}(j)|$ as
  follows. First, using \cref{as:periodic} we compute values $s =
  \Lhead(j)$, $p = |\Lroot(j)|$, and $\rend{j}$ in $\bigO(t)$
  time. Note, that then $\Lroot(j) = \T[j + s \dd j + s + p)$, i.e.,
  we have a starting position of an occurrence of $H := \Lroot(j)$ in
  $\T$.  We then calculate $k = \Lexp(j) = \lfloor \tfrac{\rend{j} -
  s}{p} \rfloor$ in $\bigO(1)$ time. Using those values, we further
  calculate $k' = \min(k, \lfloor \frac{d - s}{p} \rfloor)$ and
  $\rendcut{j}{d} = s + k'p$.  By \cref{lm:occ-eq} (this is the place
  we use the assumption $d \in [\ell \dd 2\ell]$) and the definition
  of $\Pts_H$, we now have
  \begin{align*}
    |\Occeqm_{d}(j)| &=
      \rcount{\Pts_H}{\rendcut{j}{d} - j}{n}{Yc^{\infty}} -
      \rcount{\Pts_H}{\rendcut{j}{d} - j}{n}{Y},
  \end{align*}
  where $Y = \T^{\infty}[\rendcut{j}{d} \dd j + d)$, which by
  \cref{as:periodic} we can compute in $\bigO(t)$ time using the query
  defined by \cref{prob:int-str-it-1} of \cref{prob:int-str} with the
  arguments $(i, x, q_r) = (\rendcut{j}{d}, \rendcut{j}{d} - j, j + d
  - \rendcut{j}{d})$. To check that all arguments satisfy the
  requirements of \cref{prob:int-str}, recall that $\ell < n$. Thus,
  $q_r \leq d \leq 2\ell < 2n$. On the other hand, $\rendcut{j}{d} - j
  \leq \rend{j} - j < n$ hold by $\T[n] = \texttt{\$}$. In total, the
  query takes $\bigO(t)$ time.
\end{proof}

\paragraph{Computing $|\Occgtm_{\ell}(j)|$ and $|\Occgtm_{2\ell}(j)|$}

We now describe a query algorithm, given any position $j \in \R$ and
any $d \geq \ell$, returns the cardinality of the sets
$\Occgtm_{d}(j)$ (in particular, it can return the cardinality of
$\Occgtm_{\ell}(j)$ and $\Occgtm_{2\ell}(j)$). We start with a
combinatorial result characterizing $\Occgtm_{d}(j)$ (where $d \geq
\ell$) in terms of $\rend{j} - j$ (\cref{lm:occ-gt}). We then present
a combinatorial result (\cref{lm:mod-count}) that relates the number
of positions $j'$ in $\R^{-}$ having the value $\Lexp(j')$ in a given
interval to a modular contract queries. We conclude with the algorithm
to compute $|\Occgtm_{d}(j)|$ (\cref{pr:sa-occ-gt}).

\begin{lemma}\label{lm:occ-gt}
  Assume $d \geq \ell$ and $j \in \R_{s,H}$. If $\rend{j} - j < d$,
  then it holds $\Occgtm_{d}(j) = \emptyset$. Otherwise, it holds
  $\Occgtm_{d}(j) = \{j' \in \R^{-}_{s,H} : \Lexp(j') > \lfloor
  \tfrac{d - s}{|H|} \rfloor\}$.
\end{lemma}
\begin{proof}

  Assume first $\rend{j} - j < d$. Then $k = \min(\Lexp(j), \lfloor
  \tfrac{d - s}{|H|} \rfloor) = \min(\lfloor \tfrac{\rend{j} - j -
  s}{|H|} \rfloor, \lfloor \tfrac{d - s}{|H|} \rfloor) = \lfloor
  \tfrac{\rend{j} - j - s}{|H|} \rfloor = \Lexp(j)$. Suppose that
  $\Occgtm_{d}(j) \neq \emptyset$ and let $j' \in \Occgtm_{d}(j)$.
  Then, we have $\rend{j} - j = s + \Lexp(j)|H| + \Ltail(j) < s +
  (\Lexp(j) + 1)|H| = s + (k + 1)|H| \leq s + \Lexp(j')|H| \leq
  \rend{j'} - j'$. This implies $j \neq j'$. Moreover, by definition
  of run-decomposition, if $j, j' \in \R_{s,H}$ and $\rend{j} - j \neq
  \rend{j'} - j'$, then $\LCE_{\T}(j, j') = \min(\rend{j} - j,
  \rend{j'} - j')$. Therefore, $\LCE_{\T}(j, j') = \rend{j} - j <
  d$. By $j \neq j'$ and the uniqueness of $\T[n] = \texttt{\$}$, this
  implies $\T^{\infty}[j' \dd j' + d) \neq \T[j \dd j + d)$, and
  consequently, $j' \not\in \Occ_{d}(j)$, a contradiction. Thus, we
  must have $\Occgtm_{d}(j) = \emptyset$.

  Assume now $\rend{j} - j \geq d$. Then, $k = \min(\Lexp(j), \lfloor
  \tfrac{d - s}{|H|} \rfloor) = \min(\lfloor \tfrac{\rend{j} - j -
  s}{|H|} \rfloor, \lfloor \tfrac{d - s}{|H|} \rfloor) = \lfloor
  \tfrac{d-s}{|H|} \rfloor$ and hence
  \begin{align*}
    \Occgtm_{d}(j)
      &= \Occgt_{d} \cap \R^{-}\\
      &= \{j' \in \R^{-}_{s,H} \cap \Occ_{d}(j) :
         \Lexp(j') > k\}\\
      &= \{j' \in \R^{-}_{s,H} \cap \Occ_{d}(j) :
          \Lexp(j') > \lfloor \tfrac{d - s}{|H|} \rfloor\}\\
      &\sub \{j' \in \R^{-}_{s,H} :
        \Lexp(j') > \lfloor \tfrac{d - s}{|H|} \rfloor\},
  \end{align*}
  i.e., we obtain the first inclusion. To show the opposite inclusion
  let $j' \in \R^{-}_{s,H}$ be such that $\Lexp(j') > \lfloor \tfrac{d
  - s}{|H|} \rfloor$. By $\rend{j} - j \geq d$, we can write $\T[j \dd
  j + d) = H'H^{k}H''$, where $|H'| = s$, and $H'$ (resp.\ $H''$) is a
  proper prefix (resp.\ suffix) of $H$.  Thus, $\T[j \dd j + d)$ is a
  prefix of $H'H^{k + 1}$. On the other hand, the string $H'H^{k + 1}$
  is, by $\Lexp(j') > k$ and $j' \in \R_{s,H}$, a prefix of $\T[j' \dd
  n]$. Thus, $\LCE_{\T}(j', j) \geq d$, and consequently, $j' \in
  \Occ_{d}(j)$. By $j' \in \R^{-}_{s,H}$ and $\Lexp(j') > \lfloor
  \tfrac{d - s}{|H|} \rfloor = k$ we thus have $j' \in
  \Occgtm_{d}(j)$.
\end{proof}

\begin{lemma}\label{lm:mod-count}
  Let $H \in \Lroots$ and $s \in [0 \dd |H|)$. For any $k \in \Zp$,
  let us define $Q^{-}_k \,{:=}\, \{j \,{\in}\, \R^{-}_{s,H} :
  \Lexp(j) \,{\leq}\, k\}$. Let also $I_i \,{=}\, (a_i, b_i,i)$
  (\cref{def:intervals}). Then:
  \begin{enumerate}
  \item\label{lm:mod-count-item-1} For $i \,{\in}\, \R^{-}_H$, it
    holds $|Q^{-}_k \cap [i \dd \rend{i}-3\tau+2)| = |\{j \in [a_i \dd
    b_i) : j \,{\bmod}\, |H| \,{=}\, s\text{ and }\lfloor
    \tfrac{j}{|H|} \rfloor \,{\leq}\, k\}|$.
  \item\label{lm:mod-count-item-2} It holds $|Q^{-}_k| =
    \mcount{\mathcal{I}^{-}_{H}}{|H|}{s}{k}$.
  \end{enumerate}
\end{lemma}
\begin{proof}

  1. Denote $e = \rendfull{i}-i$, $t = \rend{i}-i-3\tau+1$, and
  $b=e-t$.  As shown at the beginning of the proof of
  \cref{lm:occ-eq}, for any $i \in \R^{-}_{H}$, we have $t > 0$, $[i
  \dd i + t) \sub \R^{-}_{H}$, and for any $\delta \in [0 \dd t)$ it
  holds $\rend{i + \delta} = \rend{i}$ and $\rendfull{i + \delta} =
  \rendfull{i}$.  This implies $\Lhead(i + \delta) + \Lexp(i +
  \delta)|H| = \rendfull{i + \delta} - (i + \delta) = \rendfull{i} -
  (i + \delta) = (\rendfull{i} - i) - \delta = e - \delta$.  Thus,
  $\Lhead(i + \delta) = (e - \delta) \bmod |H|$ and $\Lexp(i + \delta)
  = \lfloor \frac{e - \delta}{|H|} \rfloor$. Hence:
  \begin{align*}
    |Q^{-}_k \cap [i \dd\rend{i}-3\tau+2)|
      &= |\{i {+} \delta : \delta \in [0 \dd t),\,
          \Lhead(i + \delta) = s, \text{ and }\Lexp(i + \delta)
          \leq k\}| \\
      &= |\{i {+} \delta : \delta \in [0 \dd t),\ (e {-} \delta)
         \bmod |H| = s,\text{ and }\lfloor
         \tfrac{e - \delta}{|H|} \rfloor \leq k\}| \\
      &= |\{i {+} e {-} j : j \in [b {+} 1 \dd e {+} 1),\, j \bmod
         |H| = s, \text{ and }\lfloor \tfrac{j}{|H|} \rfloor
         \leq k\}| \\
      &= |\{j \in [b{+}1 \dd e{+}1) : j \bmod |H| = s
         \text{ and }\lfloor \tfrac{j}{|H|} \rfloor \leq
         k\}| \\
      &= |\{j \in [a_i \dd b_i) : j \bmod |H| = s
         \text{ and }\lfloor \tfrac{j}{|H|} \rfloor \leq k\}|,
  \end{align*}
  where the third equation utilizes that if $\delta \in [0 \dd t)$,
  then letting $j = e - \delta$, we have $i + \delta = i + e - j$ and
  $j \in (e {-} t \dd e] = [b {+} 1 \dd e {+} 1)$. We then (fourth
  equality) used the fact that since $i$ and $e$ are fixed, the size
  of the ``shifted'' set does not change.  Note that, $b = (\Lhead(i)
  + \Lexp(i)|H|) - (\rend{i} - i - 3\tau + 2) = (\rend{i} - i -
  \Ltail(i)) - (\rend{i} - i - 3\tau + 2) = 3\tau - 2 - \Ltail(i) > 0$
  follows by $\Ltail(i) < |H| \leq \tau$. Thus, the interval $[a_i \dd
  b_i)$ contains only nonnegative integers.

  2. By \cref{lm:mod-count-item-1}, for any $i \in \R^{-}_H$, we can
  reduce the computation of $|Q^{-}_k \cap [i \dd i + t)\}|$, where $t
  = \rend{i} - i - 3\tau + 2$, to a modular constraint counting
  query. Observe that by definition of $\rend{i}$, if additionally it
  holds $i \in \R'$, then the interval $[i \dd i + t)$ is a
  \emph{maximal} interval of positions in $\R$, i.e., $i - 1, i + t
  \not \in \R$. Thus, letting $\mathcal{I}^{-}_H$ be the collection of
  weighted intervals (with weights corresponding to multiplicities)
  corresponding to all $i \in \R'^{-}_{H}$ (\cref{def:intervals}), we
  obtain the claim by definition of the modular constraint counting
  query (see \cref{sec:mod-queries}).
\end{proof}

\begin{proposition}\label{pr:sa-occ-gt}
  Under \cref{as:periodic}, given any $j \in \R$ and $d \geq \ell$, we
  can compute $|\Occgtm_{d}(j)|$ in $\bigO(t)$ time.
\end{proposition}
\begin{proof}

  The main idea of the algorithm is as follows. By
  \cref{lm:mod-count}, we can reduce the computation of
  $|\Occgtm_{d}(j) \cap [i \dd i + t)\}|$, where $i \in \R^{-}$,
  $\Lroot(i) = \Lroot(j)$, and $t = \rend{i} - i - 3\tau + 2$, to an
  unweighted modular constraint counting query. Therefore, we can
  compute $|\Occgtm_{d}(j)|$ using the general (weighted) modular
  constraint counting queries (\cref{sec:mod-queries}) on
  $\mathcal{I}^{-}_H$ as defined in \cref{lm:mod-count}. Note that we
  need to provide efficient modular on each collection
  $\mathcal{I}^{-}_H$ separately for every $H \in \Lroots$.  Since
  $|H|$ could be large, $H$ cannot be given explicitly during the
  query. Thus, we specify $H$ using its length and the position of
  some occurrence in $\T$.

  Given any $j \in \R$, we compute $|\Occgtm_{d}(j)|$ as follows.
  First, using \cref{as:periodic} we compute values $s = \Lhead(j)$,
  $p = |\Lroot(j)|$, and $\rend{j}$ in $\bigO(t)$ time. Note, that
  then $\Lroot(j) = \T[j + s \dd j + s + p)$, i.e., we have a starting
  position of an occurrence of $H := \Lroot(j)$ in $\T$.  If $\rend{j}
  - j < d$, then by \cref{lm:occ-gt} we have $\Occgtm_{d}(j) =
  \emptyset$, and thus we return $|\Occgtm_{d}(j)| = 0$.  Let us thus
  assume $\rend{j} - j \geq d$.  We then calculate $\lfloor \frac{d -
    s}{p} \rfloor$ and by the combination of \cref{lm:occ-gt} and
  \cref{lm:mod-count-item-2} of \cref{lm:mod-count}, we obtain
  \begin{align*}
    |\Occgtm_{d}(j)|
      &= |\{j' \in \R^{-}_{s,H} : \lfloor \tfrac{d-s}{|H|} \rfloor <
         \Lexp(j') \leq n\}|\\
      &= \mcount{\mathcal{I}^{-}_H}{p}{s}{n} -
         \mcount{\mathcal{I}^{-}_H}{p}{s}{\lfloor \tfrac{d-s}{|H|} \rfloor}
  \end{align*}
  which by \cref{as:periodic} we can compute in $\bigO(t)$ time.  In
  total, the query takes $\bigO(t)$ time.
\end{proof}

\paragraph{Summary}

By combining the above results, we obtain the following query
algorithm to compute the values $|\Occm_{\ell}(j)|$ and
$|\Occm_{2\ell}(j)|$, given any position $j \in \R$.

\begin{proposition}\label{pr:sa-occ-size}
  Under \cref{as:periodic}, given any position $j \in \R$, we can in
  $\bigO(t)$ time compute the values $|\Occm_{\ell}(j)|$ and
  $|\Occm_{2\ell}(j)|$.
\end{proposition}
\begin{proof}

  First, using \cref{pr:sa-occ-eq}, we compute $\delta^{\rm eq}_{\ell}
  := |\Occeqm_{\ell}(j)|$ and $\delta^{\rm eq}_{2\ell} :=
  |\Occeqm_{2\ell}(j)|$ in $\bigO(t)$ time. Then, using
  \cref{pr:sa-occ-gt}, we compute $\delta^{\rm gt}_{\ell} :=
  |\Occgtm_{\ell}(j)$ and $\delta^{\rm gt}_{2\ell} :=
  |\Occgtm_{2\ell}(j)|$ in $\bigO(t)$ time. By \cref{cor:occ}, we then
  have $|\Occm_{\ell}(j)| = \delta^{\rm eq}_{\ell} + \delta^{\rm
  gt}_{\ell}$ and $|\Occm_{2\ell}(j)| = \delta^{\rm eq}_{2\ell} +
  \delta^{\rm gt}_{2\ell}$.
\end{proof}

\subsubsection{Computing the Type}\label{sec:sa-periodic-type}

Assume that $i \in [1 \dd n]$ satisfies $\SA[i] \in \R$. In this
section, we show how under \cref{as:periodic}, given $i$ along with
$\LB_{\ell}(\SA[i])$, $\UB_{\ell}(\SA[i])$ and some $j \in
\Occ_{\ell}(\SA[i])$ (note, that we do not assume anything about
$\type(j)$; the query works correctly even if $\type(j) \neq
\type(\SA[i])$) to efficiently compute $\type(\SA[i])$, i.e., whether
it holds $\SA[i] \in \R^{-}$ or $\SA[i] \in \R^{+}$.

The section is organized as follows.  Given the input parameters as
described above, our query algorithm computes $\type(\SA[i])$ by
checking if it holds $\SA[i] \in \Occm_{\ell}(\SA[i])$.  To implement
such check, we first present a combinatorial result proving that
$\Occm_{\ell}(\SA[i])$ occupies a contiguous block of positions in
$\SA$ and showing what are the endpoints of this block
(\cref{lm:type}). We then use this characterization to develop an
efficient method of checking if $\SA[i] \in \Occm_{\ell}(\SA[i])$
holds (\cref{cor:type}). Finally, we develop a query algorithm that
efficiently computes $\type(\SA[i])$ in \cref{pr:sa-type}.

\begin{lemma}\label{lm:type}
  Let $i \in [1 \dd n]$ be such that $\SA[i] \in \R$.  Denote $b =
  \LB_{\ell}(\SA[i])$ and $e = b + |\Occm_{\ell}(\SA[i])|$.  Then, it
  holds $\Occm_{\ell}(\SA[i]) = \{\SA[i] : i \in (b \dd e]\}$.
\end{lemma}
\begin{proof}
  Let $s = \Lhead(\SA[i])$ and $H = \Lroot(\SA[i])$ By
  \cref{lm:occ-exp-item-1} of \cref{lm:occ-exp}, it holds
  $\Occ_{\ell}(\SA[i]) \sub \R_{s,H}$. On the other hand, by
  \cref{lm:exp}, all elements of $\R_{s,H}$ occupy a contiguous block
  of positions in $\SA$.  Moreover, all elements of $\R^{-}_{s,H}$ are
  earlier in the lexicographical order than the elements of
  $\R^{+}_{s,H}$. Thus, since $\Occ_{\ell}(\SA[i])$ by definition also
  occupies a contiguous block in $\SA$ starting at index $b + 1$, the
  elements of $\Occm_{\ell}(\SA[i])$ (assuming the set is nonempty)
  must start at this position, and occupy the block of
  $|\Occ_{\ell}(\SA[i])|$ consecutive positions.
\end{proof}

\begin{corollary}\label{cor:type}
  For any $i \in [1 \dd n]$ such that $\SA[i] \in \R$, $\SA[i] \in
  \Occm_{\ell}(\SA[i])$ holds if and only if $i - \LB_{\ell}(\SA[i])
  \leq |\Occm_{\ell}(\SA[i])|$.
\end{corollary}
\begin{proof}

  Assume $\SA[i] \in \Occm_{\ell}(\SA[i])$. By \cref{lm:poslow}, $i
  \in (b \dd e]$, where $b = \LB_{\ell}(\SA[i])$ and $e = b +
  |\Occm_{\ell}(\SA[i])|$. In particular, $i \leq e =
  \LB_{\ell}(\SA[i]) + |\Occm_{\ell}(\SA[i])|$.

  Assume now $i - \LB_{\ell}(\SA[i]) \leq |\Occm_{\ell}(\SA[i])|$.
  Let $b = \LB_{\ell}(\SA[i])$ and $e = b + |\Occm_{\ell}(\SA[i])|$.
  Then, $i \leq e$. On the other hand, by definition of
  $\Occ_{\ell}(\SA[i])$, we have $i \in (\LB_{\ell}(\SA[i]) \dd
  \UB_{\ell}(\SA[i])]$. In particular, $i > \LB_{\ell}(\SA[i]) =
  b$. Therefore, we obtain $i \in (b \dd e]$. By \cref{lm:type}, this
  implies $\SA[i] \in \Occm_{\ell}(\SA[i])$.
\end{proof}

\begin{proposition}\label{pr:sa-type}
  Let $i \in [1 \dd n]$ be such that $\SA[i] \in \R$. Under
  \cref{as:periodic}, given $i$, $\LB_{\ell}(\SA[i])$,
  $\UB_{\ell}(\SA[i])$, and some position $j \in \Occ_{\ell}(\SA[i])$,
  we can compute $\type(\SA[i])$ in $\bigO(t)$ time.
\end{proposition}
\begin{proof}

  The main idea of the query is as follows. By \cref{cor:type}, to
  compute $\type(\SA[i])$ it suffices to know $i$,
  $\LB_{\ell}(\SA[i])$ and $|\Occm_{\ell}(\SA[i])|$. The first two
  values are given as input.  The third is computed using
  \cref{pr:sa-occ-size}. Note, however, that the query in
  \cref{pr:sa-occ-size} can compute $\Occm_{\ell}(j)$ only if given
  $j$. In our case we do not have $\SA[i]$.  We observe, however, that
  by definition, for any $j \in \Occ_{\ell}(\SA[i])$ (note that
  $\type(j)$ can be either $-1$ or $+1$), it holds
  $|\Occm_{\ell}(\SA[i])| = |\Occm_{\ell}(j)|$. Thus, we can use $j$
  instead of $\SA[i]$.

  Given the index $i$, along with values $\LB_{\ell}(\SA[i])$,
  $\UB_{\ell}(\SA[i])$, and some position $j \in \Occ_{\ell}(\SA[i])$,
  we compute $\type(\SA[i])$ as follows. First, using the query from
  \cref{pr:sa-occ-size}, we compute the value $\delta :=
  |\Occm_{\ell}(j)|$ in $\bigO(t)$ time. By the above discussion, it
  holds $|\Occm_{\ell}(\SA[i])| = \delta$. By \cref{cor:type}, we then
  have $\type(\SA[i]) = -1$ if and only if $i - \LB_{\ell}(\SA[i])
  \leq \delta$, which we can evaluate in $\bigO(1)$ time.
\end{proof}

\subsubsection{Computing the Size of 
  \texorpdfstring{$\Poslow(j)$}{Poslow} and
  \texorpdfstring{$\Poshigh(j)$}{Poshigh}}\label{sec:sa-periodic-poslow-poshigh}

In this section, we present an algorithm to compute the values
$\deltalow(j)$ and $\deltahigh(j)$ for $j \in \R^{-}$.

The section is organized as follows. We start with the combinatorial
result (\cref{lm:delta-low-high}) that shows how to count the elements
of $\Poslow(j)$ and $\Poshigh(j)$ that belong to any right-maximal
contiguous block of elements of $\R^{-}$. This result is a very
general extension of \cref{lm:occ-eq}. This generalization, however,
is nontrivial and thus below we provide its complete proof (omitting
and referring to the appropriate places in the proof of
\cref{lm:occ-eq} only for parts that are identical).  We then use this
characterization to develop a query algorithm to compute
$\deltalow(j)$ and $\deltahigh(j)$ given any $j \in \R^{-}$
(\cref{pr:sa-delta-low-high}). We finally prove
(\cref{pr:sa-delta-low-high-2}) that the computation of $\deltalow(j)$
(resp.\ $\deltahigh(j)$) does not actually need the value of $j$, but
it is sufficient to only know that $\type(j) = -1$ and some $j' \in
\Occ_{\ell}(j)$ (resp.\ $j' \in \Occ_{2\ell}(j)$).

\begin{lemma}\label{lm:delta-low-high}
  Assume $i, j \in \R^{-}_{H}$ and let $t = \rend{i} - i - 3\tau +
  2$. Then, $|\Poslow(j) \cap [i \dd i + t)| \leq 1$ and $|\Poshigh(j)
  \cap [i \dd i + t)| \leq 1$.  Moreover, $|\Poslow(j) \cap [i \dd i +
  t)| = 1$ $($resp.\ $|\Poshigh(j) \cap [i \dd i + t)| = 1)$ holds if
  and only if
  \begin{itemize}
  \item $\rendfull{i} - i \geq \rendlow{j} - j$
    $($resp.\ $\rendfull{i} - i \geq \rendhigh{j} - j)$ and
  \item $\T^{\infty}[\rendfull{i} \dd \rendfull{i} + 7\tau) \succeq
    \T^{\infty}[\rendlow{j} \dd j +
    \ell)$\\ $($resp.\ $\T^{\infty}[\rendfull{i} \dd \rendfull{i} +
    7\tau) \succeq \T^{\infty}[\rendhigh{j} \dd j + 2\ell))$.
  \end{itemize}
\end{lemma}
\begin{proof}

  As shown at the beginning of the proof of \cref{lm:occ-eq}, for any
  $i \in \R^{-}_{H}$, letting $t = \rend{j} - i - 3\tau + 2$, we have
  $t > 0$, $[i \dd i + t) \sub \R^{-}_{H}$, and for any $\delta \in [0
  \dd t)$ it holds $\rend{i + \delta} = \rend{i}$ and $\rendfull{i +
  \delta} = \rendfull{i}$, which in turn implies $\rendfull{i +
  \delta} - (i + \delta) = \rendfull{i} - i - \delta$. Let $s =
  \Lhead(j)$. Recall that for any $j' \in \Poslow(j)$ (resp.\ $j' \in
  \Poshigh(j)$), we have $\rendfull{j'} - j' = s + \Lexp(j')|H| = s +
  k_1|H| = \rendlow{j} - j$ (resp.\ $\rendfull{j'} - j' = s +
  \Lexp(j')|H| = s + k_2|H| = \rendhigh{j} - j$).  Thus, $i + \delta
  \in \Poslow(j)$ (resp.\ $i + \delta \in \Poshigh(j)$) implies
  $\rendfull{i + \delta} - (i + \delta) = \rendfull{i} - i - \delta =
  \rendlow{j} - j$ (resp.\ $\rendfull{i + \delta} - (i + \delta) =
  \rendfull{i} - i - \delta = \rendhigh{j} - j$), or equivalently,
  $\delta = (\rendfull{i} - i) - (\rendlow{j} - j)$ (resp.\ $\delta =
  (\rendfull{i} - i) - (\rendhigh{j} - j)$), and thus $|\Poslow(j)
  \cap [i \dd i + t)| \leq 1$ (resp.\ $|\Poshigh(j) \cap [i \dd i +
  t)| \leq 1$).

  Next, we prove the equivalence.  Let us first assume $|\Poslow(j)
  \cap [i \dd i + t)| = 1$ (resp.\ $|\Poshigh(j) \cap [i \dd i + t)| =
  1$), i.e., that $i + \delta \in \Poslow(j)$ (resp.\ $i + \delta \in
  \Poshigh(j)$) holds for some $\delta \in [0 \dd t)$. Then, as noted
  above, $\rendlow{j} - j = \rendfull{i} - i - \delta \leq
  \rendfull{i} - i$.  (resp.\ $\rendhigh{j} - j = \rendfull{i} - i -
  \delta \leq \rendfull{i} - i$) holds. This establishes the first
  condition. To show the second condition, let $s = \Lhead(j)$. By
  $k_1 \leq \Lexp(j)$ (resp.\ $k_2 \leq \Lexp(j)$), we have $\T[j \dd
  \rendlow{j}) = H'H^{k_1}$ (resp.\ $\T[j \dd \rendhigh{j}) =
  H'H^{k_2}$), where $H'$ is a length-$s$ prefix of $H$. On the other
  hand, by definition, $i + \delta \in \Poslow(j)$ (resp.\ $i + \delta
  \in \Poshigh(j)$) implies $\Lhead(i + \delta) = s$ and $\Lexp(i +
  \delta) = k_1$ (resp.\ $\Lexp(i + \delta) = k_2$). Thus, $\T[i +
  \delta \dd \rendfull{i+\delta}) = \T[i + \delta \dd \rendfull{i}) =
  H'H^{k_1} = \T[j \dd \rendlow{j})$ (resp.\ $\T[i + \delta \dd
  \rendfull{i+\delta}) = \T[i + \delta \dd \rendfull{i}) = H'H^{k_2} =
  \T[j \dd \rendhigh{j})$).  Therefore, the assumption $\T[i + \delta
  \dd n] \succeq \T[j \dd n]$ or $\LCE_{\T}(i + \delta, j) \geq \ell$
  (resp.\ $\LCE_{\T}(i + \delta, j) \geq 2\ell$) following from $i +
  \delta \in \Poslow(j)$ (resp.\ $i+\delta \in \Poshigh(j)$) is
  equivalent to $\T[\rendfull{i} \dd n] \succeq \T[\rendlow{j} \dd n]$
  (resp.\ $\T[\rendfull{i} \dd n] \succeq \T[\rendhigh{j} \dd n]$) or
  $\LCE_{\T}(i + \delta, j) = \LCE_{\T}(\rendfull{i}, \rendlow{j})
  \geq \ell - (\rendlow{j} - j)$ (resp.\ $\LCE_{\T}(i + \delta, h) =
  \LCE_{\T}(\rendfull{i}, \rendhigh{j}) \geq 2\ell - (\rendhigh{j} -
  j)$). We thus consider two cases:
  \begin{itemize}
  \item If $\LCE_{\T}(\rendfull{i}, \rendlow{j}) \geq \ell -
    (\rendlow{j} - j)$ (resp.\ $\LCE_{\T}(\rendfull{i}, \rendhigh{j})
    \geq 2\ell - (\rendhigh{j} - j)$) holds, then by the assumption
    $\ell \leq 7\tau$ (resp.\ $2\ell \leq 7\tau$), we obtain
    $\T^{\infty}[\rendfull{i} \dd \rendfull{i} + 7\tau) \succeq
    \T^{\infty}[\rendfull{i} \dd \rendfull{i} + \ell - (\rendlow{j} -
    j)) = \T^{\infty}[\rendlow{j} \dd j + \ell)$
    (resp.\ $\T^{\infty}[\rendfull{i} \dd \rendfull{i} + 7\tau)
    \succeq \T^{\infty}[\rendfull{i} \dd \rendfull{i} + 2\ell -
    (\rendhigh{j} - j)) = \T^{\infty}[\rendhigh{j} \dd j + 2\ell)$).
  \item On the other hand, if $\T[\rendfull{i} \dd n] \succeq
    \T[\rendlow{j} \dd n]$ (resp.\ $\T[\rendfull{i} \dd n] \succeq
    \T[\rendhigh{j} \dd n]$) holds, then $\T[n]$ being smallest in
    $\T$ implies $\T^{\infty}[\rendfull{i} \dd \rendfull{i} + q)
    \succeq \T^{\infty}[\rendlow{j} \dd \rendlow{j} + q)$
    (resp.\ $\T^{\infty}[\rendfull{i} \dd \rendfull{i} + q) \succeq
    \T^{\infty}[\rendhigh{j} \dd \rendhigh{j} + q)$) for any $q \geq
    0$. In particular, for $q = 7\tau$ we have
    $\T^{\infty}[\rendfull{i} \dd \rendfull{i} + 7\tau) \succeq
    \T^{\infty}[\rendlow{j} \dd \rendlow{j} + 7\tau) \succeq
    \T^{\infty}[\rendlow{j} \dd j + \ell)$
    (resp.\ $\T^{\infty}[\rendfull{i} \dd \rendfull{i} + 7\tau)
    \succeq \T^{\infty}[\rendhigh{j} \dd \rendhigh{j} + 7\tau) \succeq
    \T^{\infty}[\rendhigh{j} \dd j + 2\ell)$), where the last
    inequality follows by $j + \ell \leq \rendlow{j} + 7\tau$
    (resp.\ $j + 2\ell \leq \rendhigh{j} + 7\tau$).
  \end{itemize}
  To prove the opposite implication, assume $\rendfull{i} - i \geq
  \rendlow{j} - j$ (resp.\ $\rendfull{i} - i \geq \rendhigh{j} - j$)
  and $\T^{\infty}[\rendfull{i} \dd \rendfull{i} + 7\tau) \succeq
  \T^{\infty}[\rendlow{j} \dd j + \ell)$
  (resp.\ $\T^{\infty}[\rendfull{i} \dd \rendfull{i} + 7\tau) \succeq
  \T^{\infty}[\rendhigh{j} \dd j + 2\ell)$). Let $\delta =
  (\rendfull{i} - i) - (\rendlow{j} - j)$ (resp.\ $\delta =
  (\rendfull{i} - i) - (\rendhigh{j} - j)$). We will prove that
  $\delta \in [0 \dd t)$ and $i + \delta \in \Poslow(j)$ (resp.\ $i +
  \delta \in \Poshigh(j)$).  The inequality $\delta \geq 0$ follows
  from the definition of $\delta$ and our assumptions. To show $\delta
  < t$, we consider two cases:
  \begin{itemize}
  \item Let $\rend{j} - j < \ell$ (resp.\ $\rend{j} - j < 2\ell$). We
    start by noting $k_1 = \min(\Lexp(j), \lfloor \frac{\ell-s}{|H|}
    \rfloor) = \min(\lfloor \frac{\rend{j} - j - s}{|H|} \rfloor,
    \lfloor \frac{\ell - s}{|H|} \rfloor) = \lfloor \frac{\rend{j} - j
    - s}{|H|} \rfloor = \Lexp(j)$ (resp.\ $k_2 = \Lexp(j)$). Thus,
    $\rendlow{j} = j + s + k_1|H| = j + s + \Lexp(j)|H| =
    \rendfull{j}$ (resp.\ $\rendhigh{j} = j + s + k_2|H| = j + s +
    \Lexp(j)|H| = \rendfull{j}$). Let $q = \rend{i} - \rendfull{i}$.
    Using the same argument as in the proof of \cref{lm:occ-eq} (when
    considering the two cases for the value $\rend{j} - j$), we obtain
    that assuming $q < \rend{j} - \rendfull{j}$ implies
    $\T[\rendfull{i} \dd \rendfull{i} + q] \prec \T[\rendfull{j} \dd
    \rendfull{j} + q]$. By $q < \rend{j} - \rendfull{j} < j + \ell -
    \rendfull{j} = j + \ell - \rendlow{j} \leq \ell \leq 7\tau$
    (resp.\ $q < \rend{j} - \rendfull{j} < j + 2\ell - \rendfull{j} =
    j + 2\ell - \rendhigh{j} \leq 2\ell \leq 7\tau$), this implies
    $\T^{\infty}[\rendfull{i} \dd \rendfull{i} + 7\tau) \prec
    \T^{\infty}[\rendlow{j} \dd j + \ell)$
    (resp.\ $\T^{\infty}[\rendfull{i} \dd \rendfull{i} + 7\tau) \prec
    \T^{\infty}[\rendhigh{j} \dd j + 2\ell)$), contradicting our
    assumption.  Thus, it holds $\rend{i} - \rendfull{i} \geq \rend{j}
    - \rendfull{j}$ and hence using the same chain of inequalities as
    in the proof of \cref{lm:occ-eq} with $\rendhigh{j}$ replaced by
    $\rendlow{j}$ (resp.\ without any change), we obtain $\rend{i} -
    (i + \delta) \geq 3\tau - 1$, or equivalently, $\delta \leq
    \rend{i} - i - 3\tau + 1 < t$.
  \item Let $\rend{j} - j \geq \ell$ (resp.\ $\rend{j} - j \geq
    2\ell$).  Denote $q = \rend{i} - \rendfull{i}$. Using the same
    argument as in the proof \cref{lm:occ-eq} (when considering the
    two cases for the value $\rend{j} - j$), we obtain that assuming
    $q < j + \ell - \rendlow{j}$ (resp.\ $q < j + 2\ell -
    \rendhigh{j}$), it holds $\T[\rendfull{i} \dd \rendfull{i} + q]
    \prec \T[\rendlow{j} \dd \rendlow{j} + q]$
    (resp.\ $\T[\rendfull{i} \dd \rendfull{i} + q] \prec
    \T[\rendhigh{j} \dd \rendhigh{j} + q]$).  By $q < j + \ell -
    \rendlow{j} \leq \ell \leq 7\tau$ (resp.\ $q < j + 2\ell -
    \rendhigh{j} \leq 2\ell \leq 7\tau$), this implies
    $\T^{\infty}[\rendfull{i} \dd \rendfull{i} + 7\tau) \prec
    \T^{\infty}[\rendlow{j} \dd j + \ell)$
    (resp.\ $\T^{\infty}[\rendfull{i} \dd \rendfull{i} + 7\tau) \prec
    \T^{\infty}[\rendhigh{j} \dd j + 2\ell)$), contradicting our
    assumption.  Thus, $\rend{i} - \rendfull{i} \geq j + \ell -
    \rendlow{j}$ (resp.\ $\rend{i} - \rendfull{i} \geq j + 2\ell -
    \rendhigh{j}$) and hence using the same chain of inequalities as
    in the proof of \cref{lm:occ-eq} with $\rendhigh{j}$ replaced by
    $\rendlow{j}$ and $2\ell$ replaced by $\ell$ (resp.\ without any
    change), we obtain $\rend{i} - (i + \delta) \geq 3\tau - 1$, or
    equivalently, $\delta \leq \rend{i} - i - 3\tau + 1 < t$.
  \end{itemize}
  It remains to show $i + \delta \in \Poslow(j)$ (resp.\ $i + \delta
  \in \Poshigh(j)$). For this, we first observe that $\rendfull{i +
  \delta} - (i + \delta) = \rendfull{i} - (i + \delta) = \rendlow{j} -
  j$ (resp.\ $\rendfull{i + \delta} - (i + \delta) = \rendfull{i} - (i
  + \delta) = \rendhigh{j} - j$). Combining this with $i + \delta, j
  \in \R_H$ gives $\T[i + \delta \dd \rendfull{i + \delta}) = \T[j \dd
  \rendlow{j})$ (resp.\ $\T[i + \delta \dd \rendfull{i + \delta}) =
  \T[j \dd \rendhigh{j})$).  This implies $\Lhead(i + \delta) =
  \Lhead(j)$ and $\Lexp(i + \delta) = \lfloor (\rendlow{j} - j)/|H|
  \rfloor = k_1$ (resp.\ $\Lexp(i + \delta) = \lfloor (\rendhigh{j} -
  j)/|H| \rfloor = k_2$).  Moreover, by \cref{lm:R-block}, $\type(i +
  \delta) = \type(i) = -1$. Thus, we obtain $i + \delta \in
  \R^{-}_{s,H}$.  It remains to show that it holds $\T[i + \delta \dd
  n] \succeq \T[j \dd n]$ or $\LCE_{\T}(i + \delta, j) \geq \ell$
  (resp.\ $\LCE_{\T}(i + \delta, j) \geq 2\ell$).  For this, we first
  note that by $\rendfull{i+\delta} = \rendfull{i}$ and $\T[i + \delta
  \dd \rendfull{i + \delta}) = \T[j \dd \rendlow{j})$ (resp.\ $\T[i +
  \delta \dd \rendfull{i + \delta}) = \T[j \dd \rendhigh{j})$), we
  have $\LCE_{\T}(i + \delta, j) = \rendlow{j} - j +
  \LCE_{\T}(\rendfull{i}, \rendlow{j})$ (resp.\ $\LCE_{\T}(i + \delta,
  j) = \rendhigh{j} - j + \LCE_{\T}(\rendfull{i}, \rendhigh{j})$).
  Let $q' = \LCE_{\T}(\rendfull{i}, \rendlow{j})$ (resp.\ $q' =
  \LCE_{\T}(\rendfull{i}, \rendhigh{j})$).  We consider three cases:
  \begin{enumerate}
  \item First, if $q' \geq \ell - (\rendlow{j} - j)$ (resp.\ $q' \geq
    2\ell - (\rendhigh{j} - j)$), then we obtain $\LCE_{\T}(i +
    \delta, j) = \rendlow{j} - j + q' \geq \ell$ (resp.\ $\LCE_{\T}(i +
    \delta, j) = \rendhigh{j} - j + q' \geq 2\ell$).
  \item Second, if $q' < \ell - (\rendlow{j} - j)$ (resp.\ $q' < 2\ell
    - (\rendhigh{j} - j)$) and $\max(\rendfull{i} + q', \rendlow{j} +
    q') = n + 1$ (resp.\ $\max(\rendfull{i} + q', \rendhigh{j} + q') =
    n + 1$), then $\T[n] = \texttt{\$}$ being unique in $\T$ gives
    $\rendfull{i} = \rendlow{j}$ (resp.\ $\rendfull{i} =
    \rendhigh{j}$), and hence $\T[i + \delta \dd n] = \T[j \dd n]$.
  \item Finally, if $q' < \ell - (\rendlow{j} - j)$ (resp.\ $q' <
    2\ell - (\rendhigh{j} - j)$) and $\max(\rendfull{i} + q',
    \rendlow{j} + q') \leq n$ (resp.\ $\max(\rendfull{i} + q',
    \rendhigh{j} + q') \leq n$), then by $\T^{\infty}[\rendfull{i} \dd
    \rendfull{i} + 7\tau) \succeq \T^{\infty}[\rendlow{j} \dd j +
    \ell)$ (resp.\ $\T^{\infty}[\rendfull{i} \dd \rendfull{i} + 7\tau)
    \succeq \T^{\infty}[\rendhigh{j} \dd j + 2\ell)$), it holds
    $\T[\rendfull{i} + q'] \succ \T[\rendlow{j} + q']$
    (resp.\ $\T[\rendfull{i} + q'] \succ \T[\rendhigh{j} + q']$) and
    hence $\T[i + \delta \dd n] \succ \T[j \dd n]$.\qedhere
  \end{enumerate}
\end{proof}

\begin{proposition}\label{pr:sa-delta-low-high}
  Under \cref{as:periodic}, given any position $j \in \R^{-}$, we can
  in $\bigO(t)$ time compute $\deltalow(j)$ and $\deltahigh(j)$.
\end{proposition}
\begin{proof}

  By \cref{lm:delta-low-high}, in order to compute $\deltalow(j)$
  (resp.\ $\deltahigh(j)$) given $j \in \R^{-}_{H}$, it suffices to
  count the all positions $i \in \R'^{-}_{H}$ that satisfy
  $\rendfull{i} - i \geq \rendlow{j} - j$ (resp.\ $\rendfull{i} - i
  \geq \rendhigh{j} - j$) and $\T^{\infty}[\rendfull{i} \dd
  \rendfull{i} + 7\tau) \succeq \T^{\infty}[\rendlow{j} \dd j + \ell)$
  (resp.\ $\T^{\infty}[\rendfull{i} \dd \rendfull{i} + 7\tau) \succeq
  \T^{\infty}[\rendhigh{j} \dd j + 2\ell)$). The latter task is
  equivalent to the generalized orthogonal range counting queries
  (\cref{sec:range-queries}) on a set of points $\Pts_H \sub
  \mathcal{X} \times \mathcal{Y}$ (with $\mathcal{X} = \Zz$ and
  $\mathcal{Y} = \Sigma^{*}$) containing the value $\rendfull{i} - i$
  on the $\mathcal{X}$-coordinate and the string
  $\T^{\infty}[\rendfull{i} \dd \rendfull{i} + 7\tau)$ on the
  $\mathcal{Y}$-coordinate for every $i \in \R'^{-}_{H}$.

  Given any position $j \in \R^{-}$, we compute $\deltalow(j)$ and
  $\deltahigh(j)$ as follows.  First, using \cref{as:periodic} we
  compute values $s = \Lhead(j)$, $p = |\Lroot(j)|$, and $\rend{j}$ in
  $\bigO(t)$ time. Note, that then $\Lroot(j) = \T[j + s \dd j + s +
  p)$, i.e., we have a starting position of an occurrence of $H :=
  \Lroot(j)$ in $\T$.  We then calculate $k = \Lexp(j) = \lfloor
  \tfrac{\rend{j} - s}{p} \rfloor$ in $\bigO(1)$ time. Using those
  values, we further calculate $k_1 = \Lexpcut{j}{\ell} = \min(k,
  \lfloor \frac{\ell - s}{p} \rfloor)$, $k_2 = \Lexpcut{j}{2\ell} =
  \min(k, \lfloor \frac{2\ell - s}{p} \rfloor)$, $\rendlow{j} = s +
  k_1 p$, and $\rendhigh{j} = s + k_2 p$. By \cref{lm:delta-low-high}
  and the definition of $\Pts_H$, we now have
  \begin{align*}
    \deltalow(j) &=
      \rcount{\Pts_H}{\rendlow{j} - j}{n}{c^{\infty}} -
      \rcount{\Pts_H}{\rendlow{j} - j}{n}{\T^{\infty}[\rendlow{j}
        \dd j + \ell)}\text{ and }\\
    \deltahigh(j) &=
      \rcount{\Pts_H}{\rendhigh{j} - j}{n}{c^{\infty}} -
      \rcount{\Pts_H}{\rendhigh{j} - j}{n}{\T^{\infty}[\rendhigh{j}
        \dd j + 2\ell)},
  \end{align*}
  which by \cref{as:periodic} we can compute in $\bigO(t)$ time using
  the query defined by \cref{prob:int-str-it-1} of
  \cref{prob:int-str}, first with the query arguments $(i, x, q_r) =
  (\rendlow{j}, \rendlow{j} - j, j + \ell - \rendlow{j})$ and then
  with arguments $(i, x, q_r) = (\rendhigh{j}, \rendhigh{j} - j, j +
  2\ell - \rendhigh{j})$. To check that all arguments satisfy the
  requirements of \cref{prob:int-str}, recall that $\ell < n$. Thus,
  $q_r \leq 2\ell < 2n$. On the other hand, $\rendlow{j} - j \leq
  \rend{j} - j < n$ and $\rendhigh{j} - j \leq \rend{j} - j < n$ hold
  by $\T[n] = \texttt{\$}$.
\end{proof}

\begin{proposition}\label{pr:sa-delta-low-high-2}
  Let $j \in \R^{-}$. Under \cref{as:periodic}, given any position $j'
  \in \Occ_{\ell}(j)$ $($resp.\ $j' \in \Occ_{2\ell}(j))$, we can in
  $\bigO(t)$ time compute $\deltalow(j)$ $($resp.\ $\deltahigh(j))$.
\end{proposition}
\begin{proof}

  The main idea of the query is as follows.  The algorithm in the
  proof of \cref{pr:sa-delta-low-high} needs $s = \Lhead(j)$, $p =
  |\Lroot(j)|$, the value $\rendlow{j} - j$ (resp.\ $\rendhigh{j} -
  j$), some position $x \in [1 \dd n]$ such that
  $\T^{\infty}[\rendlow{j} \dd j + \ell) = \T^{\infty}[x \dd x + j +
  \ell - \rendlow{j})$ (resp.\ $\T^{\infty}[\rendhigh{j} \dd j +
  2\ell) = \T^{\infty}[x \dd x + j + 2\ell - \rendhigh{j})$), and some
  position $r \in [1 \dd n]$ such that $\Lroot(j) = \T[r \dd r +
  p)$. By \cref{lm:occ-exp-item-1} of \cref{lm:occ-exp}, we have
  $\Lhead(j') = \Lhead(j)$ and $\Lroot(j') = \Lroot(j)$.  This implies
  that we can determine $s$, $p$, and position $r$ using $j'$.  On the
  other hand, to determine $\rendlow{j} - j$ (resp.\ $\rendhigh{j} -
  j$) and the position $x$ we utilize that by
  \cref{lm:occ-exp-item-2,lm:occ-exp-item-3} of \cref{lm:occ-exp} it
  holds $\rendlow{j} - j = \rendlow{j'} - j'$ (resp.\ $\rendhigh{j} -
  j = \rendhigh{j'} - j'$) and position $x = \rendlow{j'}$ (resp.\ $x
  = \rendhigh{j'}$) satisfies $\T^{\infty}[x \dd x + j + \ell -
  \rendlow{j}) = \T^{\infty}[\rendlow{j} \dd j + \ell)$
  (resp.\ $\T^{\infty}[x \dd x + j + 2\ell - \rendhigh{j}) =
  \T^{\infty}[\rendhigh{j} \dd j + 2\ell)$).

  Given any $j' \in \Occ_{\ell}(j)$ (resp.\ $j' \in \Occ_{2\ell}(j)$),
  we compute $\deltalow(j)$ (resp.\ $\deltahigh(j)$) as follows.
  First, using \cref{as:periodic} we compute values $s = \Lhead(j') =
  \Lhead(j)$, $p = |\Lroot(j')| = |\Lroot(j)|$, and $\rend{j'}$ in
  $\bigO(t)$ time. Note, that then the position $r = j' + s$ satisfies
  $\Lroot(j) = \T[r \dd r + p)$, i.e., we have a starting position of
  an occurrence of $H := \Lroot(j') = \Lroot(j)$ in $\T$.  We then
  calculate $k = \Lexp(j') = \lfloor \tfrac{\rend{j'} - s}{p} \rfloor$
  in $\bigO(1)$ time. Using those values, we further calculate $k_1 :=
  \Lexpcut{j'}{\ell} = \min(k, \lfloor \tfrac{\ell - s}{p} \rfloor)$
  (resp.\ $k_2 := \Lexpcut{j'}{2\ell} = \min(k, \lfloor \tfrac{2\ell -
  s}{p} \rfloor)$) and $\rendlow{j'} - j' = s + k_1p$
  (resp.\ $\rendhigh{j'} - j' = s + k_2p$).  Finally, as in the proof
  of \cref{pr:sa-delta-low-high}, we obtain
  \begin{align*}
    \deltalow(j)
      &= \rcount{\Pts_H}{\rendlow{j} - j}{n}{c^{\infty}} -
         \rcount{\Pts_H}{\rendlow{j} - j}{n}{\T^{\infty}[\rendlow{j}
                 \dd j + \ell)}\\
      &= \rcount{\Pts_H}{\rendlow{j'} - j'}{n}{c^{\infty}} -
         \rcount{\Pts_H}{\rendlow{j'} - j'}{n}{\T^{\infty}[\rendlow{j'}
                 \dd j' + \ell)}
  \end{align*}
  (resp.\ $\deltahigh(j) = \rcount{\Pts_H}{\rendhigh{j} -
  j}{n}{c^{\infty}} - \rcount{\Pts_H}{\rendhigh{j} -
  j}{n}{\T^{\infty}[\rendhigh{j} \dd j + 2\ell)} =
  \rcount{\Pts_H}{\rendhigh{j'} - j'}{n}{c^{\infty}} -
  \rcount{\Pts_H}{\rendhigh{j'} - j'}{n}{\T^{\infty}[\rendhigh{j'} \dd
  j' + 2\ell)}$) in $\bigO(t)$ time, by using the query defined by
  \cref{prob:int-str-it-1} of \cref{prob:int-str} with the query
  arguments $(i, x, q_r) = (\rendlow{j'}, \rendlow{j'} - j', j' + \ell
  - \rendlow{j'})$ (resp.\ $(i, x, q_r) = (\rendhigh{j'},
  \rendhigh{j'} - j', j' + 2\ell - \rendhigh{j'})$).  The arguments
  satisfy the requirements of \cref{prob:int-str} by the same logic as
  in the proof of \cref{pr:sa-delta-low-high}.  In total, the query
  takes $\bigO(t)$ time.
\end{proof}

\begin{remark}
  Although it may seem that \cref{pr:sa-delta-low-high-2} can be
  simplified by proving that for any $j' \in \Occ_{\ell}(j)$
  (resp.\ $j' \in \Occ_{2\ell}(j)$) it holds $\Poslow(j') =
  \Poslow(j)$ (resp.\ $\Poshigh(j') = \Poshigh(j)$) and it suffices to
  use \cref{pr:sa-delta-low-high} on position $j'$, this does not
  hold. The problem occurs when $\rend{j} - j \geq \ell$
  (resp.\ $\rend{j} - j \geq 2\ell$).  Then, it may hold $\type(j') =
  +1$ and $\Poslow(j) \neq \Poslow(j')$ ($\Poshigh(j) \neq
  \Poshigh(j')$), where $\Poslow(j)$ (resp.\ $\Poshigh(j)$) is
  generalized to $j$ satisfying $\type(j) = +1$ as shown in the proof
  of \cref{pr:sa-periodic}. We shall later see that in our application
  of \cref{pr:sa-delta-low-high-2} it is not possible to guarantee
  $\type(j') = -1$ when $\type(j) = -1$, and the current form of
  \cref{pr:sa-delta-low-high-2}, in which we do not assume
  $\type(j')$, is in fact necessary.
\end{remark}

\subsubsection{Computing the Exponent}\label{sec:sa-periodic-exp}

Assume that $i \in [1 \dd n]$ satisfies $\SA[i] \in \R^{-}$ (positions
$i \in [1 \dd n]$ satisfying $\SA[i] \in \R^{+}$ are processed
symmetrically; see the proof of \cref{pr:sa-periodic}). In this
section, we show that under \cref{as:periodic}, given the index $i$
along with values $\LB_{\ell}(\SA[i])$, $\UB_{\ell}(\SA[i])$,
$\deltalow(\SA[i])$, and some position $j \in \Occ_{\ell}(\SA[i])$, we
can efficiently compute $\Lexp(\SA[i])$.

The section is organized as follows.  Given the input parameters as
described above, our query algorithm first checks if it holds $\SA[i]
\in \Poslow(\SA[i])$. To implement such check, we first present a
combinatorial result proving that $\Poslow(\SA[i])$ occupies a
contiguous block of positions in $\SA$ and showing what are the
endpoints of this block (\cref{lm:poslow}). We then use this
characterization to develop an efficient method of checking if $\SA[i]
\in \Poslow(\SA[i])$ holds (\cref{cor:poslow}). By definition of
$\Poslow(\SA[i])$, if $\SA[i] \in \Poslow(\SA[i])$, then
$\Lexp(\SA[i]) = \Lexpcut{\SA[i]}{\ell}$, which by \cref{lm:occ-exp}
is equal to $\Lexpcut{j}{\ell}$. Thus, the main difficulty is to
compute $\Lexp(\SA[i])$ when $\SA[i] \not\in \Poslow(\SA[i])$. In
\cref{lm:select-exp} we show that in such case the computation of
$\Lexp(\SA[i])$ can be reduced to modular constraint queries (see
\cref{sec:mod-queries}). We finally put everything together in
\cref{pr:sa-exp} to obtain a general query algorithm for computing
$\Lexp(\SA[i])$.

\begin{lemma}\label{lm:poslow}
  Let $i \in [1 \dd n]$ be such that $\SA[i] \in \R^{-}$.  Denote $b =
  \LB_{\ell}(\SA[i])$ and $e = b + \deltalow(\SA[i])$.  Then, it holds
  $\Poslow(\SA[i]) = \{\SA[i] : i \in (b \dd e]\}$.
\end{lemma}
\begin{proof}

  The proof consists of two steps. First, we show that
  $\Poslow(\SA[i]) \sub \{\SA[i] : i \in (b \dd n]\}$. Then, we show
  that for $i' \,{\in}\, (b {+} 1 \dd n]$, $\SA[i'] \,{\in}\,
  \Poslow(\SA[i])$ implies $\SA[i' {-} 1] \,{\in}\,
  \Poslow(\SA[i])$. This proves that $\Poslow(\SA[i]) = \{\SA[i] : i
  \in (b \dd e]\}$, where $e = b + |\Poslow(\SA[i])| = b +
  \deltalow(\SA[i])$, i.e., the claim.

  By $\SA[i] \in \Occ_{\ell}(\SA[i])$, the range $(\LB_{\ell}(\SA[i])
  \dd \UB_{\ell}(\SA[i])]$ is nonempty. In particular, $\SA[b + 1] \in
  \Occ_{\ell}(\SA[i])$, i.e., $\T^{\infty}[\SA[b + 1] \dd \SA[b + 1] +
  \ell) = \T^{\infty}[\SA[i] \dd \SA[i] + \ell)$.  Let $i' \in [1 \dd
  b]$ and denote $j' = \SA[i']$. The condition $i' \not\in (b \dd
  \UB_{\ell}(\SA[i])]$ implies $\T^{\infty}[j' \dd j' + \ell) \neq
  \T^{\infty}[\SA[i] \dd \SA[i] + \ell)$.  On the other hand, by
  definition of lexicographical order, $i' < b + 1$ implies
  $\T^{\infty}[j' \dd j' + \ell) \preceq \T^{\infty}[\SA[b + 1] \dd
  \SA[b + 1] + \ell) = \T^{\infty}[\SA[i] \dd \SA[i] + \ell)$. Thus,
  we must have $\T^{\infty}[j' \dd j' + \ell) \prec \T^{\infty}[\SA[i]
  \dd \SA[i] + \ell)$.  Let now $i'' \in [1 \dd n]$ be such that for
  $j'' = \SA[i'']$ it holds $j'' \in \Poslow(\SA[i])$. By definition
  of $\Poslow(\SA[i])$, this implies $\T[j'' \dd n] \succeq \T[\SA[i]
  \dd n]$ or $\LCE_{\T}(\SA[i], j'') \geq \ell$.  By \cref{lm:utils-2}
  in \cref{lm:utils} this is equivalent to $\T^{\infty}[j'' \dd j'' +
  \ell'') \succeq \T^{\infty}[\SA[i] \dd \SA[i] + \ell)$ for any
  $\ell'' \geq \ell$.  In particular, $\T^{\infty}[j'' \dd j'' + \ell)
  \succeq \T^{\infty}[\SA[i] \dd \SA[i] + \ell)$.  We have thus proved
  that $\T^{\infty}[j' \dd j' + \ell) \prec \T^{\infty}[\SA[i] \dd
  \SA[i] + \ell) \preceq \T^{\infty}[j'' \dd j'' + \ell)$.  In
  particular, $\T[j' \dd j' + \ell) \prec \T[j'' \dd j'' + \ell)$, and
  hence $i' < i''$. Since $i'$ was an arbitrary element of $[1 \dd
  b]$, we thus obtain $\Poslow(\SA[i]) \sub \{\SA[i] : i \in (b \dd
  n]\}$.

  Assume now that for some $i' \in (b + 1 \dd n]$ it holds $\SA[i']
  \in \Poslow(\SA[i])$. We will show that this implies $\SA[i' - 1]
  \in \Poslow(\SA[i])$. Let $s = \Lhead(\SA[i])$ and $H =
  \Lroot(\SA[i])$.
  \begin{itemize}
  \item First, observe that by $b + 1 \leq i' - 1$, the fact that
    $\SA[b + 1] \in \Occ_{\ell}(\SA[i])$ (see above), and the
    definition of the lexicographical order, we obtain
    $\T^{\infty}[\SA[i] \dd \SA[i] + \ell) = \T^{\infty}[\SA[b + 1]
    \dd \SA[b + 1] + \ell) \preceq \T^{\infty}[\SA[i' - 1] \dd \SA[i'
    - 1] + \ell)$.  By \cref{lm:utils-2} of \cref{lm:utils}, this
    implies that $\T[\SA[i' - 1] \dd n] \succeq \T[\SA[i] \dd n]$ or
    $\LCE_{\T}(\SA[i' - 1], \SA[i]) \geq \ell$.
  \item Second, by $\SA[b + 1] \in \Occ_{\ell}(\SA[i])$ and
    \cref{lm:occ-exp-item-1} of \cref{lm:occ-exp},
    we have $\SA[b + 1] \in \R_{s,H}$.  Furthermore, by $b + 1 < i'$,
    $\type(\SA[i']) = -1$, and \cref{lm:exp}, it holds $\type(\SA[b +
    1]) = -1$. Thus, $\SA[b + 1] \in \R^{-}_{s,H}$.  On the other
    hand, by definition of $\Poslow(\SA[i])$, $\SA[i'] \in
    \Poslow(\SA[i])$ implies $\SA[i] \in \R^{-}_{s,H}$. Thus, since by
    \cref{lm:exp} the positions in $\R^{-}_{s,H}$ occupy a contiguous
    block in $\SA$ and $b + 1 \leq i' - 1 < i'$, it holds $\SA[i' - 1]
    \in \R^{-}_{s,H}$.
  \item Denote $k_1 = \Lexpcut{\SA[i]}{\ell}$. Applying \cref{lm:exp},
    we obtain from $i' - 1 < i'$ and $\SA[i'] \in \Poslow(\SA[i])$
    that $\Lexp(\SA[i' - 1]) \leq \Lexp(\SA[i']) = k_1$. On the other
    hand, by $\SA[b + 1] \in \Occ_{\ell}(\SA[i])$ and
    \cref{lm:occ-exp-item-1,lm:occ-exp-item-2} of
    \cref{lm:occ-exp}, it holds $\Lhead(\SA[b + 1]) = s$,
    $\Lroot(\SA[b + 1]) = H$, and $\rendlow{\SA[b + 1]} - \SA[b + 1] =
    \rendlow{\SA[i]} - \SA[i]$.  Consequently,
    \begin{align*}
      \Lexpcut{\SA[b + 1]}{\ell}
        &= \left\lfloor \frac{\rendlow{\SA[b + 1]} -
            \SA[b + 1] - s}{|\Lroot(\SA[b + 1])|} \right\rfloor\\
        &= \left\lfloor \frac{\rendlow{\SA[i]} - \SA[i] -
            s}{|\Lroot(\SA[i])|} \right\rfloor\\[1.5ex]
        &= \Lexpcut{\SA[i]}{\ell}.
    \end{align*}
    Therefore, utilizing one last time \cref{lm:exp} for $\SA[b + 1]$
    and $\SA[i' - 1]$ we obtain from $b + 1 \leq i' - 1$ that $k_1 =
    \Lexpcut{\SA[b + 1]}{\ell} \leq \Lexp(\SA[b + 1]) \leq
    \Lexp(\SA[i' - 1])$.  Combining with the earlier bound
    $\Lexp(\SA[i' - 1]) \leq k_1$, this implies $\Lexp(\SA[i' - 1]) =
    k_1$.
  \end{itemize}
  Combining the above three conditions yields (by definition) $\SA[i'
  - 1] \in \Poslow(\SA[i])$.
\end{proof}

\begin{remark}
  By the above characterization, the set $\Poslow(\SA[i])$ occurs as a
  contiguous block of position in $\SA$ starting at index
  $\LB_{\ell}(\SA[i])$. Note, that the set $\Occ_{\ell}(\SA[i])$ has
  the same property. These two sets should not be confused, however,
  and they are not equal, nor one is always a subset of the other.
\end{remark}

\begin{corollary}\label{cor:poslow}
  For any $i \in [1 \dd n]$ such that $\SA[i] \in \R^{-}$, $\SA[i] \in
  \Poslow(\SA[i])$ holds if and only if $i - \LB_{\ell}(\SA[i]) \leq
  \deltalow(\SA[i])$.
\end{corollary}
\begin{proof}

  Assume $\SA[i] \in \Poslow(\SA[i])$. By \cref{lm:poslow}, we then
  must have $i \in (b \dd e]$, where $b = \LB_{\ell}(\SA[i])$ and $e =
  b + \deltalow(\SA[i])$. In particular, $i \leq e =
  \LB_{\ell}(\SA[i]) + \deltalow(\SA[i])$.

  Assume now $i - \LB_{\ell}(\SA[i]) \leq \deltalow(\SA[i])$.  Let $b
  = \LB_{\ell}(\SA[i])$ and $e = b + \deltalow(\SA[i])$.  Then, $i
  \leq e$. On the other hand, by definition of the set
  $\Occ_{\ell}(\SA[i])$, we have $i \in (\LB_{\ell}(\SA[i]) \dd
  \UB_{\ell}(\SA[i])]$. In particular, $i > \LB_{\ell}(\SA[i]) =
  b$. Therefore, we obtain $i \in (b \dd e]$. By \cref{lm:poslow},
  this implies $\SA[i] \in \Poslow(\SA[i])$.
\end{proof}

\begin{lemma}\label{lm:select-exp}
  Assume that $i \in [1 \dd n]$ is such that $\SA[i] \in \R^{-}_{H}$
  $($where $H \in \Lroots)$ and $\SA[i] \not\in
  \Poslow(\SA[i])$. Denote $\mathcal{I} \,{=}\, \mathcal{I}^{-}_H$
  (\cref{def:intervals}), $s \,{=}\, \Lhead(\SA[i])$, $p = |H|$, $k_1
  = \Lexpcut{\SA[i]}{\ell}$, $c = \mcount{\mathcal{I}}{p}{s}{k_1}$,
  and $i' = \LB_{\ell}(\SA[i]) + \deltalow(\SA[i])$.  Then, it holds
  $\Lexp(\SA[i]) = \mselect{\mathcal{I}}{p}{s}{c + (i - i')}$.
\end{lemma}
\begin{proof}
  For any $k \geq k_1$, denote $P_k = \{j' \in \R^{-}_{s,H} :
  \Lexp(j') \in (k_1 \dd k]\}$ and let $m_k = |P_k|$.  Our proof
  consists of two parts.
  \begin{enumerate}
  \item\label{lm:select-exp-it-1} We start by showing that for any $k
    \geq k_1$, it holds $P_k = \{\SA[i''] : i'' \in (i' \dd i' +
    m_k]\}$.  We will show this in two substeps. First, we will prove
    that for any $i'' \in [1 \dd n]$ satisfying $\SA[i''] \in P_k$, it
    holds $i'' \in (i' \dd n]$. Second, we show that for any $i'' \in
    (i'+1 \dd n]$, $\SA[i''] \in P_k$ implies $\SA[i'' - 1] \in
    P_k$. These two facts immediately imply the claim.
    \begin{itemize}
    \item Let $i'' \in [1 \dd n]$ be such that $\SA[i''] \in P_k$.
      Since $P_{k_1} = \emptyset$, this implies $k > k_1$.  Note that
      $\SA[i]$ satisfies all the conditions in the definition of
      $\Poslow(\SA[i])$, except possibly $\Lexp(\SA[i]) = k_1$. Thus,
      $\SA[i] \not\in \Poslow(\SA[i])$ implies $\Lexp(\SA[i]) \neq
      \Lexpcut{\SA[i]}{\ell}$. By $\Lexpcut{\SA[i]}{\ell} =
      \min(\Lexp(\SA[i]), \lfloor \tfrac{\ell-s}{|H|} \rfloor) \leq
      \Lexp(\SA[i])$, we then must have $k_1 = \lfloor \tfrac{\ell-s}{|H|}
      \rfloor$ and $\Lexpcut{\SA[i]}{\ell} < \Lexp(\SA[i])$.  Then, the
      string $H'H^{k_1+1}$ (where $H'$ is a length-$s$ prefix of $H$)
      is a prefix of $\T[\SA[i] \dd n]$. But since $|H'H^{k_1+1}| \geq
      \ell$, we also obtain that $\T^{\infty}[\SA[i] \dd \SA[i] +
      \ell)$ is a prefix of $H'H^{k_1+1}$.  Similarly, $\SA[i''] \in
      \R^{-}_{s,H}$ and $\Lexp(\SA[i'']) > k_1$ imply that
      $\T^{\infty}[\SA[i''] \dd \SA[i''] + \ell)$ is a prefix of
      $H'H^{k_1+1}$. Therefore, $\SA[i''] \in \Occ_{\ell}(\SA[i])$
      and, consequently, $i'' > \LB_{\ell}(\SA[i])$. Moreover, since
      by \cref{lm:poslow}, $\Poslow(\SA[i]) = \{\SA[j] : j \,{\in}\,
      (\LB_{\ell}(\SA[i]) \dd \LB_{\ell}(\SA[i]) +
      \deltalow(\SA[i])]\}$ and $\SA[i''] \not\in \Poslow(\SA[i])$, we
      must have $i'' > \LB_{\ell}(\SA[i]) + \deltalow(\SA[i]) = i'$,
      or equivalently, $i'' \in (i' \dd n]$.
    \item Assume now that for some $i'' \in (i' + 1 \dd n]$ it holds
      $\SA[i''] \in P_{k}$. We will show that this implies $\SA[i''-1]
      \in P_{k}$. As observed above, if $\SA[i''] \in P_{k}$ for $k >
      k_1$, then $\SA[i''] \in \Occ_{\ell}(\SA[i])$.  Thus, $i'' - 1 <
      \UB_{\ell}(\SA[i])$. On the other hand, $\LB_{\ell}(\SA[i]) \leq
      i' < i'' - 1$. Consequently, $\LB_{\ell}(\SA[i]) < i'' - 1 \leq
      \UB_{\ell}(\SA[i])$, or equivalently, $\SA[i'' - 1] \in
      \Occ_{\ell}(\SA[i])$.  By
      \cref{lm:occ-exp-item-1,lm:occ-exp-item-2} of \cref{lm:occ-exp},
      this implies $\SA[i''-1] \in \R^{-}_{s,H}$ and $k_1 =
      \Lexpcut{\SA[i]}{\ell} = \Lexpcut{\SA[i''-1]}{\ell} \leq
      \Lexp(\SA[i'' - 1])$. To obtain $\SA[i''-1] \in P_k$ it thus
      remains to show $\Lexp(\SA[i''-1]) \neq k_1$. This follows by
      $\SA[i''-1] \not\in \Poslow(\SA[i])$ (which holds since $i' <
      i'' - 1$ and by \cref{lm:poslow} $\Poslow(\SA[i]) \sub \{\SA[j]
      : j \in [1 \dd i']\}$) because $\SA[i''-1]$ being in
      $\Occ_{\ell}(\SA[i])$ implies that it satisfies all other
      conditions in the definition of $\Poslow(\SA[i])$.
    \end{itemize}
  \item We can now show $\Lexp(\SA[i]) = \mselect{\mathcal{I}}{p}{s}{c
    + (i-i')}$. Denote $k = \Lexp(\SA[i])$. As noted above $\SA[i]
    \not\in \Poslow(\SA[i])$ implies $k_1 < k$. By definition of
    $P_{k-1}$ and $P_{k}$, we have $\SA[i] \in P_k \sm P_{k-1}$. From
    \cref{lm:select-exp-it-1}, we have $P_{k-1} = \{\SA[i''] : i'' \in
    (i' \dd i' + m_{k-1}]\}$ and $P_k = \{\SA[i''] : i'' \in (i' \dd
    i' + m_k]\}$. Therefore, $i \in (i' + m_{k-1} \dd i' + m_k]$. On
    the other hand, by \cref{lm:mod-count}, $m_{k-1} =
    \mcount{\mathcal{I}}{p}{s}{k-1} - c$ and $m_k =
    \mcount{\mathcal{I}}{p}{s}{k} - c$.  Therefore, $i \in (i' +
    \mcount{\mathcal{I}}{p}{s}{k-1} - c \dd i' +
    \mcount{\mathcal{I}}{p}{s}{k} - c]$, or equivalently, $c + (i -
    i') \in (\mcount{\mathcal{I}}{p}{s}{k-1} \dd
    \mcount{\mathcal{I}}{p}{s}{k}]$.  By definition of the select
    query, this implies $k = \mselect{\mathcal{I}}{p}{s}{c + (i -
    i')}$. Since we defined $k = \Lexp(\SA[i])$, we obtain the main
    claim. \qedhere
  \end{enumerate}
\end{proof}

\begin{proposition}\label{pr:sa-exp}
  Let $i \in [1 \dd n]$ be such that $\SA[i] \in \R^{-}$.  Under
  \cref{as:periodic}, given the values $i$, $\LB_{\ell}(\SA[i])$,
  $\UB_{\ell}(\SA[i])$, $\deltalow(\SA[i])$, and some $j \in
  \Occ_{\ell}(\SA[i])$, we can compute $\Lexp(\SA[i])$ in $\bigO(t)$
  time.
\end{proposition}
\begin{proof}

  The main idea of the algorithm is as follows. The query algorithm
  first tests if $\SA[i] \in \Poslow(\SA[i])$ holds. By
  \cref{cor:poslow} such test can be performed quickly given the
  values $i$, $\LB_{\ell}(\SA[i])$, and $\deltalow(\SA[i])$, which are
  all given as input. The rest of the query algorithm depends on this
  test. If $\SA[i] \in \Poslow(\SA[i])$, then by definition of
  $\Poslow(\SA[i])$ and \cref{lm:occ-exp} (see below for details) we
  immediately obtain $\Lexp(\SA[i]) = \Lexpcut{\SA[i]}{\ell} =
  \Lexpcut{j}{\ell}$. If $\SA[i] \not\in \Poslow(\SA[i])$, then
  $\Lexp(\SA[i])$ is determined according to \cref{lm:select-exp}, and
  it is computed using the modular constraint queries on the set of
  intervals $\mathcal{I} = \mathcal{I}^{-}_H$ (where $H =
  \Lroot(\SA[i])$), which are supported under \cref{as:periodic}.

  Given any index $i \in [1 \dd n]$ such that $\SA[i] \in \R^{-}$,
  along with values $\LB_{\ell}(\SA[i])$, $\deltalow(\SA[i])$, and
  some $j \in \Occ_{\ell}(\SA[i])$, we compute $\Lexp(\SA[i])$ as
  follows.  First, using \cref{as:periodic} we compute values $s =
  \Lhead(j) = \Lhead(\SA[i])$, $p = |\Lroot(j)| = |\Lroot(\SA[i])|$,
  and $\rend{j}$ in $\bigO(t)$ time. Note, that then the position $r =
  j + s$ satisfies $\Lroot(\SA[i]) = \T[r \dd r + p)$, i.e., we have a
  starting position of an occurrence of $H := \Lroot(j) =
  \Lroot(\SA[i])$ in $\T$.  Denote $\mathcal{I} := \mathcal{I}^{-}_H$
  (see \cref{def:intervals}).  We then calculate $k = \Lexp(j) =
  \lfloor \tfrac{\rend{j} - s}{p} \rfloor$ in $\bigO(1)$ time. Using
  those values, we further calculate $k_1 := \min(k, \lfloor
  \tfrac{\ell - s}{p} \rfloor) = \Lexpcut{j}{\ell} =
  \Lexpcut{\SA[i]}{\ell}$ (the last equality follows by $j \in
  \Occ_{\ell}(\SA[i])$ and \cref{lm:occ-exp-item-1,lm:occ-exp-item-2}
  of \cref{lm:occ-exp}).  Next, in $\bigO(1)$ time we determine if
  $\SA[i] \in \Poslow(\SA[i])$. By \cref{cor:poslow}, $\SA[i] \in
  \Poshigh(\SA[i])$ holds if and only if $i \leq \LB_{\ell}(\SA[i]) +
  \deltalow(\SA[i])$. Consider two cases:
  \begin{itemize}
  \item Assume $i \leq \LB_{\ell}(\SA[i]) + \deltalow(\SA[i])$ (i.e.,
    $\SA[i] \in \Poslow(\SA[i])$). Then, by definition of
    $\Poslow(\SA[i])$, we have $\Lexp(\SA[i]) = k_1$.
  \item Assume now $i > \LB_{\ell}(\SA[i]) + \deltalow(\SA[i])$ (i.e.,
    $\SA[i] \not\in \Poslow(\SA[i])$). First, using a modular
    constraint counting query (\cref{sec:mod-queries}) we compute the
    value $c = \mcount{\mathcal{I}}{p}{s}{k_1}$.  By
    \cref{as:periodic}, this takes $\bigO(t)$ time.  Next, we compute
    $i' = \LB_{\ell}(\SA[i]) + \deltalow(\SA[i])$ in $\bigO(1)$
    time. Finally, using a modular constraint selection query
    (\cref{sec:mod-queries}) we compute $r =
    \mselect{\mathcal{I}}{p}{s}{c+(i-i')}$.  By \cref{as:periodic},
    this takes $\bigO(t)$ time and by \cref{lm:select-exp}, we then
    have $\Lexp(\SA[i]) = r$. \qedhere
  \end{itemize}
\end{proof}

\subsubsection{Computing the Size of 
  \texorpdfstring{$\Posmid(j)$}{Posmid}}\label{sec:sa-periodic-posmid}

In this section, we show that under \cref{as:periodic}, we can
efficiently compute $\deltamid(j) = |\Posmid(j)|$ for any $j \in
\R^{-}$ ($j \in \R^{+}$ can be processed symmetrically; see
\cref{pr:sa-periodic}).

The section is organized as follows.  We first develop an algorithm
that takes the position $j$ as input and returns $\deltamid(j)$
(\cref{pr:sa-delta-mid}). We then prove (\cref{pr:sa-delta-mid-2})
that the computation of $\deltamid(j)$ does not actually need the
value of $j$, but it is sufficient to only know $\Lexp(j)$ and some
$j' \in \Occ_{\ell}(j)$.

\begin{proposition}\label{pr:sa-delta-mid}
  Under \cref{as:periodic}, given any position $j \in \R^{-}$, we can
  in $\bigO(t)$ time compute $\deltamid(j)$.
\end{proposition}
\begin{proof}

  The main idea of the query is as follows. By \cref{lm:mod-count}, we
  can reduce the computation of $|\Posmid(j) \cap [i \dd i + t)\}|$,
  where $i \in \R^{-}$, $\Lroot(i) = \Lroot(j)$, and $t = \rend{i} - i
  - 3\tau + 2$, to two modular constraint counting queries. Observe
  that if additionally, $i \in \R'$, then this interval $[i \dd i +
  t)$ is a maximal interval of positions in $\R$, i.e., $i - 1, i + t
  \not \in \R$. Thus, letting $\mathcal{I}^{-}_H$ be the collection of
  weighted intervals (with weights corresponding to multiplicities)
  constructed as in \cref{def:intervals} for all $i \in \R'^{-}$
  satisfying $\Lroot(i) = \Lroot(j)$, we can compute $\deltamid(j)$
  using the general (weighted) modular constraint counting queries
  (\cref{sec:mod-queries}) on $\mathcal{I}^{-}_{H}$.

  Given any $j \in \R^{-}$, we compute $\deltamid(j)$ as follows.
  First, using \cref{as:periodic} we compute values $s = \Lhead(j)$,
  $p = |\Lroot(j)|$, and $\rend{j}$ in $\bigO(t)$ time. Note, that
  then the position $r = j + s$ satisfies $\Lroot(j) = \T[r \dd r +
  p)$, i.e., we have a starting position of an occurrence of $H :=
  \Lroot(j)$ in $\T$.  We then calculate $k = \Lexp(j) = \lfloor
  \tfrac{\rend{j} - s}{p} \rfloor$ in $\bigO(1)$ time. Using those
  values, we further calculate $k_1 := \Lexpcut{j}{\ell} = \min(k,
  \lfloor \frac{\ell - s}{p} \rfloor)$ and $k_2 := \Lexpcut{j}{2\ell}
  = \min(k, \lfloor \frac{2\ell - s}{p} \rfloor)$. By
  \cref{lm:mod-count} we then obtain
  \[
    \deltamid(j) = \mcount{\mathcal{I}^{-}_H}{p}{s}{k_2} -
                   \mcount{\mathcal{I}^{-}_H}{p}{s}{k_1}
  \]
  which by \cref{as:periodic} we can compute in $\bigO(t)$ time.  In
  total, the query takes $\bigO(t)$ time.
\end{proof}

\begin{proposition}\label{pr:sa-delta-mid-2}
  Let $j \in \R^{-}$. Under \cref{as:periodic}, given some position
  $j' \in \Occ_{\ell}(j)$, and the value $\Lexp(j)$, we can in
  $\bigO(t)$ time compute $\deltamid(j)$.
\end{proposition}
\begin{proof}

  The main idea is as follows. The algorithm in the proof of
  \cref{pr:sa-delta-mid} needs $s = \Lhead(j)$, $p = |\Lroot(j)|$, $k
  = \Lexp(j)$, and some $r \in [1 \dd n]$ such that $\Lroot(j) = \T[r
  \dd r + p)$. By \cref{lm:occ-exp-item-1} of \cref{lm:occ-exp}, for
  $j' \in \Occ_{\ell}(j)$, we have $\Lhead(j') = \Lhead(j)$ and
  $\Lroot(j') = \Lroot(j)$. This implies that we can determine $s$,
  $p$, and $r$ from $j'$.  The remaining information needed during the
  query ($\Lexp(j)$) is given as input.

  Given some $j' \in \Occ_{\ell}(j)$ and the value $k = \Lexp(j)$ as
  input, we compute $\deltamid(j)$ as follows.  First, using
  \cref{as:periodic} we compute $s = \Lhead(j') = \Lhead(j)$ and $p =
  |\Lroot(j')| = |\Lroot(j)|$ in $\bigO(t)$ time. Note, that then the
  position $r = j' + s$ satisfies $\Lroot(j) = \T[r \dd r + p)$, i.e.,
  we have a starting position of an occurrence of $H := \Lroot(j) =
  \Lroot(j')$ in $\T$.  Using those values, we further calculate $k_1
  := \Lexpcut{j}{\ell} = \min(k, \lfloor \tfrac{\ell - s}{p} \rfloor)$
  and $k_2 := \Lexpcut{j}{2\ell} = \min(k, \lfloor \tfrac{2\ell -
  s}{p} \rfloor)$.  Finally, as in the proof of
  \cref{pr:sa-delta-mid}, we obtain $\deltamid(j) =
  \mcount{\mathcal{I}^{-}_H}{p}{s}{k_2} -
  \mcount{\mathcal{I}^{-}_H}{p}{s}{k_1}$, which by \cref{as:periodic}
  we can compute in $\bigO(t)$ time.
\end{proof}

\subsubsection{Computing a Position in
  \texorpdfstring{$\Occ_{2\ell}(\SA[i])$}{Occ}}\label{sec:sa-periodic-occ-pos}

Assume that $i \in [1 \dd n]$ satisfies $\SA[i] \in \R^{-}$ ($i \in [1
\dd n]$ satisfying $\SA[i] \in \R^{+}$ are processed symmetrically;
see the proof of \cref{pr:sa-periodic}). In this section, we show how
under \cref{as:periodic}, given $i$ along with $\LB_{\ell}(\SA[i])$,
$\UB_{\ell}(\SA[i])$, $\deltalow(\SA[i])$, $\deltamid(\SA[i])$,
$\Lexp(\SA[i])$, and some $j \in \Occ_{\ell}(\SA[i])$, to efficiently
compute some $j' \in \Occ_{2\ell}(\SA[i])$.

The section is organized as follows.  Given the input parameters as
described above, our query algorithm first checks if it holds $\SA[i]
\in \Poshigh(\SA[i])$. To implement such check, we first present a
combinatorial result proving that $\Poshigh(\SA[i])$ occupies a
contiguous block of positions in $\SA$ and showing what are the
endpoints of this block (\cref{lm:poshigh}). We then use this
characterization to develop an efficient method of checking if $\SA[i]
\in \Poshigh(\SA[i])$ holds (\cref{cor:poshigh}). In the next two
results (\cref{lm:select,lm:select-2}), we show how in each of the two
cases reduce the computation of some position $j' \in
\Occ_{2\ell}(\SA[i])$ to a generalized range selection query (see
\cref{sec:range-queries}).  We then put everything together in
\cref{pr:sa-periodic-occ-pos}.

\begin{lemma}\label{lm:poshigh}
  Let $i \in [1 \dd n]$ be such that $\SA[i] \in \R^{-}$.  Denote $b =
  \LB_{2\ell}(\SA[i])$ and $e = b + \deltahigh(\SA[i])$.  Then, it
  holds $\Poshigh(\SA[i]) = \{\SA[i] : i \in (b \dd e]\}$.
\end{lemma}
\begin{proof}

  The proof is analogous to the proof of \cref{lm:poslow}. We will
  thus omit the parts that are identical as in the proof of
  \cref{lm:poshigh}. First, we show that $\Poshigh(\SA[i]) \sub
  \{\SA[i] : i \in (b \dd n]\}$. Then, we show that for $i' \,{\in}\,
  (b {+} 1 \dd n]$, $\SA[i'] \,{\in}\, \Poshigh(\SA[i])$ implies
  $\SA[i' {-} 1] \,{\in}\, \Poshigh(\SA[i])$. This proves that
  $\Poshigh(\SA[i]) = \{\SA[i] : i \in (b \dd e]\}$, where $e = b +
  |\Poshigh(\SA[i])| = b + \deltahigh(\SA[i])$, i.e., the claim.

  By $\SA[i] \in \Occ_{2\ell}(\SA[i])$, the range
  $(\LB_{2\ell}(\SA[i]) \dd \UB_{2\ell}(\SA[i])]$ is nonempty. In
  particular, $\SA[b + 1] \in \Occ_{2\ell}(\SA[i])$, i.e.,
  $\T^{\infty}[\SA[b + 1] \dd \SA[b + 1] + 2\ell) = \T^{\infty}[\SA[i]
  \dd \SA[i] + 2\ell)$.  Let $i' \in [1 \dd b]$ and denote $j' =
  \SA[i']$. The condition $i' \not\in (b \dd \UB_{2\ell}(\SA[i])]$
  implies $\T^{\infty}[j' \dd j' + 2\ell) \neq \T^{\infty}[\SA[i] \dd
  \SA[i] + 2\ell)$.  On the other hand, by definition of
  lexicographical order, $i' < b + 1$ implies $\T^{\infty}[j' \dd j' +
  2\ell) \preceq \T^{\infty}[\SA[b + 1] \dd \SA[b + 1] + 2\ell) =
  \T^{\infty}[\SA[i] \dd \SA[i] + 2\ell)$. Thus, we must have
  $\T^{\infty}[j' \dd j' + 2\ell) \prec \T^{\infty}[\SA[i] \dd \SA[i]
  + 2\ell)$.  Let now $i'' \in [1 \dd n]$ be such that for $j'' =
  \SA[i'']$ it holds $j'' \in \Poshigh(\SA[i])$. By definition of
  $\Poshigh(\SA[i])$, this implies $\T[j'' \dd n] \succeq \T[\SA[i]
  \dd n]$ or $\LCE_{\T}(\SA[i], j'') \geq 2\ell$.  By
  \cref{lm:utils-2} in \cref{lm:utils} this is equivalent to
  $\T^{\infty}[j'' \dd j'' + \ell'') \succeq \T^{\infty}[\SA[i] \dd
  \SA[i] + 2\ell)$ for any $\ell'' \geq 2\ell$.  In particular,
  $\T^{\infty}[j'' \dd j'' + 2\ell) \succeq \T^{\infty}[\SA[i] \dd
  \SA[i] + 2\ell)$.  We have thus proved that $\T^{\infty}[j' \dd j' +
  2\ell) \prec \T^{\infty}[\SA[i] \dd \SA[i] + 2\ell) \preceq
  \T^{\infty}[j'' \dd j'' + 2\ell)$.  In particular, $\T[j' \dd j' +
  2\ell) \prec \T[j'' \dd j'' + 2\ell)$, and hence $i' < i''$. Since
  $i'$ was an arbitrary element of $[1 \dd b]$, we thus obtain
  $\Poshigh(\SA[i]) \sub \{\SA[i] : i \in (b \dd n]\}$.

  Assume now that for some $i' \in (b + 1 \dd n]$ it holds $\SA[i']
  \in \Poshigh(\SA[i])$. We will show that this implies $\SA[i' - 1]
  \in \Poshigh(\SA[i])$. Let $s = \Lhead(\SA[i])$ and $H =
  \Lroot(\SA[i])$.
  \begin{itemize}
  \item First, observe that by $b + 1 \leq i' - 1$, the fact that
    $\SA[b + 1] \in \Occ_{2\ell}(\SA[i])$ (see above), and the
    definition of the lexicographical order, we obtain
    $\T^{\infty}[\SA[i] \dd \SA[i] + 2\ell) = \T^{\infty}[\SA[b + 1]
    \dd \SA[b + 1] + 2\ell) \preceq \T^{\infty}[\SA[i' - 1] \dd \SA[i'
    - 1] + 2\ell)$.  By \cref{lm:utils-2} of \cref{lm:utils}, this
    implies that $\T[\SA[i' - 1] \dd n] \succeq \T[\SA[i] \dd n]$ or
    $\LCE_{\T}(\SA[i' - 1], \SA[i]) \geq 2\ell$.
  \item Second, by $\SA[b + 1] \in \Occ_{2\ell}(\SA[i])$ and
    \cref{lm:occ-exp-item-1} of \cref{lm:occ-exp}, we have $\SA[b + 1]
    \in \R_{s,H}$. Thus, using the argument from the proof of
    \cref{lm:poslow}, we have $\SA[i' - 1] \in \R^{-}_{s,H}$.
  \item Denote $k_2 = \Lexpcut{\SA[i]}{2\ell}$. Applying
    \cref{lm:exp}, we obtain from $i' - 1 < i'$ and $\SA[i'] \in
    \Poshigh(\SA[i])$ that $\Lexp(\SA[i' - 1]) \leq \Lexp(\SA[i']) =
    k_2$. On the other hand, by $\SA[b + 1] \in \Occ_{2\ell}(\SA[i])$
    and \cref{lm:occ-exp-item-1,lm:occ-exp-item-2} of
    \cref{lm:occ-exp}, it holds $\Lhead(\SA[b + 1]) = s$,
    $\Lroot(\SA[b + 1]) = H$, and $\rendhigh{\SA[b + 1]} - \SA[b + 1]
    = \rendhigh{\SA[i]} - \SA[i]$.  Consequently, (see the proof of
    \cref{lm:poslow}), $\Lexpcut{\SA[b + 1]}{2\ell} =
    \Lexpcut{\SA[i]}{2\ell}$.  Therefore, utilizing one last time
    \cref{lm:exp} for $\SA[b + 1]$ and $\SA[i' - 1]$ we obtain from $b
    + 1 \leq i' - 1$ that $k_2 = \Lexpcut{\SA[b + 1]}{2\ell} \leq
    \Lexp(\SA[b + 1]) \leq \Lexp(\SA[i' - 1])$.  Combining with the
    earlier bound $\Lexp(\SA[i' - 1]) \leq k_2$, this implies
    $\Lexp(\SA[i' - 1]) = k_2$.
  \end{itemize}
  Combining the above three conditions yields (by definition) $\SA[i'
  - 1] \in \Poshigh(\SA[i])$.
\end{proof}

\begin{remark}
  Similarly as for $\Poslow(\SA[i])$ (see \cref{lm:poslow}), by the
  above characterization, the set $\Poshigh(\SA[i])$ occurs as a
  contiguous block of position in $\SA$ starting at index
  $\LB_{2\ell}(\SA[i])$. Note, that the set $\Occ_{2\ell}(\SA[i])$ has
  the same property. These two sets should not be confused, however,
  and they are not equal, nor one is always a subset of the other. We
  will, however, use $\Poshigh(\SA[i])$ to infer about
  $\Occ_{2\ell}(\SA[i])$. More precisely, to compute a position $j \in
  \Occ_{2\ell}(\SA[i])$, we will distinguish two cases: $\SA[i] \in
  \Poshigh(\SA[i])$ and $\SA[i] \not\in \Poshigh(\SA[i])$. In the
  first case, we will indeed locate and element of
  $\Occ_{2\ell}(\SA[i])$ from $\Poshigh(\SA[i])$.  We first, however,
  need to develop an efficient test of whether $\SA[i] \in
  \Poshigh(\SA[i])$ holds.
\end{remark}

\begin{corollary}\label{cor:poshigh}
  For any $i \in [1 \dd n]$ such that $\SA[i] \in \R^{-}$, $\SA[i] \in
  \Poshigh(\SA[i])$ holds if and only if $i - \LB_{\ell}(\SA[i]) \leq
  \deltalow(\SA[i]) + \deltamid(\SA[i])$.
\end{corollary}
\begin{proof}

  Assume $\SA[i] \in \Poshigh(\SA[i])$. By \cref{lm:poshigh}, we then
  must have $i \in (b \dd e]$, where $b = \LB_{2\ell}(\SA[i])$ and $e
  = b + \deltahigh(\SA[i])$. In particular, $i \leq e =
  \LB_{2\ell}(\SA[i]) + \deltahigh(\SA[i]) = \LB_{\ell}(\SA[i]) +
  \deltal_{\ell}(\SA[i]) + \deltahigh(\SA[i]) = \LB_{\ell}(\SA[i]) +
  \deltalow(\SA[i]) + \deltamid(\SA[i])$, where the last equality
  follows by \cref{lm:delta}. This is equivalent to the claim.

  Assume now $i - \LB_{\ell}(\SA[i]) \leq \deltalow(\SA[i]) +
  \deltamid(\SA[i])$. Let $b = \LB_{2\ell}(\SA[i])$ and $e = b +
  \deltahigh(\SA[i])$. Then, by \cref{lm:delta}, we equivalently have
  $i \leq \LB_{\ell}(\SA[i]) + \deltalow(\SA[i]) + \deltamid(\SA[i]) =
  \LB_{2\ell}(\SA[i]) + \deltahigh(\SA[i]) = e$. On the other hand, by
  definition, we have $i \in (\LB_{2\ell}(\SA[i]) \dd
  \UB_{2\ell}(\SA[i])]$. In particular, $i > \LB_{2\ell}(\SA[i]) =
  b$. Therefore, we obtain $i \in (b \dd e]$. By \cref{lm:poshigh},
  this implies $\SA[i] \in \Poshigh(\SA[i])$.
\end{proof}

\begin{lemma}\label{lm:select}
  Assume that $i \in [1 \dd n]$ is such that $\SA[i] \in \R^{-}_H \cap
  \Poshigh(\SA[i])$, where $H \in \Lroots$.  Let $\mathsf{P}_H =
  \{(\rendfull{r}, \rendfull{r} - r) : r \in \R'^{-}_H\}$, $\Pts =
  \Points_{7\tau}(\T, \mathsf{P}_H)$, $d = \deltahigh(\SA[i])$, $x =
  \rendhigh{\SA[i]} - \SA[i]$, $m = \rcountb{\Pts}{x}{n}$, and $e =
  \LB_{\ell}(\SA[i]) + \deltalow(\SA[i]) + \deltamid(\SA[i])$. Then $d
  \leq m$. Moreover:
  \begin{enumerate}
  \item\label{lm:select-it-1} For $\delta \in [0 \dd d)$, any position
    $p \in \rselect{\Pts}{x}{n}{m - \delta}$ satisfies $\T^{\infty}[p
    - x \dd p - x + 2\ell) = \T^{\infty}[\SA[e - \delta] \dd \SA[e -
    \delta] + 2\ell)$.
  \item\label{lm:select-it-2} For $\delta = e - i$, we have $\delta
    \in [0 \dd d)$ and any position $p \in \rselect{\Pts}{x}{n}{m -
    \delta}$ satisfies $p - x \in \Occ_{2\ell}(\SA[i])$.
  \end{enumerate}
\end{lemma}
\begin{proof}

  Note that $\Pts$ is well-defined, since positions $\rendfull{r}$ are
  distinct among $r\in \R'^{-}_{H}$.  Denote $q = |\R'^{-}_H|$. Let
  $(a_j)_{j \in [1 \dd q]}$ be a sequence containing all positions $r
  \in \R'^{-}_{H}$ ordered according to the string
  $\T^{\infty}[\rendfull{r} \dd \rendfull{r} + 7\tau)$ In other words,
  for any $j, j' \in [1 \dd q]$, $j < j'$ implies
  $\T^{\infty}[\rendfull{r} \dd \rendfull{r} + 7\tau) \prec
  \T^{\infty}[\rendfull{r'} \dd \rendfull{r'} + 7\tau)$, where $r =
  a_j$ and $r' = a_{j'}$.  Note, that the sequence $(a_j)_{j \in [1
  \dd q]}$ is not unique.  Since $\{(\rendfull{a_i}, \rendfull{a_i} -
  a_i) : i \in [1 \dd q]\} = \mathsf{P}_H$, it holds
  $|\{\rendfull{a_j} - x : j \in [1 \dd q]\text{ and } \rendfull{a_j}
  - a_j \geq x\}| = |\{j \in [1 \dd q] : \rendfull{a_j} - a_j \geq
  x\}| = |\{(p, v) \in \mathsf{P}_H : v \in [x \dd n)\}| = m$, where
  the last equality follows by the definition of
  $\rcountb{\Pts}{x}{n}$ (see \cref{sec:int-str}), and by observing
  that for any $q \in \R$, it holds $\rendfull{q} - q < n$.  On the
  other hand, observe that by \cref{lm:delta-low-high}, for any $j \in
  [1 \dd q]$, the right-maximal run of positions in $\R$ starting at
  position $r = a_j$ contains an element $p'$ of $\Poshigh(\SA[i])$ if
  and only if $\rendfull{r} - r \geq x$ and $\T^{\infty}[\rendfull{r}
  \dd \rendfull{r} + 7\tau) \succeq \T^{\infty}[\rendhigh{\SA[i]} \dd
  \SA[i] + 2\ell)$.  Moreover, if the latter condition is true, then
  since by definition on $\Poshigh(\SA[i])$ any such element $p'$ must
  satisfy $\Lexp(p') = \Lexpcut{\SA[i]}{2\ell}$ (or equivalently,
  $\rendfull{p'} - p' = s + \Lexpcut{\SA[i]}{2\ell}|H| =
  \rendhigh{\SA[i]} - \SA[i] = x$, where $s = \Lhead(\SA[i]) =
  \Lhead(p')$), we obtain $p' = \rendfull{p'} - x = \rendfull{r} -
  x$. In other words, $\Poshigh(\SA[i]) = \{\rendfull{a_j} - x : j \in
  [1 \dd q],\, \rendfull{a_j} - a_j \geq x,\text{ and
  }\T^{\infty}[\rendfull{a_j} \dd \rendfull{a_j} + 7\tau) \succeq
  \T^{\infty}[\rendhigh{\SA[i]} \dd \SA[i] + 2\tau)\}$.  The latter
  set (whose cardinality is equal to $d$) is clearly a subset of
  $\{\rendfull{a_j} - x : j \in [1 \dd q]\text{ and }\rendfull{a_j} -
  a_j \geq x\}$ (whose size, as shown above, is $m$). Thus, $d \leq
  m$.

  1. As shown above, $|\{j \in [1 \dd q] : \rendfull{a_j} - a_j \geq
  x\}| = m$.  Let $(b_j)_{j \in [1 \dd m]}$ be a subsequence of
  $(a_j)_{j \in [1 \dd q]}$ containing all the elements of the set
  $\{a_j : j \in [1 \dd q]\text{ and } \rendfull{a_j} - a_j \geq x\}$
  (in the same order as they appear in the sequence $(a_j)_{j \in [1
  \dd q]}$). Our proof consists of three steps:
  \begin{enumerate}[label=(\roman*)]
  \item\label{lm:select-it-1-1} Let $j \in [1 \dd m]$. We first show
    $b_j \in \rselect{\Pts}{x}{n}{j}$.  Denote $Y =
    \T^{\infty}[\rendfull{b_j} \dd \rendfull{b_j} + 7\tau)$. Let
    \begin{align*}
      r_{\rm beg} &=
        |\{a_t : t \in [1 \dd q],
          \rendfull{a_t} - a_t \geq x,\text{ and }
          \T^{\infty}[\rendfull{a_t} \dd \rendfull{a_t} + 7\tau)
          \prec Y\}|\text{ and}\\
      r_{\rm end} &=
        |\{a_t : t \in [1 \dd q],
          \rendfull{a_t} - a_t \geq x,\text{ and }
          \T^{\infty}[\rendfull{a_t} \dd \rendfull{a_t} + 7\tau)
          \preceq Y\}|.
    \end{align*}
    If an index $t \in [1 \dd q]$ satisfies $\rendfull{a_t} - a_t \geq
    x$, then we also have $a_t \in \{b_1, \ldots, b_m\}$.  Moreover,
    since for any indexes $t, t' \in [1 \dd m]$, $t < t'$ implies
    $\T^{\infty}[\rendfull{b_t} \dd \rendfull{b_t} + 7\tau) \preceq
    \T^{\infty}[\rendfull{b_{t'}} \dd \rendfull{b_{t'}} + 7\tau)$, if
    $t \in [1 \dd q]$ additionally satisfies
    $\T^{\infty}[\rendfull{a_t} \dd \rendfull{a_t} + 7\tau) \prec Y =
    \T^{\infty}[\rendfull{b_j} \dd \rendfull{b_j} + 7\tau)$, then $a_t
    \in \{b_1, \ldots, b_{j-1}\}$. Thus, $r_{\rm beg} < j$.  On the
    other hand, every $t \in [1 \dd j]$ satisfies $\rendfull{b_t} -
    b_t \geq x$ and $\T^{\infty}[\rendfull{b_t} \dd \rendfull{b_t} +
    7\tau) \preceq \T^{\infty}[\rendfull{b_j} \dd \rendfull{b_j} +
    7\tau) = Y$. Thus, $j \leq r_{\rm end}$. Altogether, $j \in
    (r_{\rm beg} \dd r_{\rm end}]$. By definition of $\Pts =
    \Points_{7\tau}(\T, \mathsf{P}_H)$ (\cref{def:p-right-context}),
    it holds $r_{\rm beg} = \rcount{\Pts}{x}{n}{Y}$ and $r_{\rm end} =
    \rcounti{\Pts}{x}{n}{Y}$. Thus, $j \in (\rcount{\Pts}{x}{n}{Y} \dd
    \rcounti{\Pts}{x}{n}{Y}]$.  On the other hand, by
    $(\rendfull{b_j}, \rendfull{b_j} - b_j) \in \mathsf{P}_H$, we have
    $(\rendfull{b_j} - b_j, Y, f) \in \Pts$.  Combining with $x \leq
    \rendfull{b_j} - b_j < n$ we obtain $b_j \in
    \rselect{\Pts}{x}{n}{j}$.
  \item\label{lm:select-it-1-2} Let $j \in [1 \dd m]$. We will now
    show that for any $p \in \rselect{\Pts}{x}{n}{j}$, it holds
    $\T^{\infty}[p - x \dd p + 7\tau) = \T^{\infty}[\rendfull{b_j} - x
    \dd \rendfull{b_j} + 7\tau)$. By \cref{lm:select-it-1-1} and the
    definition of $\rselect{\Pts}{x}{n}{j}$, the assumption $p \in
    \rselect{\Pts}{x}{n}{j}$ implies $\T^{\infty}[p \dd p + 7\tau) =
    \T^{\infty}[b_j \dd b_j + 7\tau)$.  Moreover, $x \leq \rendfull{p}
    - p < n$ and there exists $r \in \R'^{-}_H$ such that $p =
    \rendfull{r}$. The position $p$ is thus preceded in $\T^{\infty}$
    by a string $H'H^{z}$, where $H'$ is a prefix of $H$ of length $s
    = \Lhead(\SA[i])$ and $z = \Lexpcut{\SA[i]}{2\ell}$.  Since by
    definition of $(b_j)_{j\in [1 \dd m]}$, the position
    $\rendfull{b_j}$ is also preceded by $H'H^{z}$ in $\T$ and
    $|H'H^{z}| = x$, we obtain $\T^{\infty}[p - x \dd p + 7\tau) =
    \T^{\infty}[\rendfull{b_j} - x \dd \rendfull{b_j} + 7\tau)$.
  \item We are now ready to prove the main claim.  As noted above,
    $\Poshigh(\SA[i]) = \{\rendfull{a_j} - x : j \in [1 \dd q],\,
    \rendfull{a_j} - a_j \geq x,\text{ and }\T^{\infty}[\rendfull{a_j}
    \dd \rendfull{a_j} + 7\tau) \succeq \T^{\infty}[\rendhigh{\SA[i]}
    \dd \SA[i] + 2\tau)\}$. Since the positions $k$ in
    $(a_j)_{j \in [1 \dd q]}$ are sorted by
    $\T^{\infty}[\rendfull{k} \dd \rendfull{k} + 7\tau)$, we can
    eliminate the second condition. Denoting $j_{\rm skip} = |\{j \in
    [1 \dd q] : \T^{\infty}[\rendfull{a_j} \dd \rendfull{a_j} + 7\tau)
    \prec \T^{\infty}[\rendhigh{\SA[i]} \dd \SA[i] + 2\ell)\}|$, we
    have
    \begin{equation*}
      \Poshigh(\SA[i]) = \{\rendfull{a_j} - x : j \in (j_{\rm skip}
      \dd q]\text{ and }\rendfull{a_j} - a_a \geq x\}.
    \end{equation*}
    Consequently, by definition of $(b_j)_{j \in [1 \dd m]}$, we have
    $\Poshigh(\SA[i]) = \{\rendfull{b_j} - x : j \in (m-d \dd m]\}$.
    On the other hand, by \cref{lm:delta}, $e = \LB_{\ell}(\SA[i]) +
    \deltalow(\SA[i]) + \deltamid(\SA[i]) = \LB_{\ell}(\SA[i]) +
    \deltal_{\ell}(\SA[i]) + \deltahigh(\SA[i]) =
    \LB_{2\ell}(\SA[i]) + \deltahigh(\SA[i])$.  Thus, by
    \cref{lm:poshigh}, we have $\Poshigh(\SA[i]) = \{\SA[j] : j \in
    (e-d \dd e]\}$. We now observe:
    \begin{itemize}
    \item Let $j_1, j_2 \in (m-d \dd m]$ and assume $j_1 < j_2$. Since
      the elements of $(b_j)$ occur in the same order as in $(a_j)$,
      and positions $p$ in $(a_j)$ are sorted by $\T[\rendfull{p} \dd
      \rendfull{p} + 7\tau)$, it follows that
      $\T^{\infty}[\rendfull{p_1} \dd \rendfull{p_1} + 7\tau) \preceq
      \T^{\infty}[\rendfull{p_2} \dd \rendfull{p_2} + 7\tau)$, where
      $p_1 = b_{j_1}$ and $p_2 = b_{j_2}$. On the other hand, by
      $\rendfull{p_1} - x, \rendfull{p_2} - x \in \Poshigh(\SA[i])$,
      both positions $\rendfull{p_1} - x$ and $\rendfull{p_2} - x$ are
      followed in $\T$ by the string $H'H^{z}$, where $H'$ is a prefix
      of $H$ of length $\Lhead(\SA[i])$ and $z =
      \Lexpcut{\SA[i]}{2\ell}$. Recall, however, that $x =
      \rendhigh{\SA[i]} - \SA[i] = \Lhead(\SA[i]) +
      \Lexpcut{\SA[i]}{2\ell}$.  Thus, we obtain
      $\T^{\infty}[\rendfull{p_1} - x \dd \rendfull{p_1}) =
      \T^{\infty}[\rendfull{p_2} - x \dd \rendfull{p_2}) = H'H^{z}$,
      and consequently, $\T^{\infty}[\rendfull{p_1} - x \dd
      \rendfull{p_1} + 7\tau) \preceq \T^{\infty}[\rendfull{p_2} - x
      \dd \rendfull{p_2} + 7\tau)$.
    \item On the other hand, by definition of lexicographical order,
      for any $j_1, j_2 \in [1 \dd d]$, the assumption $j_1 < j_2$
      implies $\T^{\infty}[\SA[e - d + j_1] \dd \SA[e - d + j_1] + x +
      7\tau) \preceq \T^{\infty}[\SA[e - d + j_2] \dd \SA[e - d + j_2]
      + x + 7\tau)$.
    \end{itemize}
    We have shown that both sequences $\SA[e - d + 1], \ldots,
    \SA[e]$ and $\rendfull{b_{m-d+1}} - x, \ldots, \rendfull{b_m} - x$
    contain the same set of positions $\Poshigh(\SA[i])$ ordered
    according to the length-$(x+7\tau)$ right context in
    $\T^{\infty}$.  Therefore, regardless of how ties are resolved in
    each sequence, for any $j \in [1 \dd d]$, we have
    $\T^{\infty}[\SA[e - d + j] \dd \SA[e - d + j] + x + 7\tau) =
    \T^{\infty}[\rendfull{b_{m-d+j}} - x \dd \rendfull{b_{m-d+j}} +
    7\tau)$. Thus, by letting $\delta = d - j$, for any $\delta \in [0
    \dd d)$ we obtain
    \begin{equation*}
      \T^{\infty}[\SA[e - \delta] \dd \SA[e - \delta] + x + 7\tau) =
      \T^{\infty}[\rendfull{b_{m-\delta}} - x \dd
      \rendfull{b_{m-\delta}} + 7\tau).
    \end{equation*}
    To finalize the proof of the main claim, let $\delta \in [0 \dd
    d)$ and take any $p \in \rselect{\Pts}{x}{n}{m-\delta}$.  By
    \cref{lm:select-it-1-2} for $j = m - \delta$, we have
    $\T^{\infty}[p - x \dd p + 7\tau) = \T^{\infty}[\rendfull{b_{m -
    \delta}} - x \dd \rendfull{b_{m - \delta}} + 7\tau) =
    \T^{\infty}[\SA[e - \delta] \dd \SA[e - \delta] + x + 7\tau)$. In
    particular, by $0 \leq x$ and $2\ell \leq 7\tau$, we obtain $2\ell
    \leq x + 7\tau$ and $\T^{\infty}[p - x \dd p - x + 2\ell) =
    \T^{\infty}[\SA[e - \delta] \dd \SA[e - \delta] + 2\ell)$, i.e.,
    the claim.
  \end{enumerate}

  \noindent
  2. By \cref{lm:delta}, it holds $e = \LB_{\ell}(\SA[i]) +
  \deltal_{\ell}(\SA[i]) + \deltahigh(\SA[i]) = \LB_{2\ell}(\SA[i]) +
  \deltahigh(\SA[i])$. Thus, by \cref{lm:poshigh}, we have
  $\Poshigh(\SA[i]) = \{\SA[i] : i \in (e - d \dd e]\}$. Consequently,
  $\SA[i] \in \Poshigh(\SA[i])$ implies $i \in (e - d \dd e]$ and
  hence $\delta = e - i$ satisfies $\delta \in [0 \dd d)$. By
  \cref{lm:select-it-1} with $\delta = e - i$, any $p \in
  \rselect{\Pts}{x}{n}{m-\delta}$ satisfies $\T^{\infty}[p - x \dd p -
  x + 2\ell) = \T^{\infty}[\SA[e - \delta] \dd \SA[e - \delta] +
  2\ell) = \T^{\infty}[\SA[i] \dd \SA[i] + 2\ell)$, i.e., $p - x \in
  \Occ_{2\ell}(\SA[i])$.
\end{proof}

\begin{lemma}\label{lm:select-2}
  Assume that $i \in [1 \dd n]$ is such that $\SA[i] \in \R^{-}_{H}$
  $($where $H \in \Lroots)$ and $\SA[i] \not\in \Poshigh(\SA[i])$. Let
  $\mathsf{P}_H = \{(\rendfull{p}, \rendfull{p} - p) : p \in
  \R'^{-}_{H}\}$, $\Pts = {\rm Point}_{7\tau}(\T, \mathsf{P}_H)$, $s =
  \Lhead(\SA[i])$, $z = \lfloor \tfrac{2\ell-s}{|H|} \rfloor$, and $x
  = s + (z+1)|H|$. Then:
  \begin{enumerate}
  \item\label{lm:select-2-it-1} It holds $\rcountb{\Pts}{x}{n} \geq
    1$.
  \item\label{lm:select-2-it-2} Any $p \in \rselect{\Pts}{x}{n}{1}$
    satisfies $p - x \in \Occ_{2\ell}(\SA[i])$.
  \end{enumerate}
\end{lemma}
\begin{proof}

  1. By definition, if $\SA[i] \not\in \Poshigh(\SA[i])$, then
  $\Lexp(\SA[i]) \neq \Lexpcut{\SA[i]}{2\ell}$, since all other
  conditions in the definition of $\Poshigh(\SA[i])$ are satisfied for
  position $\SA[i]$.  However, since $\Lexpcut{\SA[i]}{2\ell} =
  \min(\Lexp(\SA[i]), \lfloor \tfrac{2\ell-s}{|H|} \rfloor)$, this
  implies $\lfloor \tfrac{2\ell-s}{|H|} \rfloor < \Lexp(\SA[i])$.
  Consequently, for $k = \SA[i]$ it holds $\rendfull{k} - k =
  \Lhead(\SA[i]) + \Lexp(\SA[i])|H| \geq s + (z+1)|H| = x$, and hence
  $(p, v) = (\rendfull{k}, \rendfull{k} - k) \in \mathsf{P}_H$
  satisfies $v \in [x \dd n)$. Thus, by definition of $\Pts$,
  $\rcountb{\Pts}{x}{n} \geq 1$.

  2. By $\Lexp(\SA[i]) \geq z + 1$, the string $H'H^{z+1}$ (where $H'$
  is a length-$s$ prefix of $H$) is a prefix of $\T[\SA[i] \dd
  n]$. But since $|H'H^{z + 1}| \geq 2\ell$, we also obtain that
  $\T^{\infty}[\SA[i] \dd \SA[i] + 2\ell)$ is a prefix of $H'H^{z+1}$.
  Thus, to find an element of $\Occ_{2\ell}(\SA[i])$, it suffices to
  find any position $p \in [1 \dd n]$ satisfying $\T^{\infty}[p - x
  \dd p) = H'H^{z+1}$ and then by the above argument, it holds $p - x
  \in \Occ_{2\ell}(\SA[i])$.  Let now $p \in \rselect{\Pts}{x}{n}{1}$.
  By definition of $\rselect{\Pts}{x}{n}{1}$, this implies $x \leq v <
  n$.  Since $(p, v) \in \mathsf{P}_H$ holds for some $v\in \Z_{\ge
  1}$, there exists $r \in \R'^{-}_H$ such that $p = \rendfull{r}$
  and $v = \rendfull{r} - r$, the position $p$ is preceded in
  $\T^{\infty}$ by a string $H'H^{z+1}$ (recall that $|H'H^{z+1}| =
  x$). By the above argument, this implies $p - x \in
  \Occ_{2\ell}(\SA[i])$.
\end{proof}

\begin{proposition}\label{pr:sa-periodic-occ-pos}
  Let $i \in [1 \dd n]$ be such that $\SA[i] \in \R^{-}$. Under
  \cref{as:periodic}, given the values $i$, $\LB_{\ell}(\SA[i])$,
  $\UB_{\ell}(\SA[i])$, $\deltalow(\SA[i])$, $\deltamid(\SA[i])$,
  $\Lexp(\SA[i])$, and some $j \in \Occ_{\ell}(\SA[i])$, we can
  compute some $j' \in \Occ_{2\ell}(\SA[i])$ in $\bigO(t)$ time.
\end{proposition}
\begin{proof}

  The main idea is as follows. The query algorithm first tests if
  $\SA[i] \in \Poshigh(\SA[i])$ holds. By \cref{cor:poshigh} such test
  can be performed quickly given the values $i$, $\LB_{\ell}(\SA[i])$,
  $\deltalow(\SA[i])$, and $\deltamid(\SA[i])$, which are all given as
  input. The rest of the query algorithm depends on this test. If
  $\SA[i] \in \Poshigh(\SA[i])$, then $j' \in \Occ_{2\ell}(\SA[i])$ is
  determined according to \cref{lm:select}, and it is computed using
  the generalized range selection query on the set of points $\Pts =
  \Points_{7\tau}(\T, \mathsf{P}_H)$ (where $H = \Lroot(\SA[i])$),
  which are supported using the algorithm from \cref{pr:sa-occ-eq}. On
  the other hand, if $\SA[i] \not\in \Poshigh(\SA[i])$, then we use
  \cref{lm:select-2} to compute some $j' \in
  \Occ_{2\ell}(\SA[i])$. This again reduces to a generalized range
  selection query on $\Pts$, but with different input parameters.

  Given any index $i \in [1 \dd n]$ such that $\SA[i] \in \R^{-}$,
  along with values $\LB_{\ell}(\SA[i])$, $k = \Lexp(\SA[i])$,
  $\deltalow(\SA[i])$, $\deltamid(\SA[i])$, and some $j \in
  \Occ_{\ell}(\SA[i])$, we compute $j' \in \Occ_{2\ell}(\SA[i])$ as
  follows. First, using \cref{as:periodic} we compute values $s =
  \Lhead(j) = \Lhead(\SA[i])$ and $p = |\Lroot(j)| = |\Lroot(\SA[i])|$
  in $\bigO(t)$ time. Note, that then the position $r = j + s$
  satisfies $\Lroot(\SA[i]) = \T[r \dd r + p)$, i.e., we have a
  starting position of an occurrence of $H := \Lroot(j) =
  \Lroot(\SA[i])$ in $\T$. Denote $\Pts = \Points_{7\tau}(\T,
  E^{-}_H)$ (see \cref{def:p-right-context}). In $\bigO(1)$ time we
  calculate $k_2 := \Lexpcut{\SA[i]}{2\ell} = \min(k, \lfloor
  \tfrac{2\ell - s}{p} \rfloor)$.  Next, in $\bigO(1)$ time we
  determine if $\SA[i] \in \Poshigh(\SA[i])$. By \cref{cor:poshigh},
  $\SA[i] \in \Poshigh(\SA[i])$ holds if and only if $i \leq
  \LB_{\ell}(\SA[i]) + \deltalow(\SA[i]) +
  \deltamid(\SA[i])$. Consider two cases:
  \begin{itemize}
  \item Assume $i \leq \LB_{\ell}(\SA[i]) + \deltalow(\SA[i]) +
    \deltamid(\SA[i])$ (i.e., $\SA[i] \in \Poshigh(\SA[i])$). First,
    we compute $x = \rendhigh{\SA[i]} - \SA[i] = \Lhead(\SA[i]) +
    \Lexpcut{\SA[i]}{2\ell}|\Lroot(\SA[i])| = s + k_2p$, $e =
    \LB_{\ell}(\SA[i]) + \deltalow(\SA[i]) + \deltamid(\SA[i])$, and
    $\delta = e - i$ in $\bigO(1)$ time. Then, using the query defined
    by \cref{prob:int-str-it-1} of \cref{prob:int-str} we compute the
    value $m = \rcountb{\Pts}{x}{n}$. By \cref{as:periodic}, this
    takes $\bigO(t)$ time. Next, using the query defined by
    \cref{prob:int-str-it-2} of \cref{prob:int-str} we compute some $p
    \in \rselect{\Pts}{x}{n}{m - \delta}$. This again takes $\bigO(t)$
    time.  Finally, we calculate $j' = p - x$.  By
    \cref{lm:select-it-2} from \cref{lm:select}, it holds $j' \in
    \Occ_{2\ell}(\SA[i])$.
  \item Assume now $i > \LB_{\ell}(\SA[i]) + \deltalow(\SA[i]) +
    \deltamid(\SA[i])$ (i.e., $\SA[i] \not\in \Poshigh(\SA[i])$).
    First, calculate $z = \lfloor \tfrac{2\ell - s}{p} \rfloor$ and $x
    = s + (z + 1)p$.  Next, using \cref{prob:int-str-it-2} of
    \cref{prob:int-str} we compute some $p \in
    \rselect{\Pts}{x}{n}{1}$. By \cref{as:periodic}, this takes
    $\bigO(t)$ time.  Finally, we calculate $j' = p - x$. By
    \cref{lm:select-2-it-2} from \cref{lm:select-2}, it holds $j' \in
    \Occ_{2\ell}(\SA[i])$. \qedhere
  \end{itemize}
\end{proof}

\subsubsection{The Data Structure}\label{sec:sa-periodic-structure}


By combining the above results, under \cref{as:periodic}, we obtain
the following efficient algorithm to ``refine'' $(\LB_{\ell}(\SA[i],
\UB_{\ell}(\SA[i]))$ into $(\LB_{2\ell}(\SA[i]), \UB_{2\ell}(\SA[i]))$
and to simultaneously compute some $j' \in \Occ_{2\ell}(\SA[i])$.

\begin{proposition}\label{pr:sa-periodic}
  Let $i \in [1 \dd n]$ be such that $\SA[i] \in \R$.  Under
  \cref{as:periodic}, given the values $i$, $\LB_{\ell}(\SA[i])$,
  $\UB_{\ell}(\SA[i])$, and some $j \in \Occ_{\ell}(\SA[i])$, we can
  compute $(\LB_{2\ell}(\SA[i]), \UB_{2\ell}(\SA[i]))$ and some $j'
  \in \Occ_{2\ell}(\SA[i])$ in $\bigO(t)$ time.
\end{proposition}
\begin{proof}

  The algorithm consists of seven steps:
  \begin{enumerate}
  \item\label{sa:step-1} First, using \cref{pr:sa-type}, we compute in
    $\bigO(t)$ time the value of $\type(\SA[i])$. The query algorithm
    needs the index $i$ along with $\LB_{\ell}(\SA[i])$,
    $\UB_{\ell}(\SA[i])$, and some $j \in \Occ_{\ell}(\SA[i])$, which
    at this point we have as input. Let us assume $\type(\SA[i]) = -1$
    (the case $\type(\SA[i]) = +1$ is explained at the end of the
    proof).
  \item\label{sa:step-2} Next, using \cref{pr:sa-delta-low-high-2}, we
    compute in $\bigO(t)$ time the value of $\deltalow(\SA[i])$. The
    algorithm only needs some $j \in \Occ_{\ell}(\SA[i])$, which we
    are given as input.
  \item\label{sa:step-3} In the third step, using \cref{pr:sa-exp}, we
    compute in $\bigO(t)$ time the value of $\Lexp(\SA[i])$. The
    algorithm needs the values of $i$, $\LB_{\ell}(\SA[i])$,
    $\UB_{\ell}(\SA[i])$, some $j \in \Occ_{\ell}(\SA[i])$, and
    $\deltalow(\SA[i])$. The first four values are given as input, and
    the fifth value was computed in Step~\ref{sa:step-2}.
  \item\label{sa:step-4} Next, using \cref{pr:sa-delta-mid-2}, we
    compute in $\bigO(t)$ time the value $\deltamid(\SA[i])$. The
    algorithm needs some position $j \in \Occ_{\ell}(\SA[i])$ and the
    value $\Lexp(\SA[i])$ as input. The first argument is given as
    input, and the second value we computed in Step~\ref{sa:step-3}.
  \item\label{sa:step-5} Next, using \cref{pr:sa-periodic-occ-pos}, we
    compute in $\bigO(t)$ time some position $j' \in
    \Occ_{2\ell}(\SA[i])$.  The algorithm needs values $i$,
    $\LB_{\ell}(\SA[i])$, $\UB_{\ell}(\SA[i])$, some position $j \in
    \Occ_{\ell}(\SA[i])$, $\deltalow(\SA[i])$, $\Lexp(\SA[i])$, and
    $\deltamid(\SA[i])$ as input. The first four arguments are given
    as input.  The remaining three we computed in
    Step~\ref{sa:step-2}, Step~\ref{sa:step-3}, and
    Step~\ref{sa:step-4}.
  \item\label{sa:step-6} Next, using \cref{pr:sa-delta-low-high-2}, we
    compute in $\bigO(t)$ time the value of $\deltahigh(\SA[i])$. The
    algorithm only needs some $j' \in \Occ_{2\ell}(\SA[i])$, which we
    computed in Step~\ref{sa:step-5}.
  \item\label{sa:step-7} Finally, using \cref{pr:sa-occ-size} and its
    symmetric version adapted according to \cref{lm:exp} (see also
    below), we compute in $\bigO(t)$ time the values $m^{-} :=
    |\Occm_{2\ell}(j')|$ and $m^{+} := |\Occp_{2\ell}(j')|$, where $j'
    \in \Occ_{2\ell}(\SA[i])$.  Note that by definition of
    $\Occ_{2\ell}(\SA[i])$ and $j' \in \Occ_{2\ell}(\SA[i])$, we have
    $|\Occm_{2\ell}(\SA[i])| = |\Occm_{2\ell}(j')|$ and
    $|\Occp_{2\ell}(\SA[i])| = |\Occp_{2\ell}(j')|$. Observe now that
    by \cref{lm:occ-exp-item-1} of \cref{lm:occ-exp}, for any $p \in
    \Occ_{2\ell}(\SA[i])$, it holds $p \in \R$. Thus,
    $\Occ_{2\ell}(\SA[i])$ is a disjoint union of
    $\Occm_{2\ell}(\SA[i])$ and $\Occp_{2\ell}(\SA[i])$, and hence
    denoting $m = |\Occ_{2\ell}(\SA[i])|$, we have $m = m^{-} +
    m^{+}$. To compute $m^{-}$ and $m^{+}$, the query algorithms in
    \cref{pr:sa-occ-size} and its symmetric counterpart only require
    some position $j' \in \Occ_{2\ell}(\SA[i])$ as input.  We obtained
    such position in Step~\ref{sa:step-5}.
  \end{enumerate}
  After executing the above steps, we already have $j' \in
  \Occ_{2\ell}(\SA[i])$ (computed in Step~\ref{sa:step-5}) Thus, it
  remains to calculate the pair $(\LB_{2\ell}(\SA[i]),
  \UB_{2\ell}(\SA[i]))$.  For this, in $\bigO(1)$ time we first obtain
  $\deltal_{\ell}(\SA[i]) = \deltalow(\SA[i]) + \deltamid(\SA[i]) -
  \deltahigh(\SA[i])$ by \cref{lm:delta}.  We then obtain the output
  as $(\LB_{2\ell}(\SA[i]), \UB_{2\ell}(\SA[i])) = (\LB_{\ell}(\SA[i])
  + \deltal_{\ell}(\SA[i]), \LB_{\ell}(\SA[i]) +
  \deltal_{\ell}(\SA[i]) + m)$. In total, the query takes $\bigO(t)$
  time.

  If in Step~\ref{sa:step-1}, we have $\type(\SA[i]) = +1$, then in
  Steps~\ref{sa:step-2}--\ref{sa:step-6} instead of queries on
  $\Points_{7\tau}(\T, E^{-}_H)$ and $\mathcal{I}^{-}_H$ in
  \cref{as:periodic}, we perform the queries on $\Points_{7\tau}(\T,
  E^{+}_H)$ and $\mathcal{I}^{+}_H$. By \cref{lm:exp}, all remaining
  details are symmetrical.  In particular, $\deltag_{\ell}(\SA[i]) =
  \deltalow(\SA[i]) + \deltamid(\SA[i]) - \deltahigh(\SA[i])$ and
  $(\LB_{2\ell}(\SA[i]), \UB_{2\ell}(\SA[i])) = (\UB_{\ell}(\SA[i]) -
  \deltag_{\ell}(\SA[i]) - m, \UB_{\ell}(\SA[i]) -
  \deltag_{\ell}(\SA[i]))$. Definitions of sets $\Poslow(j)$,
  $\Posmid(j)$, and $\Poshigh(j)$ are adapted for $j \in \R^{+}$ as
  follows:
  \begin{align*}
    \Poslow(j) &=
      \{j' \in \R_{s,H}^{+} : \Lexp(j') = k_1\text{ and }(\T[j' \dd n]
      \preceq \T[j \dd n]\text{ or }\LCE_{\T}(j,j') \geq \ell)\},\\
    \Posmid(j) &=
      \{j' \in \R_{s,H}^{+} : \Lexp(j') \in (k_1 \dd k_2]\},\text{
      and}\\
    \Poshigh(j) &=
      \{j' \in \R_{s,H}^{+} : \Lexp(j') = k_2\text{ and }(\T[j' \dd n]
      \preceq \T[j \dd n]\text{ or }\LCE_{\T}(j, j') \geq 2\ell)\}.
  \end{align*}
  By the symmetric version of \cref{lm:delta}, for $j \in \R^{+}$,
  $\deltag_{\ell}(j) = \deltalow(j) + \deltamid(j) - \deltahigh(j)$.
\end{proof}

\subsection{The Final Data Structure}\label{sec:sa-final}

We are ready to present the general $\SA$ query algorithm (returning
$\SA[i]$ given any $i \in [1 \dd n]$).

The section is organized as follows. First, we show how to combine
\cref{sec:sa-core}, \cref{sec:sa-nonperiodic}, and
\cref{sec:sa-periodic} to obtain a query algorithm that for any $\ell
\in [16 \dd n)$, given any index $i$ along with the values
$\LB_{\ell}(\SA[i])$, $\UB_{\ell}(\SA[i])$, and some $j \in
\Occ_{\ell}(\SA[i])$ as input, computes the pair
$(\LB_{2\ell}(\SA[i]), \UB_{2\ell}(\SA[i]))$ and some $j' \in
\Occ_{2\ell}(\SA[i])$ (\cref{pr:sa-range}). We then combine this
result with the query algorithm that given $i \in [1 \dd n]$ computes
the pair $(\LB_{16}(\SA[i]), \UB_{16}(\SA[i]))$ and some $j' \in
\Occ_{16}(\SA[i])$ (\cref{as:small}) to obtain the final algorithm for
$\SA$ queries (\cref{pr:sa}).

\begin{proposition}\label{pr:sa-range}
  For any $\ell \in [16 \dd n)$, under
  \cref{as:core,as:nonperiodic,as:periodic}, given an index $i \in [1
  \dd n]$ along with the pair $(\LB_{\ell}(\SA[i]),
  \UB_{\ell}(\SA[i]))$ and some position $j \in \Occ_{\ell}(\SA[i])$
  as input, we can compute the pair $(\LB_{2\ell}(\SA[i]),
  \UB_{2\ell}(\SA[i]))$ and some $j' \in \Occ_{2\ell}(\SA[i])$ in
  $\bigO(t)$ time.
\end{proposition}
\begin{proof}

  First, using \cref{as:core} we check in $\bigO(t)$ if $\SA[i] \in
  \R$.  If $\SA[i] \not\in \R$, then we compute the output in
  $\bigO(t)$ time using \cref{pr:sa-nonperiodic}.  Otherwise, we
  compute the output in $\bigO(t)$ time using \cref{pr:sa-periodic}.
\end{proof}

\begin{proposition}\label{pr:sa}
  Under \cref{as:small} and \cref{as:core,as:nonperiodic,as:periodic}
  for $\ell = 2^q$, where $q \in [4 \dd \lceil \log n \rceil)$, given
  any $i \in [1 \dd n]$, we can compute $\SA[i]$ in $\bigO(t \log n)$
  time.
\end{proposition}
\begin{proof}

  Given any position $i \in [1 \dd n]$, we compute $\SA[i]$ as
  follows. First, using \cref{as:small}, we compute the pair
  $(\LB_{16}(\SA[i]), \UB_{16}(\SA[i]))$ and some $j \in
  \Occ_{16}(\SA[i])$ in $\bigO(t)$ time. Then, for $q = 4, \ldots,
  \lceil \log n \rceil - 1$, we use \cref{pr:sa-range} with $\ell =
  2^q$ to compute in $\bigO(t)$ time the pair $(\LB_{2^{q+1}}(\SA[i]),
  \UB_{2^{q+1}}(\SA[i]))$ and some $j' \in \Occ_{2^{q+1}}(\SA[i])$,
  given position $i$ along with the pair $(\LB_{2^q}(\SA[i]),
  \UB_{2^q}(\SA[i]))$ and some $j \in \Occ_{2^q}(\SA[i])$ as
  input. After executing all steps, we have $(\LB_{\ell}(\SA[i]),
  \UB_{\ell}(\SA[i]))$ and some $j' \in \Occ_{\ell}(\SA[i])$, where
  $\ell = 2^{\lceil \log n \rceil} \geq n$. Since for any $k \geq n$,
  we have $\Occ_{k}(\SA[i]) = \{\SA[i]\}$, we finally return $\SA[i] =
  j'$.  In total, the query takes $\bigO(t \log n)$ time.
\end{proof}

\section{Balanced Signature Parsing}\label{sec:parsing}

\subsection{Symbols and Grammars}

For an alphabet $\Sigma$, we define the set (algebraic data type)
$\Symb$ of \emph{symbols} over $\Sigma$ as the least fixed point of
the following equation, where $\letter$ and $\block$ are two named
constructors:
\[
  \Symb = \{\letter(a): a\in \Sigma\} \cup \{\block(S) : S\in
  \Symb^+\},
\] 
We define a recursive \emph{expansion} function $\str : \Symb\to
\Sigma^+$ with
\[
  \str(X) =
  \begin{cases}
    a & \text{if }X=\letter(a)\text{ for
    }a\in\Sigma,\\ \bigodot_{i=1}^{|S|}\str(S[i]) &
    \text{if }X=\block(S)\text{ for }S\in \Symb^+,
  \end{cases}
\]
and we lift it to $\str : \Symb^*\to \Sigma^*$ with
$\str(S)=\bigodot_{i=1}^{|S|}\str(S[i])$.  We further define a
recursive \emph{level} function $\level : \Symb \to \Zz$ with
\[
  \level(X) =
  \begin{cases}
    0 & \text{if }X= \letter(a)\text{ for }a\in
    \Sigma,\\ \max_{i=1}^{|S|}\level(S[i])& \text{if
    }X=\block(S)\text{ for }S\in \Symb^+.
  \end{cases}
\]

We say that a set $\Gr\sub \Symb$ is a \emph{grammar} if $\Gr \sub
\{\letter(a):a\in \Sigma\}\cup \{\block(S) : S\in \Gr^+\}$. Our
algorithms maintain grammars using the following interface.

\newcommand{\iL}{\mathsf{insertLetter}}
\newcommand{\iT}{\mathsf{insertString}}
\newcommand{\iP}{\mathsf{insertPower}}

\begin{lemma}[Grammar Representation]\label{lem:gr}
  For every constant $C\in \Zz$, there exists a data structure that
  maintains a grammar $\Gr\sub \Symb$ subject to the following
  operations, where each symbol $X\in \Gr$ is represented by a unique
  identifier $\id(X)\in [0\dd g)$, where $g=|\Gr|$:
  \begin{description}
  \item[$\degree(X)$:] Given $X\in \Gr$, return $0$ if $X=\letter(a)$
    and $|S|$ if $X=\block(S)$.
  \item[$\letter^{-1}(X)$:] Given $X=\letter(a)$, return $a$.
  \item[$\block^{-1}(X,i)$:] Given $X=\block(S)$ and $i\in [1\dd
    |S|]$, return $S[i]$.
  \item[$\level(X)$:] Given $X\in \Gr$, return $\level(X)$.
  \item[$\length(X)$:] Given $X\in\Gr$, return $|\str(X)|$.
  \item[$\preflength(X,i)$:] Given $X=\block(S)$ and $i\in [0\dd
    |S|]$, return $|\str(S[1\dd i])|$.
  \item[$\iL(a)$:] Given $a\in \Sigma$, set $\Gr :=
    \Gr\cup\{\letter(a)\}$ and return $\letter(a)$.
  \item[$\iT(S)$:] Given $S\in \Gr^+$ with $|S|\le C$, set $\Gr :=
    \Gr\cup\{\block(S)\}$ and return $\block(S)$.
  \item[$\iP(Y,m)$:] Given $Y\in \Gr$ and $m\in \Zp$, set $\Gr :=
    \Gr\cup\{\block(Y^m)\}$ and return $\block(Y^m)$.
  \end{description}
  In the $\Omega(\log(\sigma + \sum_{X\in \Gr}\length(X)))$-bit word
  RAM model, the data structure size is $\Oh(g)$, insertions cost
  $\Oh(\frac{\log^2 \log g}{\log \log \log g})$ time, and queries cost
  $\Oh(1)$ time.
\end{lemma}
\begin{proof} 
  Internally, the symbols are stored in array $A[0\dd g)$ with the
  following entries $A[\id(X)]$ for $X\in \Gr$:
  \begin{itemize}
    \item if $X=\letter(a)$, then $A[\id(X)]=a$;
    \item if $X=\block(Y^m)$ for $Y\in \Gr$ and $m\in \Zp$, then
      $A[\id(X)]=(\id(Y),m)$;
    \item otherwise, $X=\block(S)$ for $S\in \Gr^{m}$ with $m\in [1\dd
      C]$, and $A[\id(X)]=\id(S[1])\cdots \id(S[m])$.
  \end{itemize}
  Additionally, we maintain the inverse mapping $A^{-1}$ using a
  dynamic deterministic dictionary~\cite{wexp}, as well as the
  precomputed values $\level(X)$ and $\length(X)$ for all symbols
  $X\in \Gr$.  The array $A$ and the precomputed values allow for an
  easy implementation of all queries in $\Oh(1)$ time.  As for
  insertions, we first build the internal representation of the
  inserted symbol $X$ (note that $\iT(S)$ needs to check if $S=Y^m$
  for some $Y\in \Symb$ and $m\in \Zp$). Then, we perform a dictionary
  lookup for $X$ in $A^{-1}$.  If the lookup is successful, we
  retrieve $\id(X)$ and return this identifier as the reference to
  $X$. Otherwise, we set $\id(X):=g$, increment $g$, insert $X$ to
  $\Gr$ (by adding the corresponding entries to $A$ and $A^{-1}$), and
  compute $\level(X)$ as well as $\length(X)$.  The running time is
  dominated by the cost $\Oh(\frac{\log^2 \log g}{\log \log \log g})$
  of a dictionary lookup and insertion.
\end{proof}

\subsection{Parse Trees}

Every symbol $X\in \Symb$ can be interpreted as a rooted ordered
\emph{parse tree} $\Tr(X)$ with each node $\nu$ associated to a symbol
$\symb(\nu)\in \Symb$. The root of $\Tr(X)$ is a node $\rho$ with
$\symb(\rho)=X$.  If $X=\letter(a)$ for $a\in \Sigma$, then $\rho$ has
no children.  Otherwise, $X=\block(S)$ for $S\in \Symb^m$ and $m\in
\Zp$; then, $\rho$ has $m$ children, and the subtree rooted at the
$i$th child is (a copy of) $\Tr(S[i])$.

Each node $\nu$ of $\Tr(X)$ is associated with a fragment $\str(\nu)$
of $\str(X)$ matching $\str(\symb(\nu))$.  For the root $\nu$, we
define $\str(\rho)=\str(X)[1\dd \length(X)]$ to be the whole
$\str(X)$.  Moreover, if $\nu_1,\ldots,\nu_m$ are the children of a
node $\nu$, then $\str(\nu_i)=\str(\nu)(r_{i-1}\dd r_i]$, where $r_i =
\preflength(\symb(\nu),i)$; this way, $\str(\nu)$ is the concatenation
of fragments $\str(\nu_i)$.

\begin{lemma}\label{lem:pt}
  Based on the interface of $\Gr$ provided in \cref{lem:gr}, for every
  $X\in \Gr$, one can implement references to nodes $\nu$ of $\Tr(X)$
  so that the following operations cost $\Oh(1)$ time:
  \begin{description}
  \item[$\root(X)$:] Given $X\in \Gr$, return the root of $\Tr(X)$.
  \item[$\symb(\nu)$:] Given a node $\nu$ of $\Tr(X)$, return the
    symbol $\symb(\nu)$.
  \item[$\str(\nu)$:] Given a node $\nu$ of $\Tr(X)$, return the
    fragment $\str(\nu)$.
  \item[$\child(\nu,i)$:] Given a node $\nu$ of $\Tr(X)$ and $i\in
    [1\dd \degree(\symb(\nu))]$, return the $i$th child of $\nu$.
  \item[$\parent(\nu)$:] Given a node $\nu$ of $\Tr(X)$, return $\bot$
    if $\nu=\root(X)$, and otherwise the parent of $\nu$.
  \item[$\indexop(\nu)$:] Given a node $\nu$ of $\Tr(X)$, return
    $\bot$ if $\nu=\root(X)$, and otherwise the value $i\in \Zp$ such
    that $\nu = \child(\parent(\nu),i)$.
  \end{description}
\end{lemma}
\begin{proof}
  Internally, each node $\nu$ is represented as a tuple
  $(\symb(\nu),\str(\nu),\indexop(\nu),\parent(\nu))$.  Here,
  $\parent(\nu)$ points to the parent of $\nu$ stored in the same
  representation.  In other words, the internal of $\nu$ is a singly
  linked list describing to the path from $\nu$ to the root of its
  parse tree. This is a \emph{functional} list, which means that many
  multiple references can share parts of the underlying lists.  The
  non-trivial operations are $\root(X)$, which returns $(X,
  \str(X)[1\dd |\str(X)|], \bot,\bot)$, and $\child(\nu,i)$, which
  returns
  $(\block^{-1}(\symb(\nu),i),\str(\nu)(\preflength(\symb(\nu),i-1)\dd
  \preflength(\symb(\nu),i)],i,\nu)$.
\end{proof}

\subsection{Deterministic Coin Tossing}

For an integer $h\in \Zz$, let us denote \[\tow{h} := \Big[0\dd
\underbrace{2^{2^{\cdot^{\cdot^2}}}}_h\Big)\qquad\text{and}\qquad\ptow{h}:=
\tow{h}\cup\{\bot\}.\]

\begin{theorem}\label{thm:mehlhorn}
  For every integer $h\in \Zp$, there is a function $g_h :
  \ptow{h}^{h+ 11} \to \{\zero,\one\}$ such that the function $G_h :
  \ptow{h}^\Z \to \{\zero,\one\}^\Z$, defined with $G_h(T)[i] =
  g_h(T(i-h - 7 \dd i+4])$ for $i\in \Z$, satisfies the following
  properties for all integers $n\in \Zz$ and strings $T\in
  \ptow{h}^\Z$ with $T[1\dd n]\in \tow{h}^n$, $T[0]=T[n+1]=\bot$, and
  $T[j]\ne T[j+1]$ for all $j\in [1\dd n)$:
  \begin{enumerate}[label={\rm(\arabic*)}]
  \item\label{it:bnd} $G_h(T)[0]=G_h(T)[n]=\one$.
  \item\label{it:frb} $G_h(T)[0\dd n]$ does not contain $\one\one$ and
    $\zero\zero\zero\zero\zero\zero$ as proper substrings.
  \item\label{it:dep} $G_h(T)[0\dd n]$ would not be affected by
    setting $T[j]:=\bot$ for all $j\in \Z\sm [1\dd n]$.
  \end{enumerate}
  Moreover, given $h$ and $T[1\dd n]$, the string $G_h(T)[0\dd n]$ can
  be computed in $\Oh(nh)$ time provided that each character of
  $T[1\dd n]$ is an $\Oh(w)$-bit number.
\end{theorem}
\begin{proof}
  The starting point of our construction is the function $f_h :
  \ptow{h}^{h + 11} \to \{\zero,\one\}$ behind~\cite[Lemma
  1]{Mehlhorn}, for which the derived $F_h:\ptow{h}^\Z \to
  \{\zero,\one\}^\Z$ (defined analogously to $G_h$) satisfies the
  following property: For all integers $n\in \Zz$ and strings $T\in
  \ptow{h}^\Z$ with $T[1\dd n]\in \tow{h}^n$, $T[j]=\bot$ for $j\in
  \Z\sm [1\dd n]$, and $T[j]\ne T[j+1]$ for $j\in [1\dd n)$, the
  string $F_h(T)[1\dd n]$ does not contain $\one\one$ and
  $\zero\zero\zero\zero$ as substrings.

  For a string $X\in \ptow{h}^{h + 11}$, let $X(\ell \dd r]$ be the
  maximal fragment such that $\ell \le h + 7 \le r$ and $X(\ell\dd
  r]\in \tow{h}^{*}$.  We set:
  \begin{enumerate}[label={\rm(\alph*)}]
  \item\label{it:one} $g_h(X)=\one$ if $\ell=h + 7$ or $r=h + 7$;
  \item\label{it:zero} $g_h(X)=\zero$ if $\ell=h + 6$ (and $r> h + 7$)
    or $r=h + 8$ (and $\ell < h + 7$);
  \item\label{it:f} $g_h(X)=f_h(\bot^{\ell}X(\ell\dd r]\bot^{h + 11 -
    r})$ otherwise.
  \end{enumerate}

  Let us fix a string $T\in \ptow{h}^\Z$ with $T[1\dd n]\in
  \tow{h}^n$, $T[0]=T[n+1]=\bot$, and $T[j]\ne T[j+1]$ for all $j\in
  [1\dd n)$.  Define a string $\bar{T}\in \ptow{h}^\Z$ so that
  $\bar{T}[j]=T[j]$ for $j\in [1\dd n]$ and $\bar{T}[j] =\bot$
  otherwise.

  Observe that $G_h(T)[0]=G_h(T)[n]=\one$ by~\ref{it:one}.  Moreover,
  if $n\ge 2$, then $G_h(T)[1] = G_h(T)[n-1] = \zero$
  by~\ref{it:zero}, and, if $n\ge 4$, then $G_h(T)[2\dd
  n-2]=F_h(\bar{T})[2\dd n-2]$ by~\ref{it:f}.  This immediately
  yields conditions~\ref{it:bnd} and~\ref{it:dep}.

  As for condition~\ref{it:frb}, note that any occurrence $\one\one$
  as a proper substring of $G_h(T)[0\dd n]$ would yield an occurrence
  of $\one\one$ as a substring $F_h(\bar{T})[2\dd n-2]$, and any
  occurrence of $\zero\zero\zero\zero\zero\zero$ as a substring of
  $G_h(T)[0\dd n]$ would yield an occurrence of $\zero\zero\zero\zero$
  as a substring of $F_h(\bar{T})[2\dd n-2]$.

  An efficient algorithm for computing $G_h(T)[0\dd n]$ follows from
  the corresponding algorithm for computing $F_h(\bar{T})[0\dd n]$.
\end{proof}

\begin{definition}[Restricted signature parsing]
  Consider a set $\Act\sub \Symb$ of \emph{active symbols}, an integer
  $h\in \Zz$, and an injective function $\sig : \Act \to \tow{h}$.
  Given a string $T\in \Symb^*$, define a string $\sig(T)\in
  \ptow{h}^\Z$ so that $\sig(T)[i]=\sig(T[i])$ if $i\in [1\dd |T|]$
  and $T[i]\in \Act$, and $\sig(T)[i]=\bot$ otherwise.  The
  \emph{restricted signature parsing} decomposes $T$ into
  \emph{blocks} that end at all positions $i\in [1\dd |T|]$ with
  $G_h(\sig(T))[i]=\one$, where $G_h$ is defined in
  \cref{thm:mehlhorn}.
\end{definition}

\begin{corollary}\label{cor:mehlhorn}
  Consider the restricted signature parsing of a string $T\in \Symb^*$
  such that no position $i\in [1\dd |T|)$ satisfies $T[i]=T[i+1]\in
  \Act$.  Then:
  \begin{enumerate}[label=(\arabic*)]
  \item each block contains exactly one inactive symbol or up to six
    active symbols;
  \item an active symbol forms a length-1 block if and only if it has
    no adjacent active symbols.
  \end{enumerate}
\end{corollary}
\begin{proof}
  Let $T(\ell \dd r]$ be any maximal fragment of $T$ consisting of
  active symbols.  By \cref{thm:mehlhorn}, we have
  $G_h(\sig(T))[\ell]=G_h(\sig(T))[r]=\one$, so $T(\ell\dd r]$
  consists of full blocks.  Since $G_h(\sig(T))(\ell\dd r]$ does not
  contain $\zero\zero\zero\zero\zero\zero$, these blocks are of length
  at most six. Moreover, since $G_h(\sig(T))[\ell\dd r]$ does not
  contain $\one\one$ as a proper substring, the only possibility for a
  length-1 block within $T(\ell\dd r]$ is when $r-\ell=1$, i.e., when
  $T[r]$ is an active symbol with no adjacent active symbols.  In the
  reasoning above, $G_h(\sig(T))[\ell]=G_h(\sig(T))[r]=\one$ also
  holds when $\ell=r$, so all inactive symbols form length-1 blocks.
\end{proof}

\begin{corollary}\label{cor:consistent}
  Consider the restricted signature parsing of strings $T,T'\in
  \Symb^*$.  If a block ends at $T[i]$ but no block ends at $T'[i']$,
  then there exists:
  \begin{itemize}
  \item a string in $\Act^{\le 4}$ that is a prefix of exactly one of
    the suffixes $T(i\dd |T|]$ and $T' (i'\dd |T'|]$, or
  \item a string in $\Act^{\le h + 7}$ that is a suffix of exactly one
    of the prefixes $T[1\dd i]$ and $T'[1\dd i']$.
  \end{itemize}
\end{corollary}
\begin{proof}  
  Let $T(\ell\dd r]$ and $T'(\ell'\dd r']$ be maximal fragments with
  $i-h - 7 \le \ell\le i\le r \le i+4$ and $i'-h -7 \le \ell'\le i'\le
  r'\le i'+4$, respectively, such that $T(\ell\dd r],T'(\ell'\dd
  r']\in \Act^*$.  If $T(\ell\dd i]\ne T'(\ell'\dd i']$, this yields a
  desired suffix of exactly one of the prefixes $T[1\dd i]$ and
  $T'[1\dd i']$.  Similarly, if $T(i\dd r]\ne T'(i'\dd r']$, this
  yields a desired prefix of exactly one of the suffixes $T(i\dd |T|]$
  and $T'(i'\dd |T'|]$.  By \cref{thm:mehlhorn}\ref{it:dep}, we have
  $G_h(\sig(T))[i] = g_h(\bot^{\ell-i+h + 7}\sig(T)(\ell\dd
  r]\bot^{i+4-r})$ and $G_h(\sig(T'))[i'] = g_h(\bot^{\ell'-i'+h +
  7}\sig(T')(\ell'\dd r']\bot^{i'+4-r'})$.  Due to the assumption that
  $T(\ell\dd i]=T'(\ell'\dd i']$ and $T(i\dd r]=T'(i'\dd r']$, this
  yields $G_h(\sig(T))[i]=G_h(\sig(T'))[i']$. Hence, blocks end at
  both $T[i]$ and $T'[i']$ or at neither $T[i]$ nor $T'[i']$.
\end{proof}

\subsection{Balanced Signature Parsing}\label{sec:bsp}

\begin{definition}[Restricted run parsing]
  Consider a set $\Act\sub \Symb$ of \emph{active symbols}.  The
  \emph{restricted run parsing} decomposes a string $T\in \Symb^*$
  into \emph{blocks} that end at all positions $i\in [1\dd |T|]$
  except those with $T[i]=T[i+1]\in \Act$.
\end{definition}

For every $k\in \Zz$, we define the set of level-$k$ \emph{active
symbols} $\Act_k := \{X\in \Symb : \level(X)= k\text{ and
}\length(X)\le d_k\}$, where $d_k := 2^{\floor{k/2}}$.  Moreover, we
say that a function $\sig:\Symb \to \Zz$ is a \emph{signature
function} if $\sig_k := \sig|_{\Act_k}$ injectively maps $\Act_k$ to
$\tow{h_k}$, where $h_k = \log^* |\Act_k|$.

\begin{fact}\label{rem:nk}
  Every $k\in \Zz$ satisfies $h_k = \Oh(\log^*(\sigma k))$.
\end{fact}
\begin{proof}
  For $k,\ell\in \Zz$, let $s_{k,\ell}=|\{X\in \Symb : \level(X)\le
  k\text{ and }\length(X)\le \ell\}|$.  We shall inductively prove
  that $s_{k,\ell}\le \sigma^\ell (k+1)^{\ell-1}$.  The base case of
  $k=0$ is easy to check since $s_{0,0} = 0$ and $s_{0,\ell}=\sigma$
  for $\ell \in \Zp$.  As for $k\in \Zp$, we have
  \begin{multline*}
    s_{k,\ell}\le s_{k-1,\ell}+\sum_{j=1}^{\ell-1} s_{k-1,j}\cdot
    s_{k,\ell-j} \le \sigma^{\ell}k^{\ell-1} +
    \sum_{j=1}^{\ell-1}\sigma^{j}k^{j-1}\sigma^{\ell-j}(k+1)^{\ell-j-1}
    \\ = \sigma^{\ell} (k+1)^{\ell-2}
    \left(\tfrac{k^{\ell-1}}{(k+1)^{\ell-2}} +
    \sum_{i=0}^{\ell-2}\tfrac{k^{i}}{(k+1)^{i}}\right) \le
    \sigma^{\ell} (k+1)^{\ell-2}
    \left(\tfrac{k^{\ell-1}}{(k+1)^{\ell-2}} +
    \tfrac{1-\frac{k^{\ell-1}}{(k+1)^{\ell-1}}}{1-\frac{k}{k+1}}\right)
    = \sigma^\ell (k+1)^{\ell-1}.
  \end{multline*}
  In particular, we conclude that $|\Act_k| \le
  \sigma^{2^{k/2}}(k+1)^{2^{k/2}-1}\le 2^{2^{\sigma k}}$, so $h_k \le
  2+\log^*(\sigma k)$.
\end{proof}

\begin{construction}[Balanced signature parsing]\label{constr:parsing}
  For a signature function $\sig$, the \emph{balanced signature
  parsing} of a string $T\in \Sigma^+$ is a sequence
  $(T_k)_{k=0}^\infty$ of strings $T_k\in \Symb^+$ constructed as
  follows: The string $T_0$ is obtained from $T$ by replacing each
  letter $T[j]$ with symbol $\letter(T[j])$.  For $k>0$, we perform
  the restricted signature parsing of $T_{k-1}$ with respect to
  $\Act_{k-1}$, $h_{k-1}$, and $\sig_{k-1}$ (if $k$ is even) or the
  restricted run parsing of $T_{k-1}$ with respect to $\Act_{k-1}$ (if
  $k$ is odd) and derive $T_k$ by collapsing each block $T_{k-1}[j\dd
  j+m)$ into the corresponding symbol $\block(T_{k-1}[j\dd j+m))$.
\end{construction}

Observe that $\str(T_k)=T$ for $k\in \Zz$. Hence, for every $j\in
[1\dd |T_k|]$, we can associate $T_k[j]$ with a fragment
$T(|\str(T_k[1\dd j))|\dd |\str(T_k[1\dd j])|]$ matching
$\str(T_k[j])$; these fragments are called \emph{phrases} induced by
$T_k$.  We also define a set $B_k(T)$ of \emph{phrase boundaries}
induced by $T_k$:
\[
  B_k(T) = \{|\str(T_k[1\dd j])| : j\in [0\dd |T_k|]\}.
\]

The following result lets us use \cref{cor:mehlhorn} when analyzing
the restricted signature parsing of $T_{k}$ for odd $k\in \Zp$.
\begin{fact}\label{fct:distinct}
  For every $T\in \Sigma^+$ and odd $k\in \Zp$, there is no $j\in
  [1\dd |T_k|)$ with $T_k[j]=T_k[j+1]\in \Act_{k}$.
\end{fact}
\begin{proof}
  For a proof by contradiction, suppose that $T_k[j]=T_k[j+1]\in
  \Act_{k}$ holds for some $j\in [1\dd |T_k|)$.  By definition of
  $T_k$, we have $T_k[j]=T_k[j+1]=\block(X^\ell)$ for some $X\in
  \Symb$ and $\ell\in \Zp$. In particular, $|\str(X)|\le
  |\str(X^\ell)| \le d_{k} = d_{k-1}$, so $X\in \Act_{k-1}$.  Let
  $T_{k-1}(i-\ell\dd i]$ and $T_{k-1}(i\dd i+\ell]$ be blocks of
  $T_{k-1}$ collapsed to $T_k[j]$ and $T_k[j+1]$, respectively.  Due
  to $T_{k-1}[i]=T_{k-1}[i+1]=X\in \Act_{k-1}$, no block ends at
  position $i$ in the restricted run parsing of $T_{k-1}$.  This
  contradicts the existence of the block $T_{k-1}(i-\ell\dd i]$.
\end{proof}
  
Next, we use \cref{cor:mehlhorn} to derive upper and lower bounds on
phrase lengths.
\begin{fact}\label{fct:recompr}
  For every $T\in \Sigma^+$, $k\in \Zz$, and $j\in [1\dd |T_k|]$, the
  string $\str(T_k[j])$ has length at most $3d_k$ or primitive root of
  length at most $d_k$.
\end{fact}
\begin{proof}
  We proceed by induction on $k$.  If $k=0$, then
  $|\str(T_k[j])|=1<3$.  Thus, we may assume $k>0$.  Let $T_{k-1}[i\dd
  i+\ell)$ be the fragment of $T_{k-1}$ obtained by expanding
  $T_k[j]$.  If $\ell=1$, then the inductive assumption shows that
  $\str(T_k[j])$ is of length at most $3d_{k-1} \le 3d_k$ or its
  primitive root is of length at most $d_{k-1}\le d_k$.  If $\ell\ge
  2$ and $k$ is odd, then $T_{k-1}[i\dd i+\ell)=X^\ell$ for some $X\in
  \Act_{k-1}$, and thus the primitive root of $\str(T_k[j])$ is of
  length at most $\length(X)\le d_{k-1}=d_k$.  If $\ell\ge 2$ and $k$
  is even, then, by \cref{cor:mehlhorn}, $T_{k-1}[i\dd i+\ell)\in
  \Act_{k-1}^{\ell}$ and $\ell\le 6$. Consequently, $|\str(T_k[j])|\le
  \ell\cdot d_{k-1} \le 6d_{k-1} = 3d_k$.
\end{proof}

\begin{lemma}\label{lem:sparse}
  For every $T\in \Sigma^+$, $k\in \Zz$, and $j\in [1\dd |T_k|)$, we
  have $|\str(T_k[j\dd j+1])|>\frac12 d_k$.
\end{lemma}
\begin{proof}
  We proceed by induction on $k$. For $k=0$, the claim holds
  trivially: $|\str(T_0[j\dd j+1])|=2 > \frac12$.  Otherwise, let
  $T_{k-1}[i\dd i+\ell)$ be the fragment of $T_{k-1}$ obtained by
  expanding both symbols of $T_k[j\dd j+1]$.  If $k$ is odd, then
  $|\str(T_{k-1}[i\dd i+\ell))|\ge |\str(T_{k-1}[i\dd i+1])| > \frac12
  d_{k-1} = \frac12 d_k$.  If $k$ is even and $T_{k-1}[i]\in \Symb\sm
  \Act_{k-1}$ or $T_{k-1}[i+\ell-1]\in \Symb\sm \Act_{k-1}$, then
  $|\str(T_k[j\dd j+1])|> d_{k-1} = \frac12 d_k$ by definition of
  $\Act_{k-1}$.  Otherwise, by \cref{cor:mehlhorn},
  $\degree(T_{k}[j])\ge 2$ and $\degree(T_{k}[j+1])\ge 2$, and thus
  $|\str(T_k[j\dd j+1])|\ge |\str(T_{k-1}[i\dd i+3])|> 2\cdot \frac12
  d_{k-1} = \frac12 d_k$.
\end{proof}

Let us define sequences $(\alpha_k)_{k=0}^\infty$ and
$(\beta_k)_{k=0}^\infty$ with
\[
  \alpha_k =
  \begin{cases} 0 & \text{if }k=0,\\ \alpha_{k-1}+d_{k-1} &
    \text{if }k\in \Zp\text{ is odd},\\ \alpha_{k-1} + (h_{k-1} +
    7)d_{k-1} & \text{if }k\in \Zp\text{ is even};
  \end{cases}
  \quad
  \beta_k=
  \begin{cases}
    0 & \text{if }k=0,\\ \beta_{k-1} + d_{k-1} & \text{if }k\in
    \Zp\text{ is odd},\\ \beta_{k-1} + 4d_{k-1} & \text{if }k\in
    \Zp\text{ is even}.
  \end{cases}
\]

\begin{fact}\label{fct:alphabeta}
  For every $k\in \Zp$, we have $d_k \le \alpha_k <
  (2h_{k-1}+16)d_{k-1}$ and $\beta_k < 10d_{k-1}$.
\end{fact}
\begin{proof}
  The lower bound on $\alpha_k$ is shown by simple induction.  The
  base case holds for $k=1$ due to $\alpha_1=d_0=1=d_1$.  For $k\ge
  2$, the inductive assumption yields $\alpha_{k}\ge
  \alpha_{k-1}+d_{k-1} \ge 2d_{k-1} = 2^{\floor{(k-1)/2}+1} \ge
  2^{\floor{k/2}} = d_{k}$.

  As for the upper bound on $\alpha_k$, we use induction with a
  stronger upper bound of $\alpha_k < (h_{k-1}+9)d_{k-1}$ for odd
  values $k$.  The base case of $k=1$ holds due to $\alpha_1=1 < h_0+9
  = (h_0+9)d_0$.  If $k$ is odd, then $\alpha_k = \alpha_{k-1}+d_{k-1}
  < (2h_{k-2}+16)d_{k-2}+d_{k-1} = (h_{k-2}+9)d_{k-1} \leq
  (h_{k-1}+9)d_{k-1}$ holds by the inductive assumption.  If $k$ is
  even, then $\alpha_k = \alpha_{k-1}+(h_{k-1}+7)d_{k-1} <
  (h_{k-2}+9)d_{k-2}+(h_{k-1}+7)d_{k-1} \le (2h_{k-1}+16)d_{k-1}$
  holds by the inductive assumption.

  As for the upper bound on $\beta_k$, we use induction with a
  stronger upper bound of $\beta_k < 6d_{k-1}$ for odd values $k$.
  The base case of $k=1$ holds due to $\beta_1=1 < 6 = 6d_0$.  If $k$
  is odd, then $\beta_k = \beta_{k-1}+d_{k-1} < 10d_{k-2}+d_{k-1} =
  6d_{k-1}$ holds by the inductive assumption.  If $k$ is even, then
  $\beta_k = \beta_{k-1}+4d_{k-1} < 6d_{k-2}+4d_{k-1} = 10d_{k-1}$
  holds by the inductive assumption.
\end{proof}

\begin{lemma}\label{lem:recompr1}
  Consider strings $T,T'\in \Sigma^+$, an integer $k\in \Zz$, and
  positions $i\in [0\dd |T|]$ and $i'\in [0\dd |T'|]$.  Suppose that
  the following two conditions are satisfied for every string
  $U\in\Sigma^*$:
  \begin{itemize}
  \item If $cn\le \beta_k$, then $U$ is a prefix of $T(i\dd |T|]$ if
    and only if $U$ is a prefix of $T'(i'\dd |T'|]$.
  \item If $cn\le \alpha_k$, then $U$ is a suffix of $T[1\dd i]$ if
    and only if $U$ is a suffix of $T'[1\dd i']$.
  \end{itemize}
  Then, $i\in B_k(T)$ if and only if $i'\in B_{k}(T')$.
\end{lemma}
\begin{proof}
  We proceed by induction on $k$.  The base case of $k=0$ is trivially
  satisfied due to $B_0(T)=[0\dd |T|]$, $B_0(T')=[0\dd |T'|]$, and
  $\alpha_0=\beta_0=0$.

  For the inductive step, consider $k\in\Zp$.  We first claim that,
  for each $k'\in [0\dd k)$ and $\delta \in [\alpha_{k'}-\alpha_k\dd
  \beta_k-\beta_{k'}]$, we have $i+\delta\in B_{k'}(T)$ if and only if
  $i'+\delta\in B_{k'}(T')$.  By symmetry, we may assume for a proof
  by contradiction that $i+\delta\in B_{k'}(T)$ (and, in particular,
  $i+\delta \in [0\dd |T|]$) yet $i'+\delta\notin B_{k'}(T')$. We
  consider two cases.
  \begin{description}
  \item[\boldmath{$\delta \in [0 \dd \beta_k-\beta_{k'}]$}.]  Due to
    $T(i\dd i+\delta]\in \Sigma^{\le \beta_k}$, this string must occur
    as a prefix of $T'(i'\dd |T'|]$.  Thus, $i'+\delta \le |T'|$ and
    $T(i\dd i+\delta]=T'(i'\dd i'+\delta]$.  If there is a string
    $U\in \Sigma^{\le \beta_{k'}}$ that is a prefix of exactly one of
    the suffixes $T(i+\delta\dd |T|]$ and $T'(i'+\delta\dd |T'|]$,
    then $T(i\dd i+\delta]\cdot U\in \Sigma^{\le \beta_k}$ is a prefix
    of exactly one of the suffixes $T(i\dd |T|]$ and $T'(i'\dd
    |T'|]$. Otherwise, the inductive assumption yields a string $U\in
    \Sigma^{\le \alpha_{k'}}$ that is a suffix of exactly one of the
    prefixes $T[1\dd i+\delta]$ and $T'[1\dd i'+\delta]$. In this
    case, $U(cn-\delta\dd cn]=T(i\dd i+\delta]=T'(i'\dd i'+\delta]$
    and $U[1\dd cn-\delta]\in \Sigma^{\le \alpha_{k}}$ is a suffix of
    exactly one of the prefixes $T[1\dd i]$ and $T'[1\dd i']$.
  \item[\boldmath{$\delta \in [\alpha_{k'}-\alpha_k \dd 0]$}.]  Due to
    $T(i+\delta\dd i]\in \Sigma^{\le \alpha_k}$, this string must
    occur as a suffix of $T'[1\dd i']$.  Thus, $i'+\delta \ge 0$ and
    $T(i+\delta\dd i]=T'(i'+\delta\dd i']$.  If there is a string
    $U\in \Sigma^{\le \alpha_{k'}}$ that is a suffix of exactly one of
    the prefixes $T[1\dd i+\delta]$ and $T'[1\dd i'+\delta]$, then
    $U\cdot T(i+\delta\dd i]\in \Sigma^{\le \alpha_k}$ is a suffix of
    exactly one of the prefixes $T[1\dd i]$ and $T'[1\dd
    i']$. Otherwise, the inductive assumption yields a string $U\in
    \Sigma^{\le \beta_{k'}}$ that is a prefix of exactly one of the
    suffixes $T(i+\delta\dd |T|]$ and $T'(i'+\delta\dd |T'|]$. In this
    case, $U[1\dd -\delta]=T(i+\delta \dd i]=T'(i'+\delta\dd i']$ and
    $U(-\delta \dd cn]\in \Sigma^{\le \beta_{k}}$ is a prefix of
    exactly one of the suffixes $T(i\dd |T|]$ and $T'(i'\dd |T'|]$.
  \end{description}

  If $i\notin B_{k-1}(T)$ and $i'\notin B_{k-1}(T')$, then $i\notin
  B_k(T)$ and $i'\notin B_k(T')$, so the lemma holds trivially.
  Otherwise, the claim, instantiated with $k'=k-1$ and $\delta=0$,
  implies that both $i\in B_{k-1}(T)$ and $i'\in B_{k-1}(T')$.  Let us
  set $j,j'$ so that $i = |\str(T_{k-1}[1\dd j])|$ and $i' =
  |\str(T'_{k-1}[1\dd j'])|$.  By the assumption on $i,i'$, exactly
  one of the positions $j,j'$ is an endpoint of a block of the parsing
  of $T_{k-1}$ and $T'_{k-1}$.

  If $k$ is odd, this means that either $T_{k-1}[j]=T_{k-1}[j+1]\in
  \Act_{k-1}$ or $T'_{k-1}[j']=T'_{k-1}[j'+1]\in \Act_{k-1}$ (but not
  both). Without loss of generality, suppose that
  $T_{k-1}[j]=T_{k-1}[j+1]=X$ for some $X\in \Act_{k-1}$.  By the
  claim, for each $k'\in [0\dd k)$, the fragments $T(i-|\str(X)|\dd
  i+|\str(X)|]$ and $T'(i'-|\str(X)|\dd i'+|\str(X)|]$ are parsed into
  level-$k'$ phrases in the same way.  In particular, this implies
  $T'_{k-1}[j']=T_{k-1}[j]=T_{k-1}[j+1]=T'_{k-1}[j'+1]\in \Act_{k-1}$,
  a contradiction.

  If $k$ is even, the aforementioned condition yields by
  \cref{cor:consistent} a string $S\in \Act_{k-1}^{\le 4}$ that is a
  prefix of exactly one of the suffixes $T_{k-1}(j\dd |T_{k-1}|]$ and
  $T'_{k-1}(j'\dd |T'_{k-1}|]$, or a string $S\in \Act_{k-1}^{\le
  h_{k-1}+7}$ that is a suffix of exactly one of the prefixes
  $T_{k-1}[1\dd j]$ and $T'_{k-1}[1\dd j']$.  Without loss of
  generality, suppose that $S$ is a prefix of $T_{k-1}(j\dd
  |T_{k-1}|]$ or a suffix of $T_{k-1}[1\dd j]$. If $S$ is a prefix of
  $T_{k-1}(j\dd |T_{k-1}|]$ then, due to $|\str(S)|\le
  \beta_{k}-\beta_{k-1}$, the claim implies that, for each $k'\in
  [0\dd k)$, the fragments $T(i\dd i+|\str(S)|]$ and $T'(i'\dd
  i'+|\str(S)|]$ are parsed into level-$k'$ phrases in the same way.
  In particular, $S$ is also a prefix of $T'_{k-1}(j'\dd |T'_{k-1}|]$,
  a contradiction.  Similarly, if $S$ is a suffix of $T_{k-1}[1\dd j]$
  then, due to $|\str(S)|\le \alpha_{k}-\alpha_{k-1}$, the claim
  implies that, for each $k'\in [0\dd k)$, the fragments
  $T(i-|\str(S)|\dd i]$ and $T'(i'-|\str(S)|\dd i']$ are parsed into
  level-$k'$ phrases in the same way.  In particular, $S$ is also a
  suffix of $T'_{k-1}[1\dd j']$, a contradiction.
\end{proof}

\begin{corollary}\label{cor:lcps}
  Consider strings $T,T'\in \Sigma^+$, an integer $k\in \Zz$, and a
  string $S\in \Symb^*$.
  \begin{itemize}
  \item If $S$ is a prefix of exactly one of the strings $T_k,T'_k$,
    then $\lcp(T,T')<|\str(S)|+\beta_k$.
  \item If $S$ is a suffix of exactly one of the strings $T_k,T'_k$,
    then $\lcs(T,T')<|\str(S)|+\alpha_k$.
  \end{itemize}
\end{corollary}
\begin{proof}
  We proceed by induction on $k$. The case of $k=0$ is trivial.  For a
  proof by contradiction, suppose that $\lcp(T,T') \ge
  |\str(S)|+\beta_k$ and $S$ is prefix of $T_k$ yet $S$ is not a
  prefix of $T'_k$ (the remaining cases are symmetric).  Let $S'$ be
  the prefix of $T_{k-1}$ obtained by expanding symbols in $S$.  By
  the inductive assumption, $S'$ is also a prefix of $T'_{k-1}$.  At
  the same time \cref{lem:recompr1} yields $B_{k}(T)\cap [0\dd
  \lcp(T,T')-\beta_k]=B_{k}(T')\cap [0\dd \lcp(T,T')-\beta_k]$. Hence,
  the block boundaries within $S'$ are placed in the same way in the
  parsing of $T_{k-1}$ and $T'_{k-1}$. Consequently, $S$ is also a
  prefix of $T'_k$, a contradiction.
\end{proof}

\section{Dynamic Strings}\label{sec:ds}

By \cref{lem:sparse}, for every signature function $\sig$ and every
string $T\in \Sigma^+$ $|T_k|=1$ holds for sufficiently large $k\in
\Zz$ (whenever $k \ge 2\ceil{\log (2|T|)}$).  Hence, we define
$\symb_{\sig}(T)$ as the unique symbol in the string $T_k$ for the
smallest $k\in \Zz$ with $|T_k|=1$.

We maintain a growing set of strings $\W\sub \Sigma^+$ and represent
each string $T\in \W$ as $\symb_{\sig}(T)$ for implicit signature
function $\sig$.  Our data structure consists of a grammar $\Gr\sub
\Symb$ (maintained using \cref{lem:gr}) and the values $\sig(X)$
stored for all $X\in \Gr$.  The key invariant is that
$\symb_{\sig}(T)\in \Gr$ holds for each $T\in \W$; based on this, we
represent each string $T\in \W$ using the identifier
$\id(\symb_{\sig}(T))$ of the underlying symbol.  The following simple
observation allows leaving $\sig(X)$ unspecified for $X\in \Symb\sm
\Gr$.
\begin{observation}
  Consider a grammar $\Gr\sub \Symb$ and two signature functions
  $\sig,\sig'$ such that $\sig|_\Gr=\sig'|_\Gr$.  If
  $\symb_{\sig}(T)\in \Gr$ for some string $T\in \Sigma^+$, then
  $\symb_{\sig'}(T)=\symb_{\sig}(T)$.
\end{observation}

Efficiency of our data structure is supported by an additional
invariant that, for each $k\in \Zz$, the signature function $\sig$
injectively maps $\Act_k\cap \Gr$ to $[0\dd |\Act_k\cap \Gr|)$.
Whenever a symbol $X\in \Act_k\sm \Gr$ is added to $\Gr$, it is
assigned the smallest ``free'' signature $|\Act_k\cap \Gr|$.  This
way, we make sure that the explicitly assigned signature values have
$\Oh(\log |\Gr|)$ bits.

In the following, we describe various operations supported by this
data structure.  Internally, we measure the efficiency in terms of the
grammar size $g:=|\Gr|$ (after executing the operation), the alphabet
size $\sigma := |\Sigma|$, and the lengths of the strings involved in
the operations.  We also assume that the machine word size $w$
satisfies $w=\Omega(\log (\sigma g))$.

\subsection{Access}

The $\accessop(T,i)$ operation retrieves $T[i]$ for a given string
$T\in \W$ and position $i\in [1\dd |T|]$.  For this, we traverse the
parse tree $\Tr(\symb_\sig(T))$ maintaining a pointer to a node $\nu$
such that $\str(\nu)=T(\ell\dd r]$ for $i\in (\ell\dd r]$.  If $\nu$
is a leaf, we return $\letter^{-1}(\symb(\nu))$.  If $\nu$ has $d\le
6$ children, we compute $\preflength(\symb(\nu),j)$ for $j\in [0\dd
d]$ and descend to the child $\nu_j$ such that
$\preflength(\symb(\nu),j-1)<i-\ell\le \preflength(\symb(\nu),j)$.
Otherwise, we have $\preflength(\symb(\nu),j)=j\cdot
\preflength(\symb(\nu),1)$ for $j\in [0\dd d]$, so we descend to the
$\ceil{\frac{i-\ell}{\preflength(\symb(\nu),1)}}$-th child.  Overall,
the running time is $\Oh(\level(\symb_\sig(T)))=\Oh(\log |T|)$.

\subsection{Longest Common Prefix}

The operation $\lcpop(T,U)$ computes the (length of) the longest
common prefix of any two strings $T,U\in \W$.  Let
$(T_k)_{k=0}^\infty$ and $(U_k)_{k=0}^\infty$ be the balanced
signature parsing of $T$ and $S$, respectively.  Moreover, let
$X=\symb_{\sig}(T)$ and $Y=\symb_{\sig}(U)$.  Our algorithm traverses
the parse trees $\Tr(X)$ and $\Tr(Y)$ maintaining two pointers $\nu_T$
and $\nu_U$, as illustrated in \cref{alg:lcp}.

\begin{algorithm}
  $\nu_T := \root(\symb_{\sig}(T))$\;
  $\nu_U := \root(\symb_{\sig}(U))$\;
  $\lcp := 0$\;
  \lWhile{$\level(\nu_T)>\level(\nu_U)$}{$\nu_T := \child(\nu_T,1)$}
  \lWhile{$\level(\nu_U)>\level(\nu_T)$}{$\nu_U := \child(\nu_U,1)$}\label{ln:pre}
  \While{\KwSty{true}}{\label{ln:whiletrue}
    \While{$\nu_T\ne \bot$ \KwSty{and} $\nu_U \ne \bot$ \KwSty{and} $\symb(\nu_T) = \symb(\nu_U)$}{\label{ln:fwd}
      $\lcp := \lcp + \length(\symb(\nu_T))$\;
      $\nu_T := \rightop(\nu_T)$\;
      $\nu_U := \rightop(\nu_U)$\;
    }
    \lIf{$\nu_T = \bot$ \KwSty{or} $\nu_U = \bot$ \KwSty{or} $\level(\nu_T)=0$}{\Return{$\lcp$}}
    $\nu_T := \child(\nu_T,1)$\;\label{ln:down}
    $\nu_U := \child(\nu_U,1)$\;
  }
  \caption{$\mathsf{lcp}(T,U)$}\label{alg:lcp}
\end{algorithm}
In this algorithm, we use a short-hand
$\level(\nu):=\level(\symb(\nu))$ and an extra $\rightop(\nu)$
operation that, given a node $\nu$ corresponding to $T_k[i]$, returns
the node corresponding to $T_k[i+1]$ or $\bot$ if $i=|T_k|$ (and
analogously for $U_k$).  We maintain an invariant at
\cref{ln:whiletrue} that $\level(\nu_T)=\level(\nu_U)$ and, denoting
this common value by $k$, the pointers $\nu_T$ and $\nu_U$ correspond
to symbols $T_k[i]$ and $U_k[i]$ such that $T_k[1\dd i)=U_k[1\dd i)$
and $\lcp = |\str(T_k[1\dd i))|$.  It is easy to see that this
invariant is indeed satisfied after executing \cref{ln:pre} with
$i=\lcp=0$ and $k=\min(\level(\symb_\sig(T)),\level(\symb_\sig(U)))$.
The while loop of \cref{ln:fwd} maintains the invariant and increments
$i$ as far as possible.  If $i$ reaches $|T_k|+1$ or $|U_k|+1$, then
we conclude that $T_k$ is a prefix of $U_k$ (or vice versa), and thus
the longest common prefix of $T$ and $S$ is of length $\lcp =
\min(|\str(T_k)|,|\str(U_k)|)=\min(|T|,cn)$; it is then reported
correctly.  Similarly, if $k=0$ and $T_k[i]\ne U_k[i]$, then the
longest common prefix of $T$ and $S$ is of length $i-1=\lcp$.  In the
remaining case, by moving $\nu_T$ and $\nu_U$ to the leftmost
children, we maintain the invariant and decrement the level $k$.

An efficient implementation of \cref{alg:lcp} requires one
optimization: if, at \cref{ln:down}, we have $\symb(\nu_T) =
\block(S_X)$ and $\symb(\nu_U)=\block(S_Y)$, then we immediately
simulate the first $\ell:=\lcp(S_X,S_Y)$ subsequent steps of
\cref{ln:whiletrue}. (Since $S_X$ and $S_Y$ are symbol powers or
strings of constant length, the value $\ell$ can be determined in
$\Oh(1)$ time.)  To carry out such simulation, we set $\lcp:=\lcp +
\preflength(\nu_T,\ell)$, $\nu_T := \child(\nu_T,\ell+1)$, and $\nu_U
:= \child(\nu_U,\ell+1)$, where
$\child(\nu,\degree(\nu)+1)=\child(\rightop(\nu),1)$.

As our traversal of $\Tr(X)$ and $\Tr(Y)$ always proceeds forward in
the pre-order, the time complexity is proportional to the number of
visited nodes (including the ancestors visited while traversing the
paths from $\nu$ to $\rightop(\nu)$).  \cref{cor:lcps} guarantees
that, for each $k$, the visited nodes at level $k$ correspond to
level-$k$ phrases overlapping $T(\max(0,\lcp(T,U)-\beta_{k+1})\dd
\lcp(T,U)]$. By \cref{lem:sparse,fct:alphabeta}, the number of such
phrases is $\Oh(\frac{\beta_{k+1}}{d_k})=\Oh(1)$.  Consequently, the
overall running time is $\Oh(\level(X)+\level(Y))=\Oh(\log|TU|)$.

\subsection{Longest Common Suffix}

A symmetric operation $\lcsop(T,U)$ computes the (length of) the
longest common suffix of any two strings $T,U\in \W$.  Its
implementation is analogous to that of $\lcpop(T,U)$.  However, in the
analysis, the number of level-$k$ nodes visited is
$\Oh(\frac{\alpha_{k+1}}{d_k})=\Oh(h_{k+1})=\Oh(\log^*(k\sigma))$,
which yields the total running time of
$\Oh(\log|TU|\log^*(|TU|\sigma))$.

\subsection{Insertion}

The $\makeop(T)$ operation inserts to $\W$ a given string $T\in
\Sigma^+$.  For this, we construct the subsequent levels of the
balanced signature parsing $(T_k)_{k=0}^\infty$ until reaching
$|T_k|=1$ so that $\symb_{\sig}(T)=T_k[1]$.  The construction of $T_k$
for $k=0$ costs $\Oh(|T|\frac{\log^2 \log g}{\log \log \log g})$ time.
The cost for odd values $k\in \Zp$ is $\Oh(|T_{k-1}|\frac{\log^2 \log
g}{\log \log \log g})$, whereas the cost for even values $k\in \Zp$ is
$\Oh(|T_{k-1}|(h_{k-1}+\frac{\log^2 \log g}{\log \log \log g})) =
\Oh(|T_{k-1}|(\log^*(\sigma k)+\frac{\log^2 \log g}{\log \log \log
g}))$ by \cref{thm:mehlhorn,rem:nk}.  Since \cref{lem:sparse} yields
$|T_k|=\Oh(1+\frac{|T|}{d_k})=\Oh(1+\frac{|T|}{2^{k/2}})$, the overall
cost is $\Oh(|T|(\log^*\sigma + \frac{\log^2 \log g}{\log \log \log
g}))$.

\subsection{Concatenation}

The $\concatop(L,R)$ operation, given two strings $L,R\in \W$, inserts
their concatenation $T:=L\cdot R$ to $\W$.  Consider the balanced
signature parsing $(L_k)_{k=0}^\infty$ of $L$, $(R_k)_{k=0}^\infty$ of
$R$, and $(T_k)_{k=0}^\infty$ of $T$.  For each $k\in \Zz$, let $P_k$
denote the longest prefix of $T_k$ with $|\str(P_k)|\le
\max(0,|L|-\beta_k)$, and let $S_k$ denote the longest suffix of $T_k$
with $|\str(S_k)|\le \max(0,|R|-\alpha_k)$.  By \cref{cor:lcps}, $P_k$
is also the longest prefix of $L_k$ with $|\str(P_k)|\le
\max(0,|L|-\beta_k)$, and $S_k$ is also the longest suffix of $R_k$
with $|\str(S_k)|\le \max(0,|R|-\alpha_k)$.  Moreover, define
$L'_k,R'_k,C_k$ so that $L_k=P_k L'_k$, $R_k = R'_kS_k$, and $T_k =
P_k C_k S_k$.  Furthermore, let $P'_k$ be the suffix of $P_k$ such
that $\str(P_k)=\str(P_{k+1})\str(P'_k)$, and let $S'_k$ be the prefix
of $S_k$ such that $\str(S_k)=\str(S'_k)\str(S_{k+1})$.

Our implementation of $\concatop(L,R)$ constructs $C_k$ for subsequent
integers $k\in \Zz$.  For $k=0$, the string $C_0$ is empty.  For odd
$k\in \Zp$, we build $P'_{k-1}C_{k-1}S'_{k-1}$ and apply the
restricted run-length parsing with respect to $\Act_{k-1}$.  Finally,
we retrieve $C_k$ by collapsing each block into the corresponding
symbol, inserted to $\Gr$ using $\iP$.  For even $k\in \Zp$, we
additionally retrieve the longest string $P''_{k-1}\in \Act_{k-1}^{\le
h_{k-1}+7}$ such that $P''_{k-1}P'_{k-1}$ is a suffix of $P_{k-1}$,
and the longest string $S''_{k-1}\in \Act_{k-1}^{\le 4}$ such that
$S'_{k-1}S''_{k-1}$ is a prefix of $S_{k-1}$.  Then, we perform the
balanced signature parsing of
$P''_{k-1}P'_{k-1}C_{k-1}S'_{k-1}S''_{k-1}$ with respect to
$\Act_{k-1}$, $h_{k-1}$, and $\sig_{k-1}$.  Finally, we cut
$P'_{k-1}C_{k-1}S'_{k-1}$ (which is guaranteed to consist of full
blocks by \cref{cor:consistent}), and we obtain $C_{k}$ by collapsing
each block into the corresponding symbol, inserted to $\Gr$ using
$\iT$.  We stop as soon as $|C_k|=1$ and $|P_k|=|S_k|=0$; the only
symbol of $C_k$ is then $Z:=\symb_{\sig}(T)$.

The efficiency of this procedure follows from the fact that, by
\cref{lem:sparse,fct:alphabeta}, $|C_k|=\Oh(h_k)$, $|L'_k| = \Oh(1)$,
and $|R'_k|=\Oh(h_k)$.  For odd $k\in \Zp$, we also have
$|P''_{k-1}P'_{k-1}|=\Oh(h_k)$ and $|S'_{k-1}S''_{k-1}|=\Oh(h_k)$,
whereas for even $k\in \Zp$, the strings $P'_{k-1}$ and $S'_{k-1}$
have run-length parsing of size $\Oh(1)$ and $\Oh(h_k)$, respectively.
In particular, all the required strings $P'_{k-1}$ and $P''_{k-1}$ can
be generated in $\Oh(\log |T| \log^*(|T|\sigma))$ time by traversing
the parse tree $\Tr(X)$, whereas all the required strings $S'_{k-1}$
and $S''_{k-1}$ can be generated in $\Oh(\log |T| \log^*(|T|\sigma))$
time by traversing the parse tree $\Tr(Y)$.  The cost of run-length
parsing at even levels $k\in \Zp$ is $\Oh(h_{k})=\Oh(\log^*
(|T|\sigma))$, whereas the cost of signature parsing at odd levels
$k\in \Zp$ is $\Oh(h_{k-1} h_k)=\Oh((\log^* |T|\sigma)^2)$, for a
total of $\Oh(\log |T| (\log^*(|T|\sigma))^2)$ across all levels.
Collapsing each block into the corresponding symbol costs
$\Oh(\frac{\log^2 \log g}{\log \log \log g})$ time, for a total of
$\Oh(\log|T|\log^*(|T|\sigma)(\log^*(|T|\sigma)+\frac{\log^2\log
g}{\log \log \log g}))$ time.

\subsection{Split}

The $\splitop(T,i)$ operation, given $T\in \W$ and $i\in [1\dd |T|)$,
inserts $L:= T[1\dd i]$ and $R:=T(i\dd |T|]$ to $\W$.  We utilize
the same notation as in the implementation of the operation
$\concatop(L,R)$.

To derive $\symb(L,\sig)$, we build $L'_k$ for subsequent integers
$k\in \Zz$.  For $k=0$, the string $L'_0$ is empty.  For odd $k\in
\Zp$, we build $P'_{k-1}L'_{k-1}$ and apply the restricted run-length
parsing with respect to $\Act_{k-1}$.  Finally, we retrieve $L'_k$ by
collapsing each block into the corresponding symbol, inserted to $\Gr$
using $\iP$.  For even $k\in \Zp$, we build
$P''_{k-1}P'_{k-1}L'_{k-1}$ and apply the balanced signature parsing
with respect to $\Act_{k-1}$, $h_{k-1}$, and $\sig_{k-1}$.  Finally,
we cut $P'_{k-1}L'_{k-1}$ (which is guaranteed to consist of full
blocks by \cref{cor:consistent}), and we obtain $L'_{k}$ by collapsing
each block into the corresponding symbol, inserted to $\Gr$ using
$\iT$.  We stop as soon as $|L'_k|=1$ and $|P_k|=0$; the only symbol
of $L'_k$ is then $\symb(L,\sig)$.

Symmetrically, to derive $\symb(R,\sig)$, we build $R'_k$ for
subsequent integers $k\in \Zz$.  For $k=0$, the string $R'_0$ is
empty.  For odd $k\in \Zp$, we build $R'_{k-1}S'_{k-1}$ and apply the
restricted run-length parsing with respect to $\Act_{k-1}$.  Finally,
we retrieve $R'_k$ by collapsing each block into the corresponding
symbol, inserted to $\Gr$ using $\iP$.  For even $k\in \Zp$, we build
$R'_{k-1}S'_{k-1}S''_{k-1}$ and apply the balanced signature parsing
with respect to $\Act_{k-1}$, $h_{k-1}$, and $\sig_{k-1}$.  Finally,
we cut $R'_{k-1}S'_{k-1}$ (which is guaranteed to consist of full
blocks by \cref{cor:consistent}), and we obtain $R'_{k}$ by collapsing
each block into the corresponding symbol, inserted to $\Gr$ using
$\iT$.  We stop as soon as $|R'_k|=1$ and $|S_k|=0$; the only symbol
of $R'_k$ is then $\symb(R,\sig)$.  The complexity analysis is similar
to that for $\concatop(L,R)$.

\subsection{Lexicographic Order of Cyclic Fragments}\label{sec:lex}

\newcommand{\CF}{\mathsf{CF}}
\newcommand{\CFp}{\mathsf{CF}^+}
\newcommand{\compareop}{\mathsf{compare}}

For a string $T\in \Sigma^+$, let $\CF(T)=\{T^\infty(i\dd j] : i,j\in
\Z\text{ such that }i\le j\}$; note that each fragment in $\CF(T)$ can
be represented using its endpoints $i,j$.  We also define
$\CFp(T)=\CF(T)\cup\{F\cdot c^\infty : F\in \CF(T)\}$, where $c=\max
\Sigma$.  Finally, $\CF(\W)=\bigcup_{T\in \W}\CF(T)$ and
$\CFp(\W)=\bigcup_{T\in \W}\CFp(T)$.

Our next operation $\compareop(F_1,F_2)$, given $F_1,F_2\in \CFp(\W)$,
decides whether $F_1\prec F_2$.  For each $i\in \{1,2\}$, let
$G_i=T_i^\infty(\ell_i\dd r_i]\in \CF(T_i)$ be such that $F_i\in
\{G_i,G_i\cdot c^\infty\}$ and $T_i\in \W$.  Our initial goal is to
determine $\lcp(G_1,G_2)$.  As the first step, we use the $\splitop$
and $\concatop$ operations in order to insert $T_i^\infty(\ell_i\dd
\ell_i+|T_i|]$. Effectively, this lets us assume that $\ell_i = 0$.
Moreover, by symmetry, we assume that $|T_1| \le |T_2|$.  Next, we
compute $\min(r_1,r_2,\lcpop(T_1,T_2))$. If this value is less than
$|T_1|$, then it is equal to the sought longest common prefix
$\lcp(T_1^\infty(0\dd r_1],T_2^\infty(0\dd r_2])$.  Otherwise, we use
the $\splitop$ and $\concatop$ operations to insert
$T_2^\infty(|T_1|\dd |T_1T_2|]$ to $\W$, and we compute
$\min(r_1,r_2,|T_1|+\lcpop(T_2, T_2^\infty(|T_1|\dd |T_1T_2|]))$.  If
this value is less than $|T_1T_2|$, then it is equal to the sought
longest common prefix $\lcp(T_1^\infty(0\dd r_1],T_2^\infty(0\dd
r_2])$.  In the remaining case, we simply have $\lcp(T_1^\infty(0\dd
r_1],T_2^\infty(0\dd r_2])=\min(r_1,r_2)$.  If the reported value $t$
is smaller than $\min(r_1,r_2)$, then we use the $\accessop$ operation
to check $G_1[t+1]\prec G_2[t+1]$ and return the reported answer.  If
$G_1$ is a proper prefix of $G_2$, we return YES if and only if $F_1
\in \CF(T_1)$.  If $G_2$ is a proper prefix of $G_1$, we return YES if
and only if $F_2 \notin \CF(T_2)$.  The running time is
$\Oh(\log|T|\log^*(|T|\sigma)(\log^*(|T|\sigma)+\frac{\log^2\log
g}{\log \log \log g}))$.

A symmetric procedure lets us implement $\compareop(F_1,F_2)$ if
$F_1,F_2 \in \CFp(\revstr{\W})$, where $\revstr{W}=\{\revstr{T}: T\in
\W\}$.

\subsection{Internal Pattern Matching}

The $\ipmop(P,T)$ operation, given $P,T\in \W$ with $|P|\le |T|\le
2|P|$, computes $\Occ(P,T):=\{o\in [0\dd |T|-|P|] : P=T(o\dd
o+|P|]\}$.  As shown in~\cite{phdtomek,Kociumaka2015}, $\Occ(P,T)$
forms an arithmetic progression.

Our implementation of $\ipmop(P,T)$ first identifies the maximum level
$k\in \Zz$ with $|P|\ge \alpha_k+\beta_k+3d_k$.\footnote{If there is
no such level, i.e., if $|P|\le 2$, we answer the query naively by
decompressing $P$ and $T$.}  We consider two cases depending on
whether $B_k(P)\cap [\alpha_k\dd |P|-\beta_k]=\emptyset$.

If there exists $i\in B_k(P)\cap [\alpha_k\dd |P|-\beta_k]$ then,
\cref{lem:recompr1} shows that $i+o\in B_k(T)$ holds for any $o\in
\Occ(P,T)$.  Thus, for each position $j\in B_k(T)\cap [i\dd |T|]$, we
compute $o:=j-i$, extract $T(o\dd |T|]$ via $\splitop(T,o)$, and check
whether $P$ is a prefix of $T(o\dd |T|]$ via $\lcpop(P, T(o\dd |T|])$.

Thus, it remains to consider the case when $B_k(P)\cap [\alpha_k\dd
|P|-\beta_k]=\emptyset$, i.e., when $P[\alpha_k\dd \allowbreak
|P|-\beta_k+1]$ is contained in a single level-$k$ phrase.  Suppose
that this phrase is $P(\ell\dd r]$; due to $r-\ell > 3d_k$,
\cref{fct:recompr} shows that $P(\ell\dd r]=Q^c$ for some integer
$c\in \Z_{\ge 2}$ and primitive string $Q$ with $|Q|\le d_k$.
Moreover, if $k'\in \Zz$ is the maximum level such that $P(\ell\dd r]$
consists of multiple level-$k'$ phrases, then all these phrases match
$Q$.  By \cref{lem:recompr1}, for each $o\in \Occ(P,T)$, the fragment
$T(o+\alpha_k\dd o+|P|-\beta_k]$ is also contained in a single
level-$k$ phrase $T(\ell'\dd r']$.  Moreover, another application of
\cref{lem:recompr1} (at level $k'$) shows that the phrase $T(\ell'\dd
r']$ is decomposed into $c'\in \Z_{\ge 2}$ level-$k'$ phrases matching
$Q$.  Our algorithm computes the symbol of $P_k$ corresponding to
$P(\ell\dd r]$, the symbol of $P_{k'}$ corresponding to the primitive
root $Q$, and all the candidate phrases $T(\ell'\dd r']$ with the same
primitive root $Q$.  We select an arbitrary phrase boundary $i\in
B_{k'}(P)\cap [\alpha_k\dd |P|-\beta_k]$ and observe that, by
\cref{lem:recompr1}, for each $o\in \Occ(P,T)$, we have $i+o\in
B_{k'}(T)$.  Given that we assume that $T(o+\alpha_k\dd
o+|P|-\beta_k]$ is contained in $T(\ell'\dd r']$, we further have
$i+o\in (\ell'\dd r')$, i.e., $j:=i+o$ is one of the boundaries in the
decomposition of $T(\ell'\dd r']=Q^{c'}$ into $c'$ individual
occurrences of $Q$.  Our goal is to verify each candidate $j$ by
checking whether $P(i\dd |P|]$ is a prefix of $T(j\dd |T|]$ and
whether $P[1\dd i]$ is a suffix of $T[1\dd j]$.  A naive
implementation performs one $\splitop$, one $\lcpop$, and one $\lcsop$
operation per candidate position $j$.  The positions $j$ form a
so-called \emph{periodic progression} (with period $Q$); thus, as
shown in~\cite[Lemma 7.1.4(c)]{phdtomek}, positions $j$ maximizing
$\lcp(P(i\dd |P|],T(j\dd |T|])$ are contiguous elements of the
periodic progression, and they can be retrieved using a constant
number of $\lcpop$ queries concerning suffixes of $P$ and $T$. Thus,
using $\Oh(1)$ calls to $\splitop$ and $\lcpop$, we can filter
positions $j$ for which $P(i\dd |P|]$ is a prefix of $T(j\dd |T|]$.
As symmetric procedure using $\Oh(1)$ calls to $\splitop$ and $\lcsop$
filters positions $j$ for which $P(i\dd |P|]$ is a prefix of $T(j\dd
|T|]$.  The set of positions $o\in \Occ(P,T)$ corresponding to
$T(\ell'\dd r']$ is obtained by intersecting the two ranges and
shifting it by $i$ positions to the left.  Finally, the occurrences
corresponding to various candidate phrase $T(\ell'\dd r']$ are
retrieved by combining all $o\in \Occ(P,T)$ into a single arithmetic
progression.

As for the running time, we observe that the number of candidate
position $j$ (in the first case) and the number of candidate phrases
$T(\ell'\dd r']$ (in the second case) is bounded by $|T_{k}| =
\Oh(1+\frac{|T|}{d_k}) = \Oh(1+\frac{|P|}{d_k}) = \Oh(h_{k}) =
\Oh(\log^*(|T|\sigma))$ by \cref{fct:alphabeta}. The overall running
time is therefore
$\Oh(\log|T|(\log^*(|T|\sigma))^2(\log^*(|T|\sigma)+\frac{\log^2 \log
g}{\log \log \log g}))$.

\subsection{2-Period Queries}

A 2-period query $\perop(T)$, given $T\in \W$, asks to compute the
shortest period $\per(T)$ or to report that $\per(T)>\frac12|T|$.  As
shown in~\cite{phdtomek,Kociumaka2015}, each of these queries can be
reduced to a constant number of $\ipmop$ and $\lcpop$ queries on
substrings of $T$.  Consequently, $\perop(T)$ queries can also be
answered in
$\Oh(\log|T|(\log^*(|T|\sigma))^2(\log^*(|T|\sigma)+\frac{\log^2 \log
g}{\log \log \log g}))$ time.

\subsection{Canonical Cyclic Shift}\label{sec:ccs}

As shown in~\cite{Kociumaka2015}, cyclic equivalence of $T$ and $T'$
can be tested using a constant number of $\ipmop$ and $\lcpop$ queries
on substrings of $T$; the resulting algorithm also reports a shift $s$
with $T'=\rot^{s}(T)$.  In this paper, however, we need a stronger
operation that, given a string $T\in \W$, computes a \emph{canonical
cyclic shift} $\canshop(T)$ such that
$f_{\sig}(T):=\rot^{\canshop(T)}(T)$ is a necklace-consistent function
(as defined in \cref{def:nc}).

Before implementing an algorithm computing $\canshop(T)$, let us
provide a synthetic definition of the underlying function $f$ in terms
of the signature function $\sig$.

\begin{construction}\label{cons:nc}
  For a string $T\in \Sigma^+$, let $k=\min\{t\in \Zz : d_t \ge
  4|T|\}$ and $U=T^c$, where
  $c=\lceil{\frac{\alpha_k+\beta_k+3d_k}{|T|}}\rceil$.  For a
  signature function $\sig$, we set $f_{\sig}(T)=\rot^{\max(B_k(U)\cap
  [0\dd \alpha_k))}(T)$.%
  \footnote{Recall that the function $B_k$ is defined in
    \cref{sec:bsp} based on the signature function $\sig$.}
\end{construction}

\begin{lemma}\label{lem:nc}
  For every signature function $\sig$, the function
  $f_{\sig}:\Sigma^+\to \Sigma^+$ defined in \cref{cons:nc} is a
  necklace-consistent function.
\end{lemma}
\begin{proof}
  Let us fix $T\in \Sigma^n$ with $n\in \Zp$.  By construction, the
  string $f_{\sig}(T)$ is cyclically equivalent to $T$.  However,
  before proving that $f_{\sig}(T)=f_{\sig}(\rot^s(T))$ holds for
  every $s\in \Z$, let us analyze the properties of the level-$k$
  phrase $U(\ell\dd r]$ in the balanced signature parsing of $U$ such
  that $\ell=\max(B_k(U)\cap[0\dd \alpha_k))$.  By definition, $\ell <
  \alpha_k \le r$; our first claim is that $r> cn-\beta_k$.  Thus, for
  a proof by contradiction, suppose that $r\in B_k(U)\cap [\alpha_k\dd
  cn-\beta_k]$.  By \cref{lem:recompr1}, $B_k(U) \supseteq \{r'\in
  [\alpha_k\dd cn-\beta_k] : r'\equiv r \pmod{n}\}$.  Due to
  $cn-\alpha_k - \beta_k \ge 3d_k > 2n$, this induces two subsequent
  phrases of total length at most $2n\le \frac12d_k$, contradicting
  \cref{lem:sparse}. The contradiction completes the proof that $r >
  cn-\beta_k$.  Consequently, $r-\ell > cn-\alpha_k-\beta_k \ge 3d_k$;
  by \cref{fct:recompr}, this means that $U(\ell\dd r]$ has primitive
  root of length at most $d_k$.  At the same time, $U(\ell\dd r]$ has
  period $n$ (inherited from $U$).  Using the periodicity lemma, we
  conclude that $U(\ell\dd \ell+n]$ is the primitive root of
  $U(\ell\dd r]$.  In particular, $U(\ell\dd r]$ has been created at
  some level $j\in [1\dd k]$ by merging several phrases matching
  $U(\ell\dd \ell+n]$. Thus, $B_{j-1}(U)\cap [\alpha_k \dd cn-\beta_k]
  = \{i\in [\alpha_k \dd cn-\beta_k] : i\equiv \ell \bmod{n}\}$ and
  $B_{j}(U)\cap [\alpha_k \dd cn-\beta_k]=\emptyset$.

  We are now ready to consider the corresponding phrase $U'(\ell'\dd
  r']$ in the balanced signature parsing of
  $U'=\rot^s(U)=(\rot^s(T))^c$. By the above argument,
  $B_{j'-1}(U')\cap [\alpha_k \dd cn-\beta_k] = \{i\in [\alpha_k \dd
  cn-\beta_k] : i\equiv \ell' \pmod{n}\}$ and $B_{j'}(U')\cap
  [\alpha_k \dd cn-\beta_k]=\emptyset$ hold for some $j'\in [1\dd k]$.
  However, \cref{lem:recompr1} also yields $B_{j-1}(U')\cap [\alpha_k
  \dd cn-\beta_k] = \{i\in [\alpha_k \dd cn-\beta_k] : i\equiv \ell-s
  \pmod{n}\}$ and $B_{j}(U')\cap [\alpha_k \dd cn-\beta_k]=\emptyset$.
  Consequently, $j'=j$ and $\ell'\equiv \ell-s \pmod{n}$.  In
  particular,
  $\rot^{\ell'}(\rot^s(T))=\rot^{\ell'+s}(T)=\rot^{\ell}(T)$ holds as
  claimed.
\end{proof}

Our algorithm computing $\canshop(T)$ simply computes $k$ and $c$,
builds $U=T^c$ by repeated calls to $\concatop$, and then builds
$B_k(U)$ by traversing the parse tree of $\symb_\sig(U)$.  Finally, it
returns $\max(B_k(U)\cap [0\dd \alpha_k))$.

As for the complexity analysis, we note \cref{lem:sparse} implies
$|B_k(U)|=\Oh(\frac{1}{d_k}|U|)=\Oh(c)$, whereas \cref{fct:alphabeta}
implies $c=\Oh(h_k)=\Oh(\log^*(|T|\sigma))$, Consequently, the total
running time does not exceed
$\Oh(\log|T|(\log^*(|T|\sigma))^2(\log^*(|T|\sigma)+\frac{\log^2 \log
g}{\log \log \log g}))$.

\section{From Balanced Signature Parsing to Synchronizing Sets}\label{sec:bsp2sss}

\begin{definition}[$\tau$-runs]
  For a string $T\in \Sigma^+$ and an integer $\tau\in [1\dd n]$, we
  define the set $\RUNS_{\tau}(T)$ of \emph{$\tau$-runs} in $T$ that
  consists of all fragments $T[p\dd q]$ of length at least $\tau$ that
  satisfy $\per(T[p\dd q])\le \frac13\tau$ yet cannot be extended (in
  any direction) while preserving the shortest period.
\end{definition}

By the Periodicity Lemma, distinct $\tau$-runs $\gamma,\gamma'$
satisfy $|\gamma \cap \gamma'|\le \frac23\tau$.  Consequently, each
$\tau$-run $\gamma$ contains at least $\frac13\tau$ (trailing)
positions which are disjoint from all $\tau$-runs starting to the left
of $\gamma$.  Hence, $\RUNS_{\tau}(T) \le \frac{3n}{\tau}$.

\begin{construction}\label{cons:sss}
  Consider the signature encoding of a string $T\in \Sigma^n$ with
  respect to a signature function $\sig$ and an integer $\tau\in \Zp$.
  Based on the set $B_k(T)$ for the largest $k\in \Zz$ with $\tau \ge
  3\alpha_k$, we defined the set of $\S_\sig(\tau,T)$ so that it
  consists of all positions $i\in [1\dd n-2\tau+1]$ that satisfy at
  least one of the following conditions:
  \begin{enumerate}[label={\rm (\arabic*)}]
  \item\label{it:sss:reg} $i+\tau-1\in B_k(T)$ and $T[i\dd i+2\tau)$
    is not contained in any $\tau$-run,
  \item\label{it:sss:beg} $i=p-1$ for some $\tau$-run $T[p\dd q]\in
    \RUNS_\tau(T)$,
  \item\label{it:sss:end} $i=q-2\tau+2$ for some $\tau$-run $T[p\dd
    q]\in \RUNS_\tau(T)$.
  \end{enumerate}
\end{construction}

Note that $\S_{\sig}(T,\tau)=\emptyset$ if $\tau > \frac12n$.

\begin{lemma}\label{lem:sss}
  For every $T\in \Sigma^n$ and $\tau\in [1\dd \floor{\frac12n}]$, the
  set $\S_\sig(\tau,T)$ obtained using \cref{cons:sss} is a
  $\tau$-synchronizing set.  Moreover, for every $i\in [1\dd
  n-3\tau+2]$, $|\S_\sig(\tau,T)\cap [i\dd
  i+\tau)|=\Oh(\log^*(\tau\sigma))$.
\end{lemma}
\begin{proof}
  First, suppose that $i,i'\in [1\dd n-2\tau+1]$ satisfy $T[i\dd
  i+2\tau)= T[i'\dd i'+2\tau)$ and $i\in \S_\sig(\tau,T)$. We will
  show that if $i$ satisfies
  conditions~\ref{it:sss:reg}--\ref{it:sss:end}, then $i'$ satisfies
  the same condition. If $i$ satisfies condition~\ref{it:sss:reg},
  then $\per(T[i'\dd i'+2\tau))=\per(T[i\dd i+2\tau))>\frac13\tau$, so
  $T[i'\dd i'+2\tau)$ is not contained in any $\tau$-run. At the same
  time, due to $T(i+\tau-1-\alpha_k\dd
  i+\tau-1+\beta_k]=T(i'+\tau-1-\alpha_k\dd i'+\tau+\beta_k]$, by
  \cref{lem:recompr1}, $i+\tau-1\in B_k(T)$ implies $i'+\tau-1\in
  B_k(T)$.  Consequently, $i'$ also satisfies
  condition~\ref{it:sss:reg}.  If $i$ satisfies
  condition~\ref{it:sss:beg}, then $\per(T[i'+1\dd
  i'+\tau])=\per(T[i+1\dd i+\tau]) \le \frac13\tau < \per(T[i\dd
  i+\tau])=\per(T[i'\dd i'+\tau])$. Hence, $T[i'+1\dd i'+\tau]$ can be
  extended to a $\tau$-run $T[p'\dd q']$ that starts at position
  $p'=i+1$, and thus $i'$ satisfies condition~\ref{it:sss:beg}.
  Similarly, if $i$ satisfies condition~\ref{it:sss:end}, then
  $\per(T[i'+\tau-1\dd i'+2\tau-2])=\per(T[i+\tau-1\dd i+2\tau-2]) \le
  \frac13\tau < \per(T[i+\tau-1\dd i+2\tau-1])=\per(T[i'+\tau-1\dd
  i'+2\tau-1])$.  Hence, $T[i'+\tau-1\dd i'+2\tau-2]$ can be extended
  to a $\tau$-run $T[p'\dd q']$ that ends at position $q'=i+2\tau-2$,
  and thus $i'$ satisfies condition~\ref{it:sss:end}.

  For a proof of the density condition, consider a position $i\in
  [1\dd n-3\tau+2]$ with $[i\dd i+\tau)\cap
  \S_\sig(\tau,T)=\emptyset$.  We start by identifying a $\tau$-run
  $T[p\dd q]$ with $p \le i+\tau$ and $q\ge i+2\tau-2$.  First,
  suppose that there exists a position $b\in [i+\tau-1\dd
  i+2\tau-1)\cap B_k(T)$.  Since $b-\tau+1\in [i\dd i+\tau)$ has not
  been added to $\S_\sig(\tau,T)$, the fragment $T[\beta-\tau+1\dd
  \beta+\tau]$ must be contained in a $\tau$-run $T[p\dd q]$ that
  satisfies $p\le b-\tau+1 \le i+\tau-1$ and $q\ge \beta+\tau \ge
  i+2\tau-1$.

  Next, suppose that $[i+\tau-1\dd i+2\tau-1)\cap B_k(T) = \emptyset$.
  Then, $T[i+\tau-1\dd i+2\tau]$ is contained in a single phrase
  induced by $T_k$, and the length of this phrase is at least
  $\tau+1$.  If $k=0$, this contradicts that all level-$0$ phrases are
  of length $1$.  Due to $\tau+1 \ge 3\alpha_k + 1 > 3d_k$ (see
  \cref{fct:alphabeta}), so \cref{fct:recompr} yields
  $\per(T[i+\tau-1\dd i+2\tau])\le d_k \le \alpha_k \le \frac13\tau$
  (again by \cref{fct:alphabeta}). The run $T[p\dd q]$ extending
  $T[i+\tau-1\dd i+2\tau]$ satisfies $p\le i+\tau-1$ and $q\ge
  i+2\tau$.

  Note that $p-1$ and $q-2\tau+2$ satisfy conditions~\ref{it:sss:beg}
  and~\ref{it:sss:end}, respectively.  Due to $[i\dd i+\tau)\cap
  \S_\sig(\tau,T) = \emptyset$, this implies $p\le i$ and $q\ge
  i+3\tau-2$, which means that $\per(T[i\dd i+3\tau-1))\le
  \frac13\tau$ holds as claimed.

  For the converse implication, note that if $s\in [i\dd i+\tau)\cap
  \S_\sig(\tau,T)$, then $\per(T[i\dd i+3\tau-1))\ge \per(T[s\dd
  s+2\tau)) > \frac13\tau$ because $T[s\dd s+2\tau)$ is not contained
  in any $\tau$-run (the latter observation is trivial if $s$
  satisfies condition~\ref{it:sss:reg}; in the remaining two cases, it
  follows from the upper bound of $\frac23\tau$ on the overlap of two
  $\tau$-runs).

  As for the size, observe that $s\in \S_\sig(\tau,T)$ can be
  accounted to $s+\tau-1\in B_k(T)$ if it satisfies
  condition~\ref{it:sss:reg}, to $\gamma \in \RUNS_\tau(T)$ starting
  at position $s+1$ if it satisfies condition~\ref{it:sss:beg}, and to
  $\gamma\in \RUNS_\tau(T)$ ending at position $s+2\tau-2$ if it
  satisfies condition~\ref{it:sss:end}.  Since any two distinct
  $\tau$-runs $\gamma,\gamma'$ satisfy $|\gamma\cap \gamma'|\le
  \frac23\tau$, the number of positions $s\in \S\cap[i\dd i+\tau)$
  satisfying~\ref{it:sss:beg} or~\ref{it:sss:end} is $\Oh(1)$.  By
  \cref{lem:sparse}, the number of positions
  satisfying~\ref{it:sss:reg} is $\Oh(\frac{\tau}{d_k}) =
  \Oh(\frac{\alpha_{k+1}}{d_k}) = \Oh(h_k) = \Oh(\log^*(k\sigma)) =
  \Oh(\log^*(\tau\sigma))$ by \cref{fct:alphabeta,rem:nk}.
\end{proof}

An additional property of \cref{cons:sss} is consistency across
different strings.  The following observation can be proved by
adapting the first paragraph of the proof of \cref{lem:sss}.

\begin{observation}
  Consider strings $T,T'\in \Sigma^+$ and positions $i\in [1\dd
  |T|-2\tau+1]$, $i'\in [1\dd |T'|-2\tau+1]$.  If $T[i\dd
  i+2\tau)=T'[i'\dd i'+2\tau)$, then $i\in \S_\sig(\T,\tau)$ if and
  only if $i'\in \S_\sig(\tau,T')$.
\end{observation}

\begin{corollary}\label{cor:sssub}
  Consider a string $T\in \Sigma^+$ and its substring $U=T(i\dd j]$.
  Then, $\S_\sig(\tau,U)=\{s-i : s\in \S_\sig(\tau,T)\cap (i\dd
  j-2\tau+1]\}$.
\end{corollary}

\begin{proposition}\label{prp:tauruns}
  The data structure of \cref{sec:ds} can be extended so that, given
  $T\in \W$ and $\tau\in \Zp$, the set $\RUNS_\tau(T)$, with
  $\tau$-runs ordered by their starting positions, can be constructed
  in $\Oh(\frac{|T|\log
  |T|}{\tau}(\log^*(|T|\sigma))^2(\log^*(|T|\sigma) +\frac{\log^2\log
  g}{\log \log \log g}))$ time.
\end{proposition}
\begin{proof}
  We partition $T$ into blocks of length $\floor{\frac13\tau}$
  (leaving up to $\floor{\frac13\tau}$ trailing characters
  behind). For any two consecutive blocks, we extract the
  corresponding fragment (using $\splitop$) and apply $\perop$ to
  retrieve its shortest period (provided that it does not exceed
  $\frac13\tau$).  If the period indeed does not exceed $\frac13\tau$,
  we maximally extend the fragment while preserving its shortest
  period. This is implemented using an $\lcpop$ query on two suffixes
  of $T$ and an $\lcsop$ queries on two prefixes of $T$. If the
  maximal fragment is of length at least $\tau$, we include it in
  $\RUNS_\tau(T)$.  By the Periodicity Lemma, each $\tau$-run is
  generated this way (perhaps multiple times).  Moreover, if we
  process the block pairs in the left-to-right order, then the
  $\tau$-runs are also generated in the left-to-right order.
\end{proof}

\begin{proposition}\label{prp:sss}
  The data structure of \cref{sec:ds} can be extended so that, given
  $T\in \W$ and $\tau\in \Zp$, the set $\S_\sig(\tau,T)$ obtained
  using \cref{cons:sss} can be built in $\Oh(\frac{|T|\log
  |T|}{\tau}(\log^*(|T|\sigma))^2\cdot \allowbreak (\log^*(|T|\sigma)
  +\frac{\log^2\log g}{\log \log \log g}))$ time.
\end{proposition}
\begin{proof}
  We first compute the largest $k\in \Zz$ with $2\tau \ge
  \alpha_k+\beta_k$ and construct $T_k$ by traversing the parse tree
  $\Tr(\symb_{\sig}(T))$.  In particular, this yields the set
  $B_k(T)$.  Next, we generate $\RUNS_\tau(T)$ using
  \cref{prp:tauruns}.  To construct $\S_\sig(\tau,T)$, we
  simultaneously scan $B_k(T)$ and $\RUNS_\tau(T)$.  For every
  position $i\in B_k(T)\cap [\alpha_k+1\dd n-2\tau-\alpha_k]$, we add
  $i-\tau+1$ to $\S_\sig(\tau,T)$ provided that $T(i-\tau\dd i-2\tau]$
  is not contained in any $\tau$-run.  Moreover, for every $\tau$-run
  $T[p\dd q]$, we add to $\S_\sig(\tau,T)$ positions $p-1$ and
  $q-2\tau+2$ (provided that they are within $[1\dd n-2\tau+1]$).  The
  query time is dominated by the cost of generating runs using
  \cref{prp:tauruns}.
\end{proof}

\section{Dynamic Text Implementation}\label{sec:app}

In this section, we develop a data structure that maintains a dynamic
text $T\in \Sigma^+$ subject to character insertions and deletions, as
well as substring swaps (the `cut-paste' operation), and supports
queries specified in
\cref{as:small,as:core,as:nonperiodic,as:periodic}.

For an alphabet $\Sigma$, we say that a \emph{labelled string} over
$\Sigma$ is a string over $\Sigma \times \Zz$.  For $c:=(a,\ell)\in
\Sigma \times \Zz$, we say that $\val(c):=a$ is the \emph{value} of
$c$ and $\labelop(c):=\ell$ is the \emph{label} of $c$.  For a
labelled string $S\in (\Sigma \times \Zz)^*$, we define the set of
labels $L(S)=\{\labelop(S[i]) : i\in [1\dd |S|\}$ and the string of
values $\val(S)=\val(S[1])\cdots \val(S[|S|])$.

Instead of maintaining a single labelled string representing $T$, our
data structure internally allows maintaining multiple labelled strings
(with character labels unique across the entire collection).  This
lets us decompose each update into smaller building blocks; for
example, a $\swapop$ (cut-paste) operation can be implemented using
three splits followed by three concatenations. This internal interface
matches the setting considered in a large body of previous work on
dynamic strings~\cite{Mehlhorn,Alstrup2000,dynstr}.

For a finite family $\LW \sub (\Sigma \times \Zz)^*$, we set
$L(\LW)=\bigcup_{S\in \LW} L(S)$ and $\|\LW\|=\sum_{S\in \LW} |S|$.
We say $\LW$ is \emph{uniquely labelled} if $|L(\LW)|=\|\LW\|$;
equivalently, for each label $\ell\in L(\LW)$ there exist unique $S\in
\LW$ and $i\in [1\dd |S|]$ such that $\ell=\labelop(S[i])$.

\begin{lemma}\label{lem:lf}
  There is a data structure maintaining a uniquely labelled family
  $\LW\sub (\Sigma \times \Zz)^+$ using the following interface, where
  $\labelop(S[1])$ is used as a reference to any string $S\in \LW$:
  \begin{description}
  \item[$\concatop(R,S)$:] Given distinct $R,S\in \LW$, set $\LW :=
    \LW\sm (\{R,S\})\cup \{R\cdot S\}$.
  \item[$\splitop(S,i)$:] Given $S\in \LW$ and $i\in [1\dd |S|)$, set
    $\LW := \LW\sm (\{S\})\cup \{S[1\dd i], S(i\dd |S|]\}$.
  \item[$\makeop(a,\ell)$:] Given $a\in \Sigma$ and $\ell\in \Zz\sm
    L(\LW)$, set $\LW:=\LW\cup\{(a,\ell)\}$.
  \item[$\deleteop(S)$:] Given $S\in \LW$ with $|S|=1$, set $\LW :=
    \LW\sm\{S\}$.
  \item[$\labelop(S,i)$:] Given $S\in \LW$ and $i\in [1\dd |S|]$,
    return $\labelop(S[i])$.
  \item[$\val(S,i)$:] Given $S\in \LW$ and $i\in [1\dd |S|]$, return
    $\val(S[i])$.
  \item[$\unlabel(\ell)$:] Given $\ell\in L(\LW)$, return $(S,i)$,
    where $S\in \LW$ and $i\in [1\dd |S|]$ are such that
    $\ell=\labelop(S[i])$.
  \end{description}
  In the word RAM model with word size $w$ satisfying $L(\LW)\sub
  [0\dd 2^w)$, each of these operations can be implemented in
  $\Oh(\log n)$ time, where $n=\|\LW\|$.
\end{lemma}
\begin{proof}
  Each string $S\in \LW$ is stored as a balanced binary search tree
  (such as an AVL tree~\cite{AVL}) whose in-order traversal yields
  $S$.  Additionally, each node of the BST is augmented with the size
  of its subtree and the label of the leftmost node.  We also maintain
  a node dictionary that maps each label $\ell \in L(\LW)$ to the node
  representing the character with label $\ell$.

  The $\unlabel(\ell)$ operation uses the node dictionary to reach the
  node representing the character with label $\ell$, and then
  traverses the path from $\ell$ to the root of the corresponding
  tree, using the subtree sizes to determine the number of nodes to
  the left of the path traversed.  The remaining operations use the
  same approach to locate the roots of the trees representing strings
  $S\in \LW$ represented by $\labelop(S[1])$.  Then, the $\val$ and
  $\labelop$ operations traverse the tree, using the subtree sizes to
  descend to the $i$th leftmost node, and they return the value and
  the label, respectively, of the corresponding character.  The
  $\splitop$ operation also descends to the $i$th node of the tree
  representing $S$, but then it splits the tree.  The $\concatop$
  operation joins two trees.  The $\deleteop$ operation deletes a
  single-element tree and removes the corresponding entry from the
  node dictionary.  The $\makeop$ operation creates a single-element
  tree and inserts a new entry to the node dictionary.
\end{proof}

As a warm-up, we show that \cref{lem:lf} allows for efficient random
access to a dynamic text.

\begin{corollary}\label{cor:access}
  A dynamic text $T\in \Sigma^n$ can be implemented so that
  initialization takes $\Oh(1)$ time, updates take $\Oh(\log n)$ time,
  and the following $\accessop(i)$ queries take $\Oh(\log n)$ time:
  \begin{description}
  \item[$\accessop(i)$:] given $i\in [1\dd n]$, return $T[i]$.
  \end{description}
\end{corollary}
\begin{proof}
  We maintain a uniquely labelled string family $\LW$ that satisfies
  the following invariant after each update is completed: $\LW =
  \{S\}$ for a labelled string $S$ such that $\val(S)=T$.  We also
  store a count $c$ of the insertions performed so that we can assign
  a fresh label to every inserted character, and a label
  $\ell:=\labelop(S[1])$ needed to access the only string in $\LW$.

  To implement $\initop(\sigma)$, we initialize an empty labelled
  family $\LW$ over $\Sigma=[0\dd \sigma)$ and perform an
  $\LW.\makeop(\$,0)$ operation to make sure that $\LW= \{S\}$ with
  $\val(S)=\$$.  We also set $c:=0$ and $\ell := 0$.
  
  As for the $T.\makeop(i,a)$ operation, we increment the counter $c$
  and perform the $\LW.\makeop(a, c)$ operation adding to $\LW$ a new
  labelled string $R$ with $\val(R)=a$.  If $i=1$, we simply perform
  $\LW.\concatop(\ell,c)$ that results in $\LW=\{R\cdot S\}$, and we
  update the label of the first character of the only string in $\LW$,
  setting $\ell := c$.  If $i>1$, on the other hand, we retrieve the
  label $\ell_i := \LW.\labelop(\ell,i)$ of $S[i]$, and perform the
  following calls: $\LW.\splitop(\ell,i-1)$ (resulting in $\LW =
  \{S[1\dd i), S[i\dd n],R\}$), $\LW.\concatop(\ell,c)$ (resulting in
  $\LW = \{S[1\dd i)\cdot R, S[i\dd n]\}$), and
  $\LW.\concatop(\ell,\ell_i)$ (resulting in $\LW = \{S[1\dd i)\cdot
  R\cdot S[i\dd n]\}$).

  As for the $T.\deleteop(i)$ operation, we retrieve the label
  $\ell_i:= \LW.\labelop(\ell,i)$ of the character to be deleted and
  the label $\ell_{i+1}:= \LW.\labelop(\ell,i+1)$ of the subsequent
  character.  Next, we perform the $\LW.\splitop(\ell, i)$ operation
  that results in $\LW = \{S[1\dd i],S(i\dd n]\}$.  If $i=1$, we then
  set $\ell:=\ell_{i+1}$ because $S[1\dd i]=S[i]$ and $S(i\dd
  n]=S[1\dd i)\cdot S(i\dd n]$.  Otherwise, we perform the
  $\LW.\splitop(\ell,i-1)$ operation (resulting in $\LW = \{S[1\dd i),
  S[i],S(i\dd n]\}$) and the $\LW.\concatop(\ell,\ell_{i+1})$
  operation (resulting in $\LW = \{S[1\dd i)\cdot S(i\dd n], S[i]\}$).
  In both cases, we conclude with a call $\LW.\deleteop(\ell_i)$ that
  removes $S[i]$ from $\LW$ and results in desired state $\LW =
  \{S[1\dd i)\cdot S(i\dd n]\}$.

  As for the $T.\swapop(i,j,k)$ operation, we first check whether $i <
  j < k$; otherwise, there is nothing to do.  Next, we retrieve the
  labels $\ell_i := \LW.\labelop(\ell,i)$, $\ell_j :=
  \LW.\labelop(\ell,j)$, and $\ell_k := \LW.\labelop(\ell,k)$ of
  $S[i]$, $S[j]$, and $S[k]$, respectively.  Then, we perform the
  $\LW.\splitop(\ell, k-1)$ operation (resulting in $\LW = \{S[1\dd
  k), S[k\dd n]\}$), the $\LW.\splitop(\ell,j-1)$ operation (resulting
  in $\LW = \{S[1\dd j), S[j\dd k), S[k\dd n]\}$), and the
  $\LW.\concatop(\ell,\ell_k)$ operation (resulting in $\LW = \{S[1\dd
  j)\cdot S[k\dd n], S[j\dd k)\}$).  If $i=1$, we proceed with the
  $\LW.\concatop(\ell_j,\ell)$ operation (resulting in $\LW = \{S[j\dd
  k)\cdot S[1\dd j)\cdot S[k\dd n]\}=\{S[1\dd i)\cdot S[j\dd k)\cdot
  S[i\dd j)\cdot S[k\dd n]\}$ and set $\ell:= \ell_j$ to update the
  label of the first character of the only string in $\LW$.)
  Otherwise, we perform the $\LW.\splitop(\ell,i-1)$ operation
  (resulting in $\LW = \{S[1\dd i), S[i\dd j)\cdot S[k\dd n], S[j\dd
  k)\}$), the $\LW.\concatop(\ell,\ell_j)$ operation (resulting in
  $\LW = \{S[1\dd i) \cdot S[j\dd k), S[i\dd j)\cdot S[k\dd n]\}$),
  and finally the $\LW.\concatop(\ell,\ell_i)$ operation (resulting in
  $\LW = \{S[1\dd i)\cdot S[j\dd k)\cdot S[i\dd j)\cdot S[k\dd n]\}$).

  Finally, the $T.\accessop(i)$ operation simply queries
  $\LW.\accessop(\ell,i)$ and forwards the obtained answer.
    
  It is easy to see that the initialization takes $\Oh(1)$ time
  whereas the remaining operations take $\Oh(\log n)$ time.
\end{proof}

\subsection{Implementing Assumption~\ref{as:small}}

For a labelled family $\LW$ and an integer $q\in \Zp$, we
define \[\CF_{q}(\LW)=\bigcup_{S\in \LW} \{(\val(S)^\infty[j\dd j+q),
\labelop(S[j])) : j\in [1\dd |S|]\}\]

\begin{lemma}\label{lem:rksel}
  For any fixed $q\in \Zp$, the data structure of \cref{lem:lf} can be
  augmented so that, given $r\in [1\dd \|\LW\|]$, the $r$th smallest
  element $(X,\ell)\in \CF_{q}(\LW)$ can be computed in $\Oh(q \log
  n)$ time, along with the counts $|\{(X',\ell')\in \CF_q(\LW)\} : X'
  \prec X\}|$ and $|\{(X',\ell')\in \CF_q(\LW) : X' \preceq
  X\}|$. This comes at the price of increasing the cost of all updates
  to $\Oh(q^2\log n)$.
\end{lemma}
\begin{proof}
  We explicitly store $\CF_{q}(\LW)$ in a balanced binary search tree
  (ordered lexicograpically) and a dictionary that maps each label
  $\ell\in L(\LW)$ to the corresponding node $(X,\ell)\in
  \CF_{q}(\LW)$.  With this implementation, a query takes $\Oh(q\log
  n)$ time ($\Oh(\log n)$ comparisons in $\Oh(q)$ time each).

  As for the $\deleteop(S)$ operation, we just remove from
  $\CF_{q}(\LW)$ the corresponding pair $(S^q,\labelop(S) )$, which
  takes $\Oh(q\log n)$ time. Symmetrically, $\makeop(a,\ell)$ inserts
  $(a^q,\ell)$ in $\Oh(q\log n)$ time.

  The $\concatop(R,S)$ operation is implemented as follows: For each
  $j\in (\max(0,|R|-q+1)\dd |R|]$, we replace $(\val(R)^\infty[j\dd
  j+q), \labelop(R[j]))$ with $(\val(RS)^\infty[j\dd j+q),
  \labelop(R[j]))$ and, analogously, for each $j\in
  (\max(0,|S|-q+1)\dd |S|]$, we replace $(\val(S)^\infty[j\dd j+q),
  \labelop(S[j]))$ with $(\val(RS)^\infty[j+|R|\dd j+|R|+q),
  \labelop(S[j]))$.  Each new entry can be constructed in $\Oh(q\log
  n)$ time using $q$ calls to the $\val$ and $\labelop$ operations,
  for the overall running time of $\Oh(q^2\log n)$.

  The implementation of $\splitop(S,i)$ is symmetric. For each $j\in
  (\max(0,i-q+1)\dd i]$, we replace $(\val(S)^\infty[j\dd j+q),
  \labelop(S[j]))$ with $(\val(S[1\dd i])^\infty[j\dd j+q),
  \labelop(S[j]))$, and, analogously, for each $j\in
  (\max(i,|S|-q+1)\dd |S|]$, we replace $(\val(S)^\infty[j\dd j+q),
  \labelop(S[j]))$ with $(\val(S(i\dd |S|])^\infty[j-i\dd j-i+q)),
  \labelop(S[j]))$.
\end{proof}

\begin{proposition}\label{prp:small}
  A dynamic text $T\in \Sigma^n$ can be implemented so that
  initialization takes $\Oh(1)$ time, updates take $\Oh(\log n)$ time,
  and the queries of \cref{as:small} take $\Oh(\log n)$ time.
\end{proposition}
\begin{proof}
  We proceed as in the proof of \cref{cor:access}, but instead of
  implementing the uniquely labelled family $\LW$ using the vanilla
  version of \cref{lem:lf}, we apply the extension of \cref{lem:rksel}
  for $q=16$.  Due to $q=\Oh(1)$, this preserves the update times.

  As for the query, we forward the argument $i$ to the query of
  \cref{lem:rksel}.  This results in the $i$th smallest pair
  $(X,\ell)$ in the set $\{T^\infty[j\dd j+16), \labelop(S[j])) : j\in
  [1\dd n]\}$, as well as the counts $|\{j\in [1\dd n] : T^\infty[j\dd
  j+16)\prec X\}|$ and $|\{j\in [1\dd n] : T^\infty[j\dd j+16)\preceq
  X\}|$.  By \cref{cor:lbub}, we have $X = T^\infty[\SA[i]\dd
  \SA[i]+16)$, and thus the counts are equal to $\LB_{16}(\SA[i])$ and
  $\UB_{16}(\SA[i])$, respectively. Moreover, a position $j\in
  \Occ_{16}(\SA[i])$ can be retrieved using an $\unlabel(\ell)$ call.
\end{proof}

\subsection{Common Tools for Assumptions~\ref{as:core},~\ref{as:nonperiodic}, and~\ref{as:periodic}}

\newcommand{\Lb}{\mathsf{L}}

For a family $\Lb\sub L(\LW)$ and a label $\ell\in \LW$, we define
$\Predecessor_{\Lb}(\ell)$ and $\Successor_{\Lb}(\ell)$ as follows.
Let $(S,i)=\unlabel(\ell)$.  If $L(S[1\dd i])\cap \Lb = \emptyset$,
then we set $\Predecessor_{\Lb}(\ell)=\bot$.  Otherwise,
$\Predecessor_{\Lb}(\ell)=\labelop(S[\max\{t\in [1\dd i] :
\labelop(S[t])\in \Lb\}])$.  Symmetrically, if $L(S[i\dd |S|])\cap \Lb
= \emptyset$, then we set $\Successor_{\Lb}(\ell)=\bot$.  Otherwise,
$\Successor_{\Lb}(\ell)=\labelop(S[\min\{t\in [i\dd |S|] :
\labelop(S[t])\in \Lb\}])$.

\begin{lemma}\label{lem:predsucc}
  The data structure of \cref{lem:lf} can be augmented to maintain a
  set $\Lb\sub L(\LW)$ of \emph{marked} labels so that, given $\ell\in
  L(\LW)$, the values $\Predecessor_{\Lb}(\ell)$ and
  $\Successor_{\Lb}(\ell)$ can be computed in $\Oh(\log n)$
  time. Moreover, marking and unmarking a label $\ell\in L(\LW)$ also
  costs $\Oh(\log n)$ time.
\end{lemma}
\begin{proof}
  Compared to the implementation in the proof of \cref{lem:lf}, each
  node $\nu$ stores two extra bits, specifying whether the
  corresponding label is marked and whether the subtree of $\nu$
  contains a node with a marked label.

  To determine the predecessor of $\Predecessor_{\Lb}(\ell)$, we first
  locate the node $\nu$ with label $\ell$.  Then, we check whether
  $\ell\in \Lb$; if so, we simply return $\ell$.  Otherwise, we
  traverse the path from $\nu$ the root of the corresponding tree.
  For $\nu$ as well as for each node that we reach from its left
  subtree, we check whether the right subtree contains a node with a
  marked label node. If so, we descend to the leftmost such node and
  return its label.  If our traversal reaches the root without
  success, we report that $\Predecessor_{\Lb}(\ell)=\bot$.  It is easy
  to see that this procedure is correct and takes $\Oh(\log n)$ time.
  A symmetric procedure computes $\Successor_{\Lb}(\ell)$.

  As for marking and unmarking, we use the node dictionary to locate
  the node with label $\ell$, and update its status with respect to
  marking. Then, we update the cumulative bits on the path towards the
  root; this costs $\Oh(\log n)$ time.  As for the remaining updates,
  we recompute the cumulative bits for all nodes visited; this does
  not increase the asymptotic time of these updates.
\end{proof}

Based on the function $\CFp(\cdot)$ defined for unlabeled string
families in \cref{sec:lex}, let us denote
$\CFp(\LW)=\CFp(\{\val(S):S\in \LW\})$.

\begin{lemma}\label{lem:comp}
  The data structure of \cref{lem:lf} can be augmented so that any two
  fragments in $\CFp(\LW)$ and any two fragments in
  $\CFp(\revstr{\LW})$ can be compared lexicographically in $\Oh(\log
  n\frac{\log^2 \log m\cdot \log^* m}{\log \log \log m})$ time, where
  $m=\sigma + t$ and $t$ is the total number of instruction that the
  data structure has performed so far.  This comes at the price of
  increasing the cost of $\concatop$ to $\Oh(\log n\frac{\log^2 \log
  m\cdot \log^* m}{\log \log \log m})$, $\splitop$ to $\Oh(\log
  n\frac{\log^2 \log m\cdot \log^* m}{\log \log \log m})$, and
  $\makeop$ to $\Oh(\frac{\log^2 \log m\cdot \log^* m}{\log \log \log
  m})$.
\end{lemma}
\begin{proof}
  On top of the data structure of \cref{lem:lf}, we also store
  $\W:=\{\val(S) : S\in \LW\}$ using dynamic strings of \cref{sec:ds}.
  The $\makeop(a,\ell)$ operation adds a new string $a$ to $\W$ using
  $\W.\makeop(a)$.  The $\splitop(S,i)$ operation adds strings
  $\val(S)[1\dd i]$ and $\val(S)(i\dd |S|]$ to $\W$ using
  $\W.\splitop(\val(S),i)$.  The $\concatop(R,S)$ operation adds the
  string $\val(R)\cdot \val(S)$ to $\W$ using
  $\W.\concatop(\val(R),\val(S))$.  As for lexicographic comparisons,
  we use the $\compareop$ operation implemented in \cref{sec:lex}.
\end{proof}

\subsection{Implementing Assumption~\ref{as:core}}

Recall that, for a string $T$ and an integer $\tau\in \Zp$, we defined
$\R(\tau,T) = \{i\in [1\dd |T|-3\tau+2] : \per(T[i\dd i+3\tau-2])\le
\frac13\tau\}$ and $\R'(\tau,T)=\{i\in \R(\tau,T) : i-1\notin
\R(\tau,T)\}$.  We also define a symmetric set $\R''(\tau,T) = \{i\in
\R(\tau,T) : i+1\notin \R(\tau,T)\}$.  For a uniquely labelled family
$\LW$, we generalize these notions as follows:
\begin{align*}
  \R(\tau,\LW)&=\{\labelop(S[i]) : S\in \LW\text{ and } i\in \R(\tau,\val(S))\}\\
  \R'(\tau,\LW)&=\{\labelop(S[i]) : S\in \LW\text{ and } i\in \R'(\tau,\val(S))\}\\
  \R''(\tau,\LW)&=\{\labelop(S[i]) : S\in \LW\text{ and } i\in \R''(\tau,\val(S))\}.
\end{align*}

\begin{lemma}\label{lem:runs}
  For any fixed $\tau\in \Zp$, the data structure of \cref{lem:comp}
  can be augmented so that, given $\ell\in L(\LW)$, we can check in
  $\Oh(\log n)$ time whether $\ell\in \R(\tau,\LW)$.  This comes at
  the price of increasing the cost of $\concatop$ to $\Oh(\log n
  \frac{\log^2 \log m\cdot (\log^* m)^2}{\log \log \log m})$.
\end{lemma}
\begin{proof}
  On top of the data structure of \cref{lem:comp}, we store two
  instances of the component of \cref{lem:predsucc}, for
  $\R'(\tau,\LW)$ and $\R''(\tau,\LW)$, respectively.
  
  In the query algorithm, we first find
  $\ell':=\Predecessor_{\R'(\tau,\LW)}(\ell)$.  If $\ell'=\bot$, we
  report that $\ell \notin \R(\tau,\LW)$.  Otherwise, we find
  $\ell'':=\Successor_{\R''(\tau,\LW)}(\ell')$ (it is never $\bot$),
  and we use the $\unlabel(\ell)$, $\unlabel(\ell')$, and
  $\unlabel(\ell'')$ operations to determine the positions $j$, $j'$,
  and $j''$ corresponding to these labels (these are positions in the
  same string $S\in \LW$). We report that $\ell \in \R(\tau,\LW)$ if
  and only if $j' \le j \le j''$.

  After executing the $\concatop(R,S)$ operation of \cref{lem:comp},
  we perform the following steps.  First, we remove $\labelop(S[1])$
  from $\R'(\tau,\LW)$ and $\labelop(R[|R|])$ from $\R''(\tau,\LW)$
  (if present in the respective sets).  Next, we compute the
  $\tau$-runs in $\val(U)$, where $U=R(\max(0,|R|-3\tau+1)\dd
  |R|]\cdot S[1\dd \min(|S|,3\tau-1)]$. For this, we use construct
  $\val(U)$ using $\concatop$ and $\splitop$ operations on $\W$, and
  then run the algorithm of \cref{prp:tauruns} on $\val(U)$.  We
  iterate over fragments $(R\cdot S)[x\dd y]$ corresponding to
  $\tau$-runs of length at least $3\tau-1$ in $\val(U)$.  For each
  such fragment, we add $\labelop((R\cdot S)[x])$ to $\R'(\tau,\LW)$
  if $x > |R|-3\tau+2$, and we add $\labelop((R\cdot S)[y])$ to
  $\R''(\tau \LW)$ if $y \le |R|+3\tau-2$.  The number of newly marked
  labels is $\Oh(1)$, so the extra cost is $\Oh(\log n \cdot
  \frac{\log^2 \log m\cdot \log^* m}{\log \log \log m})$ time,
  dominated by the procedure of \cref{prp:tauruns}.

  Before executing the $\splitop(S,i)$ operation of \cref{lem:comp},
  we perform the following steps.  First, we run the query algorithm
  to check whether $\labelop(S[i])\in \R(\tau,\LW)$ and
  $\labelop(S[i+1])\in \R(\tau,\LW)$. If $\labelop(S[i])\in
  \R(\tau,\LW)$, we add $\labelop(S[i])$ to $\R''(\tau,\LW)$, and if
  $\labelop(S[i+1])\in \R(\tau,\LW)$, we add $\labelop(S[i+1])\in
  \R'(\tau,\LW)$.  Then, we remove from $\R'(\tau,\LW)$ all labels
  $\ell$ with $\unlabel(\ell)=(S,j)$ for $j\in (i-3\tau+2\dd i]$, and
  from $\R''(\tau,\LW)$ all labels $\ell$ with $\unlabel(\ell)=(S,j)$
  for $j\in (i-3\tau+2\dd i]$ for $j\in (i\dd i+3\tau-2]$.  The labels
  to be removed are listed with $\Oh(\log n)$-time delay using
  predecessor queries on $\R'(\tau,\LW)$ and successor queries on
  $\R''(\tau,\LW)$, respectively.  The number of newly unmarked labels
  in $\Oh(1)$, so the extra cost compared to \cref{lem:comp} is
  $\Oh(\log n)$, dominated by the original cost of $\W.\splitop(S,i)$.
\end{proof}
  
\begin{proposition}\label{prp:core}
  For any fixed $\ell\in \Zp$, a dynamic text $T\in \Sigma^n$ can be
  implemented so that initialization takes $\Oh(\frac{\log^2 \log
  m\cdot \log^* m}{\log \log \log m})$ time, updates take $\Oh(\log
  n \cdot \frac{\log^2 \log m\cdot (\log^* m)^2}{\log \log \log m})$
  time, and the queries of \cref{as:core} take $\Oh(\log n)$ time,
  where $m=|\Sigma|+t$ and $t$ is the total number of instructions
  that the data structure has performed so far.
\end{proposition}
\begin{proof}
  We proceed as in the proof of \cref{cor:access}, but instead of
  implementing family $\LW$ using \cref{lem:lf}, we apply the
  extension of \cref{lem:runs} for $\tau=\floor{\frac{\ell}{3}}$. This
  increases the cost of $\initop$ to $\Oh(\frac{\log^2 \log m\cdot
  \log^* m}{\log \log \log m})$ and the cost of updates to $\Oh(\log
  n \cdot \frac{\log^2 \log m\cdot (\log^* m)^2}{\log \log \log m})$
  (dominated by $\LW.\concatop$).

  As for a query, note that $\LW = \{S\}$, where $S$ is a labelled
  string with $\val(S)=T$.  Given a position $i\in [1\dd |T|]$, we
  compute $\labelop(S[i])$ using a call $\LW.\labelop(S,i)$.  Next, we
  apply the query algorithm of \cref{lem:runs} to check whether
  $\labelop(S[i])\in \R(\tau,\LW)$, which is equivalent to testing
  whether $i\in \R(\tau,T)$.  The query time is therefore $\Oh(\log
  n)$.
\end{proof}

\subsection{Implementing Assumption~\ref{as:nonperiodic}}

Based on \cref{cons:sss}, we define $\S_\sig(\tau,\LW) = \bigcup_{S\in
\LW}\{\labelop(S[i]) : i\in \S_\sig(\tau,\val(S))\}$.

\begin{lemma}\label{lem:succ}
  For any fixed $\tau\in \Zz$, the data structure of \cref{lem:comp}
  can be augmented so that, given $S\in \LW$ and $j\in [1\dd |S|]$,
  the value $\Successor_{\S_\sig(\tau,\val(S))}(j)$ can be computed in
  $\Oh(\log n)$ time, where $\sig$ is an (implicit) signature
  function. This comes at the price of increasing the cost of
  $\concatop$ to $\Oh(\log n\cdot \frac{\log^2 \log m\cdot (\log^*
  m)^2}{\log \log \log m})$.
\end{lemma}
\begin{proof}
  On top of the data structure of \cref{lem:comp}, we store the
  component of \cref{lem:predsucc} with $\Lb:=\S_\sig(\tau,\LW)$.  As
  for the query algorithm, we compute $\ell = \labelop(S[j])$ and find
  $\ell' = \Successor_{\Lb}(\ell)$. If $\ell'=\bot$, then
  $\Successor_{\S_\sig(\tau,\val(S))}(j)=|S|-2\tau+2$ (see
  \cref{sec:sa-nonperiodic}). Otherwise,
  $\Successor_{\S_\sig(\tau,\val(S))}(j)=j'$, where
  $(S,j')=\unlabel(\ell')$.  Thus, the query time is $\Oh(\log n)$.

  While executing the $\splitop(S,i)$ operation, we unmark all nodes
  corresponding to $S[i-2\tau+2\dd i]$.  For this, we repeatedly
  compute $\Predecessor_{\S_\sig(\tau,\LW)}(\labelop(S[i]))$ and, if
  the predecessor exists, use the $\unlabel$ operation to retrieve its
  position $j$.  If $j\ge i-2\tau+2$, we unmark the corresponding
  label. Otherwise, we terminate the procedure.  By \cref{cor:sssub},
  this correctly maintains the set of marked nodes.  The running time
  of this step is $\Oh(\log n)$ per unmarked node and, by
  \cref{lem:sss}, $\Oh(\log n \log^*(\tau\sigma))$ in total. This cost
  is dominated by the cost $\Oh(\log n \frac{\log^2 \log m\cdot \log^*
  m}{\log \log \log m})$ of $\splitop(\val(S),i)$.

  While executing the $\concatop(R,S)$ operation, by \cref{cor:sssub},
  we need to mark $\S_{\sig}(U,\tau)$ in the tree representing $R\cdot
  S$, where $U=R(\max(0,|R|-2\tau+1)\dd |R|]\cdot S[1\dd
  \min(2\tau,|S|)]$.  For this, we construct $\val(U)$ using
  $\splitop$ and $\concatop$ operations on $\W$, and then apply
  \cref{prp:sss} to derive $\S_{\sig}(\val(U),\tau)$.  For each of the
  obtained positions, we locate the corresponding label using the
  $\labelop(R\cdot S, j)$.  The cost $\Oh(\log n\frac{\log^2 \log
  m\cdot (\log^* m)^2}{\log \log \log m})$ is dominated by the
  application of \cref{prp:sss}.
\end{proof}

For a uniquely labelled family $\LW$, an integer $\tau\in \Zp$, and a
signature function $\sig$, define
\[
  \Pts_\sig(\tau,\LW)=\bigcup_{S\in \LW}\{(\revstr{\val(S)^{\infty}[i
      - 7\tau \dd i)}, \val(S)^{\infty}[i \dd i+7\tau),\labelop(S[i]))
      : S\in \S_\sig(\tau,\val(S))\}
\]
(cf.\ \cref{def:p-context}).  We interpret $\Pts_\sig(\tau,\LW)\sub
\X\times \Y \times \Z$ as a family of labelled points with $\X =
\CFp(\revstr{\LW})$ and $\Y=\CFp(\LW)$, both ordered
lexicographically.

\begin{lemma}\label{lem:range}
  The data structure of \cref{lem:succ} can be augmented so that range
  queries (\cref{sec:range-queries}) on $\Pts_\sig(\tau,\LW)$ can be
  supported in time $\Oh(\log^3 n + \log^2 n \cdot \frac{\log^2 \log
  m\cdot \log^* m}{\log \log \log m})$. This comes at the price of
  increasing the cost of $\concatop$ and $\splitop$ to $\Oh(\log^2 n
  \cdot \frac{\log^2 \log m\cdot (\log^* m)^2}{\log \log \log m})$
  time.
\end{lemma}
\begin{proof}
  Compared to the data structure of \cref{lem:succ}, we also maintain
  $\Pts_\sig(\tau,\LW)$ in the data structure of \cref{thm:range}.  By
  \cref{lem:comp}, each comparison on $\X$ and $\Y$ costs $\Oh(\log
  n\cdot \frac{\log^2 \log m\cdot \log^* m}{\log \log \log m})$ time.
  As a result, the cost of a range query is $\Oh(\log^3 n + \log^2
  n\cdot \frac{\log^2 \log m\cdot \log^* m}{\log \log \log m})$.  As
  for updates, $\makeop$ and $\deleteop$ do not incur updates to
  $\Pts$ (strings of length $1$ have empty synchronizing sets). On the
  other hand, $\splitop$ and $\concatop$ may incur $\Oh(\log^*m)$
  insertions and deletions, both corresponding to changes in the
  synchronizing set and to synchronizing positions whose context has
  changed; the latter can be generated by repeated calls to the
  $\Successor_{\S_\sig(\tau,\LW)}$ operation.  The running time is
  therefore increased to $\Oh(\log^2 n \cdot \frac{\log^2 \log m\cdot
  (\log^* m)^2}{\log \log \log m})$.
\end{proof}

\begin{proposition}\label{prp:nonperiodic}
  For any fixed $\ell\in \Zp$, a dynamic text $T\in \Sigma^n$ can be
  implemented so that initialization takes $\Oh(\frac{\log^2 \log
  m\cdot \log^* m}{\log \log \log m})$ time, updates take
  $\Oh(\log^2 n\cdot \frac{\log^2 \log m\cdot (\log^* m)^2}{\log \log
  \log m})$ time, and the queries of \cref{as:nonperiodic} take
  $\Oh(\log^3 n + \log^2 n\cdot \frac{\log^2 \log m\cdot \log^*
  m}{\log \log \log m})$ time, where $m=|\Sigma|+t$ and $t$ is the
  total number of instructions that the data structure has performed
  so far.
\end{proposition}
\begin{proof}
  We proceed as in the proof of \cref{cor:access}, but instead of
  implementing family $\LW$ using \cref{lem:lf}, we apply the
  extensions of \cref{lem:succ,lem:range} for
  $\tau=\floor{\frac{\ell}{3}}$. This increases the cost of $\initop$
  to $\Oh(\frac{\log^2 \log m\cdot \log^* m}{\log \log \log m})$ and
  the cost of updates to $\Oh(\log^2 n \cdot \frac{\log^2 \log m\cdot
  (\log^* m)^2}{\log \log \log m})$ (dominated by $\LW.\concatop$
  and $\LW.\splitop$).

  We use $\S:=\S_\sig(\tau,T)$ as the synchronizing set of $T$.  As
  for a $\Successor_{\S}(i)$ query for $i\in [1\dd n-3\tau+1]$, we
  simply forward the query to the component of \cref{lem:succ}.  The
  cost of this query is $\Oh(\log n)$.

  As for the string-string range queries on $\Points_{7\tau}(T,\S)$,
  we use the equivalent range queries on $\Pts_\sig(\tau,\LW)$
  instead. The only work needed is to convert the query arguments from
  indices to fragments in $\CFp(T)$ and $\CFp(\revstr{T})$ (as
  specified in \cref{prob:str-str}).  Moreover, the label returned by
  the range selection on $\Pts_\sig(\tau,\LW)$ is converted to a
  position in $T$ using the $\LW.\unlabel(\cdot)$ operation.  The
  query time is $\Oh(\log^3 n + \log^2 n\cdot \frac{\log^2 \log m\cdot
  \log^* m}{\log \log \log m})$, dominated by the cost of range
  queries of \cref{lem:range}.
\end{proof}

\subsection{Implementing Assumption~\ref{as:periodic}}
\begin{lemma}\label{lem:runs2}
  For any fixed $\tau\in \Zp$, the data structure of \cref{lem:runs}
  can be augmented so that, given $S\in \LW$ and $j\in
  \R(\tau,\val(S))$, in $\Oh(\log n)$ time we can compute
  $\rendgen{\tau}{\val(S)}{j}$,
  $|\Lrootgen{f_\sig}{\tau}{\val(S)}{j}|$, and
  $\Lheadgen{f_\sig}{\tau}{\val(S)}{j}$, where $f_\sig$ is the
  necklace-consistent function defined in \cref{cons:nc}.  This comes
  at the price of increasing the cost of $\splitop$ and $\concatop$ to
  $\Oh(\log n\cdot \frac{\log^2 \log m\cdot (\log^* m)^2}{\log \log
  \log m})$.
\end{lemma}
\begin{proof}
  On top of the data structure of \cref{lem:runs}, for each $\ell \in
  \R'(\tau,\LW)$, we store $|\Lrootgen{f_\sig}{\tau}{\val(S)}{j}|$ and
  $\Lheadgen{f_\sig}{\tau}{\val(S)}{j}$, where $(S,j)=\unlabel(\ell)$.
  
  In the query algorithm, we first compute $\ell=\labelop(S[j])$.
  Then, we find $\ell'=\Predecessor_{\R'(\tau,\LW)}(\ell)$ and
  $\ell''=\Successor_{\R''(\tau,\LW)}(\ell')$, and we use the
  $\unlabel(\ell')$ and $\unlabel(\ell'')$ operations to retrieve the
  corresponding positions $j'$ and $j''$ in $S$.  We note that
  $\rendgen{\tau}{\val(S)}{j}=\rendgen{\tau}{\val(S)}{j'}=j''+3\tau-2$.
  Additionally, we retrieve the stored values $p:=
  |\Lrootgen{f_\sig}{\tau}{\val(S)}{j'}|$ and
  $h:=\Lheadgen{f_\sig}{\tau}{\val(S)}{j'}$, and we return
  $|\Lrootgen{f_\sig}{\tau}{\val(S)}{j}|=p$ and
  $\Lheadgen{f_\sig}{\tau}{\val(S)}{j'}=(h+j'-j)\bmod p$.  The
  correctness of this procedure follows from \cref{lm:R-block}, and
  the running time is clearly $\Oh(\log n)$.

  In order to maintain the extra information stored for $\ell\in
  \R'(\tau,\LW)$, we execute the following procedure whenever a label
  $\ell$ is added to $\R'(\tau,\LW)$: first, we determine
  $(S,j)=\unlabel(\ell)$.  Next, we extract $\val(S)[j\dd j+3\tau-1)$
  using $\W.\splitop$ operations, compute $p:= \perop(\val(S)[j\dd
  j+3\tau-1))$, and note that
  $|\Lrootgen{f_\sig}{\tau}{\val(S)}{j}|=p$.  Then, we extract
  $\val(S)[j\dd j+p)$ using another $\W.\splitop$ operation, compute
  $h=\W.\canshop(\val(S)[j\dd j+p))$, and note that
  $\Lheadgen{f_\sig}{\tau}{\val(S)}{j}=h\bmod p$.  Since
  $\splitop(S,i)$ and $\concatop(R,S)$ involve $\Oh(1)$ insertions to
  $\R'(\tau,\LW)$, the extra cost of these operations is $\Oh(\log
  n\cdot \frac{\log^2 \log m\cdot (\log^* m)^2}{\log \log \log m})$.
  The correctness follows from the fact that $\ell=\labelop(S[j])$ is
  deleted from $\R'(\tau,\LW)$ whenever $\splitop(S,i)$ with $i\in
  [j\dd j+3\tau-2)$ is performed.
\end{proof}

\newcommand{\I}{\mathcal{I}}

For a uniquely labelled family $\LW$, an integer $\tau\in\Zp$, a
string $H$, and a necklace-consistent function $f$, we define
(cf.\ \cref{def:p-right-context,def:intervals}):
\begin{align*}
  \Pts^-_{f,H}(\tau,\LW)&=\bigcup_{S\in \LW} \{(d, \val(S)^{\infty}[j \dd j+7\tau), \labelop(S[j])) : (j,d)\in E^{-}_{f,H}(\tau,\val(S))\},\\
  \Pts^+_{f,H}(\tau,\LW)&=\bigcup_{S\in \LW} \{(d, \val(S)^{\infty}[j \dd j+7\tau), \labelop(S[j])) : (j,d)\in E^{+}_{f,H}(\tau,\val(S))\},\\
  \I^-_{f,H}(\tau,\LW) &= \bigcup_{S\in \LW} \{(a,b,\labelop(S[j])) : (a,b,j)\in \I^-_{f,H}(\tau,\val(S))\},\\
  \I^+_{f,H}(\tau,\LW) &= \bigcup_{S\in \LW} \{(a,b,\labelop(S[j])) : (a,b,j)\in \I^+_{f,H}(\tau,\val(S))\}.
\end{align*}
We interpret $\Pts^\pm_{f,H}(\tau,\LW)$ as a family of labelled points
with $\X = \Z$ and $\Y=\CFp(\LW)$.

\begin{lemma}\label{lem:range2}
  The data structure of \cref{lem:runs2} can be augmented so that, for
  each string $H$, given as a fragment of $\val(S)$ for some $S\in
  \LW$, range queries (\cref{sec:range-queries}) on
  $\Pts^\pm_{f_\sig,H}(\tau,\LW)$ and modular constraint queries
  (\cref{sec:mod-queries}) on $\I^\pm_{f_\sig,H}(\tau,\LW)$ can be
  supported in time $\Oh(\log^3 n + \log^2 n\cdot \frac{\log^2 \log
  m\cdot \log^* m}{\log \log \log m})$.  This comes at the price of
  increasing the cost of $\concatop$ and $\splitop$ to $\Oh(\log^2 n
  \cdot \frac{\log^2 \log m\cdot (\log^* m)^2}{\log \log \log m})$
  time.
\end{lemma}
\begin{proof}
  The non-empty sets $\Pts^\pm_{f_\sig,H}(\tau,\LW)$ (henceforth
  denoted $\Pts^\pm_H$) are maintained in the data structure of
  \cref{thm:range}, whereas the non-empty sets
  $\I^\pm_{f_\sig,H}(\tau,\LW)$ (henceforth denoted $\I^\pm_H$) are
  maintained in the data structure of \cref{cor:mod-queries}. Pointers
  to these components are stored in a deterministic dynamic
  dictionary~\cite{wexp} indexed by $\symb_\sig(H)$.

  The first step of the query algorithm is to compute $\symb_\sig(H)$
  by extracting the appropriate fragment using $\W.\splitop$.  This
  lets us identify the data structure responsible for the query in
  $\Oh(\log n\cdot \frac{\log^2 \log m \cdot \log^* m}{\log \log \log
  m})$ time.  By \cref{lem:comp}, each comparison on $\CFp(\LW)$
  costs $\Oh(\log n \cdot \frac{\log^2 \log m\cdot \log^* m}{\log \log
  \log m})$ time.  As a result, the cost of a range query is
  $\Oh(\log^3 n + \log^2 n \cdot \frac{\log^2 \log m\cdot \log^*
  m}{\log \log \log m})$.  On the other hand, the cost of a modular
  constraint query is $\Oh(\log^3 n)$.

  Before executing the $\splitop(S,i)$ operation of \cref{lem:runs},
  we identify all maximal intervals $[j'\dd j'']\in \R(\tau,\val(S))$
  (with $j'\in \R'(\tau,\val(S))$ and $j''\in \R''(\tau,\val(S))$)
  such that $i\in [j'\dd j''+7\tau)$ or $i+|S|\in [j'\dd j''+7\tau)$.
  There are $\Oh(1)$ such intervals, and they can be generated in
  $\Oh(1)$ time using the $\Successor$ and $\Predecessor$ operations
  on $\R'(\tau,\LW)$ and $\R''(\tau,\LW)$.  For each identified
  interval $[j'\dd j'']$, we remove the entries with label
  $\labelop(S[j'])$ from the sets $\Pts^\pm_H$ and $\I^\pm_H$
  containing it, and we add $\labelop(S[j'])$ to a temporary set $A$
  containing all labels $\ell\in \R'(\tau,\LW)$ which are missing
  their elements in $\Pts^\pm_H$ and $\I^\pm_H$.  Then, we execute the
  $\splitop(S,i)$ operation of \cref{lem:runs}, updating the set $A$
  whenever a label is removed from or added to $\R'(\tau,\LW)$.
  Finally, we iterate over the set $A$ to add the missing elements to
  $\Pts^\pm_H$ and $\I^\pm_H$.  For such a label $\ell\in A$, we
  retrieve $(T,j)=\unlabel(\ell)$ and ask the query of \cref{lem:runs}
  to determine $p:=|\Lrootgen{f_\sig}{\tau}{\val(T)}{j}|$,
  $h:=\Lheadgen{f_\sig}{\tau}{\val(T)}{j}$, and $e:=
  \rendgen{\tau}{\val(T)}{j}$.  This lets us extract $H=\val(T)[j+h\dd
  j+h+p)$ via a $\W.\splitop$ operation and, in particular identify
  $\symb_\sig(H)$. Moreover, we compute $\typegen{\tau}{\val(T)}{j}$
  by comparing $\val(T)^\infty[j\dd)$ with $\val(T)^\infty[j+p\dd )$
  via \cref{lem:comp} so that we know whether the elements should be
  inserted to $\Pts^-_H$ and $\I^-_H$ or to $\Pts^+_H$ and $\I^+_H$.
  Finally, we determine
  $e':=\rendfullgen{f_\sig}{\tau}{\val(T)}{j}=e-(e-j-h)\bmod p$ and
  insert a point $(e', \val(T)^\infty[e'\dd e'+7\tau),\labelop(j))$ to
  $\Pts^\pm_H$, and a tuple $(e'-e+3\tau-1, e'-j+1,\labelop(j))$ to
  $\I^\pm_H$.  The overall running time is increased to $\Oh(\log^2
  n\cdot \frac{\log^2 \log m\cdot \log^* m}{\log \log \log m})$,
  dominated by the cost of $\Oh(1)$ insertions and deletions on the
  sets $\Pts^\pm_H$.

  The implementation of the $\concatop(R,S)$ operation is symmetric.
  The only difference is that the auxiliary set $A$ is initialized
  based on maximal intervals $[j'\dd j'']\in \R(\tau,\val(R))$ with
  $|R|+1\in [j'\dd j''+7\tau)$ and on maximal intervals $[j'\dd
  j'']\in \R(\tau,\val(S))$ with $|S|+1\in [j'\dd j''+7\tau)$.
\end{proof}

\begin{proposition}\label{prp:periodic}
  For any fixed $\ell\in \Zp$, a dynamic text $T\in \Sigma^n$ can be
  implemented so that initialization takes $\Oh(\frac{\log^2 \log
  m\cdot \log^* m}{\log \log \log m})$ time, updates take
  $\Oh(\log^2 n \cdot \frac{\log^2 \log m\cdot (\log^* m)^2}{\log \log
  \log m})$ time, and the queries of \cref{as:periodic} take
  $\Oh(\log^3 n + \log^2 n\cdot \frac{\log^2 \log m\cdot \log^*
  m}{\log \log \log m})$ time, where $m=|\Sigma|+t$ and $t$ is the
  total number of instructions that the data structure has performed
  so far.
\end{proposition}
\begin{proof}
  We proceed as in the proof of \cref{cor:access}, but instead of
  implementing family $\LW$ using \cref{lem:lf}, we apply the
  extensions of \cref{lem:runs2,lem:range2} for
  $\tau=\floor{\frac{\ell}{3}}$. This increases the cost of $\initop$
  to $\Oh(\frac{\log^2 \log m\cdot \log^* m}{\log \log \log m})$ and
  the cost of updates to $\Oh(\log^2 n \cdot \frac{\log^2 \log m\cdot
  (\log^* m)^2}{\log \log \log m})$ (dominated by $\LW.\splitop$ and
  $\LW.\concatop$).

  We use $f:=f_\sig$ as the necklace-consistent function.  The queries
  asking to compute $|\Lrootgen{f}{\tau}{T}{j}|$,
  $\Lheadgen{f}{\tau}{T}{j}$, and $\rendgen{\tau}{T}{j}$ for $j\in
  \R(\tau,T)$ are answered in $\Oh(\log n)$ time directly using
  \cref{lem:runs2}.  As for the int-string range queries for the sets
  $\Points_{7\tau}(T, E^\pm_{f,H}(\tau,T))$, we use the equivalent
  range queries on $\Pts^\pm_{f,H}(\tau,\LW)$ answered using
  \cref{lem:range}.  The only work needed is to covert the query
  arguments from indices to fragments in $\CFp(T)$ (as specified in
  \cref{prob:int-str}).  Moreover, the label returned by the range
  selection on $\Pts^\pm_{\sig,H}(\tau,\LW)$ is converted to a
  position in $T$ using the $\LW.\unlabel(\cdot)$ operation.  The
  query time is $\Oh(\log^3 n + \log^2 n \cdot \frac{\log^2 \log
  m\cdot \log^* m}{\log \log \log m})$, dominated by the cost of
  range queries of \cref{lem:range2}.

  Finally, the modular constraint queries on the sets
  $\I^\pm_{f,H}(\tau,T)$ are implemented using the equivalent modular
  constraint queries on the set $\I^\pm_{f,H}(\tau,\LW)$. No
  transformation of arguments and query outputs are necessary here,
  and the overall query time is $\Oh(\log^3 n + \log^2 n\cdot
  \frac{\log^2 \log m\cdot \log^* m}{\log \log \log m})$.
\end{proof}

\subsection{Summary}

We are now ready to describe the main result of our work: a dynamic
text implementation that can answer suffix array queries. We start
with a version with a bounded lifespan: it takes an additional
parameter $N$ at initialization time, and it is only able to handle
$N$ operations.  Then, we use this solution as a black box to develop
an `everlasting' dynamic suffix array.

\begin{proposition}\label{prp:saisa}
  For any given integer $N\ge \sigma$, a dynamic text $T\in [0\dd
  \sigma)^+$ can be implemented so that initialization takes $\Oh(\log
  N \cdot \frac{\log^2 \log N\cdot \log^* N}{\log \log \log N})$ time,
  updates take $\Oh(\log^3 N \cdot \frac{\log^2 \log N\cdot (\log^*
  N)^2}{\log \log \log N})$ time, the suffix array queries take
  $\Oh(\log^4 N)$ time, and the inverse suffix array queries take
  $\Oh(\log^5 N)$ time, provided that the total number of updates and
  queries does not exceed $N$.
\end{proposition}
\begin{proof}
  We maintain $T$ using data structures of \cref{prp:small,lem:comp},
  as well as several instances of the data structures of
  \cref{prp:core,prp:nonperiodic,prp:periodic} for $\ell=2^q$, where
  $q\in [4\dd \ceil{\log N}]$.  The initialization and each update
  operation needs to be replicated in all these components.

  The suffix queries are implemented using \cref{pr:sa}. Due to the
  fact that $|T|\le N$, the components maintained are sufficient to
  satisfy the assumption required in \cref{pr:sa}.

  As for the inverse suffix array queries, we perform binary
  search. In each of the $\Oh(\log |T|)=\Oh(\log N)$ steps, we compare
  the specified suffix $T[j\dd |T|]$ with the suffix $T[\SA[i]\dd
  |T|]$; here, we use a suffix array query of to determine $\SA[i]$
  and the lexicographic comparison (of \cref{lem:comp}) to compare the
  two suffixes lexicographically.

  Recall that the running times of all the components are expressed in
  terms of parameters $n=|T|$ (which does not exceed $N$) and
  $m=\sigma+t$, where $t$ is the total number of instructions
  performed so far by the respective component. This value may differ
  across components, but we bound it from above by the total number of
  instructions performed so far by all the components; let us call
  this value $M$.  Note that each update and query costs
  $\Oh(\log^{\Oh(1)}(N+M))$ time, which means that $M = \Oh(\sigma + N
  \log^{\Oh(1)} N) = \Oh(N\log^{\Oh(1)} N)$, where the last step
  follows from the assumption $N\ge \sigma$.

  Consequently, the initialization takes $\Oh(\log N \cdot
  \frac{\log^2 \log M\cdot \log^* M}{\log \log \log M}) = \Oh(\log N
  \cdot \frac{\log^2 \log N\cdot \log^* N}{\log \log \log N})$ time
  and the updates take $\Oh(\log N \cdot \log^2 n \cdot \frac{\log^2
  \log M\cdot (\log^* M)^2}{\log \log \log M}) = \Oh(\log^3 N \cdot
  \frac{\log^2 \log N\cdot (\log^* N)^3}{\log \log \log N})$ time.  By
  \cref{pr:sa}, suffix queries take $\Oh(\log n \cdot (\log^3 n +
  \log^2 n \frac{\log^2 \log M\cdot \log^* M}{\log \log \log
  M}))=\Oh(\log^4 N)$ time.  On the other hand, the inverse suffix
  array queries cost $\Oh(\log N \cdot (\log^4 N + \log n\cdot
  \frac{\log^2 \log M\cdot \log^* M}{\log \log \log M}))=\Oh(\log^5
  N)$ time.
\end{proof}

\begin{theorem}\label{thm:saisa}
  A dynamic text $T\in [0\dd \sigma)^n$ can be implemented so that
  initialization takes $\Oh(\log \sigma\cdot \frac{\log^2 \log \sigma
  \cdot \log^* \sigma}{\log \log \log \sigma})$ time, updates take
  $\Oh(\log^3 (n\sigma) \cdot \frac{\log^2 \log (n\sigma)\cdot (\log^*
  (n\sigma))^2}{\log \log \log (n\sigma)})$ time, the suffix array
  queries take $\Oh(\log^4 (n\sigma))$ time, and the inverse suffix
  array take queries $\Oh(\log^5 (n\sigma))$ time.
\end{theorem}
\begin{proof}
  We first describe an amortized-time solution which performs a
  \emph{reorganization} every $\Omega(n)$ operations. This
  reorganization takes $\Oh(n\cdot U(n,\sigma))$ time, where
  $U(n,\sigma)=\Oh(\log^3 (n\sigma) \frac{\log^2 \log (n\sigma)\cdot
  (\log^* (n\sigma))^2}{\log \log \log (n\sigma)})$ is the update
  time.

  The text $T$ is stored using the data structures of both
  \cref{cor:access,prp:saisa}.  Moreover, we maintain a counter $t$
  representing the number of operations that can be performed before
  reorganization.  At initialization time, we set $t=1$ and initialize
  both components, setting $N=\sigma$ for \cref{prp:saisa}.  The
  updates and queries are forwarded to the component of
  \cref{prp:saisa}, but we first perform reorganization (if $t=0$) and
  decrement $t$ (unconditionally). As for the reorganization, we set
  $t=\ceil{\frac12 |T|}$, discard the component of \cref{prp:saisa},
  and initialize a fresh copy using $N=\max(\sigma,
  \ceil{\frac32|T|}-1)$; we then insert characters of $T$ one by one
  using the $\accessop$ operation of \cref{cor:access} and the
  $\makeop$ operation of \cref{prp:saisa}.

  To prove that this implementation is correct, we must argue that
  each instance of \cref{prp:saisa} performs no more than $N$
  operations. The instance created at initialization time is limited
  to a single operation, which is no more than the allowance of
  $N=\sigma$ operations. On the other hand, an instance created during
  a reorganization performs $|T|-1$ insertions during the
  reorganization, and is then limited to $\ceil{\frac12 |T|}$
  operations. In total, this does not exceed the allowance of
  $N=\max(\sigma, \ceil{\frac32|T|}-1)$ operations.

  It remains to analyze the time complexity. For this, we observe
  that, if $N>\sigma$, then $|T| \ge |T|-t \ge \floor{\frac13N}$ is
  preserved as an invariant. This means that
  $N=\Oh(\max(\sigma,|T|))=\Oh(\sigma|T|)$, and thus the operation
  times of \cref{prp:saisa} can be expressed using $n\sigma$ instead
  of $N$.  This also applies to the cost of reorganization, which uses
  initialization and $n-1$ updates.

  As for the deamortization, we use the standard technique of
  maintaining two instances of the above data structure. At any time,
  one them is active (handles updates and queries), whereas the other
  undergoes reorganization. The lifetime of the entire solution is
  organized into \emph{epochs}.  At the beginning of each epoch, the
  active instance is ready to handle $t\ge \frac12n$ forthcoming
  operations, whereas the other instance needs to be reorganized.  The
  epoch lasts for $t$ operations. During the first half of the epoch,
  the reorganization is performed in the background and the updates
  are buffered in a queue. For each operation in the second half of
  the epoch, at most one update in buffered (none if the operation is
  a query) and two buffered updates are executed (unless there are
  already fewer updates in the buffer). Since the reorganization cost
  is bounded by $\Oh(t \cdot U(n,t))$ and since the query cost is
  larger than the update cost, the deamortized solution has the same
  asymptotic time complexity as the amortized one.
\end{proof}

\section{Conditional Lower Bound for Copy-Pastes}

A natural extension of the dynamic text interface provided in
\cref{sec:app} would be to support not only cut-pastes (that move
fragments of $T$), but also a copy-pastes (that copy fragments of
$T$).  Such operation is readily supported by the underlying
implementation of dynamic strings (\cref{sec:ds}), but it is
incompatible with the concept of labelling characters -- a single
copy-paste may add multiple new characters, and we cannot afford to
assign them labels one by one.  In this section, we provide a
conditional lower bound showing that this is not due to a limitation
of our techniques, but rather due to inherent difficulty of the
(inverse) suffix array queries in dynamic texts.  Our lower bound is
conditioned on the Online Matrix-Vector Multiplication
Conjecture~\cite{DBLP:conf/stoc/HenzingerKNS15}, which is often used
in the context of dynamic algorithms.  The underlying reduction
resembles one from the Dynamic Internal Dictionary Matching
problem~\cite{CKMRRW21}, where the queries ask if a given fragment of
a static text contains an occurrence of any pattern from a dynamic
dictionary.

In the Online Boolean Matrix-Vector Multiplication (OMv) problem, we
are given as input an $n \times n$ boolean matrix $M$. Then, a
sequence $n$ vectors $v_1, \ldots, v_n$, each of size $n$, arrives in
an online fashion. For each such vector $v_i$, we are required to
output $Mv_i$ before receiving $v_{i+1}$.

\begin{conjecture}[OMv Conjecture~\cite{DBLP:conf/stoc/HenzingerKNS15}]\label{conj:OMv}
  For any constant $\epsilon>0$, there is no
  $\Oh(n^{3-\epsilon})$-time algorithm that solves OMv correctly with
  probability at least $\frac23$.
\end{conjecture}

We use the following simplified version of~\cite[Theorem
2.2]{DBLP:conf/stoc/HenzingerKNS15}.

\begin{theorem}[\cite{DBLP:conf/stoc/HenzingerKNS15}]\label{thm:gOMv}
  For all constants $\gamma,\epsilon>0$, the OMv Conjecture implies
  that there is no algorithm that, given as input a $p \times q$
  matrix $M$, with $p=\lfloor q^{\gamma} \rfloor$, preprocesses $M$ in
  time polynomial in $p\cdot q$, and then, presented with a vector $v$
  of size $q$, computes $Mv$ in time $\Oh(q^{1+\gamma-\epsilon})$
  correctly with probability at least $\frac23$.
\end{theorem}

\begin{theorem}\label{thm:lowerbound}
  For all constants $\alpha,\beta>0$ with $\alpha+\beta < 1$, the OMv
  Conjecture implies that there is no dynamic algorithm that
  preprocesses a text $T$ in time polynomial in~$|T|$, supports
  copy-pastes in $\Oh(|T|^{\alpha})$ time, and inverse suffix array
  queries in $\Oh(|T|^{\beta})$ time, with each answer correct with
  probability at least $\frac23$.
\end{theorem}
\begin{proof}
  Let us suppose that there is such an algorithm and set $\gamma =
  \frac{1+\alpha-\beta}{1+\beta-\alpha}$.  Given a $p \times q$ matrix
  $M$ satisfying $p=\lfloor q^{\gamma} \rfloor$, we construct a text
  $T$ of length $pq+2p+3$ over an alphabet $\{\$,a_0,\ldots,a_p,\#\}$
  (with $\$ \prec a_0 \prec \cdots \prec a_p \prec \$$) using the
  following formula:
  \[
    T = \left(\bigodot_{j=1}^q \left(\bigodot_{i=1}^p a_{i\cdot
      M[i,j]}\right)\right) \cdot \left(\bigodot_{i=0}^p a_i\#\right)
    \cdot \$.
  \]
  In other words, we first write down the columns of $M$ in the
  increasing order, replacing each zero with $a_0$ and each one with
  $a_i$, where $i\in [1\dd p]$ is the row index of the underlying
  matrix entry.  Then, we append $a_i\#$ for all $i\in [0\dd p]$, and
  finally we place $\$$ at the very end of $T$.  At the preprocessing
  time, we also precompute, for each $i\in [1\dd p]$, the number $c_i$
  of occurrences of $a_i$ in $T$.

  Given a query vector $v\in \{0,1\}^q$, we proceed as follows.  For
  each $j\in [1\dd q]$ with $v_j = 1$, we copy $\T((j-1)p\dd jp]$ to
  the end of $T$ (setting $T:= T[1\dd |T|) \cdot \T((j-1)q\dd jq]
  \cdot \T[|T|]$). Next, we construct the answer vector $w\in
  \{0,1\}^p$, setting $w_i := 1$ if and only if
  $\ISA[pq+2i+1]-\ISA[pq+2i-1]>c_i$ for $i\in [1\dd p]$.

  Let us prove that $w=Mv$ holds assuming that all the $\ISA$ queries
  are answered correctly.  For this, note that, for $i\in [0\dd p]$,
  the value $\ISA[pq+2i+1]$ represents total number of occurrences of
  symbols $\$,a_0,\ldots,a_i$ in $T$; that is because $T[pq+2i+1\dd
  pq+2i+2]$ is the unique occurrence of $a_i\#$ in $T$, and $\#$ is
  the largest symbol in $\Sigma$.  In particular,
  $\ISA[pq+2i+1]-\ISA[pq+2i-1]>c_i$ holds if and only if one of the
  copy-pastes involved a substring $\T((j-1)p\dd jp]$ containing
  $a_i$.  This character could have only occurred at position
  $\T[(j-1)p+i]$, indicating that $M[i,j]=1$.  Moreover, the whole
  copy-paste is executed if and only if $v_j=1$.  Consequently,
  $\ISA[pq+2i+1]-\ISA[pq+2i-1]>c_i$ holds if and only if there exists
  $j\in [1\dd q]$ with $M[i,j]=v_j=1$; this is precisely when
  $(Mv)_i=1$.

  If the answers to $\ISA$ queries are correct with probability at
  least $\frac23$, we can guarantee that the whole vector $w$ is
  correct with probability at least $1-n^{-\Omega(1)}\ge \frac23$ by
  maintaining $\Theta(\log n)$ independent instances of the algorithm
  and taking the dominant answer to each $\ISA$ query.  In total, we
  perform $\Ohtilde(q)$ copy-pastes and $\Ohtilde(p)$ $\ISA$ queries.
  Furthermore, the copy-pastes have disjoint sources, so the length of
  $T$ increases at most twofold.  Hence, the total time required
  is
  \[
    \Ohtilde\left(q |T|^\alpha+p
    |T|^\beta\right)=\Ohtilde\left(p^\alpha q^{1+\alpha} + p^{1+\beta}
    q^\beta\right)=\Ohtilde\left(q^{1+(1+\gamma)\alpha} + q^{\beta +
    \gamma(1+\beta)}\right) =
    \Ohtilde\left(q^{\frac{1+\alpha+\beta}{1+\beta-\alpha}}\right).
  \]
  In the light of \cref{thm:gOMv}, this would disprove \cref{conj:OMv}
  due to $\frac{1+\alpha+\beta}{1+\beta-\alpha} <
  \frac{2}{1+\beta-\alpha} = 1+\gamma$.
\end{proof}

\begin{remark}
  The reduction behind \cref{thm:lowerbound} not only proves hardness
  of inverse suffix array queries, but also of the counting version of
  the dynamic indexing problem: It shows that the counting queries are
  hard already for patterns of length one.  Moreover, since each
  $\ISA$ query can be reduced to a logarithmic number of $\SA$
  queries, we get the same lower bound for $\SA$ queries. (The
  underlying reduction is described in the proof of \cref{prp:saisa};
  the dynamic strings of \cref{sec:ds} support lexicographic
  comparisons and copy-pastes in $\Ohtilde(1)$ time.)
\end{remark}

\appendix

\section{Omitted Proofs}\label{app:extra-proofs}

\lmexp*
\begin{proof}
  Let $j' \in [1 \dd k]$ be such that $\LCE_{S}(j, j') \geq 3\tau -
  1$.  Then, by definition, $\Lrootgen{f}{\tau}{S}{j'} =
  \Lrootstrgen{f}{S[j' \dd j' + 3\tau - 1)} = \Lrootstrgen{f}{S[j \dd
  j + 3\tau - 1)} = \Lrootgen{f}{\tau}{S}{j}$. To show
  $\Lheadgen{f}{\tau}{S}{j'} = s$, note that by $|H| \leq \tau$, the
  string $H'H^2$ (where $H'$ is a length-$s$ suffix of $H$) is a
  prefix of $S[j \dd j + 3\tau - 1) = S[j' \dd j' + 3\tau - 1)$.  On
  the other hand, $\Lheadgen{f}{\tau}{S}{j'} = s'$ implies that
  $\widehat{H}'H^2$ (where $\widehat{H}'$ is a length-$s'$ suffix of
  $H$) is a prefix of $S[j' \dd j' + 3\tau - 1)$.  Thus, by the
  synchronization property of primitive
  strings~\cite[Lemma~1.11]{AlgorithmsOnStrings} applied to the two
  copies of $H$, we have $s' = s$, and consequently, $j' \in
  \R_{f,s,H}$.

  For the converse implication, assume $j, j' \in \R_{f,s,H}$. This
  implies that both $S[j \dd \rendgen{\tau}{S}{j})$ and $S[j' \dd
  \rendgen{\tau}{S}{j'})$ are prefixes of $H'H^{\infty}$ (where $H'$
  is as above). Thus, by $\rendgen{\tau}{S}{j} - j,\,
  \rendgen{\tau}{S}{j'} - j' \geq 3\tau - 1$, we obtain $\LCE_{S}(j,
  j') \geq 3\tau - 1$.

  1. Let $Q = H'H^{\infty}$, where $H'$ is a length-$s$ suffix of
  $H$. We will prove $S[j \dd k] \prec Q \prec S[j' \dd k]$, which
  implies the claim. First, we note that $\typegen{\tau}{S}{j} = -1$
  implies that either $\rendgen{\tau}{S}{j} = k + 1$, or
  $\rendgen{\tau}{S}{j} \leq k$ and $S[\rendgen{\tau}{S}{j}] \prec
  S[\rendgen{\tau}{S}{j} - |H|]$.  In the first case, $S[j \dd
  \rendgen{\tau}{S}{j}) = S[j \dd k]$ is a proper prefix of $Q$ and
  hence $S[j \dd k] \prec Q$. In the second case, letting $\ell =
  \rendgen{\tau}{S}{j} - j$, we have $S[j \dd j + \ell) = Q[1 \dd
  \ell]$ and $S[j + \ell] \prec S[j + \ell - |H|] = Q[1 + \ell - |H|]
  = Q[1 + \ell]$. Consequently, $S[j \dd k] \prec Q$. To show $Q \prec
  S[j' \dd k]$ we observe that $\typegen{\tau}{S}{j'} = +1$ implies
  $\rendgen{\tau}{S}{j'} \leq k$. Thus, letting $\ell' =
  \rendgen{\tau}{S}{j'} - j'$, we have $Q[1 \dd \ell'] = S[j' \dd j' +
  \ell')$ and $Q[1 + \ell'] = Q[1 + \ell' - |H|] = S[j' + \ell' - |H|]
  \prec S[j' + \ell']$. Hence, we obtain $Q \prec S[j' \dd k]$.

  2. Similarly as above, we consider two cases for
  $\rendgen{\tau}{S}{j}$. If $\rendgen{\tau}{S}{j} = k + 1$, then by
  $\rendgen{\tau}{S}{j} - j < \rendgen{\tau}{S}{j'} - j'$, $S[j \dd
  \rendgen{\tau}{S}{j}) = S[j \dd k]$ is a proper prefix of $S[j' \dd
  \rendgen{\tau}{S}{j'})$ and hence $S[j \dd k] \prec S[j' \dd
  \rendgen{\tau}{S}{j'}) \preceq S[j' \dd k]$. If
  $\rendgen{\tau}{S}{j} \leq k$, then letting $\ell =
  \rendgen{\tau}{S}{j} - j$, we have $S[j \dd j + \ell) = S[j' \dd j'
  + \ell)$ and by $\rendgen{\tau}{S}{j} - j < \rendgen{\tau}{S}{j'} -
  j'$, $S[j + \ell] \prec S[j + \ell - |H|] = S[j' + \ell - |H|] =
  S[j' + \ell]$.  Consequently, $S[j \dd k] \prec S[j' \dd k]$.

  3. By $\typegen{\tau}{S}{j'} = +1$ we have $\rendgen{\tau}{S}{j'}
  \leq k$. Thus, letting $\ell' = \rendgen{\tau}{S}{j'} - j'$, by
  $\rendgen{\tau}{S}{j} - j > \rendgen{\tau}{S}{j'} - j'$, we have
  $S[j \dd j + \ell') = S[j' \dd j' + \ell')$ and $S[j + \ell'] = S[j
  + \ell' - |H|] = S[j' + \ell' - |H|] \prec S[j' + \ell']$.
  Consequently, $S[j \dd k] \prec S[j' \dd k]$.
\end{proof}

\lmRblock*
\begin{proof}

  Denote $p = \per(S[j{-}1 \dd j{-}1{+}3\tau{-}1))$ and $p' = \per(S[j
  \dd j {+} 3\tau {-} 1))$. By $j - 1, j \in \R(\tau,S)$ we have $p,
  p' \leq \tfrac{\tau}{3}$.  Consider $Q = S[j \dd j + \tau)$. The
  string $Q$ has both periods $p$ and $p'$. If $p \neq p'$, then by
  the weak periodicity lemma, $Q$ has a period $p'' = \gcd(p, p') <
  p'$. Since $p'' \mid p'$ we obtain that $S[j \dd j + p')$ is not
  primitive, which contradicts $\per(S[j \dd j + 3\tau - 1)) =
  p'$. Thus, $p = p'$.  Then, by $p \leq \tfrac{\tau}{3}$, we have
  $S[j {-} 1 \dd j {-} 1 {+} p) = S[j {-} 1 {+} p \dd j {-} 1 {+}
  2p)$.  Consequently, $\{S[j {-} 1 {+} t \dd j {-} 1 {+} t {+} p) : t
  \in [0 \dd p)\} = \{S[j {+} t \dd j {+} t {+} p) : t \in [0 \dd
  p)\}$. This implies that $S[j {-} 1 \dd j {-} 1 {+} p)$ and $S[j \dd
  j {+} p)$ are cyclically equivalent. Thus, it holds
  $\Lrootgen{f}{\tau}{S}{j - 1} = f(S[j - 1 \dd j - 1 + p)) = f(S[j
  \dd j + p)) = \Lrootgen{f}{\tau}{S}{j}$.

  By \cite[Fact~3.2]{sss}, $S[j - 1 \dd \rendgen{\tau}{S}{j - 1})$
  (resp.\ $S[j \dd \rendgen{\tau}{S}{j})$) is the longest substring
  starting at position $j - 1$ (resp.\ $j$) with period $p$
  (resp.\ $p'$).  Equivalently, $\rendgen{\tau}{S}{j - 1} = j {-} 1
  {+} p + \LCE_{S}(j {-} 1, j {-} 1 {+} p)$ and $\rendgen{\tau}{S}{j}
  = j {+} p' + \LCE_{S}(j, j {+} p')$.  Thus, by $p = p'$ and $S[j {-}
  1] = S[j {-} 1 {+} p]$, we have $\rendgen{\tau}{S}{j - 1} = j {-}
  1 {+} p + \LCE_{S}(j {-} 1, j {-} 1 {+} p) = j {+} p + \LCE_{S}(j, j
  {+} p) = j {+} p' + \LCE_{S}(j, j {+} p') = \rendgen{\tau}{S}{j}$.

  The third claim follows from the definition of type and equalities
  $\rendgen{\tau}{S}{j - 1} = \rendgen{\tau}{S}{j}$ and $p = p'$.
\end{proof}

\bibliographystyle{alphaurl}
\bibliography{paper}

\end{document}